\documentclass[letterpaper,11pt]{article}
\usepackage{amsfonts}
\usepackage{amsmath, amsthm, amssymb, bbm, astron,verbatim}
\usepackage{xcolor}
\usepackage{graphicx}

\usepackage{xr}   

\numberwithin{equation}{section} 

\newtheorem{assumption}{Assumption}[section]
\newtheorem{theorem}{Theorem}[section]
\newtheorem{example}{Example}

\newtheorem{lemma}[theorem]{Lemma}
\newtheorem{corollary}[theorem]{Corollary}
\newtheorem{remark}{Remark}

\newcommand{\plim}{\operatorname*{plim}}
\newcommand{\diag}{{\rm diag}}
\newcommand{\argmin}{\operatorname*{argmin}}
\newcommand{\argmax}{\operatorname*{argmax}}
\newcommand{\E}{\mathbb{E}_{\phi}}
\newcommand{\EE}{\overline{\mathbb{E}}}
\newcommand{\Ep}{\mathbb{E}}
\newcommand{\T}{\mathcal{T}}
\newcommand{\N}{\mathcal{N}}

\def\ft#1#2{{\textstyle {\frac{#1}{#2}} }}

\begin{document}

\title{Individual and Time Effects
       in  Nonlinear Panel Models with Large $N$, $T$\thanks{
         \scriptsize
      We would like to thank  the editor Jianqing Fan, an associate editor, two anonymous referees, Richard Blundell, Ben Boury, Mingli Chen, Geert Dhaene, Andreas Dzemski, Antonio Galvao, 
       Jinyong Hahn, Bo Honor\'{e}, Koen Jochmans, Hiroaki Kaido,
        Dennis Kristensen,  Guido Kuersteiner, Kevin Lang, Roger Moon, Fabien Postel-Vinay, Zhongjun Qu, Bernard Salani{\'e}, Sami Stouli,
       Francis Vella,   Fa Wang, and Jeff Wooldridge for very useful comments that helped improve the paper.
       We also benefitted from discussions with
       seminar participants at Berkeley, Boston University, Brown, Cambridge, Columbia, Georgetown, Iowa, John Hopkins, Maryland, Michigan State, Princeton,  Stanford, UCL, and USC,
               and by conference participants
       at the 16th and 19th International Conference on Panel Data, as well as the 
       2013 North American Summer Meetings of the Econometric Society.
       Financial support from the National Science Foundation and the Economic and Social Research Council through the ESRC Centre for Microdata Methods and Practice grant RES-589-28-0001 is gratefully
       acknowledged.
       }}

\author{\setcounter{footnote}{2} Iv\'{a}n Fern\'{a}ndez-Val\footnote{
     Department of Economics, Boston University,
     270 Bay State Road,
     Boston, MA 02215-1403, USA.
     Email:~{\tt ivanf@bu.edu}
      }
   \and Martin Weidner\footnote{
                   Department of Economics,
                   University College London,
                   Gower Street,
                   London WC1E~6BT,
                   UK,
                   and and CeMMaP.
                   Email:~{\tt m.weidner@ucl.ac.uk}
                   }
                   }

\date{\today}
\maketitle

\abstract{
\medskip
\begin{center}
\begin{minipage}{0.8\textwidth}
We derive fixed effects estimators of parameters and average partial effects in (possibly dynamic) nonlinear panel data models with individual and time effects. They cover logit, probit, ordered probit, Poisson and Tobit models that are important for many empirical applications in micro and macroeconomics. Our estimators use analytical and jackknife bias corrections to  deal with the incidental parameter problem, and are asymptotically unbiased under asymptotic sequences where $N/T$ converges to a constant. We develop inference methods and show that they perform well in numerical examples.
\end{minipage}
\end{center}
\medskip


%
}

\bigskip
\bigskip

{\footnotesize

}

\newpage



\newpage

\section{Introduction}

Fixed effects estimators of nonlinear panel data models can be severely biased
because of the  incidental parameter problem \cite{Neyman:1948p881}.
A growing literature, surveyed in Arellano and Hahn \cite*{ArellanoHahn2007}, 
shows that the leading term of an asymptotic expansion of the bias as both the cross-sectional dimension $N$ and time series dimension $T$ of the panel
grow,  can be characterized and corrected for. In models with individual effects,  the leading bias term is of order $1/T$ and comes from the estimation of the individual effects. 
This result, however, does not apply to models with  individual and  time effects, where both of these effects are treated as  parameters to be estimated. In this paper we show that the estimation of the time effects causes an additional incidental parameter bias of order $1/N$. Thus, if $N$ and $T$ are similarly large, the bias produced by the estimation of the time effects is of similar order of magnitude to the bias produced by the estimation of the individual effects, and both biases need to be corrected. We provide the corresponding  analytical and jackknife bias corrections.

The asymptotic approximation to the fixed effects estimators that lets the two dimensions of the panel grow with the sample size is motivated by the recent availability of long panels and other large pseudo-panel data structures where the indexes might not correspond to individuals and time periods.  Examples of  these datasets  include traditional microeconomic panel surveys with a long history of data such as the PSID and  NLSY,   international cross-country panels  such as the Penn World Table, U.S. state level  panels over time such as the CPS, and square pseudo-panels of trade flows across countries such as the Feenstra's World Trade Flows and CEPII, where the indexes correspond to the same countries indexed as importers and exporters.

We focus on semi-parametric models with log-likelihood functions that are concave in all parameters, and where each individual effect $\alpha_i$ and  time effect $\gamma_t$ enter the log-likelihood for observation $(i,t)$ additively as $\alpha_i + \gamma_t$.
This is the most common specification for the individual and time effects  in linear models and is also a natural specification in the nonlinear models that we consider.   Imposing concavity of the log-likelihood function
greatly  facilitates showing consistency in our setting where the dimension of the parameter space 
grows with the sample size.  The
most popular limited dependent variable models, including logit,
probit, ordered probit, Tobit and Poisson models have concave log-likelihood functions, possibly after reparametrization (Olsen \cite*{Olsen:1978p3375}, and Pratt \cite*{Pratt:1981p654}).  We note here that the general expansion that we derive in Appendix B  do not impose additivity and concavity, but we use these restrictions to apply the expansion to fixed effects estimators. 
The models that we consider are semi-parametric because the joint distribution of the explanatory variables and the unobserved effects is left unspecified. The explanatory variables can be either strictly exogenous or predetermined. 

We derive bias expansions and corrections for fixed effects estimators of common parameters $\beta$ and average partial effects (APEs). The vector $\beta$ includes all the unknown parameters that enter the log-likelihood function other than the individual and time effects, such as index coefficients in a probit model.  The APEs are functions of the data, the common parameters, and the individual and time effects in nonlinear models. We find that the properties of the fixed effects estimators of $\beta$ and the APEs are different. For $\beta$, the order of the  bias is  $1/T + 1/N$, which is of the same as the rate of convergence $1/\sqrt{NT}$ under sequences where $N/T$ converge to a constant.  For the APEs,  we uncover that the incidental parameter problem is negligible asymptotically 
because the order of the bias, $1/N + 1/T$, is smaller than the rate of convergence, which is $1/\sqrt{N} + 1/\sqrt{T}$, slower  than for model parameters. To the best of our knowledge, this rate result is new for fixed effects estimators of average partial effects in nonlinear panel models with individual and time effects.\footnote{Galvao and Kato~\cite*{GalvaoKato:2013} also found slow rates of convergence for fixed effects estimators in linear models with individual effects under misspecification. Fernandez-Val and Lee~\cite*{FL13}  pointed out this issue in nonlinear models with only individual effects.} In numerical examples we find that the bias corrections, while not necessary to center the asymptotic distribution  of  APE estimators, do improve their finite-sample properties, specially in dynamic models.

The bias correction  eliminates the bias terms of orders $1/T$ and $1/N$  from the fixed effects estimators. We considerer two methods to implement the correction: an analytical bias correction similar to Hahn
and Newey \cite*{Hahn:2004p882} and Hahn and Kuersteiner \cite*{HahnKuersteiner2011}, and a suitable modification of the split panel jackknife of  Dhaene and Jochmans \cite*{DhaeneJochmans2015}.\footnote{A similar split panel jackknife bias correction method was outlined in Hu~\cite*{Hu2002}.}
However, the theory of the previous papers does not cover the models that we consider, because, in addition to not allowing for
time effects, it
assumes either identical distribution or stationarity over time for the processes of the observed variables,
conditional on the unobserved effects. These assumptions are  violated
in our models due to the presence of the time effects,
 so we need to adjust the asymptotic theory accordingly. 
 The individual and time effects introduce strong correlation in both dimensions of the panel.
Conditional on the unobserved effects, we impose cross-sectional independence and weak time-serial dependence, and we allow for heterogeneity in both dimensions.

Simulation evidence indicates that our corrections improve the estimation and inference performance of the fixed effects estimators of parameters and average effects. The analytical corrections dominate the jackknife corrections in a probit model for sample sizes that are relevant for empirical practice.  
In the online supplement, Fern{\'a}ndez-Val and Weidner \cite*{Supp2015}, we illustrate the corrections with an empirical application on the relationship between competition and innovation using a panel of U.K. industries, following Aghion, Bloom, Blundell, Griffith and Howitt \cite*{AghionBloomBlundellGriffithHowitt2005}. We find that the inverted-U pattern relationship found by Aghion \textit{et al} is robust to relaxing the strict exogeneity assumption of competition with respect to the innovation process and to the inclusion of innovation dynamics. We also uncover substantial state dependence in the innovation process.

\paragraph{Literature review.} The Neyman and Scott incidental parameter problem has been extensively discussed in the econometric literature; see, for example, Heckman \cite*{Heckman:1981p2940},
Lancaster \cite*{Lancaster:2000p879}, and Greene \cite*{Greene:2004p3125}. There is also a vast literature that shows how to tackle the problem in specific models under asymptotic sequences where $T$ is fixed and $N$ grows to infinity. However, there are results,
e.g.~from Honor{\'e} and Tamer~\cite*{HonoreTamer2006}, 
Chamberlain~\cite*{Chamberlain2010},
and Chernozhukov, Fern{\'a}ndez-Val, Hahn and Newey \cite*{CFHN13}, showing 
that model parameters and APEs are not point identified in important nonlinear panel data models under fixed-$T$ asymptotic sequences, implying  that no fixed-$T$ consistent point estimators exist in  these models.

A recent response to the incidental parameter problem is to adopt an alternative asymptotic approximation where  both $N$ and $T$ grow with the sample size. Under these large-$T$ sequences,  the fixed effects estimator is consistent  but has bias  in the asymptotic distribution. This asymptotic bias is the large-$T$ version of the incidental parameter problem and has motivated the development of bias corrections. Examples of papers that use this approximation include  Phillips and Moon \cite*{Phillips:1999p733}, Hahn and Kuersteiner \cite*{Hahn:2002p717},
Lancaster \cite*{Lancaster:2002p875}, Woutersen~\cite*{Woutersen:2002p3683}, 
Alvarez and Arellano~\cite*{AlvarezArellano2003},
Hahn and Newey~\cite*{Hahn:2004p882},
Carro \cite*{Carro:2007p3601}, 
Arellano and Bonhomme~\cite*{ArellanoBonhomme2009},
Fernandez-Val~\cite*{FernandezVal:2009p3313},
Hahn and Kuersteiner~\cite*{HahnKuersteiner2011}, 
Fernandez-Val and Vella~\cite*{FernandezValVella2011},
and Kato, Galvao and Montes-Rojas~\cite*{KatoGalvaoMontes-Rojas2012}.
This previous work, however, does not cover models with time effects.\footnote{An 
exception is
Woutersen~\cite*{Woutersen:2002p3683}, which considers a special type of grouped  time effects whose number is  fixed with $T$. We instead consider an unrestricted set of $T$ time effects, one for each time period.} Our contribution to this literature is to extend the large-$T$ bias corrections to models with two-way unobserved effects such as the individual and time effects commonly included in linear models.



The large-$T$ panel literature on models with both individual and time effects is
sparse.  Pesaran~\cite*{Pesaran2006}, Bai~\cite*{Bai:2009p3321}, and Moon and Weidner~\cite*{MoonWeidner2015a,MoonWeidner2015b} study linear regression models  with interactive individual and time fixed effects. The fixed effects estimators in these models also have asymptotic bias of order $1/T + 1/N$, but the methods used to derive this bias rely on linearity and therefore cannot be applied to the nonlinear models that we consider. Hahn and Moon \cite*{HahnMoon2006} consider bias corrected fixed effects estimators in panel linear autoregressive models with additive individual and
time effects. Regarding non-linear models, there is 
 independent and contemporaneous work by Charbonneau~\cite*{Charbonneau2011,Charbonneau2014}, which extends the conditional fixed effects estimators to logit and Poisson models
with  individual and time effects. She differences out the individual and time effects
by conditioning on sufficient statistics. The conditional approach 
completely eliminates the asymptotic bias coming from the estimation of the incidental parameters, 
but it does not permit estimation of average partial effects and has not been developed for models with predetermined regressors. We instead consider
estimators of model parameters and average partial effects in nonlinear models with predetermined
regressors. The two approaches can therefore be considered as complementary. 



\paragraph{Outline of the paper.} The rest of the paper is organized as follows. Section \ref{sec:model} introduces the model and fixed effects estimators. Section  \ref{sec:bc}  describes the bias corrections to deal with the incidental parameters problem and illustrates how the bias corrections work through an example. Section  \ref{sec:app_panel} provides
 the asymptotic theory.  
 Section \ref{sec: MC} presents Monte Carlo results. The Appendix collects the proofs of  the main results, and  an online supplement to the paper   contains additional technical derivations, numerical examples, and an empirical application \cite{Supp2015}.



\section{Model and Estimators}
\label{sec:model}

\subsection{Model} The data consist of $N \times T$ observations $\{(Y_{it}, X'_{it})': 1 \leq i \leq N, 1 \leq t \leq T \},$ for a scalar outcome variable of interest $Y_{it}$ and a vector
of explanatory variables $X_{it}$.  We assume that the outcome for
individual $i$ at time $t$ is generated by the sequential process:
\begin{equation*}
Y_{it}  \mid X^t_i, \alpha, \gamma, \beta   \sim f_{Y}(\cdot \mid X_{it},\alpha_i, \gamma_t, \beta), \ \ (i = 1,...,N; t = 1, ..., T),
\end{equation*}
where  $X^t_i = (X_{i1}, \ldots, X_{it}),$  $\alpha = (\alpha_1, \ldots, \alpha_N)$, $\gamma = (\gamma_1, \ldots, \gamma_T)$,
$f_{Y}$ is a known probability function, and $\beta$ is a finite dimensional parameter vector.
The variables $\alpha_i$ and $\gamma_t$ are unobserved individual and time effects that in economic applications capture individual heterogeneity and aggregate shocks, respectively. The model is semiparametric because  we do not specify the distribution of these effects nor their relationship with the explanatory variables. The conditional distribution $f_{Y}$ represents the parametric part of the model. 
The vector $X_{it}$ contains predetermined variables
 with respect to $Y_{it}$. Note that $X_{it}$ can include lags of $Y_{it}$ to accommodate dynamic models.

We consider two running examples  throughout the analysis:
\begin{example}[Binary response model]\label{example: probit} Let $Y_{it}$ be a binary outcome and $F$ be a cumulative
distribution function, e.g. the standard normal or standard logistic distribution. We can model the conditional distribution of $Y_{it}$ using the single-index specification with individual and time effects
\begin{equation*}
f_{Y}(y \mid X_{it}, \alpha_i, \gamma_t, \beta) = F(X_{it}'\beta+ \alpha_i + \gamma_t)^y[1 - F(X_{it}'\beta+ \alpha_i + \gamma_t)]^{1-y}, \ \ y \in \{0,1\}.
\end{equation*}
In a labor economics application, $Y$ can be an indicator for female labor force participation and $X$ can include fertility indicators and other socio-economic characteristics.
\end{example}


\begin{example}[Poisson model]\label{example: poisson} Let $Y_{it}$ be a non-negative interger-valued outcome, and $f(\cdot; \lambda)$ be the
probability mass function of a Poisson random variable with mean $\lambda > 0$. We can model the conditional distribution of $Y_{it}$ using the single index specification  with individual and time effects
\begin{equation*}
f_{Y }(y \mid X_{it}, \alpha_i, \gamma_t, \beta) = f(y;  \exp[X_{it}'\beta + \alpha_i +  \gamma_t ]), \ \ y \in \{0, 1,  2, .... \}.
\end{equation*}
In an industrial organization application, $Y$ can be the number of patents that a firm produces and $X$ can include investment in R\&D and other firm characteristics.
\end{example}

For estimation, we adopt a fixed effects approach, treating the realization of the unobserved individual and time effects as parameters to be estimated. We collect all these effects in the vector $\phi_{NT} = (\alpha_1, ..., \alpha_N, \gamma_1, ..., \gamma_T)'$.  The model parameter  $\beta$ usually includes regression coefficients of interest, while the vector  $\phi_{NT}$ is  treated as a nuisance parameter. 
The true values of the parameters, denoted by $\beta^0$
and $\phi_{NT}^0 = ({\alpha^{0}_1}, ..., {\alpha^{0}_N}, {\gamma^{0}_1}, ..., {\gamma^{0}_T})'$, are the solution to the population conditional maximum likelihood problem
\begin{align}\label{population_prob}
& \max_{(\beta, \phi_{NT})   \in \mathbb{R}^{\dim \beta + \dim \phi_{NT}}} \; \E[ \mathcal{L}_{NT}(\beta,\phi_{NT})], 
\nonumber \\ 
&  \qquad \mathcal{L}_{NT}(\beta,\phi_{NT}) := (NT)^{-1/2} \left\{ \sum_{i,t}  \log f_{Y}(Y_{it}  \mid X_{it}, \alpha_i, \gamma_t, \beta) -  b  (v_{NT}'\phi_{NT})^2/2 \right\},
\end{align}
for every $N,T$, where $\E$ denotes the expectation with respect to
the distribution of the data conditional on the unobserved effects and initial conditions including strictly exogenous variables, $b > 0$ is an arbitrary constant, 
$v_{NT} = ( 1_N' ,- 1_T' )'$, and $1_N$ and $1_T$ denote vectors of ones with dimensions $N$ and $T$.
Existence and uniqueness of the solution to the population problem will be 
guaranteed by our assumptions in Section~\ref{sec:app_panel} below, including concavity 
of the objective function in all parameters.
 The second term of $\mathcal{L}_{NT}$  is a penalty that imposes a normalization needed to identify  $\phi_{NT}$ in 
models with scalar individual and time effects that enter 
additively into the log-likelihood function
as $\alpha_i + \gamma_t$.\footnote{
In Appendix~\ref{app:expansion} we derive asymptotic expansions that apply to  general models with multiple unobseved effects.
In order to use these expansions to obtain the asymptotic distribution of the panel fixed effects estimators, we need to derive  the properties of the expected Hessian of the incidental parameters, a matrix with increasing dimension, and to show the consistency of the estimator of the incidental parameter vector. The additive specification $\alpha_i + \gamma_t$ is useful to characterize the Hessian and we impose strict concavity of the objective function to show the consistency.
}
In this case, adding a constant  to all $\alpha_i$, while
subtracting it from all $\gamma_t$,  does not change $\alpha_{i} + \gamma_t$.
To eliminate this ambiguity, we normalize
$\phi^0_{NT}$ to satisfy $v_{NT}' \phi^0_{NT} = 0$, i.e. $\sum_i \alpha_i^0 = \sum_t
\gamma_t^0$. The penalty produces a maximizer of ${\cal L}_{NT}$ that  is automatically normalized.
We could  equivalently impose $v_{NT}' \phi_{NT} = 0$ as a constraint, but for technical reasons we prefer to work with an unconstrained optimization problem. There are other possible normalizations for $\phi_{NT}$, such as $\alpha_1 = 0$. The model parameter $\beta$ is invariant to the choice of normalization, that is, our asymptotic results on the estimator for $\beta$
are independent of this choice of normalization.
 Our choice
is convenient for certain intermediate results that involve the incidental parameter~$\phi_{NT}$, its score vector and its Hessian matrix.
The pre-factor $(NT)^{-1/2}$ in $\mathcal{L}_{NT}(\beta,\phi_{NT})$ is just a rescaling.

Other quantities of interest involve averages over the data and unobserved effects
\begin{equation} \label{eq: meffs}
 \delta^0_{NT} =    \mathbb{E} [\Delta_{NT}(\beta^0, \phi^0_{NT})], \ \ \Delta_{NT}(\beta, \phi_{NT}) = (NT)^{-1} \sum_{i,t}   \Delta(X_{it}, \beta, \alpha_i, \gamma_t),
\end{equation}
where  $\mathbb{E}$ denotes the expectation with respect to
the joint distribution of the data and the unobserved  effects, provided that the expectation exists. 
$\delta^0_{NT} $ is indexed by $N$ and $T$ because the marginal distribution of $\{(X_{it}, \alpha_i, \gamma_t)  : 1 \leq i \leq N, 1 \leq t \leq T\}$ can be heterogeneous across $i$ and/or $t$; see Section~\ref{subsec:apes}. These averages
include average partial effects (APEs), which are often the ultimate
quantities of interest in nonlinear models. The APEs are invariant to the choice of normalization for $\phi_{NT}$ if $\alpha_i$ and $\gamma_t$ enter $\Delta(X_{it}, \beta, \alpha_i, \gamma_t)$ as $\alpha_i + \gamma_t$.
Some examples of partial effects that satisfy this condition are the following:

\medskip

\noindent  \textbf{Example \ref{example: probit}} (Binary response model).  \textit{If $X_{it,k}$, the $k$th element of $X_{it}$, is
binary, its partial effect on the conditional probability of $Y_{it}$ is
\begin{equation}\label{example: probit: meff1}
\Delta(X_{it}, \beta , \alpha_i, \gamma_t) =  F(\beta_k +
X_{it,-k}'\beta_{-k} +  \alpha_i + \gamma_t) -
F(X_{it,-k}'\beta_{-k} + \alpha_i +  \gamma_t),
\end{equation}
where $\beta_k$ is the $k$th element of $\beta$, and $X_{it,-k}$ and $\beta_{-k}$ include all elements of $X_{it}$ and $\beta$ except for the $k$th element. If $X_{it,k}$ is continuous and $F$ is differentiable, the partial effect of $X_{it,k}$
on the conditional probability of $Y_{it}$ is
\begin{equation}\label{example: probit: meff2}
\Delta(X_{it}, \beta,\alpha_i, \gamma_t) = \beta_k  \partial F(X_{it}'\beta +
\alpha_i + \gamma_t),
\end{equation}
where $\partial F$ is the derivative of $F$.
}

\medskip

\noindent  \textbf{Example \ref{example: poisson}} (Poisson model). \textit{If $X_{it}$ includes $Z_{it}$ 
and some known transformation $H(Z_{it})$ with coefficients $\beta_k$ and $\beta_j$, the partial effect of $Z_{it}$
on the conditional expectation  of $Y_{it}$ is
\begin{equation}\label{example: poisson: meff}
\Delta(X_{it}, \beta, \alpha_i, \gamma_t) = [\beta_k +  \beta_j
\partial H(Z_{it})] \exp(X_{it}'\beta + \alpha_i +  \gamma_t).
\end{equation}
}



\subsection{Fixed effects estimators} 

We estimate the parameters by solving the  sample analog of problem \eqref{population_prob}, i.e.
\begin{equation}
\max_{(\beta, \phi_{NT}) \in  \mathbb{R}^{\dim \beta + \dim \phi_{NT}}} \;  {\cal L}_{NT}(\beta,\phi_{NT}).
   \label{LobjMAX}
\end{equation}
As in the population case, we shall impose conditions guaranteeing that the solution to this maximization problem
 exists  and is unique with probability approaching one as $N$ and $T$ become large. 
For computational purposes, we note that the solution to the program~\eqref{LobjMAX} for $\beta$ is  the same as the solution to the program that imposes $v_{NT}' \phi_{NT} = 0$ directly as a constraint in the optimization, and is invariant to the normalization. 
  In our numerical examples we impose either $\alpha_1 = 0$ or $\gamma_1 = 0$ directly by dropping  the first individual or time effect. This constrained program  has good computational properties because its objective function is concave and smooth in all the parameters. We have developed the commands  \texttt{probitfe} and \texttt{logitfe}  in \texttt{Stata} to implement the methods of the paper for probit and logit models  \cite{Stata2015}.\footnote{We refer to this companion work for computational details.} When $N$ and $T$ are  large, e.g., $N > 2,000$ and $T > 50$, we recommend the use of optimization routines  that exploit the  sparsity of the design matrix of the model to speed up computation such as the package \texttt{Speedglm} in \texttt{R} \cite{Speedglm2012}. For a probit model  with $N=2,000$ and $T = 52$, \texttt{Speedglm} computes the fixed effects estimator in less than 2 minutes with a 2 x 2.66 GHz 6-Core Intel Xeon processor, more than 7.5 times faster than our  \texttt{Stata} command \texttt{probitfe} and more than 30 times faster than the \texttt{R} command \texttt{glm}.\footnote{Additional comparisons of computational times  are available from the authors upon request.}

To analyze the statistical properties of the estimator of $\beta$ it is
convenient to first concentrate out the nuisance parameter $\phi_{NT}$. For
given $\beta$, we define the optimal $\widehat \phi_{NT}(\beta)$  as
\begin{equation}
   \widehat \phi_{NT}(\beta) = \argmax_{\phi_{NT} \in \mathbb{R}^{\dim \phi_{NT}}}
       \,  {\cal L}_{NT}(\beta,\,\phi_{NT}) \;.
   \label{DefEstPhi}
\end{equation}
The fixed effects estimators of $\beta^0$ and $\phi_{NT}^0$ are
\begin{equation}
   \widehat \beta_{NT} \, = \,
      \argmax_{\beta \in  \mathbb{R}^{\dim \beta}} \; {\cal L}_{NT}(\beta,\widehat \phi_{NT}(\beta)) \;, 
       \qquad
   \widehat \phi_{NT} \, = \, \widehat \phi_{NT}(\widehat \beta).
   \label{DefEst}
\end{equation}
%


Estimators of APEs can be
formed by plugging-in the estimators of the model parameters in the sample version of \eqref{eq: meffs}, i.e.
\begin{equation}
    \label{eq:ape}
\widehat \delta_{NT}  =  \Delta_{NT}(\widehat \beta, \widehat \phi_{NT}).
\end{equation}
Again, $\widehat \delta_{NT}$ is invariant to the normalization chosen for $\phi_{NT}$ if $\alpha_i$ and $\gamma_t$ enter $\Delta(X_{it}, \beta, \alpha_i, \gamma_t)$ as $\alpha_i + \gamma_t$.

\section{Incidental parameter problem and bias corrections} \label{sec:bc} 

In this section we give a heuristic discussion of the main results, leaving the technical details to Section~\ref{sec:app_panel}. We illustrate the analysis with numerical calculations based on a variation of the classical Neyman and Scott~\cite*{Neyman:1948p881} variance example.

\subsection{Incidental parameter problem}

Fixed effects estimators in nonlinear  models
suffer from the
incidental parameter problem \cite{Neyman:1948p881}.  The source of the problem is that the dimension of the nuisance parameter $\phi_{NT}$ increases with the sample size under asymptotic approximations where either $N$ or $T$ pass to infinity. 
To describe the problem  let
\begin{equation}\label{eq:pop_program}
\overline{\beta}_{NT} := \argmax_{\beta \in  \mathbb{R}^{\dim \beta}}  \;
    \E\left[ {\cal L}_{NT}(\beta,\widehat \phi_{NT}(\beta)) \right].
\end{equation}
 The fixed effects estimator  is inconsistent under the traditional Neyman and Scott asymptotic sequences where $N \to \infty$  and $T$ is fixed,  i.e.,  $\plim_{N \to \infty} \overline{\beta}_{NT} \neq \beta^0$. Similarly,  the fixed effects estimator is inconsistent under asymptotic sequences where $T \to \infty$  and $N$ is fixed,  i.e.,  $\plim_{T \to \infty} \overline{\beta}_{NT} \neq \beta^0$. Note that $\overline \beta_{NT} = \beta^0$  if $\widehat \phi_{NT}(\beta)$ is replaced by $\phi_{NT}(\beta) = \argmax_{\phi_{NT} \in \mathbb{R}^{\dim \phi_{NT}}}
       \,  \E[ {{\cal L}}_{NT}(\beta,\,\phi_{NT})]$. Under asymptotic approximations where either $N$ or $T$ are fixed, there is only  a fixed number of observations to estimate some of the components of $\phi_{NT}$, $T$ for each individual effect or $N$ for each time effect, rendering the  estimator $\widehat \phi_{NT}(\beta)$ inconsistent for $\phi_{NT}(\beta)$. The nonlinearity of the model propagates the inconsistency to the estimator of $\beta$.  

A key insight of  the large-$T$ panel data literature is that the incidental parameter problem becomes an asymptotic bias problem under an asymptotic approximation where $N \to \infty$ and $T \to \infty$  (e.g., Arellano and Hahn, \citeyear{ArellanoHahn2007}). For models with only individual effects,  this literature derived the expansion $\overline{\beta}_{NT} = \beta^0 + B/T + o_P(T^{-1})$ as $N,T \to \infty$, for some constant $B$. The fixed effects estimator is  consistent because $\plim_{N,T \to \infty} \overline{\beta}_{NT} = \beta^0$, but has bias in the asymptotic distribution if $B/T$ is not negligible relative to $1/\sqrt{NT}$, 
 the order of the standard deviation of the estimator.
This asymptotic bias problem, however, is easier to correct than the
  inconsistency problem that arises under the traditional Neyman and Scott asymptotic approximation.  We show that the same insight still applies to models with individual and time effects, but with a different expansion for $\overline{\beta}_{NT}$. We characterize the expansion and develop bias corrections.

\subsection{Bias Expansions and Bias Corrections}
Some expansions can be used to explain our corrections. 
For smooth likelihoods and under appropriate regularity conditions, as $N,T \to \infty$,
\begin{equation}\label{eq: expansion}
\overline{\beta}_{NT} = \beta^0 + \overline B_{\infty}^{\beta}/T
+ \overline D_{\infty}^{\beta} /N + o_P(T^{-1} \vee N^{-1}),
\end{equation}
for some $\overline B_{\infty}^{\beta}$ and $\overline D_{\infty}^{\beta}$ that we characterize in Theorem \ref{th:BothEffects} and explain in Remark \ref{remark:ad},
where $a \vee b := \max(a,b)$.
 Unlike in nonlinear models without incidental parameters, the order of the bias is higher than the inverse of the sample size $(NT)^{-1}$ due to the slow rate of convergence of  $\widehat \phi_{NT}$. Note also that by the properties of the maximum likelihood estimator
$$
\sqrt{NT}(\widehat \beta_{NT} - \overline{\beta}_{NT}) \to_d \mathcal{N}(0, \overline V_{\infty}),
$$
for some $\overline V_{\infty}$ that we also characterize in Theorem \ref{th:BothEffects}.
Under asymptotic sequences where $N/T \to \kappa^2$ as $N,T \to \infty$, the fixed effects estimator is asymptotically 
biased because
\begin{align}
\sqrt{NT}(\widehat \beta_{NT} - \beta^0) &= \sqrt{NT}(\widehat \beta_{NT} - \overline{\beta}_{NT}) + \sqrt{NT}(\overline B_{\infty}^{\beta}/T
+ \overline D_{\infty}^{\beta} /N + o_P(T^{-1} \vee N^{-1})) \nonumber \\
 &\to_d \mathcal{N}( \kappa \overline B_{\infty}^{\beta}
+ \kappa^{-1} \overline D_{\infty}^{\beta} , \overline V_{\infty}).
\label{eq:MainResult}
\end{align}
Relative to fixed effects estimators with only individual effects, the presence of time effects 
introduces additional asymptotic bias through $\overline D_{\infty}^{\beta}$. 
This asymptotic result predicts that the fixed effects estimator can have significant bias relative to its dispersion. Moreover, confidence intervals constructed around the fixed effects estimator can severely undercover the true value of the parameter even in large samples. We show that these predictions provide a good approximations to the finite sample behavior of the fixed effects estimator through analytical and simulation examples in Sections \ref{subsec: ns} and \ref{sec: MC}.

The analytical bias correction consists of subtracting estimates of the leading terms of the bias from the fixed effect estimator
of $\beta^0$. Let $\widehat B_{NT}^{\beta}$ and $\widehat D_{NT}^{\beta}$ be  estimators of $\overline B_{\infty}^{\beta}$ and $\overline D_{\infty}^{\beta}$ as defined in \eqref{eq: bias_parameter}. The bias corrected estimator can be formed as
\begin{equation*}\label{abc}
\widetilde{\beta}_{NT}^A = \widehat{\beta}_{NT} -  \widehat{B}_{NT}^{\beta}/T - \widehat D_{NT}^{\beta}/N.
\end{equation*}
If $N/T \to \kappa^2$,  $\widehat B_{NT}^{\beta} \to_P \overline B_{\infty}^{\beta}$, and $\widehat D_{NT}^{\beta} \to_P \overline D_{\infty}^{\beta},$ then
\begin{equation*}
\sqrt{NT}(\widetilde \beta_{NT}^A - \beta^0) 
\to_d \mathcal{N}(0, \overline V_{\infty}).
\end{equation*}
The analytical correction therefore centers the asymptotic distribution at the true value of the parameter, without increasing asymptotic variance. This asymptotic result predicts that in large samples the corrected estimator has small bias relative to dispersion, the correction does not increase dispersion, and the confidence intervals constructed around the corrected estimator have coverage probabilities close to the nominal levels. We show that these predictions provide a good approximations to the behavior of the corrections in Sections \ref{subsec: ns} and \ref{sec: MC} even in small panels with $N < 60$ and $T < 15$.

We also consider a jackknife bias correction method
that does not require explicit estimation of the bias.
 This method is based on the split panel jackknife (SPJ) of Dhaene and Jochmans~\cite*{DhaeneJochmans2015} applied to the
 time and cross-section dimension of the panel.
Alternative jackknife corrections based on the leave-one-observation-out panel jackknife (PJ) of Hahn and Newey \cite*{Hahn:2004p882}  and combinations of PJ and SPJ are also possible. We do not consider corrections based on PJ because they are theoretically justified by second-order expansions of $\overline{\beta}_{NT}$ that are beyond the scope of this paper.

To describe  our generalization of the SPJ, define the fixed effects estimator of $\beta$ in the subpanel with cross sectional indexes $A$ and time series indexes $B$ as
$$
\widehat \beta_{A,B} \in \argmax_{\beta \in \mathbb{R}^{\dim \beta} } \;
   \max_{\alpha(A) \in \mathbb{R}^{|A|}}  \;
    \max_{\gamma(B) \in \mathbb{R}^{|B|}} 
\; \sum_{i,t}  d_{it}(A,B) \, \log f_{Y}(Y_{it}  \mid X_{it}, \alpha_i, \gamma_t, \beta) ,
$$
where $\alpha(A)=( \alpha_i  :  i \in A )$,
$\gamma(B)=( \gamma_t   :   t \in B )$, and
$d_{it}(A,B) = 1(i \in A) \times 1(t \in B)$.
Let  $\widetilde{\beta}_{N,T/2}$ be the average of the 2 split jackknife
estimators in the subpanels with $A = \{1,2,\ldots,N\}$, and $B = \{1,2,\ldots,T/2\}$ or $B = \{T/2+1, T/2+2, \ldots, T\}$, i.e. including all the individuals and leaving out the first and second halves of the time
periods. Let $\widetilde{\beta}_{N/2,T}$ be the average of the 2 split jackknife
estimators in the subpanels with $B = \{1,2,\ldots,T\}$, and $A = \{1,2,\ldots,N/2\}$ or $A = \{N/2+1,N/2+2,\ldots,N\}$, i.e. including all the time periods and leaving out half of the individuals of the panel.\footnote{When $T$ is odd we define $\widetilde{\beta}_{N,T/2}$ as the average of the 2 split jackknife
estimators that use overlapping subpanels with $B = \{1,2,\ldots,(T+1)/2\}$ and $B = \{(T+1)/2, (T+1)/2+1, \ldots,T\}$. We define $\widetilde{\beta}_{N/2,T}$ similarly when $N$ is odd.}
In choosing the cross sectional
indexing of the panel, one might want to take into account individual clustering structures and other dependencies  to
preserve them in the SPJ. For example, all the individuals belonging to the same cluster should be indexed such that they remain in the same subpanel after the cross sectional split.  
If there are no cross sectional dependencies,  the indexing of the individuals is unrestricted. We recommend to construct $\widetilde{\beta}_{N/2,T}$ as the average
of the estimators obtained from all possible partitions of $N/2$ individuals to avoid ambiguity and arbitrariness in the choice of the division.\footnote{There are $P =  {N \choose N/2}$ different cross sectional partitions with $N/2$ individuals. When $N$
is large, we can approximate the average over all possible partitions  by the average over $S \ll P$ randomly chosen partitions to speed up computation.}  
The bias corrected estimator is
\begin{equation}\label{eq: jackknife2}
\widetilde{\beta}_{NT}^{J} = 3  \widehat \beta_{NT} - \widetilde{\beta}_{N,T/2} - \widetilde{\beta}_{N/2,T}.
\end{equation}
To give some intuition about how the corrections works, note that
\begin{equation*}
\widetilde{\beta}_{NT}^{J}  - \beta_0 =  (\widehat \beta_{NT} - \beta_0)  - (\widetilde{\beta}_{N,T/2} - \widehat \beta_{NT}) - (\widetilde{\beta}_{N/2,T} - \widehat \beta_{NT}),
\end{equation*}
where 
 $ \widetilde{\beta}_{N,T/2} - \widehat \beta_{NT} = \overline B_{\infty}^{\beta}/T + o_P(T^{-1} \vee N^{-1})$ and  $\widetilde{\beta}_{N/2,T} - \widehat \beta_{NT} = \overline D_{\infty}^{\beta}/N + o_P(T^{-1} \vee N^{-1}).$ Relative to $\widehat \beta_{NT}$, $ \widetilde{\beta}_{N,T/2} $  has double the bias coming from the estimation of the individual effects  because it is based on subpanels with half of the time periods, and $ \widetilde{\beta}_{N/2,T} $  has double the bias coming from the estimation of the time effects because it is based on subpanels with half of the individuals. The time series split  removes the bias term $\overline B_{\infty}^{\beta}$ and the cross sectional split removes the bias term $\overline D_{\infty}^{\beta}.$

\subsection{Illustrative Example}\label{subsec: ns}

To illustrate how the bias corrections work in finite samples, we consider a simple model where  the solution to the population program  \eqref{eq:pop_program} 
has closed form. This model corresponds to a variation of the classical Neyman and Scott~\cite*{Neyman:1948p881} variance example that includes both  individual and time effects, $Y_{it} \mid \alpha, \gamma, \beta \sim \mathcal{N}(\alpha_i + \gamma_t, \beta)$.
It is well-know that in this case
$$
\widehat \beta_{NT} = (NT)^{-1} \sum_{i,t} \left( Y_{it} - \bar Y_{i.} - \bar Y_{.t} + \bar Y_{..} \right)^2,
$$
where $\bar Y_{i.} = T^{-1} \sum_t Y_{it}$, $\bar Y_{.t} = N^{-1} \sum_i Y_{it},$ and $\bar Y_{..} = (NT)^{-1} \sum_{i,t} Y_{it}.$ Moreover, 
 from the well-known results on the degrees of freedom adjustment of the estimated variance
$$
\overline{\beta}_{NT} = \E [\widehat \beta_{NT}] =  \beta^0 \frac{(N-1)(T-1)}{NT}  = \beta^0 \left(1 - \frac{1}{ T} - \frac{1}{ N} + \frac{1}{NT} \right),
$$
so that $\overline{B}_{\infty}^{\beta} = -\beta^0$ and $\overline{D}_{\infty}^{\beta} = -\beta^0$.\footnote{Okui~\cite*{Okui2013} derived the bias of fixed effects estimators of autocovariances and autocorrelations in this model.}

To form the analytical bias correction we can set $\widehat{B}_{NT}^{\beta} = - \widehat \beta_{NT}$ and $\widehat{D}_{NT}^{\beta} = - \widehat \beta_{NT}$. This yields $\widetilde \beta^A_{NT} = \widehat \beta_{NT} (1 + 1/T + 1/N)$ with
$$
\overline{\beta}^A_{NT} = \E[\widetilde \beta^A_{NT}] = \beta^0 \left( 1 - \frac{1}{T^2} - \frac{1}{N^2}  - \frac{1}{NT} + \frac{1}{NT^2} + \frac{1}{N^2T}\right).
$$
This correction reduces the order of the bias from $(T^{-1} \vee N^{-1})$ to $(T^{-2} \vee N^{-2}),$ and introduces additional higher order
terms. The analytical correction increases finite-sample variance because the factor $(1 + 1/T + 1/N) > 1$. We compare the biases and standard deviations of the fixed effects estimator and the corrected estimator in a numerical example below.

For the Jackknife correction, straightforward calculations give
$$
\overline{\beta}^{J}_{NT} = \E[\widetilde \beta^{J}_{NT}] =  3 \overline{\beta}_{NT} - \overline{\beta}_{N,T/2} - \overline{\beta}_{N/2,T}  = \beta^0 \left( 1 - \frac{1}{NT} \right).$$
The correction therefore reduces the order of the bias from $(T^{-1} \vee N^{-1})$ to $(TN)^{-1}.$\footnote{ In this example it is possible to develop higher-order jackknife corrections that completely eliminate the bias because we know the entire expansion of $\overline{\beta}_{NT}$. For example, $\E[ 4 \widehat \beta_{NT} - 2 \widetilde \beta_{N,T/2} - 2 \widetilde \beta_{N/2,T} + \widetilde \beta_{N/2,T/2}] = \beta^0,$ where $\widetilde \beta_{N/2,T/2}$ is the average of the four split jackknife estimators that leave out half of the individuals and the first or the second halves of the time periods. See Dhaene and Jochmans~\cite*{DhaeneJochmans2015} for a discussion on higher-order bias corrections of panel fixed effects estimators.}

Table \ref{table:ns1} presents numerical results for the bias and standard deviations of the fixed effects and bias corrected estimators in finite samples. We consider panels with $N,T \in \{10,25, 50\},$ and only report the results for $T \leq N$ since all the expressions are symmetric in $N$ and $T$. All the numbers  in the table
are  in percentage of the true  parameter value, so we do not need to specify the value of $\beta^0$. We find that the analytical and jackknife corrections
offer substantial improvements over the fixed effects estimator in terms of bias. The first and fourth row of the table show that the bias of the fixed effects estimator is of the same order of magnitude as the standard deviation, where $\overline{V}_{NT} = {\rm Var}[\widehat \beta_{NT}] = 2 (N-1)(T-1) (\beta^0)^2 / (NT)^2$  under independence of $Y_{it}$ over $i$ and $t$ conditional on the unobserved effects. The fifth row shows the increase in standard deviation due to analytical bias correction is small compared to the bias reduction, where $\overline{V}_{NT}^A = {\rm Var}[\widetilde \beta_{NT}^A] = (1+1/N+1/T)^2 \overline{V}_{NT}$. The last row shows that the jackknife yields less precise estimates than the analytical correction when $T=10$.  


\begin{table}[htp]
\begin{center}\caption{\label{table:ns1} Biases and Standard Deviations for $Y_{it} \mid \alpha, \gamma, \beta \sim \mathcal{N}(\alpha_i + \gamma_t, \beta)$}
\begin{tabular}{ccccccc} \hline\hline
&   N = 10 &  \multicolumn{2}{c}{N=25} & \multicolumn{3}{c}{N=50}  \\
   &   T = 10 & T=10  & T=25  & T=10  & T=25  & T=50  \\\hline
  $ (\overline{\beta}_{NT} - \beta^0) /\beta^0$   &  -.19 & -.14  & -.08  & -.12  & -.06  & -.04  \\
  $ (\overline{\beta}^A_{NT} - \beta^0)/\beta^0$   & -.03 & -.02  & .00  & -.01  & -.01  & .00  \\
  $ (\overline{\beta}^{J}_{NT}- \beta^0)/\beta^0$   & -.01   & .00  & .00  & .00  & .00  & .00  \\\hline
  $ \sqrt{\overline{V}_{NT}}/\beta^0$   &  .13 & .08  & .05  & .06  & .04  & .03  \\
  $ \sqrt{\overline{V}^A_{NT}}/\beta^0$   &  .14 & .09  & .06  & .06  & .04  & .03  \\
  $ \sqrt{\overline{V}^J_{NT}}/\beta^0$   &  .17 & .10  & .06  & .07  & .04  & .03 \\\hline\hline
  \multicolumn{7}{l}{ \begin{footnotesize} Notes: $\overline{V}^J_{NT}$ obtained by 50,000 simulations with $\beta^0 = 1$\end{footnotesize}}
  \end{tabular}
\end{center}
\end{table}

Table~\ref{table:ns2} 
illustrates the effect of the bias on the inference based on the asymptotic distribution. It shows the 
coverage probabilities of 95\% asymptotic confidence intervals for $\beta^0$ constructed in the usual way as $$\text{CI}_{.95}(\widehat \beta) = \widehat \beta \pm 1.96  \widehat V_{NT}^{1/2} = 
\widehat \beta (1 \pm 1.96  \sqrt{2/(NT)}),$$
where $\widehat \beta = \{\widehat \beta_{NT}, \widetilde \beta_{NT}^{A}, \widetilde \beta_{NT}^{J}\}$ and $\widehat V_{NT} = 2 \widehat \beta^2/(NT)$ is an estimator of the asymptotic variance $\overline{V}_{\infty}/(NT) = 2 (\beta^0)^2/(NT)$.  To find the coverage probabilities, we use that $NT \widehat \beta_{NT}/ \beta^0 \sim \chi^2_{(N-1)(T-1)}$ and $\widetilde \beta_{NT}^{A} = (1 + 1/N + 1/T) \widehat \beta_{NT}$. These probabilities do not depend on the value of $\beta^0$ because the limits of the intervals are proportional to $\widehat \beta$.  For the Jackknife we compute the probabilities numerically by simulation with $\beta^0=1$. As a benchmark of comparison, we also consider confidence intervals constructed from the unbiased estimator $\widetilde \beta_{NT} = NT \widehat \beta_{NT}/[(N-1)(T-1)]$. Here we find that the confidence intervals based on the fixed effect estimator display severe undercoverage  for all the sample sizes. The confidence intervals based on the corrected estimators have high coverage probabilities, which approach the nominal level as the sample size grows. Moreover, the bias corrected  estimators produce confidence intervals with very similar coverage probabilities to the ones from the unbiased estimator.

\begin{table}[htp]
\begin{center} \caption{\label{table:ns2} Coverage probabilities for $Y_{it} \mid \alpha, \gamma, \beta \sim \mathcal{N}(\alpha_i + \gamma_t, \beta)$}
\bigskip

\begin{tabular}{ccccccc} \hline\hline
&   N = 10 &  \multicolumn{2}{c}{N=25} & \multicolumn{3}{c}{N=50}  \\
   &   T = 10 & T=10  & T=25  & T=10  & T=25  & T=50  \\\hline
  $\text{CI}_{.95}(\widehat \beta_{NT})$   &  .56 & .55  & .65  & .44  & .63  & .68  \\
  $\text{CI}_{.95}(\widetilde \beta^A_{NT})$   &  .89 & .92  & .93  & .92  & .94  & .94   \\
 $\text{CI}_{.95}(\widetilde \beta^J_{NT})$   &   .89 & .91  & .93  & .92  & .93  & .94  \\
  $\text{CI}_{.95}(\widetilde \beta_{NT})$   &  .91 & .93  & .94  & .93  & .94  & .94   \\\hline\hline
   \multicolumn{7}{l}{\begin{footnotesize} Notes: Nominal coverage probability is .95. $\text{CI}_{.95}(\widetilde \beta^J_{NT})$ obtained by \end{footnotesize}}\\
   \multicolumn{7}{l}{\begin{footnotesize}  50,000   simulations with $\beta^0 = 1$\end{footnotesize}}
  \end{tabular}
\end{center}
\end{table}

\section{Asymptotic Theory for Bias Corrections}
\label{sec:app_panel}

In nonlinear panel data models the population problem (\ref{eq:pop_program}) generally does not have closed form solution, so we need to
rely on asymptotic arguments to characterize the terms in the expansion of the bias (\ref{eq: expansion}) and to justify the validity of the corrections.
 
\subsection{Asymptotic distribution of model parameters}
We consider panel models with scalar individual and time effects that enter the likelihood function
additively through  $\pi_{it} = \alpha_i + \gamma_t$. In these models the dimension of the incidental parameters is $\dim \phi_{NT} = N+T$.
The leading cases are single index models, where the dependence of the likelihood function on the parameters is  through an index
$ X_{it}'\beta + \alpha_i + \gamma_t$.   These models cover  the
probit and Poisson specifications of Examples~\ref{example: probit}
and \ref{example: poisson}.  The additive structure
only applies to the unobserved effects, so we can allow for scale parameters to cover the Tobit and negative
binomial models.
 We focus on these additive models for computational tractability and because we can establish
the consistency of the fixed effects estimators under a concavity assumption in the log-likelihood function with respect to all the parameters.

The parametric part of our panel models  takes the form
\begin{equation}\label{eq: index_model}
   \log f_{Y}(Y_{it} \mid X_{it}, \alpha_i, \gamma_t, \beta)
      =: \ell_{it}(\beta, \, \pi_{it} ).
\end{equation}
We denote  the derivatives  of the log-likelihood function $\ell_{it}$ by $\partial_\beta \ell_{it}(\beta, \pi) := \partial  \ell_{it}(\beta,\pi)/\partial \beta$,
$\partial_{\beta \beta'} \ell_{it}(\beta, \pi) := \partial^2  \ell_{it}(\beta,\pi)/(\partial \beta \partial \beta')$,
$\partial_{\pi^q} \ell_{it}(\beta, \pi) := \partial^q \ell_{it}(\beta,\pi)/\partial \pi^q$, $q = 1,2,3$, \textit{etc}.
We drop the arguments $\beta$ and $\pi$ when the derivatives are evaluated at the true
parameters $\beta^0$ and $\pi^0_{it} := \alpha_i^0 + \gamma_t^0$, e.g.
$\partial_{\pi^q} \ell_{it} := \partial_{\pi^q}\ell_{it}(\beta^0,\pi^0_{it})$. We also drop the dependence on
$NT$ from all the sequences of functions and parameters, e.g. we use $\mathcal{L}$ for $\mathcal{L}_{NT}$ and $\phi$
for $\phi_{NT}$.


We make the following assumptions:

\begin{assumption}[Panel models]
   \label{ass:PanelA1}
   Let $\nu > 0$ and $\mu>4 (8+\nu)/\nu$.
   Let $\varepsilon>0$ and let 
   ${\cal B}^0_{\varepsilon}$  be a subset of $\mathbb{R}^{\dim \beta+1}$
     that contains an $\varepsilon$-neighbourhood of $(\beta^0,\pi^0_{it})$
     for all $i,t,N,T$.\footnote{ 
     For example, ${\cal B}^0_{\varepsilon} $ can be chosen to be the Cartesian product of the $\varepsilon$-ball around $\beta^0$
     and the interval $[ \pi_{\min}, \pi_{\max}]$, with $\pi_{\min} \leq \pi_{it} - \varepsilon$ and $\pi_{\max} \geq \pi_{it} + \varepsilon$
     for all $i,t,N,T$. We can have $\pi_{\min} = -\infty$ and $\pi_{\max} = \infty$, as long as this is compatible with
     Assumption~\ref{ass:PanelA1}~(iv) and (v).
     }
   \begin{itemize}
      \item[(i)] Asymptotics: we consider limits of sequences where  $N/T \rightarrow \kappa^2$, $0<\kappa<\infty$, as $N,T \rightarrow \infty$.

      \item[(ii)] Sampling:  conditional on $\phi$, $\{(Y_i^T, X_i^T) : 1 \leq i \leq N \}$ is independent across $i$ 
          and, for each $i$, $\{ (Y_{it}, X_{it}) : 1 \leq t \leq T\}$ is $\alpha$-mixing with mixing coefficients
          satisfying $\sup_i a_i(m) = {\cal O}( m^{-\mu} )$ as $m \rightarrow \infty$, where         
\begin{equation*}
 a_i(m) := \sup_t \sup_{A \in \mathcal{A}_{t}^i, B \in \mathcal{B}_{t+m}^i} |P(A \cap B) - P(A) P(B)|,
 \end{equation*}
and for $Z_{it} = (Y_{it}, X_{it})$,  $\mathcal{A}_{t}^i$ is the sigma field generated by $(Z_{it}, Z_{i,t-1}, \ldots)$, and
 $\mathcal{B}_{t}^i$ is the sigma field generated by $(Z_{it}, Z_{i,t+1}, \ldots)$.

     \item[(iii)] Model: for $X^t_i = \{X_{is}:  s = 1, ..., t\}$, we assume that for all $i,t,N,T,$
\begin{equation*}
Y_{it}  \mid X^t_i, \phi, \beta   \sim \exp[\ell_{it}(\beta,\alpha_{i} + \gamma_t)].
\end{equation*}
The realizations of the parameters and unobserved effects that generate the observed data are denoted by
$\beta^0$ and $\phi^0$.\



      \item[(iv)]  Smoothness and moments:
      We assume that
      $(\beta,\pi) \mapsto \ell_{it}(\beta,\pi)$ is
      four times continuously differentiable over ${\cal B}^0_{\varepsilon}$ a.s.
      The partial derivatives of  $\ell_{it}(\beta,\pi)$
      with  respect to the elements of $(\beta,\pi)$
      up to fourth order
      are bounded in absolute value uniformly over $(\beta,\pi) \in {\cal B}^0_{\varepsilon}$
      by a function $M(Z_{it})>0$ a.s.,
      and $\max_{i,t} \E[M(Z_{it})^{8+\nu}]$
     is a.s. uniformly bounded over  $N,T$.

        \item[(v)] Concavity: 
        For all $N,T,$
    $(\beta,\phi) \mapsto \mathcal{L}(\beta,\phi) =  (NT)^{-1/2} \{\sum_{i,t} \ell_{it} (\beta,\alpha_i + \gamma_t) - b (v'\phi)^2/2\}$
    is strictly concave over $\mathbb{R}^{\dim \beta + N + T}$ a.s.
    Furthermore, there exist constants $b_{\min}$
    and $b_{\max}$ such that for all $(\beta,\pi) \in {\cal B}^0_{\varepsilon}$,
    $0< b_{\min} \leq - \E\left[ \partial_{\pi^2} \ell_{it}(\beta,\pi) \right]   \leq b_{\max}$ a.s.  uniformly over $i,t,N,T$.

    \end{itemize}

\end{assumption}

\begin{remark}[Assumption \ref{ass:PanelA1}]
Assumption \ref{ass:PanelA1}$(i)$ defines the large-$T$ asymptotic framework and is the same as in Hahn and Kuersteiner \cite*{HahnKuersteiner2011}.  The relative rate of $N$ and $T$ exactly balances the order of the bias and variance producing a non-degenerate asymptotic distribution. 

Assumption \ref{ass:PanelA1}$(ii)$ does not impose identical distribution nor stationarity  over the time series dimension, conditional on the unobserved effects, unlike most of the large-$T$ panel literature, e.g.,
Hahn and Newey \cite*{Hahn:2004p882} and
Hahn and Kuersteiner \cite*{HahnKuersteiner2011}. 
These assumptions are violated by the presence of the time effects, because they are treated as parameters.
 The mixing condition is used to bound covariances and moments in the  application of laws of large numbers and central limit theorems -- it could replaced by other conditions that guarantee the applicability of these results.

Assumption \ref{ass:PanelA1}$(iii)$ is the parametric part of the panel model.  We rely on this assumption to guarantee that $\partial_\beta \ell_{it}$ and  $\partial_\pi \ell_{it}$ have martingale difference properties. Moreover, we use certain Bartlett identities implied by this assumption  to simplify some expressions, but those simplifications are not crucial for our results.
We provide expressions for the asymptotic bias and variance that do not apply these simplifications in Remark~\ref{rem:NoBartlett} below.

 Assumption \ref{ass:PanelA1}$(iv)$ imposes smoothness and moment conditions in the log-likelihood function and its derivatives. These conditions 
 guarantee that the higher-order stochastic expansions of the fixed effect estimator that we use to characterize the asymptotic bias are well-defined, and that the remainder terms of these expansions are bounded. 
 
 The most commonly used nonlinear models in applied economics such as logit, probit, ordered probit, Poisson, and Tobit models have
smooth log-likelihoods functions that satisfy the concavity condition of Assumption \ref{ass:PanelA1}$(v)$, provided that all the elements of $X_{it}$ have cross sectional and time series variation.
Assumption \ref{ass:PanelA1}$(v)$ guarantees that $\beta^0$ and $\phi^0$ are the unique solution to the
population problem \eqref{population_prob}, that is all the parameters are point identified.
\end{remark}

To describe the asymptotic distribution of the fixed effects estimator $\widehat \beta,$ it is convenient to introduce some additional
notation. Let $\overline{\cal H}$ be the $(N+T) \times (N+T)$ expected Hessian matrix of
the  log-likelihood with respect to the nuisance parameters evaluated at the true parameters, i.e.
\begin{align}
\overline{\cal H} = \E[ -  \partial_{\phi \phi'} {\cal L}] =  
\left(\begin{array}{cc}  \overline{\mathcal{H}}_{(\alpha\alpha)}^* & \overline{\mathcal{H}}_{(\alpha\gamma)}^*  \\  {[\overline{\mathcal{H}}_{(\alpha\gamma)}^*]}'  & \overline{\mathcal{H}}_{(\gamma\gamma)}^*
\end{array}\right) 
+  \frac{b}{\sqrt{NT}} \, vv' ,
   \label{ExpectedHessian-MAIN}
\end{align}
where $\overline{\mathcal{H}}_{(\alpha\alpha)}^* =  \text{diag}(  \sum_{t} \E[- \partial_{\pi^2} \ell_{it}])/\sqrt{NT}$,
$\overline{\mathcal{H}}_{(\alpha\gamma)it}^* = \E[- \partial_{\pi^2} \ell_{it}]/\sqrt{NT}$,
and
$\overline{\mathcal{H}}_{(\gamma\gamma)} ^*= \linebreak \text{diag}( \sum_{i} \E[- \partial_{\pi^2} \ell_{it}])/\sqrt{NT}$.
Furthermore, let
$\overline{\cal H}^{-1}_{(\alpha\alpha)}$, $\overline{\cal H}^{-1}_{(\alpha\gamma)}$,
$\overline{\cal H}^{-1}_{(\gamma\alpha)}$,
and $\overline{\cal H}^{-1}_{(\gamma\gamma)}$ denote the $N\times N$, $N\times T$, $T\times N$
and $T\times T$ blocks of the inverse $\overline{\cal H}^{-1}$ of $\overline{\cal H}$.
We define the $\dim \beta$-vector $\Xi_{it}$ and the operator  $D_{\beta \pi^q}$
as
\begin{align}
   \Xi_{it}
    & := - \frac 1 {\sqrt{NT}} \sum_{j=1}^N \sum_{\tau=1}^T \left( \overline{\cal H}^{-1}_{(\alpha\alpha)ij}
    +  \overline{\cal H}^{-1}_{(\gamma\alpha)tj}
    +  \overline{\cal H}^{-1}_{(\alpha\gamma)i\tau}
    +\overline{\cal H}^{-1}_{(\gamma\gamma)t\tau} \right)
     \; \mathbb{E}_\phi \left( \partial_{\beta \pi} \ell_{j\tau}
     \right) ,
   \nonumber \\
   D_{\beta \pi^q} \ell_{it} & :=  \partial_{\beta \pi^q} \ell_{it} -  \partial_{\pi^{q+1}} \ell_{it} \Xi_{it},
   \label{DefProXi}
\end{align}
with $q=0,1,2$. The $k$-th component of  $\Xi_{it}$ corresponds to the population least squares projection of
$\E(\partial_{\beta_k \pi} \ell_{it})/\E( \partial_{\pi^2} \ell_{it} )$
                        on the space spanned by the incidental parameters under a
metric given by $ \E( - \partial_{\pi^2} \ell_{it})$, i.e.
\begin{align*}
   \Xi_{it,k} &= \alpha^{\ast}_{i,k} + \gamma^{\ast}_{t,k} ,
   &
  \left(\alpha^{\ast}_{k} , \, \gamma^{\ast}_{k} \right)
   &= \argmin_{\alpha_{i,k},\gamma_{t,k}} \sum_{i,t}
  \E( - \partial_{\pi^2} \ell_{it} )
        \left(  \frac{\E( \partial_{\beta_k \pi} \ell_{it} )}
                        {\E( \partial_{\pi^2} \ell_{it} )}
         - \alpha_{i,k} - \gamma_{t,k} \right)^2.
\end{align*}
The operator $D_{\beta \pi^q}$ partials out individual and time effects in nonlinear models. It corresponds to  individual and time differencing when the model is linear. To see this, consider the normal  linear model $Y_{it} \mid X_{i}^t, \alpha_i,\gamma_t \sim {\cal N}(X_{it}'\beta + \alpha_i + \gamma_t, 1).$ Then, $\Xi_{it} = T^{-1} \sum_{t=1}^T \E[X_{it}] + N^{-1} \sum_{i=1}^N \E[X_{it}] - (NT)^{-1} \sum_{i=1}^N \sum_{t=1}^T \E[X_{it}]$,  $D_{\beta} \ell_{it}  = - \tilde X_{it} \varepsilon_{it}, $ $D_{\beta \pi} \ell_{it}  = - \tilde X_{it},$ and  $D_{\beta \pi^2} \ell_{it}  = 0$, where $\varepsilon_{it} = Y_{it} - X_{it}'\beta - \alpha_i - \gamma_t$ and $\tilde X_{it} = X_{it} - \Xi_{it}$ is the individual and time demeaned explanatory variables.

The following theorem establishes the asymptotic distribution of the fixed effects estimator $\widehat \beta.$

\begin{theorem}[Asymptotic distribution of $\widehat \beta$]
  \label{th:BothEffects}
Suppose that Assumption \ref{ass:PanelA1} holds, that the following limits exist
\begin{align*}
   \overline B_{\infty} &=
      \EE \left[ - \frac {1} {N}  \sum_{i=1}^{N}
            \frac{  \sum_{t=1}^T  \sum_{\tau=t}^T
        \E\left(
                \partial_{\pi} \ell_{it}  D_{\beta \pi} \ell_{i\tau}
                 \right)
                 + \frac 1 2  \sum_{t=1}^T
        \E ( D_{\beta \pi^2} \ell_{it} )   }
        {  \sum_{t=1}^T \E\left(  \partial_{\pi^2} \ell_{i t} \right) }  \right]  ,
        \nonumber \\
      \overline D_{\infty} &=    \EE \left[ -
         \frac {1} {T}  \sum_{t=1}^{T}
            \frac{  \sum_{i=1}^N
        \E\left(
                \partial_{\pi} \ell_{it} D_{\beta \pi} \ell_{it}
              +  \frac 1 2  D_{\beta \pi^2} \ell_{it}  \right)    }
        {  \sum_{i=1}^N \E\left(  \partial_{\pi^2} \ell_{i t} \right) } \right], \\
          \overline W_{\infty}  &=   \EE \left[ - \frac 1 {NT}  \sum_{i=1}^N
     \sum_{t=1}^T  \E \left(
            \partial_{\beta \beta'} \ell_{it}
              -  \partial_{\pi^2} \ell_{it} \Xi_{it} \Xi'_{it} \right) \right],
   \end{align*}
and that $\overline W_{\infty}>0$.
Then,
   \begin{align*}
      \sqrt{NT} \left( \widehat \beta - \beta^0 \right)
          \; \to_d \;
       \overline{W}_{\infty}^{-1} {\cal N}( \kappa \overline B_{\infty}
                    + \kappa^{-1}  \overline D_{\infty} ,
           \;\overline W_{\infty}),
    \end{align*}
    so that  $\overline B_{\infty}^{\beta} =   \overline W_{\infty}^{-1} \overline B_{\infty}$, $\overline D_{\infty}^{\beta} =   \overline W_{\infty}^{-1} \overline D_{\infty}$,
    and  $\overline V_{\infty} = \overline W_{\infty}^{-1}$
     in  \eqref{eq: expansion} and \eqref{eq:MainResult}.
\end{theorem}




\begin{remark}\label{remark:ad}
The complete proof of Theorem~\ref{th:BothEffects} is provided in the Appendix. Here we point out why the argument for the consistency proof in models with only individual effects does not apply to our setting, give a heuristic derivation of the asymptotic distribution, and highlight where some of the assumptions are used in the proof. 


\begin{itemize}

\item[(i)] The consistency proof for models with only individual effects relies on partitioning the log-likelihood in the sum of individual log-likelihoods that depend on a fixed number of parameters, the model parameter $\beta$ and the corresponding individual effect $\alpha_i$. The maximizers of the individual log-likelihood are then consistent estimators of all the parameters as $T$ becomes large  by standard arguments. This approach does not work in models with individual and time effects because there is no partition of the data that is only affected by a fixed number of parameters, and whose size grows with the sample size.

\item[(ii)] In the following we give a heuristic discussion of the asymptotic distribution result for $\widehat \beta$. A first-order Taylor series expansion to approximate the first order conditions of \eqref{DefEst} around $\beta^0$ gives
\begin{equation}\label{eq:heuristic_exp}
0 = \partial_\beta {\cal L}(\widehat \beta,\widehat \phi(\widehat \beta)) \approx \partial_\beta {\cal L}(\beta^0,\widehat \phi^0) - \overline W_{\infty} \sqrt{NT}(\widehat \beta - \beta^0),
\end{equation}
where $\widehat \phi^0 = \widehat \phi(\beta^0)$.
A second-order Taylor series expansion to approximate $\partial_\beta {\cal L}(\beta^0,\widehat \phi^0)$ around $\phi^0$ yields
$$
 \partial_{\beta} {\cal L}(\beta^0, \widehat \phi^0) \approx \partial_{\beta} {\cal L}(\beta^0,\phi^0) + \partial_{\beta \phi'} {\cal L}(\beta^0,\phi^0) [\widehat \phi^0 - \phi^0] + \sum_{g=1}^{\dim \phi} \partial_{\beta \phi' \phi_g}  {\cal L}(\beta^0,\phi^0) [\widehat \phi^0 - \phi^0][\widehat \phi_g^0 - \phi_g^0]/2,
$$
where the first term has zero mean and determines the asymptotic variance, and the second and third term determine the asymptotic bias. Thus, by the central limit theorem and the information equality,
$$
\partial_{\beta} {\cal L}(\beta^0,\phi^0)  \to_d \mathcal{N}(0, \overline W_{\infty}).
$$
The second and third terms satisfy
\begin{equation*}\label{eq:heuristic_exp2}
\partial_{\beta \phi'} {\cal L}(\beta^0,\phi^0) [\widehat \phi^0 - \phi^0] + \sum_{g=1}^{\dim \phi} \partial_{\beta \phi' \phi_g}  {\cal L}(\beta^0,\phi^0) [\widehat \phi^0 - \phi^0][\widehat \phi_g^0 - \phi_g^0]/2 \approx \sqrt{NT}(\overline B_{\infty}/T + \overline D_{\infty}/N),
\end{equation*}
where $\overline B_{\infty}$ and $\overline D_{\infty}$ are characterized from a  second-order Taylor series expansion to approximate $\widehat \phi^0$ around $\phi^0$.   We refer to the Appendix for the details of this derivation. There we show that  $\overline B_{\infty}$ and $\overline D_{\infty}$ originate from the elements of $\widehat \phi^0$ corresponding to the individual effects and time effects, respectively. Plugging those results into \eqref{eq:heuristic_exp}, and solving for $\sqrt{NT}(\widehat{\beta} - \beta^0)$ yields
$$
\sqrt{NT}(\widehat{\beta} - \beta^0) \approx \overline W_{\infty}^{-1} [\partial_{\beta} {\cal L}(\beta^0,\phi^0) + \overline B_{\infty}\sqrt{N/T}  + \overline D_{\infty} \sqrt{T/N} ] \to_d \overline W_{\infty}^{-1} \mathcal{N}(\kappa \overline B_{\infty} + \kappa^{-1} \overline D_{\infty}, W_{\infty}).
$$ 

This derivation shows that the source of the bias is that the score $\partial_{\beta} {\cal L}(\beta,\widehat \phi)$ is not centered at zero when $\beta = \beta^0$. This problem arises from the substitution of the incidental parameter $\phi$ by the sample analog $\widehat \phi^0$  that has a rate of convergence slower than $\sqrt{NT}$. Thus, $\overline B_{\infty}$ originates from the estimators of the individual effects in $\phi$, which have rate of convergence $\sqrt{T}$; whereas $\overline D_{\infty}$ originates from the estimators of the time effects in $\phi$, which have convergence rate $\sqrt{N}$.

\item[(iii)] The two key assumptions in the derivation of the asymptotic distribution are the additive separability of $\alpha_i$ and $\gamma_t$ in Assumption \ref{ass:PanelA1}(iii) and the concavity in Assumption \ref{ass:PanelA1}(v). We resort to concavity to prove consistency of $\widehat \beta$ and to bound the remainder terms in all the expansions. Additive separability is convenient to characterize the order of the inverse average Hessian, $\overline{\mathcal{H}}$, defined in \eqref{ExpectedHessian-MAIN}. This inverse Hessian features prominently in the second-order Taylor series expansion of $\widehat \phi$ around $\phi^0$ used to characterize $\overline B_{\infty}$ and $\overline D_{\infty}$.
\end{itemize}

\end{remark}

It is instructive to evaluate the expressions of the bias in our running examples.

\medskip

\noindent  \textbf{Example \ref{example: probit}} (Binary response model). \textit{In this case 
$$\ell_{it}(\beta,\pi) = Y_{it} \log F(X_{it}'\beta + \pi) + (1 - Y_{it}) \log [1 -  F(X_{it}'\beta + \pi)],$$ so that $\partial_{\pi} \ell_{it} = H_{it} (Y_{it} - F_{it}),$  $\partial_{\beta} \ell_{it} = \partial_{\pi} \ell_{it} X_{it},$ $ \partial_{{\pi}^2} \ell_{it} =   -  H_{it} \partial F_{it} + \partial H_{it} (Y_{it} - F_{it})$, $\partial_{\beta \beta'} \ell_{it} = \partial_{\pi^2} \ell_{it} X_{it} X_{it}'$,  $\partial_{\beta \pi} \ell_{it} = \partial_{\pi^2} \ell_{it} X_{it},$ $\partial_{{\pi}^3} \ell_{it} =  -  H_{it} \partial^2 F_{it} - 2 \partial H_{it} \partial F_{it} + \partial^2 H_{it} (Y_{it} - F_{it})$, and $\partial_{\beta {\pi}^2} \ell_{it} = \partial_{{\pi}^3} \ell_{it} X_{it}$, where $H_{it} = \partial F_{it} / [F_{it}(1-F_{it})], and $   $\partial^{j}  G_{it} := \partial^{j} G(Z)|_{Z = X_{it}'\beta^0 + \pi_{it}^0}$ for any function $G$ and $j = 0,1,2$. Substituting these values in the expressions of the bias of Theorem \ref{th:BothEffects} yields  
\begin{eqnarray*}
\overline B_{\infty} &=&    \EE \left[ -
         \frac {1} {2N}  \sum_{i=1}^{N}
            \frac{  \sum_{t=1}^T \left\{ \E[H_{it} \partial^2 F_{it} \tilde{X}_{it}] +  2 \sum_{\tau=t+1}^T
        \E \left[ H_{it}
               (Y_{it} - F_{it} ) \omega_{i\tau}  \tilde{X}_{i\tau}   \right] \right\} } {  \sum_{t=1}^T \E\left( \omega_{it} \right) } \right],\\
\overline D_{\infty} &=&  \EE \left[ -
         \frac {1} {2T}  \sum_{t=1}^{T}
            \frac{  \sum_{i=1}^N  \E[H_{it} \partial^2 F_{it} \tilde{X}_{it}] } {  \sum_{i=1}^N \E\left( \omega_{it} \right) }\right] ,   \\          
\overline W_{\infty} &=&     \EE \left[
         \frac {1} {NT}  \sum_{i=1}^{N} \sum_{t=1}^{T}
            \E[\omega_{it} \tilde{X}_{it} \tilde{X}_{it}']  \right],
\end{eqnarray*}
where $\omega_{it} = H_{it} \partial F_{it}$ and $\tilde X_{it}$ is the residual of the population projection of $X_{it}$ on the space spanned by the incidental parameters under a metric weighted by $\E(\omega_{it})$.  For the probit model where  all the components of $X_{it}$  are strictly exogenous,
$$
\overline B_{\infty} =    \EE \left[
         \frac {1} {2N}  \sum_{i=1}^{N}
            \frac{  \sum_{t=1}^T \E[\omega_{it} \tilde{X}_{it} \tilde{X}_{it}']  } {  \sum_{t=1}^T \E\left( \omega_{it} \right) }\right] \beta^0 , \ \ \overline D_{\infty} =    \EE \left[
         \frac {1} {2T}  \sum_{t=1}^{T}
            \frac{  \sum_{i=1}^N  \E[\omega_{it} \tilde{X}_{it} \tilde{X}_{it}'] } {  \sum_{i=1}^N \E\left( \omega_{it} \right) } \right] \beta^0. 
$$
The asymptotic bias  is therefore a positive definite matrix weighted average of the true parameter value as in the case of the probit model with only individual effects \cite{FernandezVal:2009p3313}.}

\medskip

\noindent  \textbf{Example \ref{example: poisson} } (Poisson model). \textit{In this case 
$$\ell_{it}(\beta,\pi) = (X_{it}'\beta + \pi) Y_{it} - \exp(X_{it}'\beta + \pi) - \log Y_{it}!,$$
 so that $\partial_{\pi} \ell_{it} = Y_{it} - \omega_{it},$  $\partial_{\beta} \ell_{it} =  \partial_{\pi} \ell_{it} X_{it},$ $\partial_{{\pi}^2} \ell_{it} = \partial_{{\pi}^3} \ell_{it}= - \omega_{it} $, $\partial_{\beta \beta'} \ell_{it} = \partial_{{\pi}^2} \ell_{it} X_{it} X_{it}',$ and $\partial_{\beta {\pi}} \ell_{it} = \partial_{\beta {\pi}^2} \ell_{it} = \partial_{{\pi}^3} \ell_{it} X_{it},$ where $\omega_{it} = \exp(X_{it}'\beta^0 + \pi_{it}^0)$. Substituting these values in the expressions of the bias of Theorem \ref{th:BothEffects} yields  
\begin{eqnarray*}
\overline B_{\infty} &=&   \EE \left[ -
         \frac {1} {N}  \sum_{i=1}^{N}
            \frac{  \sum_{t=1}^T  \sum_{\tau=t+1}^T
        \E\left[
               (Y_{it} - \omega_{it} ) \omega_{i\tau} \tilde X_{i\tau} \right]}
        {  \sum_{t=1}^T \E\left( \omega_{it} \right) } \right],\\
 \\          
\overline W_{\infty} &=&     \EE \left[
         \frac {1} {NT}  \sum_{i=1}^{N} \sum_{t=1}^{T}
            \E[\omega_{it} \tilde{X}_{it} \tilde{X}_{it}'] \right],        
\end{eqnarray*}
and $\overline D_{\infty} = 0$,  where $\tilde X_{it}$ is the residual of the population projection of $X_{it}$ on the space spanned by the incidental parameters under a metric weighted by $\E(\omega_{it})$.   If in addition
all the components of $X_{it}$ are strictly exogenous, then we get the no asymptotic bias result $\overline B_{\infty} = \overline D_{\infty} = 0$.}

\begin{remark}[Bias and Variance expressions for  Conditional Moment Models]
    \label{rem:NoBartlett}
In the derivation of the asymptotic distribution, we apply  Bartlett identities implied by  Assumption \ref{ass:PanelA1}$(iii)$ to simplify the expressions. The following expressions of the asymptotic bias and variance   do not make use of these identities and therefore remain valid in conditional moment models that do not specify the entire conditional distribution of $Y_{it}$:
\begin{align*}
   \overline B_{\infty} &=
     \EE \left[   -      \frac {1} {N}  \sum_{i=1}^{N}
            \frac{  \sum_{t=1}^T \sum_{\tau=t}^T
        \E\left(
               \partial_\pi \ell_{it} D_{\beta \pi} \ell_{i\tau} \right) }
        { \sum_{t=1}^T \E\left( \partial_{\pi^2} \ell_{i t} \right) }  \right]  \nonumber \\
       & \quad +  \, \frac 1 2  \,  \EE \left[ \,
        \frac 1 {N}  \sum_{i=1}^N
         \frac{   \sum_{t=1}^T
              \E [ ( \partial_{\pi} \ell_{it} )^2 ]
          \sum_{t=1}^T  \E ( D_{\beta \pi^2 } \ell_{it}  )  }
             {\left[  \sum_{t=1}^T  \E\left( \partial_{\pi^2} \ell_{i t} \right) \right]^2}  \right] \; ,
   \nonumber \\
      \overline D_{\infty} &=    \EE \left[ -
         \frac {1} {T}  \sum_{t=1}^{T}
            \frac{  \sum_{i=1}^N
        \E\left[
               \partial_{\pi} \ell_{it}  D_{\beta \pi} \ell_{it}  \right] }
        { \sum_{i=1}^N \E\left( \partial_{\pi^2} \ell_{i t} \right) } \right]
   \nonumber \\
       & \quad + \, \frac 1 2  \,  \EE \left[
        \frac 1 {T}  \sum_{t=1}^T
         \frac{   \sum_{i=1}^N
              \E [ ( \partial_{\pi} \ell_{it} )^2 ]
         \sum_{i=1}^N  \E (  D_{\beta \pi^2} \ell_{it}
             ) }
             {\left[  \sum_{i=1}^N  \E\left( \partial_{\pi^2} \ell_{i t} \right) \right]^2} \right],
   \nonumber \\
        \overline V_{\infty}  &= 
             \overline W_{\infty}^{-1} 
      \;   \overline \Omega_{\infty} \overline W_{\infty}^{-1}, \\
                 \overline \Omega_{\infty}  &= 
         \EE \left[ \frac 1 {NT}  \sum_{i=1}^N
       \sum_{t=1}^T  \sum_{\tau = 1}^T
     \E \left[
          D_{\beta} \ell_{it} (D_{\beta} \ell_{i\tau})'  \right] \right],
   \end{align*}
   and $ \overline W_{\infty}$ is the same as in Theorem \ref{th:BothEffects}.

For example, consider the Poisson fixed effects estimator in the conditional mean model $ \mathbb{E} [Y_{it} \mid X_{i}^t, \phi, \beta] = \omega_{it} = \exp( X_{it}'\beta + \alpha_i + \gamma_t)$. Applying the previous expressions to  $\ell_{it}(\beta,\pi) = (X_{it}'\beta + \pi) Y_{it} - \exp(X_{it}'\beta + \pi) - \log Y_{it}!$ yields the same expressions for $\overline B_{\infty}$, $\overline D_{\infty}$, $\overline W_{\infty}$ as in  Example 2, and 
$$
\overline \Omega_{\infty}  = 
             \EE \left[ \frac 1 {NT}  \sum_{i=1}^N
       \sum_{t=1}^T   
     \E \left[
          (Y_{it} - \omega_{it})^2  \tilde{X}_{it} \tilde{X}_{it}'   \right] \right],
$$
where $\tilde{X}_{it}$ is defined as in Example 2.  If all the components of $X_{it}$ are strictly exogenous, then we get again the no asymptotic bias result $\overline B_{\infty} = \overline D_{\infty} = 0$.
\end{remark}

\subsection{Asymptotic distribution of APEs}\label{subsec:apes}
In nonlinear models we are often interested in APEs, in addition to model parameters. These effects
are averages of the data, parameters and unobserved effects; see expression
\eqref{eq: meffs}. For the panel models of Assumption \ref{ass:PanelA1} we specify the partial effects as $\Delta(X_{it}, \beta, \alpha_i, \gamma_t) = \Delta_{it}(\beta, \pi_{it})$.
The restriction that the partial effects depend on $\alpha_i$ and $\gamma_t$ through $\pi_{it}$ is natural in our panel models since
$$
\mathbb{E} [Y_{it} \mid X^t_i, \alpha_i, \gamma_t, \beta] = \int y \exp [\ell_{it}(\beta, \, \pi_{it} )] dy,
$$
and the partial effects are usually defined as differences or derivatives of this conditional expectation with respect to the components of $X_{it}$. For example, the partial effects for the binary response and Poisson
models described in Section \ref{sec:model} satisfy this restriction.

The distribution of the unobserved individual and time effects is not ancillary for the APEs, unlike for model parameters. We therefore need to make  assumptions on this  distribution  to define and interpret the APEs, and to derive the asymptotic distribution of their estimators.  We control the heterogeneity of the partial effects assuming that the individual effects and explanatory variables are identically distributed cross sectionally and/or stationary over time. If $(X_{it}, \alpha_i, \gamma_t)$  is identically distributed over $i$ and can be heterogeneously distributed  over $t$, $\Ep[\Delta_{it}] = \delta_t^0$ and  $\delta_{NT}^0 = T^{-1} \sum_{t=1}^T \delta_t^0$ changes only with $T$. If $(X_{it}, \alpha_i, \gamma_t)$  is stationary over $t$ and can be heterogeneously distributed  over $i$,   $\Ep[\Delta_{it}] = \delta_i^0$ and  $\delta_{NT}^0 = N^{-1} \sum_{i=1}^N \delta_i^0$ changes only with $N$. Finally, if $(X_{it}, \alpha_i, \gamma_t)$  is identically distributed over $i$ and stationary  over $t$, $\Ep[\Delta_{it}] = \delta_{NT}^0$ and $\delta_{NT}^0 = \delta^0$ does not change with $N$ and $T.$ 
We also impose  smoothness and moment conditions on the function $\Delta$ that defines the partial effects. We use these conditions to derive  higher-order stochastic expansions for the fixed effect estimator of the APEs and to bound the remainder terms in these expansions.  Let $\{\alpha_i\}_N := \{\alpha_i : 1 \leq i \leq N\}$, $\{\gamma_t\}_T := \{\gamma_t : 1 \leq t \leq T\},$ and $\{X_{it}, \alpha_i, \gamma_t \}_{NT} := \{(X_{it}, \alpha_i, \gamma_t)  : 1 \leq i \leq N, 1 \leq t \leq T\}.$

\begin{assumption} [Partial effects]   \label{ass:PanelA2} ~
  Let $\nu>0$, $\epsilon>0$, and ${\cal B}^0_{\varepsilon}$ all be as in
  Assumption~\ref{ass:PanelA1}.
  
  \begin{itemize}
  \item[(i)] Sampling: for all $N,T,$  $\{X_{it}, \alpha_i, \gamma_t \}_{NT}$ is identically distributed across $i$ and/or stationary across $t$.\footnote{In the working paper version, 
  Fern\'andez-Val and Weidner~\cite*{ThisWorkingPaper2015}, we also consider inference conditional on the unobserved effects by assuming that $\{\alpha_i\}_N$ and $\{\gamma_t\}_T$ are deterministic sequences.}
  
  \item[(ii)] Model: for all $i,t,N,T,$
  the partial effects depend on $\alpha_i$ and $\gamma_t$ through  $\alpha_i + \gamma_t$:  
\begin{equation*}
\Delta(X_{it}, \beta, \alpha_i, \gamma_t) = \Delta_{it}(\beta, \alpha_i + \gamma_t).
\end{equation*}
The realizations of the partial effects are denoted by $ \Delta_{it} := \Delta_{it}(\beta^0, \alpha_i^0 + \gamma_t^0).$

  \item[(iii)] Smoothness and moments: The function
       $(\beta,\pi) \mapsto \Delta_{it}(\beta,\pi)$
        is
      four times continuously differentiable over ${\cal B}^0_{\varepsilon}$ a.s.
      The partial derivatives of  $\Delta_{it}(\beta,\pi)$
      with  respect to the elements of $(\beta,\pi)$
      up to fourth order
      are bounded in absolute value uniformly over $(\beta,\pi) \in {\cal B}^0_{\varepsilon}$
      by a function $M(Z_{it})>0$ a.s.,
      and $\max_{i,t} \E[M(Z_{it})^{8+\nu}]$
     is a.s. uniformly bounded over  $N,T$. 
     
    \item[(iv)] Non-degeneracy and moments:  $0 <  \min_{i,t} [\Ep (\Delta_{it}^2) - \Ep(\Delta_{it})^2] \leq \max_{i,t} [\Ep (\Delta_{it}^2) - \Ep(\Delta_{it})^2]  < \infty,$ uniformly over $N,T.$
 \end{itemize}
\end{assumption}

Analogous to $\Xi_{it}$ and $D_{\beta \pi^q} \ell_{it} $ in equation \eqref{DefProXi}  we define
\begin{align}
   \Psi_{it}
    &= - \frac 1 {\sqrt{NT}} \sum_{j=1}^N \sum_{\tau=1}^T \left(   \overline{\cal H}^{-1}_{(\alpha\alpha)ij}  
    +   \overline{\cal H}^{-1}_{(\gamma\alpha)tj} 
    +   \overline{\cal H}^{-1}_{(\alpha\gamma)i\tau} 
    +    \overline{\cal H}^{-1}_{(\gamma\gamma)t\tau}   \right)  \partial_{\pi} \Delta_{j\tau},
    \nonumber \\
   D_{\pi^q} \Delta_{it} & :=  \partial_{\pi^q} \Delta_{it} -  \partial_{\pi^{q+1}} \ell_{it} \; \E(\Psi_{it}), 
\end{align}
for $q \in \{1,2\}$.
Here, $\Psi_{it}$ is the population projection of
$ \partial_{\pi} \Delta_{it} /
                        \E[ \partial_{\pi^2} \ell_{it} ]$ 
                        on the space spanned by the incidental parameters under the
metric given by $ \E[ - \partial_{\pi^2} \ell_{it} ]$.
We use analogous notation to the previous section  for the derivatives with respect to $\beta$ and higher order derivatives with respect to $\pi$.

Let $\delta_{NT}^0$ and $\widehat \delta$ be the APE and its fixed effects estimator, defined as in equations  \eqref{eq: meffs} and  \eqref{eq:ape} with $\Delta(X_{it}, \beta, \alpha_i, \gamma_t) = \Delta_{it}(\beta, \alpha_i + \gamma_t).$\footnote{We keep the dependence of $\delta_{NT}^0$ on $NT$ to distinguish $\delta_{NT}^0$ from $\delta^0 = \lim_{N,T \to \infty} \delta_{NT}^0$.} 
The following theorem establishes the asymptotic distribution of $\widehat \delta.$

\begin{theorem}[Asymptotic distribution of $\widehat \delta$]
  \label{th:DeltaLimit}
   Suppose that the assumptions of Theorem~\ref{th:BothEffects}
   and Assumption~\ref{ass:PanelA2} hold, and that the following limits exist:\footnote{We thank Fa Wang for pointing out errors in the expressions  
          for $ \overline B_{\infty}^{\delta} $, $ \overline D_{\infty}^{\delta} $, and
          $\overline{V}_{\infty}^{\delta}$ in the published version of the paper.}
  \begin{align*}
       \overline {(D_{\beta} \Delta)}_{\infty} &= \EE \left[
       \frac 1 {NT} \sum_{i=1}^N \sum_{t=1}^T
         \E(  \partial_{\beta} \Delta_{it} -  \Xi_{it} \partial_{\pi} \Delta_{it}  ) \right],
       \nonumber \\
    \overline B_{\infty}^{\delta} &=
     \overline {(D_{\beta} \Delta)}_{\infty}' \overline W_\infty^{-1}  \overline B_{\infty}
      -  \EE \left[   \frac {1} {N}  \sum_{i=1}^{N}
            \frac{  \sum_{t=1}^T  \sum_{\tau=t}^T
        \E\left(
                \partial_{\pi} \ell_{it}  D_{\pi} \Delta_{i\tau}
                 \right)
                 + \frac 1 2  \sum_{t=1}^T
        \E ( D_{\pi^2} \Delta_{it} )   }
        {  \sum_{t=1}^T \E\left(  \partial_{\pi^2} \ell_{i t} \right) }  \right] ,
        \nonumber \\
      \overline D_{\infty}^{\delta} &=   \overline {(D_{\beta} \Delta)}_{\infty}' \overline W_\infty^{-1}  \overline D_{\infty} - 
       \EE \left[  
         \frac {1} {T}  \sum_{t=1}^{T}
            \frac{  \sum_{i=1}^N
        \E\left(
                \partial_{\pi} \ell_{it} D_{\pi} \Delta_{it}
              +  \frac 1 2  D_{\pi^2} \Delta_{it}  \right)    }
        {  \sum_{i=1}^N \E\left(  \partial_{\pi^2} \ell_{i t} \right) } \right], 
       \nonumber \\
          \overline{V}_{\infty}^{\delta} &= \EE \left\{
       \frac {r_{NT}^2} {N^2T^2}
        \Ep \left[  \left(\sum_{i=1}^N \sum_{t = 1}^T  \widetilde \Delta_{it} \right)\left(\sum_{i=1}^N \sum_{t = 1}^T  \widetilde \Delta_{it} \right)' +  \sum_{i=1}^N \sum_{t=1}^T \Gamma_{it} \Gamma_{it}'  + 2 \sum_{i=1}^N \left(\sum_{t = 1}^T  \widetilde \Delta_{it}\sum_{s = t+1}^T \Gamma_{is}' \right) \right] \right\},
  \end{align*}
for some deterministic sequence   $r_{NT} \to \infty$ such that $r_{NT} = \mathcal{O}(\sqrt{NT})$ and $ \overline{V}_{\infty}^{\delta} > 0,$  
where $ \widetilde \Delta_{it} = \Delta_{it} - \Ep(\Delta_{it})$ and $\Gamma_{it}=  \overline {(D_{\beta} \Delta)}_{\infty}' \overline W_\infty^{-1} D_{\beta} \ell_{it}
      -   \E( \Psi_{it} )
          \partial_{\pi} \ell_{it}$.
    Then,
    \begin{equation*}
    r_{NT} (\widehat \delta - \delta_{NT}^0 -  T^{-1} \overline B_{\infty}^{\delta}
                    - N^{-1}  \overline D_{\infty}^{\delta}) \to_d \mathcal{N}(0 ,
           \;\overline V_{\infty}^{\delta}) .
    \end{equation*}
\end{theorem}
 
 
\begin{remark}[Convergence rate, bias and variance] \label{remark: conv_rate} To understand the asymptotic distribution of $\widehat \delta$ is useful to decompose 
$$
\widehat \delta - \delta_{NT}^0 = [\widehat \delta - \delta] + [\delta  - \delta_{NT}^0],
$$ 
where $\delta := (NT)^{-1} \sum_{i=1}^N \sum_{t=1}^T \Delta_{it}$. In this decomposition the first term captures variation due to parameter estimation, whereas the second term captures variation due to estimation of a population mean by a sample mean.  Under Assumption~\ref{ass:PanelA2}(iv) the convergence rate $r_{NT}$ is determined by the convergence rate of  $\delta - \delta_{NT}^0$, which depends on the sampling properties of the unobserved effects. For example, if  $\{\alpha_i\}_N$  and  $\{\gamma_t\}_T$ are independent sequences, and  $\alpha_i$ and $\gamma_t$ are independent for all $i,t$, then $r_{NT} = \sqrt{NT/(N+T-1)}$, and
$$
          \overline{V}_{\infty}^{\delta} = \EE \left\{
       \frac {r_{NT}^2} {N^2 T^2}
       \sum_{i=1}^N \left[ \sum_{t, \tau = 1}^T \Ep( \widetilde \Delta_{it} \widetilde \Delta_{i \tau}' ) + \sum_{j \neq i}\sum_{t = 1}^T \Ep( \widetilde \Delta_{it} \widetilde \Delta_{j t} ') +  \sum_{t=1}^T \Ep(\Gamma_{it} \Gamma_{it}') +  2 \sum_{s>t} \Ep(\widetilde \Delta_{it} \Gamma_{is}')\right] \right\}.
$$
In the expression of $ \overline{V}_{\infty}^{\delta}$, the first two terms come from  $\delta  - \delta_{NT}^0$,  the third term comes from $\widehat \delta - \delta$, and the last term is the asymptotic covariance between $\delta  - \delta_{NT}^0$ and  $\widehat \delta - \delta$. The last term drops out when all the components of $X_{it}$ are strictly exogenous. 
The first two terms of $ \overline{V}_{\infty}^{\delta}$ are of order $NT(T + N - 1)r_{NT}^2/(NT)^2 = \mathcal{O}(1)$ by construction, the last term of  $ \overline{V}_{\infty}^{\delta}$ is of order $NT r_{NT}^2/(NT)^2 = \mathcal{O}(T^{-1}+N^{-1})$, and the asymptotic bias $ r_{NT} (T^{-1} \overline B_{\infty}^{\delta} + N^{-1}  \overline D_{\infty}^{\delta})$ is  of order $r_{NT}(T^{-1} + N^{-1}) = \mathcal{O}(T^{-1/2}+N^{-1/2})$. Thus, the bias and variance coming from parameter estimation are asymptotically negligible relative to the variances coming from the estimation of a population mean by a sample mean. In numerical examples, however, we find that correcting the mean and variance for parameter estimation improves the finite-sample estimation and inference properties of the APE estimators. 
\end{remark}

\begin{remark}[Average effects from bias corrected estimators] The first term in the expressions of the biases $ \overline B_{\infty}^{\delta}$ and $  \overline D_{\infty}^{\delta}$ comes from the bias of the estimator of $\beta$. It drops out when the APEs are 
constructed from asymptotically unbiased or bias corrected estimators of the parameter $\beta$, i.e.
\begin{equation*}
\widetilde \delta =  \Delta( \widetilde
\beta, \widehat \phi(\widetilde \beta)),
\end{equation*}
where $\widetilde \beta$ is such that $\sqrt{NT}(\widetilde \beta - \beta^0) \to_d N(0, \overline W_{\infty}^{-1})$. 
The asymptotic variance of $\widetilde \delta$ is the same as in Theorem \ref{th:DeltaLimit}.
\end{remark}



%

In the following examples we assume that the APEs are constructed from asymptotically unbiased estimators of the model parameters.

\noindent  \textbf{Example \ref{example: probit} } (Binary response model). \textit{Consider the partial effects defined in (\ref{example: probit: meff1}) and (\ref{example: probit: meff2}) with
$$
\Delta_{it}(\beta , \pi) =  F(\beta_k +
X_{it,-k}'\beta_{-k} +  \pi) -
F(X_{it,-k}'\beta_{-k} + \pi)  \text{ and }  \Delta_{it}(\beta, \pi) = \beta_k  \partial F(X_{it}'\beta +
\pi).
$$
Using the notation previously introduced for this example, the components of the asymptotic bias of $\widetilde \delta$ are
\begin{align*}
    \overline B_{\infty}^{\delta} &=  \EE \left[  \frac {1} {2N}  \sum_{i=1}^{N} \textstyle
            \frac{  \sum_{t=1}^T \left[  2 \sum_{\tau=t+1}^T
         \E\left( H_{it}(Y_{it} - F_{it})  \omega_{i\tau} \tilde \Psi_{i\tau}
                 \right)   - \E(\Psi_{it}) \E(H_{it} \partial^2 F_{it}) + \E(\partial_{\pi^2} \Delta_{it}) \right]}
        {  \sum_{t=1}^T \E\left( \omega_{it} \right) }\right],
       \\
\overline D_{\infty}^{\delta} &= \EE \left[
         \frac {1} {2T}  \sum_{t=1}^{T}
            \frac{  \sum_{i=1}^N
         \left[    - \E(\Psi_{it}) \E(H_{it} \partial^2 F_{it}) + \E(\partial_{\pi^2} \Delta_{it}) \right]}
        {  \sum_{i=1}^N \E\left(  \omega_{it} \right) } \right],
\end{align*}
where $\tilde \Psi_{it}$ is the residual of the population regression of
$ - \partial_{\pi} \Delta_{it} /
                        \E[ \omega_{it}]$ 
                        on the space spanned by the incidental parameters under the
metric given by $\E[\omega_{it}]$.  If all the components of $X_{it}$  are strictly exogenous,
the first term of $ \overline B_{\infty}^{\delta}$ is zero.}

\medskip

\noindent  \textbf{Example \ref{example: poisson} } (Poisson model). \textit{Consider the partial effect
$$
\Delta_{it}(\beta, \pi) = g_{it}(\beta) \exp(X_{it}'\beta + \pi),
$$
where $g_{it}$ does not depend on $\pi$. For example, $g_{it}(\beta) = \beta_k +  \beta_j
h(Z_{it})$ in  (\ref{example: poisson: meff}). Using the notation previously introduced for this example, the components of the asymptotic bias are
\begin{equation*}
\overline B_{\infty}^{\delta} =     \EE \left[
         \frac {1} {N}  \sum_{i=1}^{N}
            \frac{  \sum_{t=1}^T  \sum_{\tau=t+1}^T
        \E\left[
               (Y_{it} - \omega_{it} ) \omega_{i\tau} \tilde g_{i\tau} \right]}
        {  \sum_{t=1}^T \E\left( \omega_{it} \right) } \right] ,\\
\end{equation*}
and $\overline D_{\infty}^{\delta} = 0$,  where $\tilde g_{it}$ is the residual of the population projection of $g_{it}$ on the space spanned by the incidental parameters under a metric weighted by $\E[\omega_{it}]$. The asymptotic bias is zero if all the components of $X_{it}$ are strictly exogenous or $g_{it}(\beta)$ is constant. The latter arises in the leading case of the partial effect of the $k$-th component of $X_{it}$ since $g_{it}(\beta) = \beta_k$. This no asymptotic bias result applies to any type of regressor, strictly exogenous or predetermined.}

\subsection{Bias corrected estimators}
The results of the previous sections show that the asymptotic
distributions of the fixed effects estimators of the model parameters and APEs can have biases of the same
order as the  variances under sequences where $T$ grows at
the same rate as $N$. This is the large-$T$ version of the
incidental parameters problem that invalidates 
any inference based on the fixed effect estimators even in large samples.
In this section we describe how to construct  analytical and jackknife bias corrections for 
the fixed effect estimators and give conditions for the asymptotic validity of
these corrections. 

The jackknife correction for the model parameter $\beta$  in equation \eqref{eq: jackknife2} is generic and applies to the panel model. For the APEs, the jackknife correction is formed similarly as
\begin{equation*}\label{eq: jack_apes}
\widetilde{\delta}_{NT}^{J} = 3  \widehat \delta_{NT} - \widetilde{\delta}_{N,T/2} - \widetilde{\delta}_{N/2,T},
\end{equation*}
where $\widetilde{\delta}_{N,T/2}$ is the average of the 2 split jackknife
estimators of the APE that use all the individuals and leave out the first and second halves of the time
periods, and $\widetilde{\delta}_{N/2,T}$ is the average of the 2 split jackknife
estimators of the APE that use all the time periods and leave out half of the individuals. 


The analytical corrections are constructed using sample analogs of the expressions in Theorems \ref{th:BothEffects} and \ref{th:DeltaLimit}, replacing the true values of $\beta$ and $\phi$ by the  fixed effects estimators. To describe these corrections, we introduce some additional notation. For any function of the data, unobserved effects and parameters $g_{itj}(\beta,\alpha_i + \gamma_t,\alpha_i + \gamma_{t-j})$ with $0 \leq j < t$,  let $\widehat g_{itj} = g_{it}(\widehat \beta, \widehat \alpha_i + \widehat \gamma_t, \widehat \alpha_i + \widehat \gamma_{t-j})$ denote the fixed effects estimator, 
e.g., $\widehat{\E[\partial_{\pi^2} \ell_{it}]}$ denotes the fixed effects estimator of $\E[\partial_{\pi^2} \ell_{it}].$
Let $\widehat{\cal H}^{-1}_{(\alpha\alpha)}$, $\widehat{\cal H}^{-1}_{(\alpha\gamma)}$,
$\widehat{\cal H}^{-1}_{(\gamma\alpha)}$,
and $\widehat{\cal H}^{-1}_{(\gamma\gamma)}$ denote the 
blocks of the matrix $\widehat{\cal H}^{-1}$, where  
$$
\widehat{\cal H} =  
 \left(\begin{array}{cc}  \widehat{\mathcal{H}}_{(\alpha\alpha)}^* & \widehat{\mathcal{H}}_{(\alpha\gamma)}^*  \\ {[\widehat{\mathcal{H}}_{(\alpha\gamma)}^*]}' & \widehat{\mathcal{H}}_{(\gamma\gamma)}^*
\end{array}\right) 
+  \frac{b} {\sqrt{NT}} \, vv' ,$$
$\widehat{\mathcal{H}}_{(\alpha\alpha)}^* =  \text{diag}( -  \sum_{t} \widehat{\E[\partial_{\pi^2} \ell_{it}]})/\sqrt{NT}$, $\widehat{\mathcal{H}}_{(\alpha\alpha)} ^*=  \text{diag}( -  \sum_{i} \widehat{\E[\partial_{\pi^2} \ell_{it}]})/\sqrt{NT}$, and $\widehat{\mathcal{H}}_{(\alpha\gamma)it}^* =
-  \widehat{\E[\partial_{\pi^2} \ell_{it}]}/\sqrt{NT}$.
Let
\begin{align*}
  \widehat  \Xi_{it}
    &= - \frac 1 {\sqrt{NT}} \sum_{j=1}^N \sum_{\tau=1}^T \left( \widehat{\cal H}^{-1}_{(\alpha\alpha)ij}  
    +   \widehat{\cal H}^{-1}_{(\gamma\alpha)tj}  
    +    \widehat{\cal H}^{-1}_{(\alpha\gamma)i\tau}  
    +  \widehat{\cal H}^{-1}_{(\gamma\gamma)t\tau}   \right) 
     \; \widehat{\E \left( \partial_{\beta \pi} \ell_{j\tau} \right)}.
\end{align*}
The $k$-th component of  $\widehat \Xi_{it}$ corresponds to a least squares regression of
$ \widehat{\E \left( \partial_{\beta_k \pi} \ell_{it} \right)}/\widehat{\E( \partial_{\pi^2} \ell_{it} )}$ on the space spanned by the incidental parameters weighted by $ \widehat{\E( - \partial_{\pi^2} \ell_{it})}.$

The analytical bias corrected estimator of $\beta^0$ is 
\begin{equation}\label{eq: bias_parameter}
\widetilde \beta^A = \widehat \beta - \widehat B_{NT}^{\beta} /T - \widehat D_{NT}^{\beta}/N,
\end{equation}
where $\widehat B_{NT}^{\beta} = \widehat W^{-1} \widehat B$, $\widehat D_{NT}^{\beta} = \widehat W^{-1}  \widehat D$,
  \begin{align}
   \widehat{B} &=
      -    \frac {1} {N}  \sum_{i=1}^{N}
            \frac{   \sum_{j=0}^L [T/(T-j)]  \sum_{t=j+1}^{T}
        \widehat{\E\left(
                \partial_{\pi} \ell_{i,t-j}  D_{\beta \pi} \ell_{it}
                 \right)}
                 + \frac 1 2  \sum_{t=1}^T
        \widehat{\E ( D_{\beta \pi^2} \ell_{it} )}   }
        {  \sum_{t=1}^T \widehat{\E\left(  \partial_{\pi^2} \ell_{i t} \right) }}     ,
        \nonumber \\
      \widehat D &=  -  
         \frac {1} {T}  \sum_{t=1}^{T}
            \frac{  \sum_{i=1}^N \left[
        \widehat{\E\left(
                \partial_{\pi} \ell_{it} D_{\beta \pi} \ell_{it}\right)}
              +  \frac 1 2  \widehat{\E\left(D_{\beta \pi^2} \ell_{it}  \right) } \right]  }
        {  \sum_{i=1}^N \widehat{ \E\left(  \partial_{\pi^2} \ell_{i t} \right) }}, \nonumber \\
        \widehat W  &=  -  (NT)^{-1}  \sum_{i=1}^N
     \sum_{t=1}^T  \left[\widehat{\E \left(
            \partial_{\beta \beta'} \ell_{it}\right)}
              -   \widehat{\E \left(\partial_{\pi^2} \ell_{it} \Xi_{it} \Xi'_{it} \right)}\right], \label{eq: se_parameter}
   \end{align}
%
and $L$ is a trimming parameter for estimation of spectral expectations  such that $L \to \infty$ and $L/T \to 0$ \cite{HahnKuersteiner2011}.  Here we use truncation instead of kernel smoothing in the estimation of spectral expectations following Hahn and Kuersteiner \cite*{HK2007}. Note that, unlike for variance estimation, a kernel is not needed to ensure that the bias estimator be positive. Instead of choosing a value of $L$,  our recommendation for practice is to conduct a sensitivity analysis by reporting estimates for multiple values of $L$ starting from $L=1$.
From our experience  based on extensive Monte Carlo simulations, we do not recommend values of $L$ greater than $4$,  because the finite-sample dispersion of the estimator quickly increases with $L$. We refer to Section \ref{sec: MC} for an example of sensitivity analysis with respect to $L$. The factor $T/(T-j)$ is a degrees of freedom adjustment that rescales the time series averages $T^{-1} \sum_{t=j+1}^T$ by the number of observations instead of by $T$. Similar corrections  for conditional mean models can be formed using the sample analogs of the expressions of $\overline{B}_{\infty}$ and $\overline{D}_{\infty}$ in Remark \ref{rem:NoBartlett}.  We do not spell out these estimators for the sake of brevity. 

Asymptotic $(1-p)$--confidence intervals for the components of $\beta^0$ can be formed as
$$
\widetilde \beta_k^A \pm z_{1-p} \sqrt{\widehat W_{kk}^{-1} / (NT)},  \ \ k = \{1, ..., \dim \beta^0\},
$$
where $z_{1-p}$ is the $(1-p)$--quantile of the standard normal distribution, and $\widehat W_{kk}^{-1} $ is the $(k,k)$-element of the 
matrix $\widehat W^{-1}$. In conditional moment models we replace $\widehat W_{kk}$ by the $(k,k)$-element of the matrix $\widehat W^{-1} \widehat \Omega \widehat W^{-1}$, where
$$
\widehat \Omega = \frac 1 {NT}  \sum_{i=1}^N
       \sum_{t=1}^T \sum_{\tau=1}^T  
     \widehat{\E \left[
          D_{\beta} \ell_{it} (D_{\beta} \ell_{i\tau})'  \right]}.
$$
  
We have implemented the analytical correction at the level of the estimator.  Alternatively,  we can implement the correction at the level of the score or first order conditions by solving 
\begin{equation}\label{eq : sbc}
(NT)^{-1/2} \partial_{\beta} {\cal L}(\beta, \widehat \phi(\beta)) = \widehat B/T + \widehat D/N,
\end{equation}
for $\beta$. 
Global concavity of the objective function guarantees that the solution to \eqref{eq : sbc}  is unique. Other possible extensions such as continuously updated score corrections where $\overline B_{\infty}$ and $\overline D_{\infty}$ are estimated together with $\beta$, corrections at the level of the objective function, or iterative corrections are left to future research.

The analytical bias corrected estimator of $\delta^0_{NT}$ is 
$$
\widetilde \delta^A = \widehat \delta - \widehat B^\delta/T - \widehat D^\delta/N,
$$
where $\widetilde \delta$ is the APE constructed from a bias corrected estimator of $\beta$. Let
\begin{align*}
  \widehat \Psi_{it}
    &= - \frac 1 {\sqrt{NT}} \sum_{j=1}^N \sum_{\tau=1}^T \left(
     \widehat {\cal H}^{-1}_{(\alpha\alpha)ij}  
    +     \widehat  {\cal H}^{-1}_{(\gamma\alpha)tj} 
     +   \widehat  {\cal H}^{-1}_{(\alpha\gamma)i\tau}   
         +   \widehat  {\cal H}^{-1}_{(\gamma\gamma)t\tau}     \right)  \widehat{\partial_{\pi} \Delta_{j\tau}}.
\end{align*}
The fixed effects estimators of the components of the asymptotic bias  are
\begin{align*}
       \widehat B^{\delta} &=
      -    \frac {1} {N}  \sum_{i=1}^{N}
            \frac{   \sum_{j=0}^L [T/(T-j)]  \sum_{t=j+1}^{T}
        \widehat{\E\left(
                \partial_{\pi} \ell_{i,t-j}  D_{\pi} \Delta_{it}
                 \right)}
                 + \frac 1 2  \sum_{t=1}^T
        \widehat{\E ( D_{\pi^2} \Delta_{it} )}   }
        {  \sum_{t=1}^T \widehat{\E\left(  \partial_{\pi^2} \ell_{i t} \right) }}     ,
        \nonumber \\
        \widehat D^\delta  &=  -  
         \frac {1} {T}  \sum_{t=1}^{T}
            \frac{  \sum_{i=1}^N \left[
        \widehat{\E\left(
                \partial_{\pi} \ell_{it} D_{\pi} \Delta_{it}\right)}
              +  \frac 1 2  \widehat{\E\left(D_{\pi^2} \Delta_{it}  \right) } \right]  }
        {  \sum_{i=1}^N \widehat{ \E\left(  \partial_{\pi^2} \ell_{i t} \right) }}.
\end{align*}
The estimator of the asymptotic variance depends on the sampling properties of the unobserved effects. Under the independence assumption  of Remark  \ref{remark: conv_rate}  with all the components of $X_{it}$  strictly exogenous, 
 \begin{equation}\label{eq: se_ameff2}
         \widehat{V}^{\delta} = 
       \frac {r_{NT}^2} {N^2 T^2}
       \sum_{i=1}^N  \left[  \sum_{t,\tau = 1}^T \widehat{\tilde \Delta}_{it} \widehat{\tilde \Delta}_{i\tau}' +  \sum_{t = 1}^T \sum_{j \neq i}  \widehat{\tilde \Delta}_{it} \widehat{\tilde \Delta}_{jt}' + \sum_{t=1}^{T} \widehat{\E( \Gamma_{it} \Gamma_{it}' )} \right],
       \end{equation}
  where $  \widehat{\tilde \Delta}_{it}  = \widehat \Delta_{it} - N^{-1} \sum_{i=1}^N \widehat \Delta_{it}$ under identical distribution over $i$, $  \widehat{\tilde \Delta}_{it}  = \widehat \Delta_{it} - T^{-1} \sum_{t=1}^T \widehat \Delta_{it}$ under stationarity over $t$, and $ \widehat{\tilde \Delta}_{it} = \widehat \Delta_{it} - \widehat \delta$ under both. 
Note that we do not need to specify the convergence rate $r_{NT}$ to make inference because the standard errors $\sqrt{\widehat{V}^{\delta}}/r_{NT}$ do not depend on $r_{NT}$. Bias corrected estimators and confidence intervals can be constructed in the same fashion as for the model parameter. 

We use the following homogeneity assumption to show the validity of the jackknife corrections for the model parameters and APEs.
It implies that $ \widetilde{\beta}_{N,T/2} - \widehat \beta_{NT} = \overline B_{\infty}^{\beta}/T + o_P(T^{-1} \vee N^{-1})$ and  $\widetilde{\beta}_{N/2,T} - \widehat \beta_{NT} = \overline D_{\infty}^{\beta}/N + o_P(T^{-1} \vee N^{-1})$, which are weaker but higher level sufficient conditions for the validity of the jackknife for the model parameter.  For APEs, Assumption \ref{assumption: stationarity} also ensures that these effects do not change with $T$ and $N$, i.e. $\delta_{NT}^0 = \delta^0$. The analytical corrections
 {\it do not} require this assumption.   
\begin{assumption}[Unconditional homogeneity] \label{assumption: stationarity} The sequence $\{(Y_{it}, X_{it}, \alpha_i, \gamma_t): 1 \leq i \leq N, 1 \leq t \leq T\}$ is  identically distributed across $i$ and strictly stationary
across $t,$ for each $N,T.$
\end{assumption}
This assumption might seem  restrictive for dynamic models where  $X_{it}$ includes lags of the dependent variable because in this case it restricts the unconditional distribution of the initial conditions of $Y_{it}$. Note, however, that Assumption \ref{assumption: stationarity} allows  the initial conditions to depend on the unobserved effects. In other words, it does not impose that the initial conditions are generated from the stationary distribution of  $Y_{it}$ conditional on $X_{it}$ and $\phi$. Assumption \ref{assumption: stationarity} rules out time trends and structural breaks in the processes for the unobserved effects and observed variables.  

\begin{remark}[Test of homogeneity] Assumption \ref{assumption: stationarity}  is a sufficient condition for the validity of the jackknife corrections. It has the testable implications that the probability limits of the fixed effects estimator are the same in all the partitions of the panel. For example, it implies that $\beta_{N,T/2}^1 =\beta_{N,T/2}^2$, where $\beta_{N,T/2}^1$ and $\beta_{N,T/2}^2$ are the probability limits of the fixed effects estimators of $\beta$ in the subpanels that include all the individuals and the first and second halves of the time periods, respectively.   These implications can be tested using  variations of the Chow-type test proposed in Dhaene and Jochmans  \cite*{DhaeneJochmans2015}. We provide an example of the application of these tests to our setting in Section  \ref{sec: empirics} of the supplemental material.
\end{remark}

The following theorems are the main result of this section. They show that the analytical and jackknife  bias corrections eliminate the bias from the asymptotic
distribution of the fixed effects estimators of the model parameters and APEs without increasing variance, and that the estimators of the asymptotic variances  are consistent. 

\begin{theorem} [Bias corrections for $\widehat \beta$] \label{th:bc} Under the conditions of Theorems \ref{th:BothEffects}, 
$$
\widehat W \to_P \overline{W}_{\infty},
$$
and, if $L \to \infty$ and $L/T \to 0,$
$$
\sqrt{NT}(\widetilde \beta^A - \beta^0)  \to_d \mathcal{N}(0,
\overline{W}_{\infty}^{-1}).
$$
Under the conditions of Theorems \ref{th:BothEffects} and Assumption \ref{assumption: stationarity},
$$
\sqrt{NT}(\widetilde \beta^J - \beta^0)  \to_d \mathcal{N}(0,
\overline{W}_{\infty}^{-1}).
$$
\end{theorem}

\begin{theorem} [Bias corrections for $\widehat \delta$] \label{th:bc_ameff} Under the conditions of Theorems \ref{th:BothEffects} and \ref{th:DeltaLimit}, 
$$
\widehat V^{\delta} \to_P \overline{V}^{\delta}_{\infty},
$$
and, if  $L \to \infty$ and $L/T \to 0,$ 
$$
r_{NT}(\widetilde \delta^A - \delta_{NT}^0)  \to_d \mathcal{N}(0,
\overline V_{\infty}^{\delta}).
$$
Under the conditions of Theorems \ref{th:BothEffects} and \ref{th:DeltaLimit},  and Assumption \ref{assumption: stationarity},
$$
r_{NT}(\widetilde \delta^{J} - \delta^0)  \to_d \mathcal{N}(0,
\overline V_{\infty}^{\delta}).
$$
\end{theorem}

\begin{remark}[Rate of convergence] The rate of convergence  $r_{NT}$ depends on the properties of the sampling process for the explanatory variables and unobserved effects (see remark \ref{remark: conv_rate}).

\end{remark}

\section{Monte Carlo Experiments}\label{sec: MC}

This section reports evidence on the finite sample behavior
of fixed effects estimators of model parameters and APEs in static models with strictly exogenous regressors and dynamic models with predetermined regressors such as lags of the dependent variable. We analyze
the performance of uncorrected and bias-corrected
fixed effects estimators in terms of bias and inference accuracy of
their asymptotic distribution.  
In particular we compute the biases, standard deviations, and root mean squared errors of the estimators, the ratio of
average standard errors to the simulation standard deviations (SE/SD); and  the empirical coverages of confidence intervals  with 95\% nominal value (p; .95).\footnote{The standard errors are computed using the expressions (\ref{eq: se_parameter})  and (\ref{eq: se_ameff2}) with $ \widehat{\tilde \Delta}_{it} = \widehat \Delta_{it} - \widehat \delta$,  evaluated at uncorrected estimates of the parameters. We find little difference in performance of constructing standard errors based on corrected estimates.} Overall, we find that the analytically corrected estimators dominate the uncorrected and jackknife corrected estimators.\footnote{Kristensen and Salani{\'e}~\cite*{KS2013} also found that analytical corrections dominate jackknife corrections to reduce the bias of approximate estimators.}  A possible explanation for the better finite-sample performance of the analytical over the jackknife corrections is that the jackknife increases dispersion because the components of the bias are estimated from subsamples that include half of the observations of the panel.  We observe this variance increase in all our numerical examples, specially in short panels.
The jackknife corrections are also more sensitive than the analytical corrections to  Assumption \ref{assumption: stationarity}. All the results are based on 500 replications. 
The designs correspond to static and dynamic probit models. 
As in the analytical example of Section \ref{subsec: ns}, we find that our large $T$ asymptotic approximations capture well the behavior of the fixed effects estimator and the bias corrections in moderately long panels with $N=56$ and $T=14$.


\subsection{Static probit model}

The data generating process is
\begin{equation*}
Y_{it} = \mathbf{1}\left\{  X_{it} \beta + \alpha _{i} +
\gamma_{t} > \varepsilon _{it} \right\} ,
\ \ (i=1,...,N; \
t=1,...,T),
\end{equation*}%
where $\alpha_{i} \sim \mathcal{N}(0,1/16)$, $\gamma_{t} \sim
\mathcal{N}(0,1/16)$, $\varepsilon_{it} \sim \mathcal{N}(0,1)$, and
$\beta = 1$. We consider two alternative designs for $X_{it}$:
autoregressive process and linear trend process both with individual and time effects. In
the first design,  $X_{it} = X_{i,t-1} / 2 + \alpha_{i} +
\gamma_{t} + \upsilon_{it}$, $\upsilon_{it} \sim
\mathcal{N}(0,1/2)$, and $X_{i0} \sim \mathcal{N}(0,1)$. In the
second design, $X_{it} = 2 t / T + \alpha_{i} +
\gamma_{t} + \upsilon_{it}$,
$\upsilon_{it} \sim \mathcal{N}(0,3/4)$, which violates  Assumption \ref{assumption: stationarity}. In  both designs $X_{it}$ is strictly exogenous with respect to $\varepsilon_{it}$ conditional on the individual and time effects. The variables $\alpha_i$, $\gamma_t$,
$\varepsilon_{it}$, $\upsilon_{it}$, and $X_{i0}$ are independent
and  $i.i.d.$ across individuals and time periods.  We generate
panel data sets with $N=56$ individuals and three different numbers
of time periods $T$: 14, 28 and 56.\footnote{Following a suggestion from an anonymous referee, we obtained results for panel data sets with $T=56$ and $N$ in $\{14, 28, 56\}$. These results are similar to the results reported and are available from the authors upon request.}

Table 3 reports the results for the probit coefficient $\beta$, and the APE of $X_{it}$. We compute the  APE using  (\ref{example: probit: meff2}).
Throughout the table, MLE-FETE corresponds to the probit maximum
likelihood estimator with individual and  time fixed effects, Analytical is the bias corrected estimator that uses the analytical correction, 
and Jackknife is the bias corrected estimator that
uses SPJ in both the individual and time
dimensions. The cross-sectional division in the jackknife follows the
order of the observations. All the results  are
reported in percentage of the true parameter value.

We find that the bias is of the same order of magnitude as the standard deviation for the uncorrected estimator of the probit coefficient causing severe undercoverage  of the confidence intervals. This result holds for both designs and all the sample sizes considered. The bias corrections, specially Analytical, remove the bias without increasing dispersion, and produce substantial improvements in rmse and coverage probabilities. For example, Analytical reduces rmse by  50\%  and increases coverage by 26\% in the first design with $T=14$.  As in Hahn and Newey~\cite*{Hahn:2004p882} and Fernandez-Val~\cite*{FernandezVal:2009p3313}, we find very little bias in the uncorrected estimates of the APE, despite the large bias in the probit coefficients.  Jackknife performs relatively worse in the second design that does not satisfy Assumption~\ref{assumption: stationarity}.

\subsection{Dynamic probit model}

The data generating process is
\begin{eqnarray*}
Y_{it} &=& \mathbf{1}\left\{  Y_{i,t-1} \beta_Y + Z_{it} \beta_Z + \alpha _{i} +
\gamma_{t} > \varepsilon _{it} \right\} ,
\ \ (i=1,...,N; \
t=1,...,T),\\
Y_{i0} &=&  \mathbf{1}\left\{  Z_{i0} \beta_Z + \alpha _{i} +
\gamma_{0} > \varepsilon _{i0} \right\},
\end{eqnarray*}%
where $\alpha_{i} \sim \mathcal{N}(0,1/16)$, $\gamma_{t} \sim
\mathcal{N}(0,1/16)$, $\varepsilon_{it} \sim \mathcal{N}(0,1)$,
$\beta_Y = 0.5$, and $\beta_Z = 1$.  We consider two alternative designs for $Z_{it}$:
autoregressive process and linear trend process both with individual and time effects. In
the first design,  $Z_{it} = Z_{i,t-1} / 2 + \alpha_{i} +
\gamma_{t} + \upsilon_{it}$, $\upsilon_{it} \sim
\mathcal{N}(0,1/2)$, and $Z_{i0} \sim \mathcal{N}(0,1)$. In the
second design, $Z_{it} = 1.5 t / T + \alpha_{i} +
\gamma_{t} + \upsilon_{it}$,
$\upsilon_{it} \sim \mathcal{N}(0,3/4)$, which violates  Assumption \ref{assumption: stationarity}.
 The variables $\alpha_i$, $\gamma_t$,
$\varepsilon_{it}$, $\upsilon_{it}$, and $Z_{i0}$ are independent
and  $i.i.d.$ across individuals and time periods.  We generate
panel data sets with $N=56$ individuals and three different numbers
of time periods $T$: 14, 28 and 56.

Table 4 reports the simulation results for the probit coefficient $\beta_Y$ and the APE of $Y_{i,t-1}$. We compute the partial effect  of $Y_{i,t-1}$ using the
expression in equation (\ref{example: probit: meff1}) with $X_{it,k} = Y_{i,t-1}$.
This effect is commonly reported as a measure of state dependence for dynamic binary processes. Table 5 reports the simulation results for the estimators of the probit coefficient $\beta_Z$ and the APE  of $Z_{it}$. We compute the partial effect using  (\ref{example: probit: meff2}) with $X_{it,k} = Z_{it}$.
Throughout the tables, we compare the same estimators as for the static model.  For the analytical correction we consider two versions, Analytical (L=1) sets the trimming parameter to estimate spectral expectations $L$ to one, whereas Analytical (L=2) sets $L $ to two.\footnote{In results not reported for brevity, we find little difference in performance of increasing  the trimming parameters to $L=3$ and $L=4$. These results are available from the authors upon request.} 
Again, all the results in the tables are
reported in percentage of the true parameter value.

The results in table 4 show important biases toward zero for \textit{both} the probit coefficient and the APE of $Y_{i,t-1}$ in the two designs. This bias can indeed be substantially larger than the corresponding standard deviation for short panels yielding coverage probabilities below 70\% for $T=14$. The analytical corrections significantly  reduce  biases and  rmse, bring coverage probabilities close to their nominal level, and have little sensitivity to the trimming parameter $L$.   The jackknife corrections reduce bias but increase dispersion, producing less drastic improvements in rmse and coverage than the analytical corrections. The results for  the APE of $Z_{it}$ in table 5 are similar to the static probit model. There are significant bias and undercoverage of  confidence intervals for the coefficient $\beta_Z$, which are removed by the corrections, whereas there are little bias and undercoverage in the APE. As in the static model,  Jackknife performs relatively worse in the second design.

\begin{table}[ht] 
\begin{center}
\scalebox{0.9}{\includegraphics{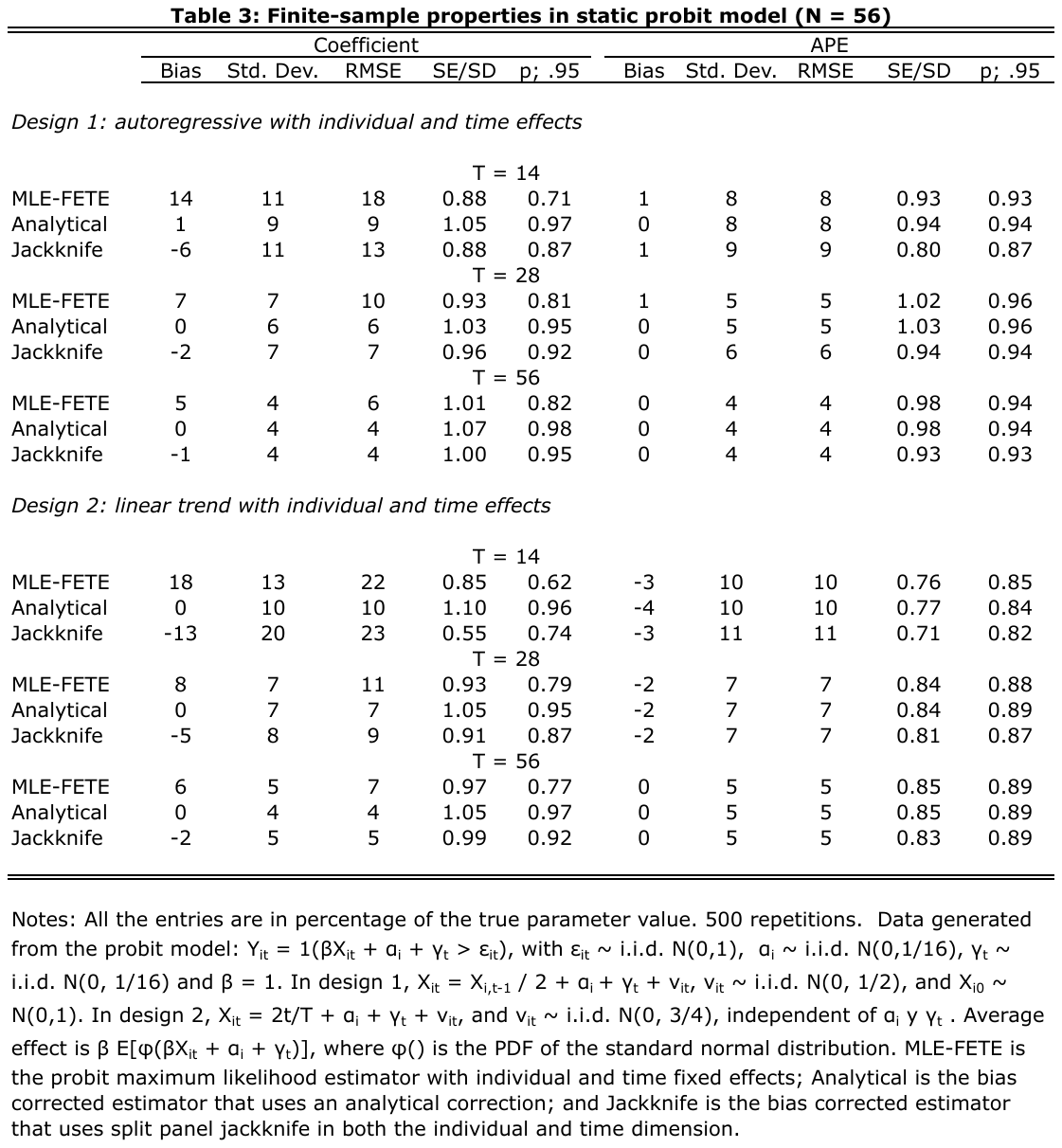}}
\end{center}
\end{table}

\begin{table}[ht]
\begin{center}
\scalebox{0.85}{\includegraphics{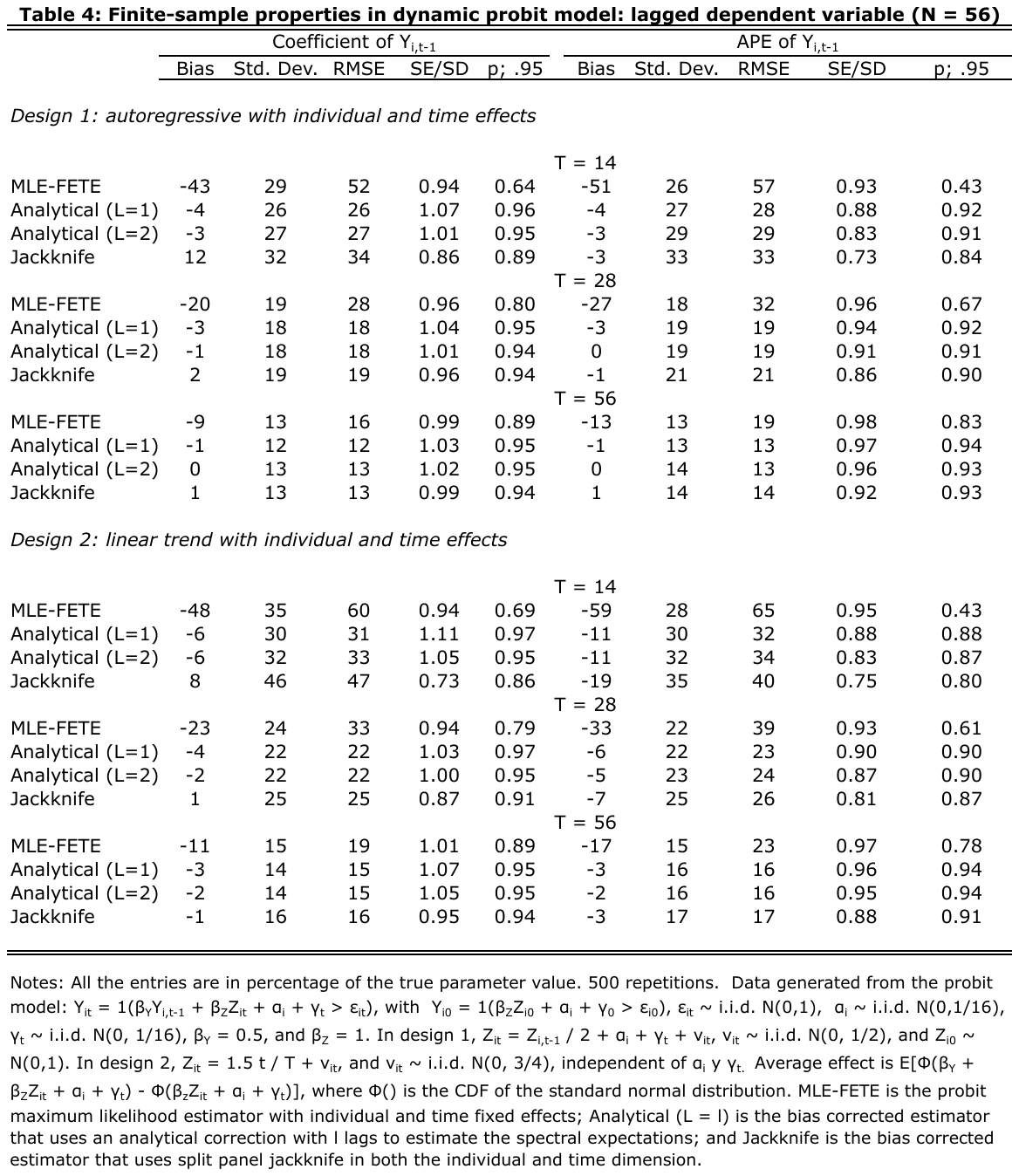}}
\end{center}
\end{table}

\begin{table}[ht]
\begin{center}
\scalebox{0.85}{\includegraphics{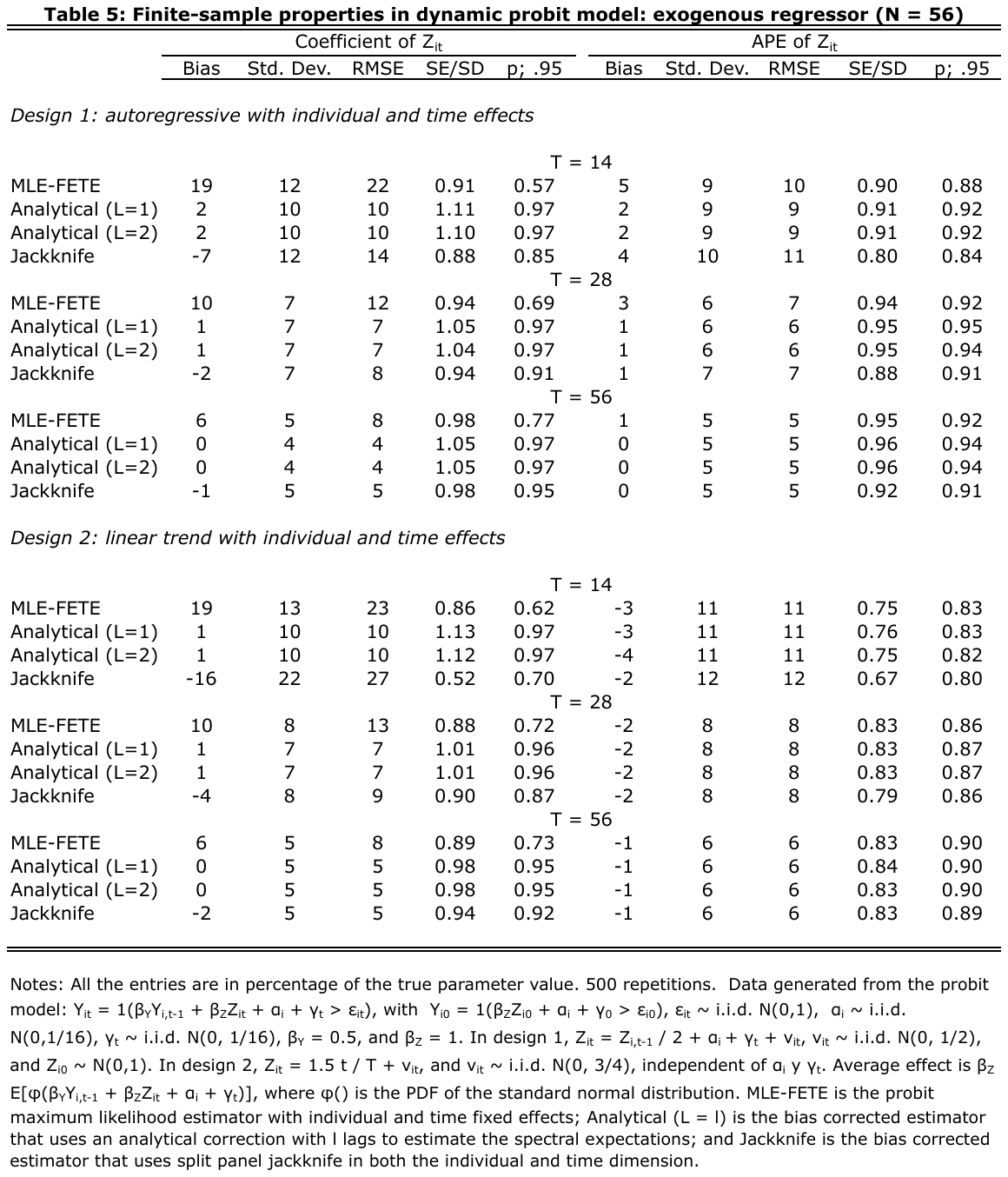}}
\end{center}
\end{table}

\section{Concluding remarks}

In this paper we develop analytical and jackknife corrections for fixed effects estimators of model parameters and APEs in semiparametric nonlinear panel models with additive individual and time effects. Our analysis applies to conditional maximum likelihood estimators with concave log-likelihood functions, and therefore covers logit, probit, ordered probit, ordered logit, Poisson, negative binomial, and Tobit estimators, which are the most popular nonlinear estimators in empirical economics. 

We are currently developing similar corrections for nonlinear models with interactive individual and time effects (Chen, Fern\'andez-Val, and Weidner \cite*{CFW2014}). Another interesting avenue of future research is to derive higher-order expansions for fixed effects estimators with individual and time effects. These expansions are needed to justify theoretically the validity of  alternative corrections based on the leave-one-observation-out panel jackknife method of Hahn and Newey~\cite*{Hahn:2004p882}.

\newpage

\bigskip

\begin{appendix}

\begin{center}
\begin{LARGE}
\textbf{Appendix}
\end{LARGE}
\end{center}

\section{Notation and Choice of Norms} 
\label{app:notation}

We write $A'$ for the transpose of a matrix or vector $A$.
We use $\mathbbm{1}_n$ for the $n\times n$ identity matrix, 
and $1_n$ for the column vector of length $n$ whose entries are all
unity. 
For
square $n \times n$ matrices $B$, $C$, we use $B>C$ (or $B\geq C$) to indicate
that $B-C$ is positive (semi) definite. 
We write wpa1 for ``with probability approaching one'' and wrt for ``with respect to''. All the limits are taken as
$N,T \to \infty$ jointly.


As in the main text, we usually suppress the dependence on $NT$
 of all the sequences of functions and parameters to lighten the notation, e.g. we write ${\cal L}$ for ${\cal L}_{NT}$ 
 and $\phi$ for $\phi_{NT}$. 
 Let
\begin{align*}
       {\cal S}(\beta,\phi) &= \partial_{\phi}  {\cal L}(\beta,\,\phi),  &
      {\cal H}(\beta,\phi) &=  - \partial_{\phi\phi'}  {\cal L}(\beta,\, \phi),   
\end{align*}
where $\partial_x f$ denotes the partial derivative of $f$ with respect to $x$, and additional subscripts denote higher-order partial derivatives.
We refer to the $\dim \phi$-vector ${\cal S}(\beta,\phi)$
as the incidental parameter score, and to
the $\dim \phi \times \dim \phi$ matrix ${\cal H}(\beta,\phi)$
as the incidental parameter Hessian.
We omit the
arguments of the functions when they are evaluated at the true
parameter values $(\beta^0, \, \phi^0)$, e.g.
${\cal H}={\cal H}(\beta^0, \phi^0)$.
We use a bar to indicate expectations conditional on $\phi$,
e.g. $\partial_{\beta} \overline {\cal L} =\mathbb{E}_\phi[ \partial_{\beta} {\cal L}]$,
and a tilde to denote variables in deviations with respect to expectations, e.g.
$\partial_{\beta} \widetilde{\cal L} = \partial_{\beta} {\cal L} - \partial_{\beta} \overline {\cal L}$.
 
We  use the Euclidian norm $\|.\|$ for vectors of dimension $\dim \beta$, and we use the norm induced by the Euclidian norm for the corresponding matrices and tensors, 
which we also denote by $\|.\|$. For matrices of dimension $\dim \beta \times \dim \beta$ this induced
norm is the spectral norm. The generalization of the spectral norm to higher order tensors
is straightforward, e.g. the induced norm of the $\dim \beta \times \dim \beta \times \dim \beta$ tensor
of third partial derivatives of ${\cal L}(\beta,\phi)$ wrt $\beta$ is given by
\begin{align*}
    \left\|  \partial_{\beta \beta \beta} {\cal L}(\beta,\phi) \right\|
      &= \max_{\left\{ u,v  \in \mathbb{R}^{\dim \beta}, \, \|u\|=1, \, \|v\|=1
       \right\}}
         \left\|   \sum_{k,l=1}^{\dim \beta}   
         u_k \, v_l \,
         \partial_{\beta \beta_k  \beta_l} {\cal L}(\beta,\phi)
         \right\| .
\end{align*}
This choice of norm is immaterial for the asymptotic analysis because $\dim \beta$ is fixed with the sample size.


In contrast, it  is important what norms we choose for vectors of dimension $\dim \phi$, and their corresponding matrices and tensors, because $\dim \phi$ is increasing with the sample size. For vectors of dimension $\dim \phi$,
we use the $\ell_q$-norm
\begin{align*}
    \| \phi \|_q = \left( \sum_{g=1}^{\dim \phi}  | \phi_g |^q \right)^{1/q} ,
\end{align*}
where $2 \leq q \leq \infty$.\footnote{We use the letter $q$ instead of $p$ to avoid confusion with the use of $p$ for probability.}
The particular value  $q=8$ will be chosen later.\footnote{The main reason not to choose $q= \infty$ is the assumption
 $\| \widetilde {\cal H} \|_q =  o_P(1)$ below, which is used to guarantee that 
 $\|  {\cal H}^{-1} \|_q$ is of the same order as  $\|  \overline {\cal H}^{-1} \|_q$.
If we assume  $\|  {\cal H}^{-1} \|_q = {\cal O}_P(1)$ directly instead of $\|  {\overline{\cal H}}^{-1} \|_q = {\cal O}_P(1)$, then we 
can set $q=\infty$.} We use the norms that are induced by the $\ell_q$-norm for the corresponding matrices and tensors, e.g. the induced $q$-norm of the $\dim \phi \times \dim \phi \times \dim \phi$ tensor
of third partial derivatives of ${\cal L}(\beta,\phi)$ wrt $\phi$ is 
\begin{align}
    \left\|  \partial_{\phi \phi \phi} {\cal L}(\beta,\phi) \right\|_q
      &= \max_{\left\{ u,v  \in \mathbb{R}^{\dim \phi}, \, \|u\|_q=1, \, \|v\|_q=1
       \right\}}
         \left\|   \sum_{g,h=1}^{\dim \phi}   
         u_g \, v_h \,
         \partial_{\phi \phi_g  \phi_h} {\cal L}(\beta,\phi)
         \right\|_q .
    \label{DefNorm}     
\end{align}
Note that in general the ordering of the indices of the tensor would matter in the definition of this norm,
with the first index having a special role. However, since partial derivatives like 
$\partial_{\phi_g \phi_h \phi_l} {\cal L}(\beta,\phi)$ are fully symmetric in the indices $g$, $h$, $l$, the ordering is
not important in their case.

For mixed partial derivatives of ${\cal L}(\beta,\phi)$ wrt $\beta$ and $\phi$, we use the norm that is induced
by the Euclidian norm on $\dim \beta$-vectors and the $q$-norm on $\dim \phi$-indices,
e.g.
\begin{align}
    \left\|  \partial_{\beta \beta \phi \phi \phi} {\cal L}(\beta,\phi) \right\|_q
      &= 
       \max_{\left\{ u,v  \in \mathbb{R}^{\dim \beta}, \, \|u\|=1, \, \|v\|=1
       \right\}}
       \max_{\left\{ w,x  \in \mathbb{R}^{\dim \phi}, \, \|w\|_q=1, \, \|x\|_q=1
       \right\}}
    \nonumber \\ & \qquad  \qquad    \qquad 
         \left\|  \sum_{k,l=1}^{\dim \beta}   \sum_{g,h=1}^{\dim \phi}   
         u_k \, v_l \, w_g \, x_h \,
         \partial_{\beta_k \beta_l \phi \phi_g  \phi_h} {\cal L}(\beta,\phi)
         \right\|_q ,
    \label{DefNorm2}     
\end{align}
where we continue to use the notation $\|.\|_q$, even though this is  a mixed norm.

Note that for $ w,x  \in \mathbb{R}^{\dim \phi}$  and $q \geq 2$, 
\begin{align*}
   |w' x| \leq  \| w \|_q \|x\|_{q/(q-1)} \leq  (\dim \phi)^{(q-2)/q} \| w \|_q \| x \|_q.
\end{align*}   
Thus, whenever we bound a scalar product of vectors, matrices and tensors in terms of the above
norms we have to account for this additional factor
$(\dim \phi)^{(q-2)/q}$.
For example, 
\begin{align*}
   & \left| \sum_{k,l=1}^{\dim \beta}   \sum_{f,g,h=1}^{\dim \phi}   
         u_k \, v_l \, w_f \, x_h \, y_f \,
         \partial_{\beta_k \beta_l \phi_f \phi_g  \phi_h} {\cal L}(\beta,\phi)
   \right| 
   \leq      
   (\dim \phi)^{(q-2)/q} 
   \| u \| \,
   \| v \| \,
   \| w \|_q \,
   \| x \|_q \,
   \| y \|_q \,
    \left\|  \partial_{ \beta \beta \phi \phi \phi} {\cal L}(\beta,\phi) \right\|_q .
\end{align*}
For higher-order tensors, we use the notation 
$ \partial_{\phi \phi \phi} {\cal L}(\beta,\phi)$  inside the $q$-norm $\|.\|_q$ defined above, while we rely on
standard index and matrix notation for all other expressions involving those partial derivatives, e.g. 
$\partial_{\phi \phi' \phi_g} {\cal L}(\beta,\phi)$ is a $\dim \phi \times \dim \phi$ matrix for every
$g=1,\ldots, \dim \phi$.
Occasionally, e.g. in Assumption~\ref{ass:A1}$(vi)$ below, we use the 
Euclidian norm for $\dim \phi$-vectors, and the spectral
norm for $\dim \phi \times \dim \phi$-matrices,  denoted by $\|.\|$,
and defined as $\|.\|_q$ with $q=2$.  Moreover, we employ  the
matrix infinity norm $\left\| A \right\|_\infty
  = \max_i \sum_j |A_{ij}|$, and the matrix maximum norm
 $\left\| A \right\|_{\max} = \max_{ij} |A_{ij}|$ to characterize the properties of the inverse of the expected Hessian of the incidental parameters in Section~\ref{sec:InverseH}.

For $r \geq 0$, we define the sets
${\cal B}(r,\beta^0)=
            \left\{ \beta: \|\beta-\beta^0\| \leq r \right\}$,
and
${\cal B}_q(r, \phi^0)=
            \left\{  \phi : 
              \|\phi-\phi^0\|_q  \leq r \right\}$,
which are closed balls of radius $r$ around the true parameter values
$\beta^0$ and $\phi^0$, respectively.

\section{Asymptotic Expansions} 
\label{app:expansion}

In this section, we derive asymptotic expansions for the score of
the profile objective function, ${\cal L}(\beta,\widehat \phi(\beta)),$ and for the fixed effects estimators
of the parameters and APEs, $\widehat \beta$ and $\widehat \delta$. We do not employ the panel
structure of the model, nor the particular form of the objective
function given in Section \ref{sec:app_panel}. Instead, we consider the
estimation of an unspecified model based on a sample of size $NT$
and a generic objective function ${\cal L}(\beta,\phi)$, which
depends on the parameter of interest $\beta$ and the incidental
parameter $\phi$. The estimators $\widehat \phi(\beta)$ and $\widehat \beta$ are
defined in \eqref{DefEstPhi} and \eqref{DefEst}. The proof of all the results in this Section are given in the supplementary material.

We make the following
high-level assumptions. These assumptions  might appear somewhat
abstract, but will be justified by  more primitive conditions in
the context of panel models. 
\begin{assumption}[Regularity conditions for asymptotic expansion of $\widehat \beta$]
  \label{ass:A1}
     Let $q>4$ and $0 \leq \epsilon < 1/8 - 1/(2q)$. 
     Let $r_\beta = r_{\beta,NT} >0$,
     $r_\phi = r_{\phi,NT}>0$,
     with $r_\beta = o\left[ (NT)^{-1/(2q)-\epsilon} \right]$
     and $r_\phi = o\left[ (NT)^{ -\epsilon}  \right]$.
      We assume that
   \begin{itemize}
      \item[(i)] $\frac{\dim \phi} {\sqrt{NT}} \rightarrow a$, $0<a<\infty$.
       \item[(ii)]  $(\beta,\phi) \mapsto {\cal L}(\beta,\, \phi)$ is four times continuously
            differentiable in ${\cal B}(r_\beta, \beta^0) \times {\cal B}_q(r_\phi, \phi^0)$, wpa1.

   \item[(iii)] $\displaystyle \sup_{\beta \in {\cal B}(r_\beta ,\beta^0)}
             \left\| \widehat \phi(\beta) - \phi^0 \right\|_q
             = o_P(r_\phi)$.

   \item[(iv)]  $\overline {\cal H} > 0$, and
     $\left\| \overline {\cal H}^{-1} \right\|_q
               =  {\cal O}_P\left(  1\right)$.

   \item[(v)]  
       For the $q$-norm defined in Appendix \ref{app:notation},
         \begin{align*}
              \| {\cal S} \|_q &= {\cal O}_P \left( (NT)^{-1/4 + 1/(2q)} \right) , 
            &
             \|   \partial_{\beta} {\cal L} \| &= {\cal O}_P(1) ,
            &        
           \| \widetilde {\cal H} \|_q &=  o_P(1) ,
              \\
                \left\| \partial_{\beta \phi'}  {\cal L} \right\|_q &=    {\cal O}_P \left( (NT)^{1/(2q)} \right) ,
            &
               \left\| \partial_{\beta \beta'}  {\cal L} \right\| &= {\cal O}_P(\sqrt{NT}) ,
             &
             \left\| \partial_{\beta \phi \phi}  {\cal L} \right\|_q    &=  {\cal O}_P( (NT)^{\epsilon} )  ,
             \\
                \left\|  \partial_{\phi \phi \phi}  {\cal L} \right\|_q   &=     {\cal O}_P\left(  (NT)^{\epsilon} \right) ,
        \end{align*}
        and
        \begin{align*}
             \sup_{\beta \in {\cal B}(r_\beta, \beta^0)}  \sup_{\phi \in {\cal B}_q(r_\phi, \phi^0)} 
        \left\|  \partial_{\beta \beta \beta}  {\cal L}(\beta,\, \phi) \right\|
                 &= {\cal O}_P\left( \sqrt{NT} \right) ,
                 \\
             \sup_{\beta \in {\cal B}(r_\beta, \beta^0)}  \sup_{\phi \in {\cal B}_q(r_\phi, \phi^0)}  
        \left\|  \partial_{\beta \beta \phi}  {\cal L}(\beta,\, \phi) \right\|_q
                 &=   {\cal O}_P\left(  (NT)^{1/(2q)} \right) ,
                 \\
             \sup_{\beta \in {\cal B}(r_\beta, \beta^0)}  \sup_{\phi \in {\cal B}_q(r_\phi, \phi^0)} 
        \left\|  \partial_{\beta \beta \phi \phi}  {\cal L}(\beta,\, \phi) \right\|_q
                 &=   {\cal O}_P\left(  (NT)^{\epsilon} \right) ,
                 \\
             \sup_{\beta \in {\cal B}(r_\beta, \beta^0)}  \sup_{\phi \in {\cal B}_q(r_\phi, \phi^0)}  
        \left\|  \partial_{\beta \phi \phi \phi}  {\cal L}(\beta,\, \phi) \right\|_q
                 &=  {\cal O}_P\left(  (NT)^{\epsilon} \right) ,
                 \\
            \sup_{\beta \in {\cal B}(r_\beta, \beta^0)}  \sup_{\phi \in {\cal B}_q(r_\phi, \phi^0)}  \left\| \partial_{\phi \phi \phi \phi} {\cal L}(\beta,\phi) \right\|_q    &=    {\cal O}_P\left(  (NT)^{\epsilon} \right) .
         \end{align*}

     \item[(vi)] For the spectral norm $\|.\|=\|.\|_2$,
        \begin{eqnarray*}
               \| \widetilde {\cal H} \| = o_P \left( (NT)^{-1/8} \right) , \ \ 
               \left\| \partial_{\beta \beta'} \widetilde {\cal L}  \right\|  = o_P( \sqrt{NT} ) , \ \ 
                                        \left\| \partial_{\beta \phi \phi}  \widetilde {\cal L} \right\|
                 =   o_P \left( (NT)^{-1/8} \right) ,\\
                               \left\| \partial_{\beta \phi'}  \widetilde {\cal L} \right\|
                 =   {\cal O}_P \left( 1 \right) , \ \ 
                \left\|  \sum_{g,h=1}^{\dim \phi} 
                \partial_{\phi \phi_g \phi_h} \widetilde {\cal L} \,
                [ \overline {\cal H}^{-1}  {\cal S} ]_g
                [ \overline {\cal H}^{-1}  {\cal S} ]_h
                   \right\|   
                     =   o_P \left( (NT)^{-1/4}  \right) .
        \end{eqnarray*}      

   \end{itemize}
\end{assumption}

 Let 
$\partial_{\beta} \mathcal{L}(\beta,\widehat \phi(\beta))$ be the score of the profile objective function.\footnote{Note that  $\frac{d} {d \beta} \mathcal{L}(\beta,\widehat \phi(\beta)) = \partial_{\beta} \mathcal{L}(\beta,\widehat \phi(\beta))$ by the envelope theorem.}
The following theorem is the main result of this appendix.

\begin{theorem}[Asymptotic expansions of $\widehat \phi(\beta)$ and $\partial_{\beta} \mathcal{L}(\beta,\widehat \phi(\beta))$]
   \label{th:ScoreExpansion}
   Let Assumption \ref{ass:A1} hold. Then
   \begin{align*}
      \widehat \phi(\beta) - \phi^0 &=  {\cal H}^{-1} {\cal S}
          +   {\cal H}^{-1}
            [\partial_{\phi \beta'}  {\cal L}] (\beta-\beta^0)
 + \ft 1 2 {\cal H}^{-1} \sum_{g=1}^{\dim \phi}
                  [\partial_{\phi \phi' \phi_g}   {\cal L}  ]   {\cal H}^{-1} {\cal S}
                   [   {\cal H}^{-1} {\cal S} ]_g
                       + R^{\phi}(\beta) ,
   \end{align*}
   and
   \begin{align*}
      \partial_{\beta} {\cal L}(\beta, \widehat \phi(\beta))
          &=   U
            - \overline W \, \sqrt{NT} (\beta-\beta^0)
            + R(\beta) ,
   \end{align*}
    where $U= U^{(0)}
            + U^{(1)}$, and
   \begin{align*}
      \overline W &= - \, \frac 1 {\sqrt{NT}} \,
                \left( \partial_{\beta \beta'} \overline {\cal L}
                     + [\partial_{\beta \phi'} \overline {\cal L}] \; \overline {\cal H}^{-1} \;
             [\partial_{\phi \beta'} \overline {\cal L}]  \right) ,
      \nonumber \\
      U^{(0)} &=   
                      \partial_{\beta} {\cal L}
                   + [\partial_{\beta \phi'} \overline {\cal L} ]\, \overline {\cal H}^{-1} {\cal S} ,
   \nonumber \\
      U^{(1)} &=  
      [\partial_{\beta \phi'} \widetilde{\cal L}] \overline {\cal H}^{-1} {\cal S}
        -  [\partial_{\beta \phi'} \overline {\cal L}] \,
                      \overline {\cal H}^{-1} \, \widetilde {\cal H} \,
                      \overline {\cal H}^{-1} \, {\cal S} 
    +  
             \frac 1 2 \,    \sum_{g=1}^{\dim \phi}
           \left( \partial_{\beta \phi' \phi_g} \overline {\cal L}
               +  [\partial_{\beta \phi'} \overline {\cal L}] \, \overline {\cal H}^{-1}
            [\partial_{\phi \phi' \phi_g} \overline {\cal L}]  \right)
              [\overline {\cal H}^{-1} {\cal S}]_g
              \overline {\cal H}^{-1} {\cal S} .
   \end{align*}
   The remainder terms of the expansions satisfy
   \begin{align*}
      \sup_{\beta \in {\cal B}(r_\beta ,\beta^0)}
        \frac{   (NT)^{1/2-1/(2q)} \, \left\| R^\phi(\beta) \right\|_q}
            {1 + \sqrt{NT} \|\beta-\beta^0\|}
              &= o_P \left( 1 \right) \; ,
    &
      \sup_{\beta \in {\cal B}(r_\beta ,\beta^0)}
        \frac{\| R(\beta) \|} {1 + \sqrt{NT} \|\beta-\beta^0\|} &= o_P(1) \; .
   \end{align*}
\end{theorem}

\begin{remark} 
    The result for $\widehat \phi(\beta) - \phi^0$  does not rely on  Assumption~\ref{ass:A1}$(vi)$. Without this assumption we can also show that
          \begin{align*}
         \partial_\beta {\cal L}(\beta,\widehat \phi(\beta)) 
        &= \partial_\beta {\cal L} +
           \left[  \partial_{\beta \beta'} {\cal L}
                                          + ( \partial_{\beta \phi'} {\cal L} ) {\cal H}^{-1}
                                            ( \partial_{\phi' \beta} {\cal L} ) \right] (\beta-\beta^0)
              +  (\partial_{\beta \phi'}  {\cal L})  {\cal H}^{-1} {\cal S}
      \nonumber \\ & \quad        
              + \frac 1 2 \sum_g  \left( \partial_{\beta \phi' \phi_g}  {\cal L}
               +  [\partial_{\beta \phi'}  {\cal L}] \,  {\cal H}^{-1}
            [\partial_{\phi \phi' \phi_g}  {\cal L}]  \right)
              [ {\cal H}^{-1} {\cal S}]_g
               {\cal H}^{-1} {\cal S}
              +  R_1 (\beta) ,
      \end{align*}
      with $R_1 (\beta)$ satisfying the same bound as $R (\beta)$.
      Thus, the spectral norm bounds in Assumption~\ref{ass:A1}$(vi)$ for $\dim \phi$-vectors, matrices and tensors
      are only used after separating expectations from deviations of expectations for certain partial derivatives.
      Otherwise, the derivation of the bounds is purely based on the $q$-norm for 
      $\dim \phi$-vectors, matrices and tensors.     
\end{remark}

The proofs are given in Section~\ref{sec:ExpansionProofs} of the supplementary material. 
 Theorem~\ref{th:ScoreExpansion} characterizes  asymptotic expansions
for the incidental parameter estimator 
and the score of the profile objective function in the incidental parameter score ${\cal S}$ up to quadratic order.
The theorem provides bounds on the
the remainder terms $R^{\phi}(\beta)$ and $R(\beta)$,
which make the expansions applicable to estimators of
$\beta$ that take values  within a shrinking $r_{\beta}$-neighborhood of $\beta^0$ wpa1.
Given such an $r_{\beta}$-consistent estimator $\widehat \beta$ that solves the first order condition $\partial_{\beta} {\cal L}(\beta, \widehat \phi(\beta)) = 0$, we can use the expansion
of the profile objective score to obtain an asymptotic expansion for $\widehat \beta$. This gives rise
to the following corollary of Theorem~\ref{th:ScoreExpansion} . Let $\overline W_{\infty} := \lim_{N,T \to \infty} \overline W$.
\begin{corollary}[Asymptotic expansion of $\widehat \beta$]
   \label{cor:LimitBeta}
   Let Assumption \ref{ass:A1} be satisfied. In addition, let $U = {\cal O}_P(1)$, let
    $\overline W_{\infty}$ exist with $\overline W_{\infty} > 0$,
   and let
   $\| \widehat \beta - \beta^0\| = o_P(r_\beta)$.
   Then
   $$\sqrt{NT} (\widehat \beta - \beta^0) = \overline W_{\infty}^{-1} U+ o_P(1).$$
\end{corollary}


The following theorem states that for strictly concave objective functions no separate consistency proof
is required for $\widehat \phi(\beta)$ and for $\widehat \beta$.
\begin{theorem}[Consistency under Concavity]
     \label{th:consistency}
     Let Assumption~\ref{ass:A1}$(i)$, $(ii)$, $(iv)$, $(v)$ and $(vi)$ hold, and let $(\beta,\phi) \mapsto {\cal L}(\beta,\phi)$ be strictly
     concave over $(\beta,\phi) \in \mathbb{R}^{\dim \beta + \dim \phi}$, wpa1.
     Assume furthermore that
      $(NT)^{-1/4+1/(2q)} = o_P(r_\phi)$
     and $(NT)^{1/(2q)} r_\beta = o_P(r_\phi)$.     
     Then,
     $$\displaystyle \sup_{\beta \in {\cal B}(r_\beta ,\beta^0)}
             \left\| \widehat \phi(\beta) - \phi^0 \right\|_q
             = o_P(r_\phi),$$
             i.e. Assumption~\ref{ass:A1}$(iii)$ is satisfied. 
     If, in addition, 
    $\overline W_{\infty}$ exists with $\overline W_{\infty} > 0$,
    then  $\| \widehat \beta - \beta^0\| = {\cal O}_P\left(  (NT)^{-1/4}  \right)$.
\end{theorem}

In the application of Theorem~\ref{th:ScoreExpansion}  to panel models, we focus on estimators with strictly concave objective functions. By Theorem  \ref{th:consistency},
we only need to check  Assumption~\ref{ass:A1}$(i)$, $(ii)$, $(iv)$, $(v)$ and $(vi)$,
as well as $U={\cal O}_P(1)$ and $\overline W_{\infty} > 0$, when we apply 
Corollary~\ref{cor:LimitBeta} to derive the limiting distribution of $\widehat \beta$.
We give the proofs of Corollary~\ref{cor:LimitBeta} and Theorem~\ref{th:consistency}
 in Section~\ref{sec:ExpansionProofs}.

\subsubsection*{Expansion for Average Effects}

We invoke the following
high-level assumption, which is verified  under more primitive
conditions for panel data models  in the next section.

\begin{assumption}[Regularity conditions for asymptotic expansion of $\widehat \delta$]
  \label{ass:A2}
     Let $q$, $\epsilon$, $r_\beta$ and $r_\phi$
     be defined as in Assumption~\ref{ass:A1}.
      We assume that
   \begin{itemize}
       \item[(i)]  $(\beta,\phi) \mapsto \Delta(\beta,\, \phi)$ is three times continuously
            differentiable in ${\cal B}(r_\beta, \beta^0) \times {\cal B}_q(r_\phi, \phi^0)$, wpa1.

   \item[(ii)]  $\left\| \partial_{\beta}  \Delta \right\| = {\cal O}_P(1),$
              $\left\| \partial_{\phi}  \Delta \right\|_q =  
                  {\cal O}_P \left( (NT)^{1/(2q)-1/2} \right),$ $\left\| \partial_{\phi \phi}  \Delta \right\|_q    =  {\cal O}_P( (NT)^{\epsilon-1/2} ),$
        and
        \begin{align*}
             \sup_{\beta \in {\cal B}(r_\beta, \beta^0)}  \sup_{\phi \in {\cal B}_q(r_\phi, \phi^0)} 
        \left\|  \partial_{\beta \beta}  \Delta(\beta,\, \phi) \right\|
                 &= {\cal O}_P\left( 1 \right) ,
                 \\
             \sup_{\beta \in {\cal B}(r_\beta, \beta^0)}  \sup_{\phi \in {\cal B}_q(r_\phi, \phi^0)}  
        \left\|  \partial_{\beta \phi'}  \Delta(\beta,\, \phi) \right\|_q
                 &=   {\cal O}_P\left(  (NT)^{1/(2q)-1/2} \right) ,
                 \\
             \sup_{\beta \in {\cal B}(r_\beta, \beta^0)}  \sup_{\phi \in {\cal B}_q(r_\phi, \phi^0)}  
        \left\|  \partial_{\phi \phi \phi}  \Delta(\beta,\, \phi) \right\|_q
                 &=  {\cal O}_P\left(  (NT)^{\epsilon-1/2} \right) .
         \end{align*}

     \item[(iii)] $\left\| \partial_{\beta} \widetilde \Delta  \right\|  = o_P( 1 ),$
             $\left\| \partial_{\phi}  \widetilde \Delta \right\|
                 =   {\cal O}_P \left( (NT)^{-1/2} \right),$ and $\left\| \partial_{\phi \phi}  \widetilde \Delta \right\|
                 =   o_P \left( (NT)^{-5/8} \right).$

   \end{itemize}
\end{assumption}

The following result gives the asymptotic expansion for the estimator,
$\widehat \delta = \Delta(\beta,\widehat \phi(\beta)),$
 wrt $\delta =  \Delta(\beta^0,\phi^0)$.

\begin{theorem}[Asymptotic expansion of $\hat \delta$]
   \label{th:DeltaExpansion}
   Let Assumptions \ref{ass:A1} and \ref{ass:A2} hold
   and let
   $\| \widehat \beta - \beta^0 \| = {\cal O}_P\left( (NT)^{-1/2} \right)
   = o_P\left( r_\beta \right)$.
    Then
   \begin{align*}
    \widehat \delta - \delta 
   &=  \left[ \partial_{\beta'}  \overline {\Delta}
      +
      (\partial_{\phi'} \overline \Delta)  
        \overline {\cal H}^{-1} (\partial_{\phi \beta'} \overline {\cal L})
    \right] (\widehat \beta - \beta^0) 
     + U^{(0)}_\Delta
     + U^{(1)}_\Delta  + o_P\left( 1/ \sqrt{NT} \right),
   \end{align*}
   where
   \begin{align*}   
       U^{(0)}_\Delta &=  (\partial_{\phi'} \overline \Delta)
           \overline {\cal H}^{-1}  {\cal S} ,
     \\ 
       U^{(1)}_\Delta
        &=  (\partial_{\phi'} \widetilde \Delta) \overline {\cal H}^{-1} {\cal S}
           - (\partial_{\phi'} \overline \Delta)
            \overline {\cal H}^{-1}  \widetilde {\cal H}
              \overline {\cal H}^{-1}  {\cal S}  
    \nonumber \\ & \quad              
             + \ft 1 2 \,  {\cal S}'  \overline {\cal H}^{-1}
         \left[\partial_{\phi \phi'} \overline \Delta + 
         \sum_{g=1}^{\dim \phi}
                  \left[\partial_{\phi \phi' \phi_g} \overline {\cal L} \right]
                  \left[ \overline {\cal H}^{-1} 
                     (\partial_{\phi} \overline \Delta) \right]_g \right]
          \overline {\cal H}^{-1} {\cal S}.
\end{align*}

\end{theorem}

\begin{remark}
The expansion of the profile score 
$\partial_{\beta_k} {\cal L}(\beta,\widehat \phi(\beta))$ in
Theorem~\ref{th:ScoreExpansion}
is a special case of the expansion in
Theorem~\ref{th:DeltaExpansion}, for
$\Delta(\beta,\phi) = \frac 1 {\sqrt{NT}} \partial_{\beta_k} {\cal L}(\beta,\phi)$.
Assumptions~\ref{ass:A2} also exactly match with the corresponding
subset of Assumption~\ref{ass:A1}.
\end{remark}

\section{Proofs of Section \ref{sec:app_panel}}
\label{app:ProofMain}

\subsection{Application of General Expansion to Panel Estimators}

We now apply the
 general expansion of appendix~\ref{app:expansion} to the panel fixed effects estimators considered in the main text.
For the objective function specified in \eqref{population_prob}
and \eqref{eq: index_model}, the
incidental parameter score evaluated at the true parameter value is
\begin{align*}
   {\cal S} &=   \left( \begin{array}{c}
                  \left[ \frac 1 {\sqrt{NT}} \sum_{t=1}^T \, \partial_{\pi} \ell_{it} \right]_{i=1,\ldots,N} \\
                  \left[ \frac 1 {\sqrt{NT}} \sum_{i=1}^N \, \partial_{\pi} \ell_{it} \right]_{t=1,\ldots,T}
                     \end{array}
                     \right) .
\end{align*}
The penalty term in the objective function does not
contribute to ${\cal S}$, because at the true parameter value 
$v'\phi^0 = 0$. The corresponding expected incidental parameter Hessian $\overline {\cal H}$ 
is given in \eqref{ExpectedHessian-MAIN}. Section~\ref{sec:InverseH} discusses
 the structure of $\overline {\cal H}$ and $\overline {\cal H}^{-1}$ in more detail.
Define 
\begin{align}
   \Lambda_{it} &:= - \frac 1 {\sqrt{NT}} \sum_{j=1}^N \sum_{\tau=1}^T \left(
   \overline {\cal H}^{-1}_{(\alpha\alpha)ij}
    + \overline  {\cal H}^{-1}_{(\gamma\alpha)tj}
     + \overline  {\cal H}^{-1}_{(\alpha\gamma)i\tau}
         + \overline  {\cal H}^{-1}_{(\gamma\gamma)t\tau} \right)  \partial_{\pi} \ell_{j\tau},
    \label{DefProLa}     
\end{align}
and the operator $D_{\beta} \Delta_{it} :=  \partial_{\beta} \Delta_{it} -  \partial_{\pi} \Delta_{it} \Xi_{it}$, which are similar to $\Xi_{it}$ and  $D_{\beta} \ell_{it}$ in equation \eqref{DefProXi}.

 The following theorem shows that Assumption~\ref{ass:PanelA1} and  Assumption~\ref{ass:PanelA2} for the panel model are sufficient for Assumption~\ref{ass:A1} and Assumption~\ref{ass:A2}  for the general expansion, and particularizes the terms of the expansion to the panel estimators. The proof  is given in the supplementary material.


\begin{theorem}
    \label{th:connection}
    Consider an estimator with objective function given by \eqref{population_prob}
    and \eqref{eq: index_model}.
     Let Assumption~\ref{ass:PanelA1} be satisfied
     and suppose that the limit $\overline W_\infty$ defined in
     Theorem~\ref{th:BothEffects} exists and is positive definite.
     Let $q=8$, $\epsilon=1/(16+2 \nu)$, 
    $r_{\beta,NT} = \log(NT) (NT)^{-1/8}$
    and
     $r_{\phi,NT} = (NT)^{-1/16}$. Then,
     \begin{itemize}
          \item[(i)] Assumption~\ref{ass:A1}
   holds  and 
  $\| \widehat \beta - \beta^0\| = {\cal O}_P( (NT)^{-1/4} )$.
  
         \item[(ii)] The approximate Hessian and the terms of the score defined in Theorem~\ref{th:ScoreExpansion}
         can be written as 
          \begin{align*}
  \overline W  &=  -  \frac 1 {NT}  \sum_{i=1}^N
     \sum_{t=1}^T  \mathbb{E}_\phi \left(
            \partial_{\beta \beta'} \ell_{it}
              -  \partial_{\pi^2} \ell_{it} \Xi_{it} \Xi'_{it} \right) ,
  \nonumber \\
   U^{(0)} &= \frac 1 {\sqrt{NT}}
        \sum_{i=1}^N \sum_{t=1}^T
           D_{\beta} \ell_{it}  ,
  \nonumber \\
   U^{(1)} &=  \frac 1 {\sqrt{NT}}
     \sum_{i=1}^N \sum_{t=1}^T \,\left\{ - \Lambda_{it}  \,
          \left[  D_{\beta \pi} \ell_{it}   - \mathbb{E}_\phi( D_{\beta \pi} \ell_{it}  ) \right] +
          \frac 1 2       \Lambda_{it}^2 \, \mathbb{E}_\phi(
           D_{\beta \pi^2} \ell_{it}) \right\}.
\end{align*}

     \item[(iii)]  In addition, let Assumption~\ref{ass:PanelA2} hold.
   Then, Assumption~\ref{ass:A2}
   is satisfied for the partial effects defined in \eqref{eq: meffs}.
   By  Theorem~\ref{th:DeltaExpansion},
   $$\sqrt{NT}
    \left( \widehat \delta - \delta \right) =  V^{(0)}_\Delta +  V^{(1)}_\Delta
    + o_P(1),$$
    where
    \begin{align*}
        V^{(0)}_\Delta &=
       \left[ \frac 1 {NT}
        \sum_{i,t}
         \mathbb{E}_\phi( D_{\beta} {\Delta_{it}} )  \right]'
         \overline W_\infty^{-1} U^{(0)}
           - \frac 1 {\sqrt{NT}} \sum_{i,t}
          \mathbb{E}_\phi( \Psi_{it} )
          \partial_{\pi} \ell_{it} ,
        \\
        V^{(1)}_\Delta &= 
          \left[ \frac 1 {NT}
        \sum_{i,t}
         \mathbb{E}_\phi( D_{\beta} {\Delta_{it}} )  \right]'
         \overline W_\infty^{-1} U^{(1)}
       + \frac 1 {\sqrt{NT}}\sum_{i,t}  \Lambda_{it} 
      \left[    \mathbb{E}_\phi( \Psi_{it} ) \partial_{\pi^2} \ell_{it} 
          -    \Psi_{it}  \mathbb{E}_\phi( \partial_{\pi^2}  \ell_{it} )
         \right]
   \nonumber \\     & \qquad       
    + 
     \frac 1 {2 \, \sqrt{NT}} \sum_{i,t}    \Lambda_{it}^2 \left[
            \mathbb{E}_\phi(  \partial_{\pi^2} \Delta_{it} )
    -  
            \mathbb{E}_\phi( \partial_{\pi^3} \ell_{it} )
          \mathbb{E}_\phi( \Psi_{it} ) 
        \right].
    \end{align*}    
    
\end{itemize}
\end{theorem}

\subsection{Proofs of Theorems~\ref{th:BothEffects} and \ref{th:DeltaLimit}}

\begin{proof}[\bf Proof of Theorem~\ref{th:BothEffects}]
   \# First, we want to show that $U^{(0)}  \to_d {\cal N}( 0 ,\; \overline W_{\infty})$.
   In our likelihood setting,
     $\mathbb{E}_\phi \partial_{\beta} {\cal L}  = 0$,  $\mathbb{E}_\phi {\cal S} = 0$,
     and,
     by the Bartlett identities,
     $\mathbb{E}_\phi( \partial_{\beta} {\cal L}  \partial_{\beta'} {\cal L} )
      =  - \frac 1 {\sqrt{NT}} \partial_{\beta \beta'} \overline {\cal L}$,
    $\mathbb{E}_\phi( \partial_{\beta} {\cal L} {\cal S}' ) =  -
    \frac 1 {\sqrt{NT}} \partial_{\beta \phi'} \overline {\cal L}$
     and $\mathbb{E}_\phi( {\cal S} {\cal S}' ) 
        = \frac 1 {\sqrt{NT}} \left( \overline {\cal H} - \frac b {\sqrt{NT}} v v' \right)$.
     Furthermore, ${\cal S}' v=0$ and   
      $ \partial_{\beta \phi'} \overline {\cal L} v =0$.
     Then, by  definition of
      $\overline W = - \, \frac 1 {\sqrt{NT}} \,
                \left( \partial_{\beta \beta'} \overline {\cal L}
                     + [\partial_{\beta \phi'} \overline {\cal L}] \; \overline {\cal H}^{-1} \;
             [\partial_{\phi \beta'} \overline {\cal L}]  \right)$
     and        
     $U^{(0)} =    \partial_{\beta} {\cal L}
                   + [\partial_{\beta \phi'} \overline {\cal L} ]\, \overline {\cal H}^{-1} {\cal S},$
    \begin{align*}
       \mathbb{E}_\phi\left( U^{(0)} \right) &= 0 , &
       {\rm Var} \left(  U^{(0)} \right) &=  \overline W,
   \end{align*}     
   which implies that $\lim_{N,T\rightarrow \infty} {\rm Var} \left(  U^{(0)} \right)  =
   \lim_{N,T\rightarrow \infty} \overline W  =  \overline W_{\infty}$. Moreover, 
    part $(ii)$
   of Theorem~\ref{th:connection} yields
   \begin{align*} 
       U^{(0)} &= \frac 1 {\sqrt{NT}}
        \sum_{i=1}^N \sum_{t=1}^T
           D_{\beta} \ell_{it},
   \end{align*}       
   where 
   $D_{\beta} \ell_{it} =  \partial_{\beta} \ell_{it} -  \partial_{\pi} \ell_{it} \Xi_{it}$
   is a martingale difference sequence for each $i$ and independent across $i$, conditional on $\phi$.
   Thus, by Lemma~\ref{lemma:martingale_diff_CLT} 
   and the Cramer-Wold device we conclude that 
   \begin{align*}
         U^{(0)}  \to_d {\cal N}\left[ 0 ,\; \lim_{N,T\rightarrow \infty} {\rm Var} \left(  U^{(0)} \right)
          \right]
           \sim  {\cal N}( 0 ,\; \overline W_{\infty}) .
   \end{align*}
   
   \# Next, we show that $U^{(1)}  \to_P
           \kappa \overline B_{\infty} + \kappa^{-1} \overline D_{\infty}$. Part $(ii)$
   of Theorem~\ref{th:connection} gives $U^{(1)}  = U^{(1a)} + U^{(1b)} $, with
   \begin{align*}
  U^{(1a)}      &=  - \frac 1 {\sqrt{NT}}
     \sum_{i=1}^N \sum_{t=1}^T \,  \Lambda_{it}  \, 
       \left[  D_{\beta \pi} \ell_{it}   - \mathbb{E}_\phi( D_{\beta \pi} \ell_{it}  ) \right] ,
     \nonumber \\  
    U^{(1b)}      & =   \frac 1 {2 \, \sqrt{NT}}
     \sum_{i=1}^N \sum_{t=1}^T       \Lambda_{it}^2 \, \mathbb{E}_\phi(
           D_{\beta \pi^2} \ell_{it}).
\end{align*} 
Plugging-in the definition of $\Lambda_{it}$, we decompose
$U^{(1a)} = U^{(1a,1)} +  U^{(1a,2)} +  U^{(1a,3)} +  U^{(1a,4)}$, where   
\begin{align*}
      U^{(1a,1)} &= \frac 1 {NT}
     \sum_{i,j}  \overline {\cal H}^{-1}_{(\alpha\alpha)ij}  
     \left( \sum_{\tau} \partial_{\pi} \ell_{j\tau} \right)
     \sum_t
       \left[  D_{\beta \pi} \ell_{it}   - \mathbb{E}_\phi( D_{\beta \pi} \ell_{it}  ) \right] ,
      \nonumber \\ 
      U^{(1a,2)} &= \frac 1 {NT}
     \sum_{j,t}  \overline  {\cal H}^{-1}_{(\gamma\alpha)tj} 
     \left( \sum_{\tau} \partial_{\pi} \ell_{j\tau} \right)
     \sum_i
       \left[  D_{\beta \pi} \ell_{it}   - \mathbb{E}_\phi( D_{\beta \pi} \ell_{it}  ) \right] ,
      \nonumber \\ 
      U^{(1a,3)} &= \frac 1 {NT}
     \sum_{i,\tau}  \overline  {\cal H}^{-1}_{(\alpha\gamma)i\tau} 
     \left( \sum_j \partial_{\pi} \ell_{j\tau} \right) 
     \sum_t
       \left[  D_{\beta \pi} \ell_{it}   - \mathbb{E}_\phi( D_{\beta \pi} \ell_{it}  ) \right] ,
      \nonumber \\ 
      U^{(1a,4)} &= \frac 1 {NT}
     \sum_{t,\tau} \overline  {\cal H}^{-1}_{(\gamma\gamma)t\tau} 
      \left( \sum_j \partial_{\pi} \ell_{j\tau} \right) 
      \sum_i \left[  D_{\beta \pi} \ell_{it}   - \mathbb{E}_\phi( D_{\beta \pi} \ell_{it}  ) \right] .
\end{align*}
By  the Cauchy-Schwarz inequality applied to the sum over $t$ in $U^{(1a,2)},$
\begin{align*}
     \left( U^{(1a,2)} \right)^2
     &\leq \frac {1} {(NT)^2} 
  \left[ \sum_{t}  
  \left( \sum_{j,\tau} \overline  {\cal H}^{-1}_{(\gamma\alpha)tj} \partial_{\pi} \ell_{j\tau}
                     \right)^2 \right]
  \left[ \sum_t  \left(   \sum_i
       \left[  D_{\beta \pi} \ell_{it}   - \mathbb{E}_\phi( D_{\beta \pi} \ell_{it}  ) \right]
  \right)^2  \right]         .             
\end{align*}
By  Lemma~\ref{lemma:HessianAdditive}, $ \overline  {\cal H}^{-1}_{(\gamma\alpha)tj} = {\cal O}_P(1/\sqrt{NT})$,
uniformly over $t,j$.
Using that both $\sqrt{NT} \, \overline {\cal H}^{-1}_{(\gamma\alpha)tj} \partial_{\pi} \ell_{j\tau}$
and $D_{\beta \pi} \ell_{it}   - \mathbb{E}_\phi( D_{\beta \pi} \ell_{it}  )$ are mean zero,
independence across $i$ and Lemma~\ref{lemma:mixing_inequality} in the supplementary material
across $t$, we obtain
\begin{align*}
     \mathbb{E}_\phi  \left( 
     \frac 1 {\sqrt{NT}} \sum_{j,\tau}  
     [\sqrt{NT} \, \overline {\cal H}^{-1}_{(\gamma\alpha)tj}] \partial_{\pi} \ell_{j\tau}
                     \right)^2 =  {\cal O}_P(1) ,
    \ \  \mathbb{E}_\phi \left( \frac 1 {\sqrt{N}}  \sum_i
       \left[  D_{\beta \pi} \ell_{it}   - \mathbb{E}_\phi( D_{\beta \pi} \ell_{it}  ) \right]
  \right)^2     =  {\cal O}_P(1) ,           
\end{align*}
 uniformly over $t$. Thus, 
 $ \sum_{t}  
  \left( \sum_{j,\tau} \overline  {\cal H}^{-1}_{(\gamma\alpha)tj} \partial_{\pi} \ell_{j\tau}
                     \right)^2 = {\cal O}_P(T)$
and
$ \sum_t  \left(   \sum_i
       \left[  D_{\beta \pi} \ell_{it}   - \mathbb{E}_\phi( D_{\beta \pi} \ell_{it}  ) \right]
  \right)^2 =      {\cal O}_P(NT)$. We conclude that
\begin{align*}
     \left( U^{(1a,2)} \right)^2
 &= \frac {1} {(NT)^2}  {\cal O}_P(T)
                 {\cal O}_P(NT)
     =  {\cal O}_P(1/N) = o_P(1),
\end{align*}                   
and therefore that $U^{(1a,2)}  = o_P(1)$. Analogously one can show   that
 $U^{(1a,3)}  = o_P(1)$.

By Lemma~\ref{lemma:HessianAdditive}, 
$\overline {\cal H}^{-1}_{(\alpha\alpha)}  = 
- {\rm diag} 
\left[ \left( \frac 1 {\sqrt{NT}} \sum_{t=1}^T \, \mathbb{E}_\phi(  \partial_{\pi^2} \ell_{it} \right)^{-1}
\right]
+ {\cal O}_P(1/\sqrt{NT})$. Analogously to the proof of $U^{(1a,2)}  = o_P(1),$
one can show that the ${\cal O}_P(1/\sqrt{NT})$ part of
$\overline {\cal H}^{-1}_{(\alpha\alpha)}$  has an
asymptotically negligible contribution to $U^{(1a,1)}$.
Thus, 
\begin{align*}
      U^{(1a,1)} &= - \frac 1 {\sqrt{NT}}
     \sum_{i}  
          \underbrace{
     \frac{
      \left( \sum_{\tau} \partial_{\pi} \ell_{i\tau} \right)
     \sum_t
       \left[  D_{\beta \pi} \ell_{it}   - \mathbb{E}_\phi( D_{\beta \pi} \ell_{it}  ) \right] }
       { \sum_t \, \mathbb{E}_\phi(  \partial_{\pi^2} \ell_{it} )}
     }_{=: U^{(1a,1)}_i}  
     +o_P(1)  .
\end{align*}       
Our assumptions guarantee that 
$\mathbb{E}_\phi\left[ \left(U^{(1a,1)}_i \right)^2 \right] = {\cal O}_P(1)$,
uniformly over $i$.
Note that both the denominator and the numerator
of $U^{(1a,1)}_i$ are of order $T$. For the denominator this is obvious because of
the sum over $T$. For the numerator there are two sums over $T$, but both 
$\partial_{\pi} \ell_{i\tau}$ and $D_{\beta \pi} \ell_{it}   - \mathbb{E}_\phi( D_{\beta \pi} \ell_{it}  )$
are mean zero weakly correlated processes, so that their sums are of order
$\sqrt{T}$. By the WLLN over $i$
(remember that we have cross-sectional independence, conditional on $\phi$, and we 
assume finite moments), 
$N^{-1} \sum_i U^{(1a,1)}_i = N^{-1}  \sum_i \mathbb{E}_\phi U^{(1a,1)}_i + o_P(1)$,
and therefore
\begin{align*}
      U^{(1a,1)} &= \underbrace{ - \sqrt{ \frac N T }  \frac 1 N
    \sum_{i=1}^{N}
            \frac{  \sum_{t=1}^T  \sum_{\tau=t}^T
        \mathbb{E}_\phi\left(
                \partial_{\pi} \ell_{it}  D_{\beta \pi} \ell_{i\tau}
                 \right)
                    }
        {  \sum_{t=1}^T \mathbb{E}_\phi\left(  \partial_{\pi^2} \ell_{i t} \right) } 
        }_{ =: \sqrt{ \frac N T }  \overline B^{(1)}}
     +o_P(1)  .
\end{align*}     
Here, we use that 
$\mathbb{E}_\phi\left(
                \partial_{\pi} \ell_{it}  D_{\beta \pi} \ell_{i\tau}
                 \right)=0$ for $t>\tau$.
Analogously,  
\begin{align*}
      U^{(1a,4)} &= - \underbrace{ \sqrt{ \frac T N }  \frac 1 T
     \sum_{t=1}^{T}
            \frac{  \sum_{i=1}^N
        \mathbb{E}_\phi\left(
                \partial_{\pi} \ell_{it} D_{\beta \pi} \ell_{it}
              \right)    }
        {  \sum_{i=1}^N \mathbb{E}_\phi\left(  \partial_{\pi^2} \ell_{i t} \right) }
        }_{=: \sqrt{ \frac T N } \overline D^{(1)}}
     +o_P(1)  .
\end{align*}     
We conclude that
$U^{(1a)} = \kappa \overline B^{(1)} + \kappa^{-1} \overline D^{(1)} + o_P(1)$.

Next, we analyze $U^{(1b)}$.
We decompose $\Lambda_{it}=\Lambda_{it}^{(1)}+\Lambda_{it}^{(2)}
 + \Lambda_{it}^{(3)} + \Lambda_{it}^{(4)}$, where
\begin{align*}
   \Lambda_{it}^{(1)} &= - \frac 1 {\sqrt{NT}} \sum_{j=1}^N 
   \overline {\cal H}^{-1}_{(\alpha\alpha)ij}
\sum_{\tau=1}^T   \partial_{\pi} \ell_{j\tau},
  &    
   \Lambda_{it}^{(2)} &= - \frac 1 {\sqrt{NT}} \sum_{j=1}^N 
  \overline  {\cal H}^{-1}_{(\gamma\alpha)tj}
    \sum_{\tau=1}^T  \partial_{\pi} \ell_{j\tau},
   \nonumber \\      
   \Lambda_{it}^{(3)} &= - \frac 1 {\sqrt{NT}}  \sum_{\tau=1}^T  
       \overline  {\cal H}^{-1}_{(\alpha\gamma)i\tau}
       \sum_{\tau=1}^T    \partial_{\pi} \ell_{j\tau},
  &   
   \Lambda_{it}^{(4)} &= - \frac 1 {\sqrt{NT}} \sum_{\tau=1}^T 
           \overline  {\cal H}^{-1}_{(\gamma\gamma)t\tau} \sum_{\tau=1}^T  \partial_{\pi} \ell_{j\tau}.
\end{align*}
This decomposition of $\Lambda_{it}$ induces the following decomposition
of $U^{(1b)}$
\begin{align*}
    U^{(1b)} &= \sum_{p,q=1}^4 U^{(1b,p,q)} , &
    U^{(1b,p,q)} &= \frac 1 {2 \, \sqrt{NT}}
     \sum_{i=1}^N \sum_{t=1}^T    
        \Lambda_{it}^{(p)} \Lambda_{it}^{(q)}
       \mathbb{E}_\phi(
           D_{\beta \pi^2} \ell_{it}).
\end{align*}
Due to the symmetry $U^{(1b,p,q)} = U^{(1b,q,p)},$ this  decomposition
has  10 distinct terms. Start with $U^{(1b,1,2)}$ noting that 
\begin{align*}
   U^{(1b,1,2)} &= \frac 1 {\sqrt{NT}} 
    \sum_{i=1}^N U^{(1b,1,2)}_i , 
  \nonumber \\  
    U^{(1b,1,2)}_i &= \frac 1 {2T} \sum_{t=1}^T    
      \mathbb{E}_\phi(
           D_{\beta \pi^2} \ell_{it})
           \frac{1} {N^2}
     \sum_{j_1,j_2=1}^N 
        \left[ NT \overline {\cal H}^{-1}_{(\alpha\alpha)ij_1}   
           \overline  {\cal H}^{-1}_{(\gamma\alpha)tj_2} \right]
      \left( \frac 1 {\sqrt{T}} \sum_{\tau=1}^T   \partial_{\pi} \ell_{j_1\tau} \right)     
      \left( \frac 1 {\sqrt{T}} \sum_{\tau=1}^T   \partial_{\pi} \ell_{j_2\tau} \right)     .
\end{align*}   
By $\mathbb{E}_\phi( \partial_{\pi} \ell_{it} ) =0$,
$\mathbb{E}_\phi( \partial_{\pi} \ell_{it}  \partial_{\pi} \ell_{j\tau}) =0$ for $(i,t) \neq (j,\tau)$, 
and the properties of the inverse expected Hessian from Lemma~\ref{lemma:HessianAdditive},
 $\mathbb{E}_\phi\left[ U^{(1b,1,2)}_i  \right]  = {\cal O}_P(1/N)$, uniformly over $i$,
  $\mathbb{E}_\phi\left[ \left(U^{(1b,1,2)}_i \right)^2 \right]  = {\cal O}_P(1)$, 
 uniformly over $i$,
 and $\mathbb{E}_\phi\left[ U^{(1b,1,2)}_i U^{(1b,1,2)}_j \right]  = {\cal O}_P(1/N)$,
 uniformly over $i \neq j$.
This implies that $\mathbb{E}_\phi \, U^{(1b,1,2)} = {\cal O}_P(1/N)$ and
$\mathbb{E}_\phi\left[ \left(U^{(1b,1,2)} - \mathbb{E}_\phi \, U^{(1b,1,2)} \right)^2 \right]  
= {\cal O}_P(1/\sqrt{N})$,
and therefore $U^{(1b,1,2)} = o_P(1)$.
By similar arguments one obtains
$U^{(1b,p,q)} = o_P(1)$ for all combinations of $p,q=1,2,3,4$, except
for $p=q=1$ and $p=q=4$.

For $p=q=1$,
\begin{align*}
   U^{(1b,1,1)} &= \frac 1 {\sqrt{NT}} 
    \sum_{i=1}^N U^{(1b,1,1)}_i , 
  \nonumber \\  
    U^{(1b,1,1)}_i &= \frac 1 {2T} \sum_{t=1}^T    
      \mathbb{E}_\phi(
           D_{\beta \pi^2} \ell_{it})
           \frac{1} {N^2}
     \sum_{j_1,j_2=1}^N 
       \left[ NT   \overline {\cal H}^{-1}_{(\alpha\alpha)ij_1}   
           \overline {\cal H}^{-1}_{(\alpha\alpha)ij_2}   \right]
      \left( \frac 1 {\sqrt{T}} \sum_{\tau=1}^T   \partial_{\pi} \ell_{j_1\tau} \right)     
      \left( \frac 1 {\sqrt{T}} \sum_{\tau=1}^T   \partial_{\pi} \ell_{j_2\tau} \right) .
\end{align*}
Analogous to the result for $U^{(1b,1,2)},$ 
$\mathbb{E}_\phi\left[ \left(U^{(1b,1,1)} - \mathbb{E}_\phi \, U^{(1b,1,1)} \right)^2 \right]  
= {\cal O}_P(1/\sqrt{N})$, and therefore 
$U^{(1b,1,1)} =  \mathbb{E}_\phi  \,   U^{(1b,1,1)} + o(1)$.
Furthermore,
\begin{align*}
     \mathbb{E}_\phi  \,   U^{(1b,1,1)}
       &=   \frac 1 {2 \sqrt{NT}} 
    \sum_{i=1}^N    \frac{ \sum_{t=1}^T    
      \mathbb{E}_\phi( D_{\beta \pi^2} \ell_{it}) 
    \sum_{\tau=1}^T \mathbb{E}_\phi\left[ \left(   \partial_{\pi} \ell_{i\tau} \right)^2  \right] }
     { \left[ \sum_{t=1}^T  \mathbb{E}_\phi\left(  \partial_{\pi^2} \ell_{it}   \right) \right]^2}
       + o(1)
   \nonumber \\  &
   =  \underbrace{ -  \sqrt{ \frac N T }  \frac 1 {2 N} 
    \sum_{i=1}^N    \frac{ \sum_{t=1}^T    
      \mathbb{E}_\phi( D_{\beta \pi^2} \ell_{it})  }
     { \sum_{t=1}^T  \mathbb{E}_\phi\left(  \partial_{\pi^2} \ell_{it}   \right) }
    }_ { =: \sqrt{ \frac N T }  \overline B^{(2)}} 
       + o(1)   .
\end{align*}   
Analogously, 
\begin{align*}
     U^{(1b,4,4)}
     &= \mathbb{E}_\phi  \,   U^{(1b,4,4)} + o_P(1)
       =  \underbrace{ -  \sqrt{ \frac T N }  \frac 1 {2 T} 
    \sum_{t=1}^T    \frac{ \sum_{i=1}^N    
      \mathbb{E}_\phi( D_{\beta \pi^2} \ell_{it})  }
     { \sum_{i=1}^N   \mathbb{E}_\phi\left(  \partial_{\pi^2} \ell_{it}   \right) }
    }_ { =: \sqrt{ \frac T N }  \overline D^{(2)}} 
       + o(1)   .
\end{align*}  
We have thus shown that
$U^{(1b)} = \kappa \overline B^{(2)} + \kappa^{-1} \overline D^{(2)} + o_P(1)$.
Since $\overline B_{\infty} =  \lim_{N,T \rightarrow \infty}[ \overline B^{(1)} + \overline B^{(2)}]$
 and $\overline D_{\infty} =  \lim_{N,T \rightarrow \infty}[ \overline D^{(1)} + \overline D^{(2)}]$             
we thus conclude  $U^{(1)} =
           \kappa \overline B_{\infty} + \kappa^{-1} \overline D_{\infty} + o_P(1)$.
       
  \# We have shown
$U^{(0)}  \to_d {\cal N}( 0 ,\; \overline W_{\infty})$, and
$U^{(1)}  \to_P
           \kappa \overline B_{\infty} + \kappa^{-1} \overline D_{\infty}$. Then, part $(ii)$ of Theorem~\ref{th:connection} yields 
$\sqrt{NT} ( \widehat \beta - \beta^0  )
          \; \to_d \;
       \overline{W}_{\infty}^{-1} {\cal N}( \kappa \overline B_{\infty}
                    + \kappa^{-1}  \overline D_{\infty} ,
           \;\overline W_{\infty})$.
 \end{proof}

\begin{proof}[\bf Proof of Theorem~\ref{th:DeltaLimit}]
 We consider the case of scalar $\Delta_{it}$  to simplify the notation. Decompose
  $$
r_{NT} (\widehat \delta - \delta_{NT}^0 -  \overline{B}_{\infty}^{\delta}/T -   \overline{D}_{\infty}^{\delta}/N) =  r_{NT} (\delta - \delta_{NT}^0) + \frac{r_{NT}}{\sqrt{NT}} \sqrt{NT} (\widehat \delta - \delta  - \overline{B}_{\infty}^{\delta} /T -  \overline{D}_{\infty}^{\delta}/N).
$$

\# Part (1): Limit of $\sqrt{NT} (\widehat \delta - \delta  -  \overline{B}_{\infty}^{\delta}/T -   \overline{D}_{\infty}^{\delta}/N)$. An argument analogous to  to the proof of Theorem~\ref{th:BothEffects} using Theorem~\ref{th:connection}$(iii)$ yields 
$$
\sqrt{NT} (\widehat \delta - \delta) \to_d \mathcal{N}\left(\kappa \overline{B}_{\infty}^{\delta} + \kappa^{-1}  \overline{D}_{\infty}^{\delta},  \overline{V}_{\infty}^{\delta(1)} \right),
$$
where $\overline{V}_{\infty}^{\delta(1)} =  \EE\left\{ (NT)^{-1} \sum_{i,t}\E[ \Gamma_{it}^2] \right\},$ for the expressions of $ \overline{B}_{\infty}^{\delta}$, $ \overline{D}_{\infty}^{\delta}$, and  $\Gamma_{it}$ given in the statement of the theorem. Then, by Mann-Wald theorem
\begin{equation*}
\sqrt{NT} (\widehat \delta - \delta  -  \overline{B}_{\infty}^{\delta}/T - \overline{D}_{\infty}^{\delta}/N) \to_d \mathcal{N}\left(0,  \overline{V}_{\infty}^{\delta(1)} \right). 
\end{equation*}

\# Part (2):  Limit of $r_{NT}(\delta - \delta_{NT}^0)$. Here we show that $r_{NT}(\delta - \delta_{NT}^0) \to_d \mathcal{N}(0,\overline{V}^{\delta(2)}_{\infty})$  
for the convergence rate $r_{NT}$ given in Remark \ref{remark: conv_rate}, and characterize the asymptotic variance $\overline{V}^{\delta(2)}_{\infty}$.
We determine $r_{NT}$ through $\Ep[(\delta - \delta_{NT}^0)^2] = \mathcal{O}(r_{NT}^{-2})$ and $r_{NT}^{-2} = \mathcal{O}(\Ep[(\delta - \delta_{NT}^0)^2])$, where  
\begin{equation}\label{eq: order_ameff}
\Ep[(\delta - \delta_{NT}^0)^2] = \Ep \left[
       \left(\frac {1} {NT} \sum_{i,t}  \widetilde \Delta_{it} \right)^2 \right] =  
       \frac {1} {N^2T^2} \sum_{i,j,t,s}\Ep \left[ \widetilde \Delta_{it} \widetilde \Delta_{js} \right],
\end{equation}       
for $\widetilde \Delta_{it} = \Delta_{it} - \Ep ( \Delta_{it} )$.   Then, we characterize  $\overline{V}^{\delta(2)}_{\infty}$ as   $\overline{V}_{\infty}^{\delta(2)} = \EE\{r_{NT}^2  \Ep[(\delta - \delta_{NT}^0)^2]\},$ because $\Ep[\delta - \delta_{NT}^0] = 0$. The order of $\Ep[(\delta - \delta_{NT}^0)^2]$ is equal to the number of terms of the sums in equation \eqref{eq: order_ameff} that are non zero, which it is determined by the sample properties of $\{(X_{it}, \alpha_i, \gamma_t) : 1 \leq i \leq N, 1 \leq t \leq T) \}$.
%
%
Under Assumption~\ref{ass:PanelA2}$(i)$, if  $\{\alpha_i\}_N$  and  $\{\gamma_t\}_T$ are independent sequences, and  $\alpha_i$ and $\gamma_t$ are independent for all $i,t$, then $\Ep [ \widetilde \Delta_{it} \widetilde \Delta_{js}  ] = \Ep [ \widetilde \Delta_{it} ] \Ep [\widetilde \Delta_{js}  ] = 0$ if $i \neq j$ and $t \neq s,$  so that 
$$
\Ep[(\delta - \delta_{NT}^0)^2] =  \frac {1} {N^2T^2} \left\{ \sum_{i,t,s}\Ep \left[ \widetilde \Delta_{it} \widetilde \Delta_{is}  \right] + \sum_{i,j,t}\Ep \left[ \widetilde \Delta_{it} \widetilde \Delta_{jt}  \right]  - \sum_{i,t} \Ep \left[ \widetilde \Delta_{it}^2  \right] \right\}= \mathcal{O}\left(\frac{N+T-1}{NT} \right),
$$
because $\Ep [ \widetilde \Delta_{it} \widetilde \Delta_{is}  ] \leq \Ep [ \E (\widetilde \Delta_{it}^2)]^{1/2} \Ep[ \E( \widetilde \Delta_{is}^2)  ]^{1/2} < C$ by the Cauchy-Schwarz inequality and Assumption~\ref{ass:PanelA2}$(ii)$. We conclude that $r_{NT} = \sqrt{NT/(N+T-1)}$ and 
$$
\overline{V}^{\delta(2)} =  \EE \left\{\frac {r_{NT}^2} {N^2T^2} \left( \sum_{i,t,s}\Ep \left[ \widetilde \Delta_{it} \widetilde \Delta_{is}  \right] + \sum_{i\neq j,t}\Ep \left[ \widetilde \Delta_{it} \widetilde \Delta_{jt}  \right] \right) \right\}.
$$

Note that  $r_{NT} \to \infty$  and $r_{NT} = \mathcal{O}(\sqrt{NT}).$

\# Part (3): Asymptotic covariance between   $r_{NT}(\delta - \delta_{NT}^0)$ and  $\sqrt{NT} (\widehat \delta - \delta  - T^{-1} \overline{B}_{\infty}^{\delta} - N^{-1}  \overline{D}_{\infty}^{\delta})$. Note that
$$
\Ep\left[(\delta - \delta_{NT}^0) \frac{1}{NT} \sum_{i,t} \Gamma_{it} \right] =  \frac{1}{N^2T^2}  \sum_{i,s>t}\Ep \left[ \widetilde \Delta_{it} \Gamma_{is}  \right] = \mathcal{O}\left(\frac{1}{N} \right)
$$
 since $\Gamma_{it}$ is a martingale difference over $t$ and independent over $i$ conditional on the unobserved effects. Let
 $$
\overline{C}^{\delta(1,2)} = \EE \left\{\frac {1} {NT^2}  \sum_{i,s>t}\Ep \left[ \widetilde \Delta_{it} \Gamma_{is}  \right] \right\}.
$$

\# Part (4): limit of $r_{NT} (\widehat \delta - \delta_{NT}^0 - T^{-1} \overline{B}_{\infty}^{\delta} - N^{-1}  \overline{D}_{\infty}^{\delta})$. The conclusion of the Theorem follows because    $\overline{V}_{\infty}^{\delta} = \overline{V}^{\delta(2)} +  \overline{V}^{\delta(1)} \lim_{N,T \to \infty} (r_{NT}/\sqrt{NT})^2 + 2 \overline{C}^{\delta(1,2)} \lim_{N,T \to \infty} (r_{NT}^2/N).$
\end{proof}

\section{Properties of the Inverse Expected Incidental Parameter Hessian}
\label{sec:InverseH}

The expected incidental parameter Hessian evaluated at the true parameter values is
\begin{align*}
\overline{\cal H} = \E[ - \partial_{\phi \phi'} {\cal L}] =  
\left(\begin{array}{cc}  \overline{\mathcal{H}}_{(\alpha\alpha)}^* & \overline{\mathcal{H}}_{(\alpha\gamma)}^*  \\ {[\overline{\mathcal{H}}_{(\alpha\gamma)}^*]}' & \overline{\mathcal{H}}_{(\gamma\gamma)}^*
\end{array}\right) 
+  \frac{b} {\sqrt{NT}} \, vv' ,
\end{align*}
where $v= v_{NT} = ( 1_N' ,- 1_T')'$,
 $\overline{\mathcal{H}}_{(\alpha\alpha)}^* =  \text{diag}(\frac 1 {\sqrt{NT}}  \sum_{t} \E[- \partial_{\pi^2} \ell_{it}])$,
$\overline{\mathcal{H}}_{(\alpha\gamma)it}^* = \frac 1 {\sqrt{NT}} \E[-\partial_{\pi^2} \ell_{it}]$,
and 
$\overline{\mathcal{H}}_{(\gamma\gamma)} ^*=  \text{diag}(\frac 1 {\sqrt{NT}}  \sum_{i} \E[- \partial_{\pi^2} \ell_{it}])$.

In panel models with only individual effects, it is straightforward to determine the order of magnitude of   $\overline  {\cal H}^{-1}$ in Assumption~\ref{ass:A1}$(iv)$,  because $\overline  {\cal H}$ contains only the diagonal matrix $ \overline {\cal H}_{(\alpha \alpha)}^*$.
In our case,
 $\overline {\cal H}$ is no longer diagonal, but it has a special structure. The diagonal terms
are of order 1, whereas the off-diagonal terms are of order $(NT)^{-1/2}$. Moreover,
$\left\| \overline {\cal H} -   \diag(  \overline {\cal H}_{(\alpha \alpha)}^*, \overline {\cal H}_{(\gamma \gamma)}^*) \right\|_{\max}
        =  \mathcal{O}_P((NT)^{-1/2})$.
These observations, however, are not sufficient to establish the order of $\overline  {\cal H}^{-1}$ because the number of non-zero off-diagonal terms is of much larger order than the number of  diagonal
terms; compare $\mathcal{O}(NT)$ to $\mathcal{O}(N+T)$. 
Note also that the expected
Hessian without penalty term $\overline{\cal H}^*$ has the same structure
as $\overline{\cal H}$ itself, but is not even invertible, i.e. the observation on the relative
size of diagonal vs. off-diagonal terms is certainly not sufficient to make statements about the
structure of $\overline {\cal H}^{-1}$.
The result of the following lemma is therefore not obvious.
It shows
that the diagonal terms of $\overline{{\cal H}}$ also dominate in determining 
the order of $\overline{{\cal H}}^{-1}$.


\begin{lemma}
   \label{lemma:HessianAdditive}
   Under Assumptions \ref{ass:PanelA1},
   \begin{equation*}
   \left\| \overline {\cal H}^{-1} -
   \diag \left(  \overline {\cal H}_{(\alpha \alpha)}^*, \overline {\cal H}_{(\gamma \gamma)}^*
    \right)^{-1}
    \right\|_{\max}
        =  \mathcal{O}_P\left( (NT)^{-1/2} \right)  .
   \end{equation*}
\end{lemma}
The proof of Lemma~\ref{lemma:HessianAdditive} is provided in the supplementary material.
The lemma result establishes that $\overline {\cal H}^{-1}$ can be uniformly approximated
by a diagonal matrix,
 which is given by
the inverse of the diagonal terms of $\overline {\cal H}$ without the penalty.
The diagonal elements of
$  \diag(  \overline {\cal H}_{(\alpha \alpha)}^*, \overline {\cal H}_{(\gamma \gamma)}^*)^{-1} $
are of order 1, i.e. the order of the difference established
by the lemma is relatively small.

Note that the choice of penalty in the objective function
is important to obtain Lemma~\ref{lemma:HessianAdditive}.
Different penalties, corresponding to other
normalizations (e.g. a penalty proportional to $\alpha_1^2$, corresponding
to the normalization $\alpha_1=0$), would fail to deliver
Lemma~\ref{lemma:HessianAdditive}. However, these alternative choices do not affect
the estimators  $\widehat \beta$ and $\widehat \delta$, i.e. which normalization
is used to compute $\widehat \beta$ and $\widehat \delta$ in practice is irrelevant 
(up to numerical precision errors).

\end{appendix}


\clearpage


\setcounter{section}{0}
\setcounter{footnote}{0}
\setcounter{page}{1}

\renewcommand{\thesection}{S.\arabic{section}}
\renewcommand{\theequation}{S.\arabic{equation}}
\renewcommand{\thetheorem}{S.\arabic{theorem}}
\renewcommand{\thetable}{S.\arabic{table}}

\begin{center}
{\bf \Large 
Supplement to `Individual and Time Effects
       in  Nonlinear Panel Models with Large $N$, $T$'} 
\end{center}

\bigskip
\bigskip
\bigskip


%

\vspace{-1cm}
\abstract{\noindent
This supplemental material contains five appendices. Appendix \ref{sec:MCpoisson} presents the results of an empirical application and a Monte Carlo simulation calibrated to the application.  Following Aghion \textit{et al.} \cite*{AghionBloomBlundellGriffithHowitt2005},  we use a  panel of U.K. industries to estimate  Poisson models with industry and time effects for the relationship between innovation and competition. Appendix~\ref{app:ProofsEst} gives the proofs of Theorems \ref{th:bc} and \ref{th:bc_ameff}. 
  Appendices~\ref{sec:ExpansionProofs}, \ref{sec:s3}, and  \ref{sec:s4} contain the proofs of Appendices~\ref{app:expansion},  \ref{app:ProofMain}, and \ref{sec:InverseH}, respectively. Appendix \ref{sec: s5} collects some useful intermediate results that are used in the proofs of the main results.}

\bigskip

\section{Relationship between Innovation and Competition}
\label{sec:MCpoisson}

\subsection{Empirical Example}\label{sec: empirics}

To illustrate the bias corrections with real data, we revisit the empirical application of Aghion, Bloom, Blundell, Griffith and Howitt \cite*{AghionBloomBlundellGriffithHowitt2005} (ABBGH) that estimated a count data model to analyze the relationship between innovation and competition. They used an unbalanced panel of seventeen U.K. industries followed over the 22 years between 1973 and 1994.\footnote{We assume that the observations are missing at random conditional on the explanatory variables and unobserved effects and apply the corrections without change since the level of attrition is low in this application.} The dependent variable, $Y_{it},$ is innovation as measured by a citation-weighted number of patents, and the explanatory variable of interest, $Z_{it},$ is competition as measured by one minus the Lerner index in the industry-year.


Following ABBGH we consider a quadratic static Poisson model with industry and year effects where
$$
Y_{it} \mid Z_i^T, \alpha_i,\gamma_t \sim \mathcal{P}(\exp[\beta_1 Z_{it} + \beta_2 Z_{it}^2 + \alpha_i + \gamma_t]),
$$
for $(i = 1,...,17; t = 1973, ..., 1994),$ and extend the analysis to a dynamic Poisson model with industry and year effects where
$$
Y_{it} \mid Y_i^{t-1}, Z_i^t, \alpha_i, \gamma^t \sim \mathcal{P}(\exp[\beta_{Y} \log(1 + Y_{i,t-1}) + \beta_1 Z_{it} + \beta_2 Z_{it}^2 + \alpha_i + \gamma_t]),
$$
for  $(i = 1,...,17; t = 1974, ..., 1994).$ In the dynamic model we use the year 1973 as the initial condition for $Y_{it}$.

\begin{table}[ht]
\begin{center} 
\scalebox{.9}{\includegraphics{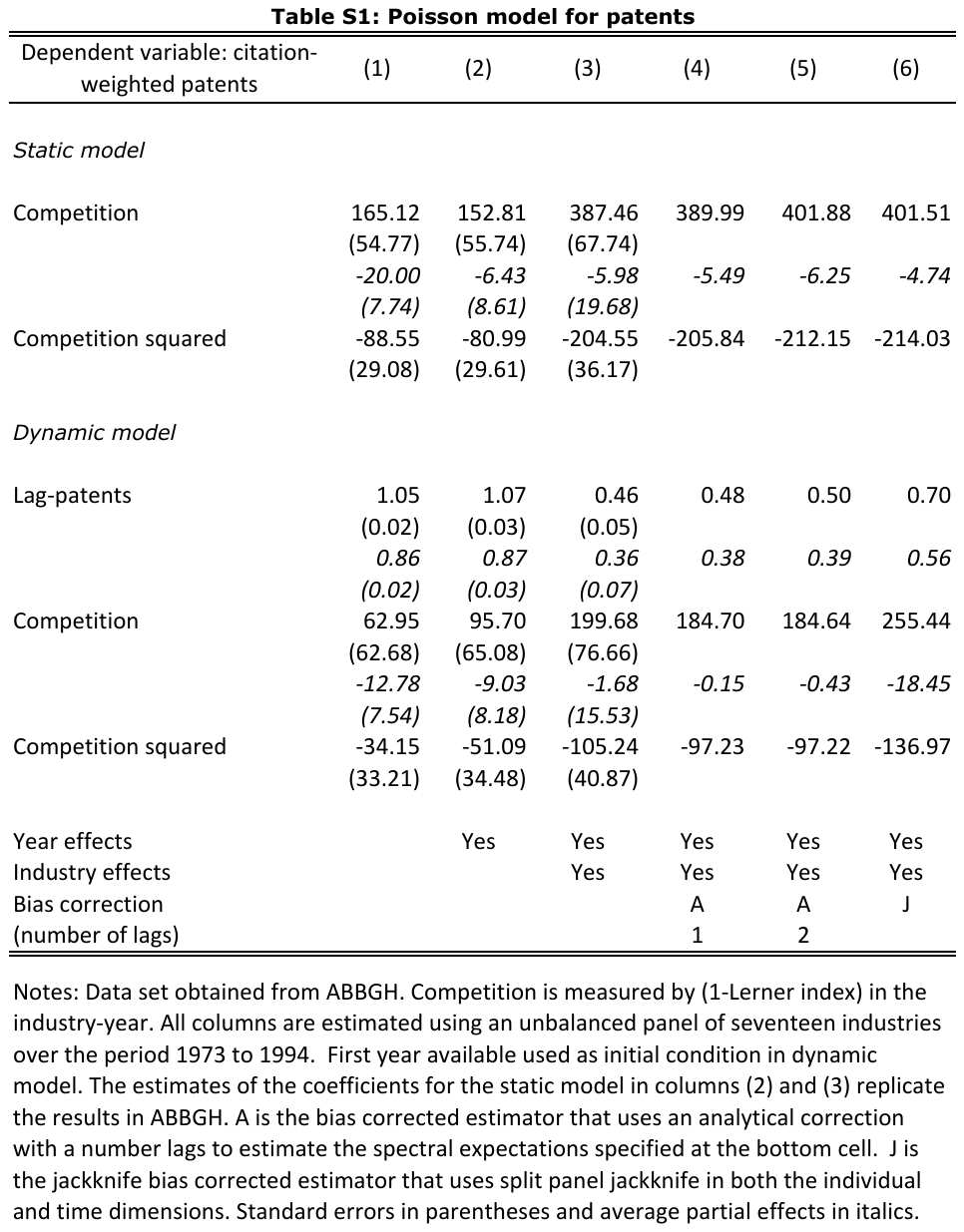}}
\end{center}
\end{table}

Table S1 reports the results of the analysis. Columns (2) and (3) for the static model replicate the empirical results of Table I in ABBGH  (p. 708), adding estimates of the APEs. Columns (4) and (5) report estimates of the analytical corrections that do not assume that competition is strictly exogenous with $L=1$ and $L=2$, and  column (6) reports estimates of the jackknife bias corrections  described in  equation (\ref{eq: jackknife2}) of the paper. Note that we do not need to report separate standard errors for the corrected estimators, because the standard errors of the uncorrected estimators are consistent for the corrected estimators under the asymptotic approximation that we consider.\footnote{In numerical examples, we find very little gains in terms of the ratio SE/SD and coverage probabilities when we reestimate the standard errors using bias corrected estimates.} Overall, the corrected estimates, while numerically different from the uncorrected estimates in column (3), agree with the inverted-U pattern in the relationship between innovation and competition found by ABBGH. The close similarity between the uncorrected and bias corrected estimates gives some evidence in favor of the strict exogeneity of competition with respect to the innovation process. 

\setcounter{table}{1}
\begin{table}[htp]
\begin{center}\caption{\label{table:test}  Homogeneity test for the jackknife}
\begin{tabular}{lccc} \hline\hline 
&  & Cross section  &  Time series   \\\hline 
  Static Model  & &  10.49 & 13.37  \\
  &  &  (0.01) & (0.00)   \\
  Dynamic Model  &  & 1.87 & 12.41 \\
 &  & (0.60) & (0.01)   \\\hline\hline
   \multicolumn{4}{l}{ \begin{footnotesize} Notes:  Wald test for equality of common parameters across sub panels.  \end{footnotesize}}\\
       \multicolumn{4}{l}{ \begin{footnotesize} P-values in parentheses \end{footnotesize}}
  \end{tabular}
\end{center}
\end{table}

The results for the dynamic model show substantial positive state dependence in the innovation process that is not explained by industry heterogeneity. 
Uncorrected fixed effects underestimates the coefficient and APE of lag patents relative to the bias corrections, specially relative to the jackknife. The pattern of the differences between the estimates is consistent with the biases that we find in the numerical example in Table S4. Accounting for state dependence does not change the inverted-U pattern, but flattens the relationship between innovation and competition.

Table \ref{table:test} implements Chow-type homogeneity tests for the validity of the jackknife corrections. These tests compare the uncorrected fixed effects estimators of the common parameters  within the elements of the cross section and time series partitions of the panel. Under time homogeneity, the probability limit of these estimators is the same, so that a standard Wald test can be applied based on the difference of the estimators in the sub panels within the partition. For the static model, the test is rejected at the 1\% level in both the cross section and time series partitions. Since the cross sectional partition is arbitrary, these rejection might be a signal of model misspecification. For the dynamic model, the test is rejected at the 1\% level in the time series partition, but it cannot be rejected at conventional levels in the cross section partition. The rejection of the time homogeneity might explain the difference between the jackknife and analytical corrections in the dynamic model.

\subsection{Calibrated Monte Carlo Simulations}
We conduct a simulation that mimics the empirical example. The designs correspond to static and dynamic Poisson models with additive individual and time effects. We calibrate all the parameters and exogenous variables using the dataset from ABBGH.

\subsubsection{Static Poisson model}

The data generating process is
\begin{equation*}
Y_{it} \mid Z_i^T,\alpha,\gamma  \sim \mathcal{P}( \exp[Z_{it} \beta_1 + Z_{it}^2 \beta_2 + \alpha _{i} +
\gamma_{t}]),
\ \ (i=1,...,N; \
t=1,...,T),
\end{equation*}%
where $\mathcal{P}$ denotes the Poisson distribution. The variable $Z_{it}$ is fixed to the values of
the competition variable in the dataset and all the parameters are set to the fixed effect estimates of the model.  We generate
unbalanced panel data sets with $T=22$ years and three different numbers
of industries $N$: 17,  34, and 51. In the second (third) case, we double (triple) the cross-sectional size by merging  two (three) independent realizations of the panel.

Table S3 reports the simulation results for the coefficients $\beta_1$ and $\beta_2$, and the APE of $Z_{it}$. We compute the  APE using the expression (\ref{example: poisson: meff}) with $H(Z_{it}) = Z_{it}^2$.
Throughout the table, MLE corresponds to the pooled Poisson maximum likelihood estimator (without individual and time effects),  MLE-TE corresponds to the Poisson estimator with only time effects, MLE-FETE corresponds to the Poisson maximum
likelihood estimator with individual and  time fixed effects, Analytical (L=l) is the bias corrected estimator that uses the analytical correction with $L=l$,
and Jackknife is the bias corrected estimator  that
uses SPJ in both the individual and time
dimensions. The analytical corrections are different from the uncorrected estimator because they do not use that the regressor $Z_{it}$ is strictly exogenous. The cross-sectional division in the jackknife follows the
order of the observations. The choice of these estimators is motivated by the empirical analysis of ABBGH.  All the results in the table are
reported in percentage of the true parameter value.

The results of the table agree with the no asymptotic  bias result for the Poisson model with exogenous regressors. Thus, the bias of MLE-FETE for the coefficients and APE is negligible relative to the standard deviation and the coverage probabilities get close to the nominal level as $N$ grows. The analytical corrections preserve the performance of the estimators and have very little sensitivity to the trimming parameter. The jackknife  correction increases dispersion and rmse, specially for the small cross-sectional size of the application. The estimators that do not control for individual effects are clearly biased.

\subsubsection{Dynamic Poisson model}

The data generating process is
\begin{equation*}
Y_{it} \mid Y_i^{t-1}, Z_i^t,\alpha, \gamma \sim \mathcal{P}(  \exp[\beta_{Y}  \log(1 + Y_{i,t-1}) +  Z_{it} \beta_1 + Z_{it}^2 \beta_2 + \alpha _{i} +
\gamma_{t}]), \ \ (i=1,...,N; t=1,...,T).
\end{equation*}%
The competition variable $Z_{it}$ and the initial condition for the number of patents $Y_{i0}$ are fixed to the values
 in the dataset and all the parameters are set to the fixed effect estimates of the model.  To generate panels, we first impute values
 to the missing observations of $Z_{it}$ using forward and backward predictions from a panel AR(1) linear  model with individual and time effects. We then draw
panel data sets with $T=21$ years and three different numbers
of industries $N$: 17, 34, and 51. As in the static model, we double (triple) the cross-sectional size by merging  two (three) independent realizations of the panel. We make the generated panels unbalanced by dropping the values corresponding to the missing observations in the original dataset.

Table S4 reports the simulation results for the coefficient $\beta^0_Y$ and the APE of $Y_{i,t-1}$. The estimators considered are the same as for the static Poisson model above. We compute the partial effect  of $Y_{i,t-1}$ using  (\ref{example: poisson: meff}) with $Z_{it} = Y_{i,t-1}$, $H(Z_{it}) = \log (1+Z_{it}),$ and dropping the linear term.
Table S5 reports the simulation results for the coefficients $\beta_1^0$ and $\beta_2^0$, and the APE  of $Z_{it}$. We compute the partial effect using  (\ref{example: poisson: meff}) with $H(Z_{it}) = Z_{it}^2$.
Again, all the results in the tables are
reported in percentage of the true parameter value.

The results in table S4 show biases of the same order of magnitude as the standard deviation for the fixed effects estimators of the coefficient and  APE of $Y_{i,t-1}$, which cause severe undercoverage of confidence intervals.  Note that in this case the rate of convergence for the estimator of the APE is $r_{NT} = \sqrt{NT}$, because the individual and time effects are hold fixed across the simulations. The analytical corrections reduce bias by more than half without increasing dispersion, substantially reducing rmse and bringing coverage probabilities closer to their nominal levels. The jackknife corrections reduce bias and increase dispersion leading to lower improvements in  rmse and  coverage probability  than the analytical corrections. The results for the coefficient of $Z_{it}$ in table 8 are similar to the static model.   The results for the APE of $Z_{it}$ are imprecise, because the true value of the effect is close to zero.

\section{Proofs of Theorems \ref{th:bc} and \ref{th:bc_ameff}} 
\label{app:ProofsEst}

We start with a lemma that shows the consistency of the fixed effects estimators of averages of the data and parameters. We will use this result to show the validity of the analytical bias corrections and the consistency of the variance estimators.
 
 \begin{lemma}\label{lemma:abc} Let $G(\beta,\phi) := [N(T-j)]^{-1} \sum_{i,t \geq j+1} g(X_{it}, X_{i,t-j}, \beta, \alpha_{i} + \gamma_t, \alpha_{i} + \gamma_{t-j})$ for $0 \leq j < T,$ and $\mathcal{B}_{\varepsilon}^0$ be a subset of $\mathbb{R}^{\dim \beta + 2}$ that contains an $\varepsilon$-neighborhood of $(\beta, \pi_{it}^0, \pi_{i,t-j}^0)$ for all $i,t,j,N, T$, and for some $\varepsilon > 0$. Assume that $(\beta, \pi_1, \pi_2) \mapsto g_{itj}(\beta, \pi_1, \pi_2) := g(X_{it}, X_{i,t-j}, \beta, \pi_1, \pi_2)$ is  Lipschitz continuous over $\mathcal{B}_{\varepsilon}^0$ a.s, i.e. $|g_{itj}(\beta_1, \pi_{11}, \pi_{21}) - g_{itj}(\beta_0, \pi_{10}, \pi_{20})| \leq M_{itj} \|(\beta_1, \pi_{11}, \pi_{21}) - (\beta, \pi_{10}, \pi_{20}) \|$ for all $(\beta_0, \pi_{10}, \pi_{20}) \in \mathcal{B}_{\varepsilon}^0$, $(\beta_1, \pi_{11}, \pi_{21}) \in \mathcal{B}_{\varepsilon}^0$, and some $M_{itj} = \mathcal{O}_P(1)$ for all $i,t,j,N, T$. Let $(\widehat \beta, \widehat \phi)$ be an estimator of $(\beta, \phi)$ such that $\|\widehat \beta - \beta^0\| \to_P 0$ and $\|\widehat \phi - \phi^0 \|_{\infty} \to_P 0.$ Then,
 $$
 G(\widehat \beta, \widehat \phi) \to_P \EE[G(\beta^0, \phi^0)],
 $$ 
 provided that the limit exists. 
 \end{lemma}
 
 \begin{proof}[\bf Proof of Lemma~\ref{lemma:abc}] By the triangle inequality
 $$
 |G(\widehat \beta, \widehat \phi) - \EE[G(\beta^0, \phi^0)]| \leq |G(\widehat \beta, \widehat \phi) - G(\beta^0, \phi^0)| + o_P(1),
 $$ 
because $|G(\beta^0, \phi^0) - \EE[G(\beta^0, \phi^0)]| = o_P(1)$. By the local Lipschitz continuity of $g_{itj}$ and the consistency of $(\widehat \beta, \widehat \phi)$, 
 \begin{multline*}
 |G(\widehat \beta, \widehat \phi) - G(\beta^0, \phi^0)| \leq \frac{1}{N(T-j)} \sum_{i,t \geq j+1} M_{itj} \|(\widehat \beta, \widehat \alpha_i + \widehat \gamma_t , \widehat \alpha_i + \widehat \gamma_{t-j}) - ( \beta^0,  \alpha_i^0 + \gamma_t^0 , \alpha_i^0 + \gamma_{t-j}^0) \| \\ \leq \frac{1}{N(T-j)} \sum_{i,t \geq j+1} M_{itj} (\|\widehat \beta - \beta^0 \| + 4 \| \widehat \phi - \phi^0 \|_{\infty} )
\end{multline*}
wpa1. The result then follows because $[N(T-j)]^{-1} \sum_{i,\tau \geq t} M_{it\tau} = \mathcal{O}_P(1)$  and $(\|\widehat \beta - \beta^0 \| + 4 \| \widehat \phi - \phi^0 \|_{\infty} ) = o_P(1)$ by assumption.
\end{proof}
 
\begin{proof}[\bf Proof of Theorem~\ref{th:bc}] We separate the proof in three parts corresponding to the three statements of the theorem.

Part I: Proof of $\widehat W \to_P \overline{W}_{\infty}$. The asymptotic variance and its fixed effects estimators can be expressed as $\overline{W}_{\infty} = \EE[W(\beta^0, \phi^0)]$ and $\widehat W = W(\widehat \beta, \widehat \phi),$ where  $W(\beta,\phi)$  has a first order representation as a continuously differentiable transformation of terms that have the form of $G(\beta,\phi)$ in Lemma \ref{lemma:abc}. The result then follows by the continuous mapping theorem noting that $\|\widehat \beta - \beta^0\| \to_P 0$ and $\|\widehat \phi - \phi^0 \|_{\infty} \leq \|\widehat \phi - \phi^0 \|_{q} \to_P 0$ by Theorem~\ref{th:connection}.

%

Part II: Proof of $\sqrt{NT}(\widetilde \beta^A - \beta^0)  \to_d \mathcal{N}(0,
\overline W_{\infty}^{-1}).$ By the argument given after equation \eqref{eq:MainResult} in the text,  we only need to show that $\widehat B \to_P \overline{B}_{\infty}$ and $\widehat D \to_P \overline{D}_{\infty}$. These asymptotic biases and their fixed effects estimators are either time-series averages of fractions of cross-sectional averages,
or vice versa. The nesting of the averages makes the analysis a bit more cumbersome than the analysis of  $\widehat W$, but the result follows by similar standard arguments, also using that
$L \to \infty$ and $L/T \to 0$ guarantee that the trimmed estimator in $\widehat B$ is
also consistent for the spectral expectations; see Lemma 6 in Hahn and Kuersteiner~\cite*{HahnKuersteiner2011}.

Part III: Proof of $\sqrt{NT}(\widetilde \beta^J - \beta^0)  \to_d \mathcal{N}(0,
\overline W_{\infty}^{-1}).$ For $\T_1 = \{1, \ldots, \lfloor (T+1)/2 \rfloor \}$, $\T_2 = \{\lfloor T/2 \rfloor + 1, \ldots, T \}$, $\T_0 = \T_1 \cup \T_2$,  $\N_1 = \{1, \ldots, \lfloor (N+1)/2 \rfloor \}$, $\N_2 = \{\lfloor N/2 \rfloor + 1, \ldots, N \}$, and $\N_0 = \N_1 \cup \N_2$, let $\widehat \beta^{(jk)}$ be the fixed effect estimator of $\beta$ in the subpanel defined by  $i \in \N_j$ and $t \in \T_k$.\footnote{Note that this definition of the subpanels covers all the cases regardless of whether $N$ and $T$ are even or odd.}  In this notation, 
$$
\widetilde \beta^J = 3 \widehat \beta^{(00)} - \widehat \beta^{(10)} / 2 - \widehat \beta^{(20)} /2 - \widehat \beta^{(01)} /2 - \widehat \beta^{(02)}/2.
$$

We derive the asymptotic distribution of $\sqrt{NT}(\widetilde \beta^J -\beta^0)$ from the joint asymptotic distribution of the vector $\widehat{ \mathbb{B}} = \sqrt{NT}(\widehat \beta^{(00)} -\beta^0,\widehat \beta^{(10)} -\beta^0,\widehat \beta^{(20)} -\beta^0,\widehat \beta^{(01)}  -\beta^0,\widehat \beta^{(02)} -\beta^0)$ with dimension $5 \times \dim \beta$. By Theorem~\ref{th:connection},
 $$
 \sqrt{NT}(\widehat \beta^{(jk)} -\beta^0) = \frac{2^{1(j>0)}2^{1(k>0)}}{\sqrt{NT}} \sum_{i \in N_j, t \in T_k} \left[\psi_{it} + b_{it} + d_{it} \right] + o_P(1), 
 $$
 for $\psi_{it} = \overline{W}_{\infty}^{-1} D_{\beta} \ell_{it}$,  $b_{it} = \overline{W}_{\infty}^{-1} [U_{it}^{(1a,1)} + U_{it}^{(1b,1,1)}]$, and  $d_{it} =\overline{W}_{\infty}^{-1} [U_{it}^{(1a,4)} + U_{it}^{(1b,4,4)}]$, where the $U_{it}^{(\cdot)}$ is implicitly defined by $U^{(\cdot)} = (NT)^{-1/2} \sum_{i,t} U_{it}^{(\cdot)}$.
 Here,  none of the terms carries a superscript $(jk)$ by  Assumption~\ref{assumption: stationarity}.  The influence function $\psi_{it}$ has zero mean and determines the asymptotic variance $\overline{W}_{\infty}^{-1}$, whereas $b_{it} $ and $d_{it}$ determine the asymptotic biases $\overline{B}_{\infty}$ and $\overline{D}_{\infty},$ but do not affect the asymptotic variance. By this representation,
$$
\widehat{ \mathbb{B}} \to_d \mathcal{N} 
\left(\kappa  \left[
\begin{array}{c}
  1 \\
  1  \\
  1  \\
  2  \\
  2 
\end{array}
\right] \otimes \overline{B}_{\infty} + \kappa^{-1}  \left[
\begin{array}{c}
  1 \\
  2  \\
  2  \\
  1 \\
  1 
\end{array}
\right] \otimes  \overline{D}_{\infty},   \left[
\begin{array}{ccccc}
1 & 1 & 1 & 1 & 1  \\
1 & 2  & 0 & 1 & 1 \\
 1 & 0 & 2 & 1 & 1  \\
1 & 1 & 1 &  2 & 0 \\
1 & 1 &1 & 0 &  2
\end{array}
\right] \otimes \overline{W}_{\infty}^{-1} \right),
$$ 
where we use that $\{\psi_{it}: 1 \leq i \leq N, 1 \leq t \leq T\}$ is independent across $i$ and martingale difference across $t$ and Assumption~\ref{assumption: stationarity}. 

The result follows by writing $\sqrt{NT}(\widetilde \beta^J - \beta^0) = (3,-1/2,-1/2,-1/2,-1/2) \widehat{ \mathbb{B}}$ and using the properties of the multivariate normal distribution.
\end{proof} 

\begin{proof}[\bf Proof of Theorem~\ref{th:bc_ameff}] We separate the proof in three parts corresponding to the three statements of the theorem.

Part I: $\widehat V^{\delta} \to_P \overline{V}^{\delta}_{\infty}$.  $ \overline{V}^{\delta}_{\infty}$ and  $\widehat V^{\delta}$ have a similar structure  to $\overline{W}_{\infty}$ and $\widehat W$ in part I of the proof of Theorem~\ref{th:bc}, so that the consistency follows by an analogous argument.

Part II: $
\sqrt{NT} (\widetilde \delta^A - \delta_{NT}^0)   \to_d \mathcal{N}(0, \overline{V}_{\infty}^{\delta})
$. As in the proof of Theorem~\ref{th:DeltaLimit},
 we  decompose
  $$
r_{NT} (\widetilde \delta^A - \delta_{NT}^0) =  r_{NT} (\delta - \delta_{NT}^0) + \frac{r_{NT}}{\sqrt{NT}} \sqrt{NT} (\widetilde \delta^A - \delta).
$$
Then,  by Mann-Wald theorem,
$$
\sqrt{NT} (\widetilde \delta^A - \delta) =  \sqrt{NT} (\widehat \delta - \widehat B^{\delta}/T - \widehat D^{\delta}/N  - \delta)  \to_d \mathcal{N}(0, \overline{V}_{\infty}^{\delta(1)}),
$$ 
 provided that $ \widehat B^{\delta} \to_P \overline{B}_{\infty}^{\delta}$ and  $ \widehat D^{\delta} \to_P \overline{D}_{\infty}^{\delta}$, and $ r_{NT} (\delta - \delta_{NT}^0) \to_d \mathcal{N}(0, \overline{V}_{\infty}^{\delta(2)})$, where $\overline{V}_{\infty}^{\delta(1)}$ and $\overline{V}_{\infty}^{\delta(2)}$ are defined as in the proof of Theorem~\ref{th:DeltaLimit}.  The statement thus follows by using a similar argument to part II of the proof of Theorem~\ref{th:bc} to show the consistency of $ \widehat B^{\delta}$ and $ \widehat D^{\delta}$, and  because  $ (\delta - \delta_{NT}^0)$ and  $(\widetilde \delta^A - \delta)$
are asymptotically independent, and  $\overline{V}_{\infty}^{\delta} = \overline{V}^{\delta(2)} +  \overline{V}^{\delta(1)} \lim_{N,T \to \infty} (r_{NT}/\sqrt{NT})^2.$

Part III: $
\sqrt{NT} (\widetilde \delta^J - \delta_{NT}^0)   \to_d \mathcal{N}(0, \overline{V}_{\infty}^{\delta})
$. As in part II,
 we  decompose
  $$
r_{NT} (\widetilde \delta^J - \delta_{NT}^0) =  r_{NT} (\delta - \delta_{NT}^0) + \frac{r_{NT}}{\sqrt{NT}} \sqrt{NT} (\widetilde \delta^J - \delta).
$$
Then, by an argument similar to part III of the proof of Theorem~\ref{th:bc}, 
$$
\sqrt{NT} (\widetilde \delta^J - \delta)  \to_d \mathcal{N}(0, \overline{V}_{\infty}^{\delta(1)}),
$$ 
and  $ r_{NT} (\delta - \delta_{NT}^0) \to_d \mathcal{N}(0, \overline{V}_{\infty}^{\delta(2)})$, where $\overline{V}_{\infty}^{\delta(1)}$ and $\overline{V}_{\infty}^{\delta(2)}$ are defined as in the proof of Theorem~\ref{th:DeltaLimit}.  The statement  follows because  $ (\delta - \delta_{NT}^0)$ and  $(\widetilde \delta^J - \delta)$
are asymptotically independent, and  $\overline{V}_{\infty}^{\delta} = \overline{V}^{\delta(2)} +  \overline{V}^{\delta(1)} \lim_{N,T \to \infty} (r_{NT}/\sqrt{NT})^2.$
\end{proof}

\section{Proofs of Appendix~\ref{app:expansion} (Asymptotic Expansions)}
\label{sec:ExpansionProofs}

The following Lemma contains some statements that are not explicitly assumed in
 Assumptions~\ref{ass:A1}, but that are implied by it.
\begin{lemma}
     \label{lemma:assA1add}
     Let Assumptions~\ref{ass:A1} be satisfied. Then
     \begin{itemize}
     \item[(i)] $ {\cal H}(\beta,\phi) >0$ 
        for all $\beta \in {\cal B}(r_\beta, \beta^0)$ and  $\phi \in {\cal B}_q(r_\phi, \phi^0)$ wpa1,
     \begin{align*}
             \sup_{\beta \in {\cal B}(r_\beta, \beta^0)}  \sup_{\phi \in {\cal B}_q(r_\phi, \phi^0)} 
        \left\|  \partial_{\beta \beta'}  {\cal L}(\beta,\, \phi) \right\|
                 &= {\cal O}_P\left( \sqrt{NT} \right) ,
                 \\
            \sup_{\beta \in {\cal B}(r_\beta, \beta^0)}  \sup_{\phi \in {\cal B}_q(r_\phi, \phi^0)} 
        \left\|  \partial_{\beta \phi'}  {\cal L}(\beta,\, \phi) \right\|_q
                 &= {\cal O}_P\left( (NT)^{1/(2q)} \right) ,
                 \\
              \sup_{\beta \in {\cal B}(r_\beta, \beta^0)}  \sup_{\phi \in {\cal B}_q(r_\phi, \phi^0)} 
        \left\|  \partial_{\phi \phi \phi}  {\cal L}(\beta,\, \phi) \right\|_q
                 &= {\cal O}_P\left( (NT)^{\epsilon} \right) ,
                 \\
            \sup_{\beta \in {\cal B}(r_\beta, \beta^0)}  \sup_{\phi \in {\cal B}_q(r_\phi, \phi^0)}  \left\| \partial_{\beta \phi \phi} {\cal L}(\beta,\phi) \right\|_q    &=   {\cal O}_P( (NT)^{\epsilon} ) ,
           \\
             \sup_{\beta \in {\cal B}(r_\beta, \beta^0)}  \sup_{\phi \in {\cal B}_q(r_\phi, \phi^0)}  
           \left\| {\cal H}^{-1}(\beta,\phi)  \right\|_q   
             &=  {\cal O}_P(1 ).
         \end{align*}
         
       \item[(ii)] Moreover, $\left\| {\cal S} \right\| = {\cal O}_P\left( 1 \right),$
$ \left\| {\cal H}^{-1}   \right\| =  {\cal O}_P\left( 1 \right) ,$ $
        \left\|  \overline {\cal H}^{-1}   \right\| =   {\cal O}_P\left( 1 \right) ,$ $
        \left\| {\cal H}^{-1}  - \overline {\cal H}^{-1}  \right\| =  o_{P}\left( (NT)^{-1/8} \right) ,$ $
        \left\| {\cal H}^{-1}  - \left( \overline {\cal H}^{-1} - \overline {\cal H}^{-1} \widetilde {\cal H} \overline {\cal H}^{-1}
            \right) \right\| = o_{P}\left( (NT)^{-1/4} \right) ,$ $
         \left\|  \partial_{\beta \phi'}  {\cal L} \right\| =   {\cal O}_P\left( (NT)^{1/4} \right) ,$ $
         \left\|  \partial_{\beta \phi \phi}  {\cal L} \right\| =   {\cal O}_P\left( (NT)^{\epsilon} \right) ,$ $
         \left\|  \sum_g  \partial_{\phi \phi' \phi_g}  {\cal L} \, [{\cal H}^{-1}  {\cal S}]_g \right\| =  
          {\cal O}_P\left( (NT)^{-1/4+1/(2q)+\epsilon} \right) ,$ and $
         \left\|  \sum_g  \partial_{\phi \phi' \phi_g}  {\cal L} \, [\overline {\cal H}^{-1}  {\cal S}]_g \right\| = 
               {\cal O}_P\left( (NT)^{-1/4+1/(2q)+\epsilon} \right).$
    \end{itemize}     
\end{lemma}

\begin{proof}[\bf Proof of Lemma~\ref{lemma:assA1add}]
   \# Part $(i)$:
    Let $v \in \mathbbm{R}^{\dim \beta}$
    and $w,u \in \mathbbm{R}^{\dim \phi}$.
    By a Taylor expansion of $ \partial_{\beta \phi' \phi_g}  {\cal L}(\beta,\, \phi) $ around $(\beta^0,\phi^0)$ 
     \begin{align*}
          & \sum_{g} u_g  v'  \left[ \partial_{\beta \phi' \phi_g}  {\cal L}(\beta,\, \phi) \right] w 
          \nonumber \\  
            &=  \sum_{g} u_g   v'  \left[ \partial_{\beta \phi' \phi_g} {\cal L}     
            +  \sum_k (\beta_k - \beta^0_k) \partial_{\beta_k \beta \phi' \phi_g}  {\cal L}(\tilde \beta,\tilde \phi)
     - \sum_{h} (\phi_{h} - \phi^0_{h}) \partial_{\beta \phi' \phi_g \phi_{h} }  {\cal L}(\tilde \beta,\tilde \phi)
      \right] w,
      \end{align*}   
     with $(\tilde \beta,\tilde \phi)$ between $(\beta^0,\phi^0)$ and $(\beta,\phi)$. 
     Thus
     \begin{align*}
           \left\| \partial_{\beta \phi \phi}  {\cal L}(\beta,\, \phi) \right\|_q
           &= \sup_{\|v\|=1} \; \; \sup_{\|u\|_q=1} \; \; \sup_{\|w\|_{q/(q-1)}=1} \; \;
               \sum_{g} u_g  v'  \left[ \partial_{\beta \phi' \phi_g}  {\cal L}(\beta,\, \phi) \right] w 
        \nonumber \\   
           &\leq   \left\|  \partial_{\beta \phi \phi}   {\cal L} \right\|_q
                     +  \| \beta - \beta^0 \| 
                     \sup_{(\tilde \beta,\tilde \phi)}   \left\| \partial_{\beta \beta \phi \phi}  {\cal L}(\tilde \beta,\tilde \phi)  \right\|_q
                      +  \| \phi - \phi^0 \|_q  
                       \sup_{(\tilde \beta,\tilde \phi)} \left\| \partial_{\beta \phi \phi \phi}  {\cal L}(\tilde \beta,\tilde \phi)  \right\|_q ,
     \end{align*}
     where the supremum over  
     $(\tilde \beta,\tilde \phi)$ is necessary, because those parameters depend
     on $v$, $w$, $u$. By Assumption~\ref{ass:A1}, for large enough $N$ and $T,$
     \begin{align*}
          \sup_{\beta \in {\cal B}(r_\beta, \beta^0)}  \sup_{\phi \in {\cal B}_q(r_\phi, \phi^0)} 
           \left\| \partial_{\beta \phi \phi}  {\cal L}(\beta,\, \phi) \right\|_q
         & \leq   \left\|  \partial_{\beta \phi \phi} {\cal L} \right\|
                     +    r_\beta     \sup_{\beta \in {\cal B}(r_\beta, \beta^0)}  \sup_{\phi \in {\cal B}_q(r_\phi, \phi^0)}  
                \left\| \partial_{\beta \beta \phi \phi}  {\cal L}(\beta, \phi)  \right\|_q
         \nonumber \\ 
           &
                      +  r_\phi 
                       \sup_{\beta \in {\cal B}(r_\beta, \beta^0)}  \sup_{\phi \in {\cal B}_q(r_\phi, \phi^0)} 
              \left\| \partial_{\beta \phi  \phi \phi}  {\cal L}( \beta, \phi)  \right\|_q 
         \nonumber \\ 
           &  =   {\cal O}_P\left[ (NT)^\epsilon + r_\beta (NT)^\epsilon +   r_\phi (NT)^\epsilon \right] 
             =  {\cal O}_P\left( (NT)^\epsilon \right).
     \end{align*}
     The proofs for the bounds on
     $ \left\|  \partial_{\beta \beta'}  {\cal L}(\beta,\, \phi) \right\|$,
     $   \left\|  \partial_{\beta \phi'}  {\cal L}(\beta,\, \phi) \right\|_q$
     and
      $  \left\|  \partial_{\phi \phi \phi}  {\cal L}(\beta,\, \phi) \right\|_q$
      are analogous. 
      
      Next, we show that ${\cal H}(\beta,\phi)$ is non-singular 
        for all $\beta \in {\cal B}(r_\beta, \beta^0)$ and  $\phi \in {\cal B}_q(r_\phi, \phi^0)$ wpa1. By a  Taylor expansion      and Assumption~\ref{ass:A1}, for large enough $N$ and $T,$
      \begin{align}
       \sup_{\beta \in {\cal B}(r_\beta, \beta^0)}  \sup_{\phi \in {\cal B}_q(r_\phi, \phi^0)} 
            \left\| {\cal H}(\beta,\phi) - {\cal H} \right\|_q   
            &\leq   r_\beta     \sup_{\beta \in {\cal B}(r_\beta, \beta^0)}  \sup_{\phi \in {\cal B}_q(r_\phi, \phi^0)}  
                \left\| \partial_{\beta \phi \phi}  {\cal L}(\beta, \phi)  \right\|_q
          \nonumber \\ 
           &     
                    +  r_\phi 
                       \sup_{\beta \in {\cal B}(r_\beta, \beta^0)}  \sup_{\phi \in {\cal B}_q(r_\phi, \phi^0)} 
              \left\| \partial_{\phi  \phi \phi}  {\cal L}( \beta, \phi)  \right\|_q 
          = o_P(1).             
             \label{eq:lb51}
      \end{align}
       Define $\Delta {\cal H}(\beta,\phi) =  \overline {\cal H} - {\cal H}(\beta,\phi)$.
      Then $\left\| \Delta {\cal H}(\beta,\phi) \right\|_q \leq \left\| {\cal H}(\beta,\phi) - {\cal H} \right\|_q
      +  \left\| \widetilde {\cal H} \right\|_q  $, and therefore
      \begin{align*}
            \sup_{\beta \in {\cal B}(r_\beta, \beta^0)}  \sup_{\phi \in {\cal B}_q(r_\phi, \phi^0)} 
            \left\| \Delta {\cal H}(\beta,\phi) \right\|_q = o_P(1) ,
      \end{align*} 
     by Assumption~\ref{ass:A1} and equation (\ref{eq:lb51}).
     
     For any square matrix with $\|A\|_q < 1$,
      $\left\| (\mathbbm{1} - A)^{-1} \right\|_q \leq \left( 1 - \|A\|_q \right)^{-1}$,
      see e.g. p.301 in Horn and Johnson~\cite*{HornJohnson1985}.
      Then
      \begin{align*}
             \sup_{\beta \in {\cal B}(r_\beta, \beta^0)}  \sup_{\phi \in {\cal B}_q(r_\phi, \phi^0)} \left\| {\cal H}^{-1}(\beta,\phi)  \right\|_q 
          \nonumber \ \     
            &=   \sup_{\beta \in {\cal B}(r_\beta, \beta^0)}  \sup_{\phi \in {\cal B}_q(r_\phi, \phi^0)} \left\| \left( \overline {\cal H} - \Delta {\cal H}(\beta,\phi) \right)^{-1}  \right\|_q  
          \nonumber \\
             &=    \sup_{\beta \in {\cal B}(r_\beta, \beta^0)}  \sup_{\phi \in {\cal B}_q(r_\phi, \phi^0)} 
              \left\| \overline {\cal H}^{-1} \left( \mathbbm{1} - \Delta {\cal H}(\beta,\phi) \overline {\cal H}^{-1} \right)^{-1}  \right\|_q  
          \nonumber \\
             &\leq   
               \left\|  \overline {\cal H}^{-1} \right\|_q \; \;
              \sup_{\beta \in {\cal B}(r_\beta, \beta^0)}  \sup_{\phi \in {\cal B}_q(r_\phi, \phi^0)} 
              \left\| \left( \mathbbm{1} - \Delta {\cal H}(\beta,\phi) \overline {\cal H}^{-1} \right)^{-1}  \right\|_q  
          \nonumber \\
             &\leq   
               \left\|  \overline {\cal H}^{-1} \right\|_q \; \;
              \sup_{\beta \in {\cal B}(r_\beta, \beta^0)}  \sup_{\phi \in {\cal B}_q(r_\phi, \phi^0)} 
              \left( 1 -  \left\| \Delta {\cal H}(\beta,\phi) \overline {\cal H}^{-1} \right\|_q   \right)^{-1}  
          \nonumber \\
             &\leq   
               \left\|  \overline {\cal H}^{-1} \right\|_q  
              \left( 1 -   o_P(1)  \right)^{-1}   = {\cal O}_P(1).
      \end{align*}

         \bigskip
   
   \noindent
   \#Part $(ii)$: 
   By the properties of the $\ell_q$-norm and Assumption~\ref{ass:A1}$(v)$,
   \begin{align*}
       \| {\cal S} \| =  \| {\cal S} \|_2 \leq  (\dim \phi)^{1/2-1/q}  \| {\cal S} \|_q = {\cal O}_p(1).
   \end{align*}
   Analogously,
   \begin{align*}
         \left\|  \partial_{\beta \phi'}  {\cal L} \right\|
          \leq  (\dim \phi)^{1/2-1/q}      \left\|  \partial_{\beta \phi'}  {\cal L} \right\|_q
          &=   {\cal O}_P\left( (NT)^{1/4} \right) .
   \end{align*}
 By Lemma~\ref{lemma:matrix-q-norm},  $\| \overline {\cal H}^{-1} \|_{q/(q-1)} = \| \overline {\cal H}^{-1} \|_q$ because $\overline {\cal H}^{-1}$ is symmetric, and
   \begin{align}
       \left\| \overline {\cal H}^{-1} \right\| = \left\|  \overline {\cal H}^{-1} \right\|_2
        \leq \sqrt{  \| \overline {\cal H}^{-1} \|_{q/(q-1)} \| \overline {\cal H}^{-1} \|_q }
        = \| \overline {\cal H}^{-1} \|_q = {\cal O}_P(1). \label{eq:lb52}
   \end{align}  
   Analogously,
   \begin{align*}
           \left\|  \partial_{\beta \phi \phi}  {\cal L} \right\| 
           & \leq \left\|  \partial_{\beta \phi \phi}  {\cal L} \right\|_q 
           =   {\cal O}_P\left( (NT)^{\epsilon} \right) ,
        \nonumber \\   
          \left\|  \sum_g  \partial_{\phi \phi' \phi_g}  {\cal L} \, [{\cal H}^{-1}  {\cal S}]_g \right\| 
          &\leq   \left\|  \sum_g  \partial_{\phi \phi' \phi_g}  {\cal L} \, [{\cal H}^{-1}  {\cal S}]_g \right\|_q
       \nonumber \\& 
       \leq
          \left\|   \partial_{\phi \phi \phi}  {\cal L} \right\|_q  \left\| {\cal H}^{-1} \right\|_q \left\| {\cal S} \right\|_q   
          =  
          {\cal O}_P\left( (NT)^{-1/4+1/(2q)+\epsilon} \right) ,
        \nonumber \\
         \left\|  \sum_g  \partial_{\phi \phi' \phi_g}  {\cal L} \, [\overline {\cal H}^{-1}  {\cal S}]_g \right\| 
        &\leq  \left\|  \sum_g  \partial_{\phi \phi' \phi_g}  {\cal L} \, [\overline {\cal H}^{-1}  {\cal S}]_g \right\|_q 
               \nonumber \\
          &\leq
          \left\|   \partial_{\phi \phi \phi}  {\cal L} \right\|_q  \left\| \overline {\cal H}^{-1} \right\|_q \left\| {\cal S} \right\|_q   
        = 
               {\cal O}_P\left( (NT)^{-1/4+1/(2q)+\epsilon} \right)   .
   \end{align*}
    Assumption~\ref{ass:A1} guarantees that $\left\| \overline {\cal H}^{-1} \right \|       \left\|  \widetilde {\cal H} \right \|<1$ wpa1. 
   Therefore, 
   \begin{align*}
    {\cal H}^{-1} =
    \overline {\cal H}^{-1} \left(  \mathbbm{1} +  \widetilde {\cal H} \overline {\cal H}^{-1} \right)^{-1}
   = 
    \overline {\cal H}^{-1}  \sum_{s=0}^\infty   (- \widetilde {\cal H} \overline {\cal H}^{-1})^{s} 
     =   \overline {\cal H}^{-1} - \overline {\cal H}^{-1} \widetilde {\cal H} \overline {\cal H}^{-1}
         +   \overline {\cal H}^{-1}  \sum_{s=2}^\infty   (- \widetilde {\cal H} \overline {\cal H}^{-1})^{s} .
\end{align*}
Note that $ \left\| \overline {\cal H}^{-1}  \sum_{s=2}^\infty   (- \widetilde {\cal H} \overline {\cal H}^{-1})^{s}  \right\|
  \leq \left\| \overline {\cal H}^{-1} \right\|  \sum_{s=2}^\infty  
  \left( \left\| \overline {\cal H}^{-1} \right \|      \left\|  \widetilde {\cal H} \right \| \right)^s$,
and therefore
\begin{align*}
   \left\| {\cal H}^{-1} - \left( \overline {\cal H}^{-1} - \overline {\cal H}^{-1} \widetilde {\cal H} \overline {\cal H}^{-1} \right)
   \right\|
   \leq  \frac{ \left\| \overline {\cal H}^{-1} \right \|^3      \left\|  \widetilde {\cal H} \right \|^2 }
         { 1  -  \left\| \overline {\cal H}^{-1} \right \|       \left\|  \widetilde {\cal H} \right \| } 
         =  o_P\left(  (NT)^{-1/4} \right),
\end{align*}
by Assumption~\ref{ass:A1}$(vi)$ and equation (\ref{eq:lb52}).

The results for 
$\left\| {\cal H}^{-1}    \right\| $ and
$\left\| {\cal H}^{-1}  - \overline {\cal H}^{-1}  \right\| $ follow immediately.
 \end{proof}

\subsection{Legendre Transformed Objective Function}

 We consider the shrinking neighborhood
${\cal B}(r_\beta, \beta^0) \times {\cal B}_q(r_\phi, \phi^0)$ of the true parameters $(\beta^0,\phi^0)$.
Statement~$(i)$ of Lemma~\ref{lemma:assA1add} implies
that the objective function ${\cal L}(\beta,\,\phi)$
is strictly concave in $\phi$ in this shrinking neighborhood wpa1.
We define
\begin{align}
   {\cal L}^*(\beta,\,S) \, &= \,
      \max_{\phi \in {\cal B}_q(r_\phi, \phi^0)} \left[ {\cal L}(\beta,\,\phi)
       -  \phi' S \right]   ,
  &
  \quad \Phi(\beta,\, S) \, &= \,
      \argmax_{\phi \in {\cal B}_q(r_\phi, \phi^0)} \left[ {\cal L}(\beta,\,\phi)
       -  \phi' S \right]  ,
    \label{DefLS}
\end{align}
where $\beta \in {\cal B}(r_\beta, \beta^0)$ and $S \in \mathbbm{R}^{\dim \phi}$.
The function ${\cal L}^*(\beta,\,S)$ is the Legendre
transformation of the objective function ${\cal
L}(\beta,\,\phi)$ in the incidental parameter $\phi$. We
denote the parameter $S$ as the dual parameter to $\phi$, and
${\cal L}^*(\beta,\,S)$ as the dual function to ${\cal
L}(\beta,\,\phi)$. We only consider ${\cal L}^*(\beta,\,S)$
and $\Phi(\beta,\, S)$ for parameters
$\beta \in {\cal B}(r_\beta, \beta^0) $ and $S \in {\cal S}(\beta,{\cal B}_q(r_\phi, \phi^0))$,
where the optimal $\phi$ is defined by the first order conditions, i.e. is not a boundary solution.
We define the corresponding set of pairs $(\beta,{\cal S})$ that is dual
to  ${\cal B}(r_\beta, \beta^0) \times {\cal B}_q(r_\phi, \phi^0)$ by
\begin{align*}
     {\cal SB}_r(\beta^0,\phi^0)
     &= \left\{ (\beta,{\cal S})\in \mathbbm{R}^{\dim \beta + \dim \phi} \; : \;
              (\beta,  \Phi(\beta,\, S)) \in {\cal B}(r_\beta, \beta^0) \times {\cal B}_q(r_\phi, \phi^0)
           \right\} .   
\end{align*}
Assumption~\ref{ass:A1} guarantees that  for $\beta \in {\cal B}(r_\beta, \beta^0)$ the domain
${\cal S}(\beta,{\cal B}_q(r_\phi, \phi^0))$ includes $S=0$, the origin of $\mathbbm{R}^{\dim \phi}$,
as an interior point, wpa1, and that ${\cal L}^*(\beta,\,S)$ is four times differentiable in a neighborhood
of $S=0$ (see Lemma \ref{lemma:Lstar} below). 
The optimal $\phi=\Phi(\beta,S)$ in equation \eqref{DefLS}
satisfies the first order condition $S = {\cal S}(\beta,\phi)$. Thus,
for given $\beta$, the functions $\Phi(\beta,S)$ and
${\cal S}(\beta,\phi)$ are inverse to each other, and the relationship
between  $\phi$ and its dual $S$ is one-to-one. This is a consequence of 
strict concavity of ${\cal L}(\beta,\,\phi)$ in the neighborhood of the true parameter value that we consider here.\footnote{Another consequence of strict concavity of ${\cal L}(\beta,\,\phi)$ is that
the dual function ${\cal L}^*(\beta,\,S)$ is strictly convex in $S$. The original ${\cal L}(\beta,\,\phi)$
can be recovered from ${\cal L}^*(\beta,\,S)$ by again performing
a Legendre transformation, namely
\begin{align*}
   {\cal L}(\beta,\,\phi) \, &= \,
      \min_{S \in \mathbbm{R}^{\dim \phi}}
      \left[ {\cal L}^*(\beta,\,S)
       +  \phi' S \right] \; .
\end{align*}
}
One can show that
\begin{align*}
   \Phi(\beta,S) = \, - \,
     \frac{\partial {\cal L}^*(\beta,\,S)}  {\partial S} \; ,
\end{align*}
which shows the dual nature of the functions ${\cal L}(\beta,\,\phi)$ and ${\cal L}^*(\beta,\,S)$.
For $S = 0$ the optimization in \eqref{DefLS} is just over
the objective function ${\cal L}(\beta,\phi)$, so that $\Phi(\beta,0) = \widehat \phi(\beta)$
and ${\cal L}^*(\beta,0)={\cal L}(\beta,\widehat \phi(\beta))$,
the profile objective function. We already introduced
${\cal S} = {\cal S}(\beta^0,\phi^0)$, i.e. at $\beta=\beta^0$ the dual of $\phi^0$
is ${\cal S}$, and vica versa.
We can write the profile objective function
${\cal L}(\beta,\widehat \phi(\beta)) = {\cal L}^*(\beta,0)$
as a Taylor series expansion of ${\cal L}^*(\beta,\, S)$
around $(\beta,S)=(\beta^0,{\cal S})$, namely
\begin{align*}
   {\cal L}(\beta,\widehat \phi(\beta))  
    &=  {\cal L}^*(\beta^0, {\cal S})
      + (\partial_{\beta'} {\cal L}^*) \Delta \beta
      - \Delta \beta'  (\partial_{\beta S'} {\cal L}^*)  {\cal S}
      + \frac 1 2 \Delta \beta'  (\partial_{\beta \beta'} {\cal L}^*)  \Delta \beta
        + \ldots \; ,
\end{align*}
where $\Delta \beta = \beta -\beta^0$,
and here and in the following we omit the arguments of  ${\cal L}^*(\beta,S)$ and of its partial derivatives when
they are evaluated at $(\beta^0,{\cal S})$.
Analogously, we can obtain Taylor expansions for the profile score
$\partial_{\beta} {\cal L}(\beta,\widehat \phi(\beta))=\partial_{ \beta} {\cal L}^*(\beta,0)$
and the estimated nuisance parameter $\widehat \phi(\beta) = - \partial_S {\cal L}^*(\beta,0)$
in $\Delta \beta$ and ${\cal S}$, see the proof of Theorem~\ref{th:ScoreExpansion} below.
Apart from combinatorial factors
those expansions feature the same coefficients
as the expansion of ${\cal L}(\beta,\widehat \phi(\beta))$ itself.
They are standard Taylor expansions that can be truncated at a certain order, and the remainder term
can be bounded by applying the mean value theorem.

The functions ${\cal L}(\beta,\, \phi)$ and its dual
${\cal L}^*(\beta,\, S)$ are closely related. In particular,
 for given $\beta$
their first derivatives with respect to the second argument
${\cal S}(\beta,\phi)$ and $\Phi(\beta,S)$ are
inverse functions of each other.
We can therefore express partial derivatives of ${\cal L}^*(\beta,\, S)$
in terms of partial derivatives of ${\cal L}(\beta,\, \phi)$.
This is done in  Lemma~\ref{lemma:Lstar}.
The norms
$\left\| \partial_{\beta S S S}  {\cal L}^*(\beta,S) \right\|_q$,
$\left\| \partial_{SSSS}  {\cal L}^*(\beta,S) \right\|_q$, etc.,
are defined as in equation \eqref{DefNorm} and \eqref{DefNorm2}.

\begin{lemma}
     \label{lemma:Lstar}
     Let assumption \ref{ass:A1} be satisfied.

     \begin{itemize}
     \item[(i)]  The function ${\cal L}^*(\beta,S)$ is well-defined and is four times continuously differentiable  
         in ${\cal SB}_r(\beta^0,\phi^0)$, wpa1.

     \item[(ii)] For  $ {\cal L}^* =  {\cal L}^*(\beta^0,{\cal S}),$
    \begin{align*}
         \partial_S {\cal L}^* &= - \phi^0 , \ \
         \partial_\beta {\cal L}^* = \partial_\beta  {\cal L} , \ \
         \partial_{SS'} {\cal L}^* =  - (\partial_{\phi \phi'} {\cal L} )^{-1} = {\cal H}^{-1} , \ \ 
         \partial_{\beta S'}  {\cal L}^* = - ( \partial_{\beta \phi'} {\cal L} ) {\cal H}^{-1} , \\
         \partial_{\beta \beta'} {\cal L}^* &=   \partial_{\beta \beta'} {\cal L}
                                          + ( \partial_{\beta \phi'} {\cal L} ) {\cal H}^{-1}
                                            ( \partial_{\phi' \beta} {\cal L} ) , 
                        \ \
         \partial_{SS' S_g} {\cal L}^*   =
             - \sum_h   {\cal H}^{-1}    (\partial_{\phi \phi' \phi_h} {\cal L})   {\cal H}^{-1}   (  {\cal H}^{-1} )_{gh}   , 
            \\
         \partial_{\beta_k S S'} {\cal L}^* &=       
             {\cal H}^{-1} (\partial_{\beta_k \phi' \phi} {\cal L})           {\cal H}^{-1}
             +  \sum_{g}     {\cal H}^{-1}  (\partial_{\phi_g \phi' \phi} {\cal L})           {\cal H}^{-1}    
                [        {\cal H}^{-1}      \partial_{\beta_k \phi} {\cal L} ]_g            ,
            \\
         \partial_{\beta_k \beta_l S'} {\cal L}^* &=       
           - ( \partial_{\beta_k \beta_l \phi'} {\cal L} ) {\cal H}^{-1}   
             -   ( \partial_{\beta_l \phi'} {\cal L} ) {\cal H}^{-1}  ( \partial_{\beta_k \phi \phi'} {\cal L} ) {\cal H}^{-1}  
           - ( \partial_{\beta_k \phi'} {\cal L} )
           {\cal H}^{-1} (\partial_{\beta_l \phi' \phi} {\cal L})           {\cal H}^{-1}
        \nonumber \\ & \qquad   
             -  \sum_{g}    ( \partial_{\beta_k \phi'} {\cal L} )   {\cal H}^{-1}  (\partial_{\phi_g \phi' \phi} {\cal L})           {\cal H}^{-1}    
                [        {\cal H}^{-1}      \partial_{\beta_l \phi} {\cal L} ]_g    ,
                        \\
         \partial_{\beta_k \beta_l \beta_m} {\cal L}^* &=   \partial_{\beta_k \beta_l \beta_m} {\cal L}
                            + \sum_{g}    ( \partial_{\beta_k \phi'} {\cal L} )    {\cal H}^{-1}  (\partial_{\phi_g \phi' \phi} {\cal L})  {\cal H}^{-1}    
                              (\partial_{\beta_l \phi} {\cal L})
                [        {\cal H}^{-1}      \partial_{\phi \beta_m} {\cal L} ]_g
                        \nonumber \\ & \quad          
                                             + ( \partial_{\beta_k \phi'} {\cal L} ) {\cal H}^{-1}
                                            ( \partial_{\beta_l \phi' \phi } {\cal L} )   {\cal H}^{-1}      \partial_{\phi \beta_m} {\cal L} 
                                             + ( \partial_{\beta_m \phi'} {\cal L} ) {\cal H}^{-1}
                                            ( \partial_{\beta_k \phi' \phi } {\cal L} )   {\cal H}^{-1}      \partial_{\phi \beta_l} {\cal L} 
                        \nonumber \\ & \quad          
                                             + ( \partial_{\beta_l \phi'} {\cal L} ) {\cal H}^{-1}
                                            ( \partial_{\beta_m \phi' \phi } {\cal L} )   {\cal H}^{-1}      \partial_{\phi \beta_k} {\cal L} 
                        \nonumber \\ & \quad          
                                  +  (\partial_{\beta_k \beta_l \phi'} {\cal L})  {\cal H}^{-1}   ( \partial_{\phi' \beta_m} {\cal L} )
                                  +  (\partial_{\beta_k \beta_m \phi'} {\cal L})  {\cal H}^{-1}   ( \partial_{\phi' \beta_l} {\cal L} )
                                  +  (\partial_{\beta_l \beta_m \phi'} {\cal L})  {\cal H}^{-1}   ( \partial_{\phi' \beta_k} {\cal L} ) ,
         \end{align*}
and
       \begin{align*}
         \partial_{SS' S_g S_h} {\cal L}^*   &=
             \sum_{f,e}   {\cal H}^{-1}    (\partial_{\phi \phi' \phi_f \phi_e} {\cal L})   {\cal H}^{-1}   
                 (  {\cal H}^{-1} )_{gf}      (  {\cal H}^{-1} )_{he}   
         \nonumber \\    & \quad
            +  3 \sum_{f,e}   {\cal H}^{-1}    (\partial_{\phi \phi' \phi_e} {\cal L})   {\cal H}^{-1}     (\partial_{\phi \phi' \phi_f} {\cal L})   {\cal H}^{-1}  
             (  {\cal H}^{-1} )_{gf}    (  {\cal H}^{-1} )_{he}       ,   
           \\
         \partial_{\beta_k SS' S_g} {\cal L}^*   &=
             - \sum_h    {\cal H}^{-1} (\partial_{\beta_k \phi' \phi} {\cal L})           {\cal H}^{-1}    (\partial_{\phi \phi' \phi_h} {\cal L})   {\cal H}^{-1}   [  {\cal H}^{-1} ]_{gh}   
         \nonumber \\    & \quad
             - \sum_h   {\cal H}^{-1}    (\partial_{\phi \phi' \phi_h} {\cal L})    {\cal H}^{-1} (\partial_{\beta_k \phi' \phi} {\cal L})           {\cal H}^{-1}   [  {\cal H}^{-1} ]_{gh}   
         \nonumber \\    & \quad
             - \sum_h   {\cal H}^{-1}    (\partial_{\phi \phi' \phi_h} {\cal L})   {\cal H}^{-1}   [   {\cal H}^{-1} (\partial_{\beta_k \phi' \phi} {\cal L})           {\cal H}^{-1} ]_{gh}   
         \nonumber \\    & \quad
             - \sum_{h,f}     {\cal H}^{-1}  (\partial_{\phi_f \phi' \phi} {\cal L})           {\cal H}^{-1}    
                  (\partial_{\phi \phi' \phi_h} {\cal L})   {\cal H}^{-1}   [  {\cal H}^{-1} ]_{gh}  
                    [        {\cal H}^{-1}      \partial_{\beta_k \phi} {\cal L} ]_f 
         \nonumber \\    & \quad
             - \sum_{h,f}   {\cal H}^{-1}    (\partial_{\phi \phi' \phi_h} {\cal L})   {\cal H}^{-1}  (\partial_{\phi_f \phi' \phi} {\cal L})           {\cal H}^{-1}    
                 [  {\cal H}^{-1} ]_{gh}   
                  [        {\cal H}^{-1}      \partial_{\beta_k \phi} {\cal L} ]_f 
         \nonumber \\    & \quad
             - \sum_{h,f}   {\cal H}^{-1}    (\partial_{\phi \phi' \phi_h} {\cal L})   {\cal H}^{-1}  
              [    {\cal H}^{-1}  (\partial_{\phi_f \phi' \phi} {\cal L})           {\cal H}^{-1}    ]_{gh}
               [        {\cal H}^{-1}      \partial_{\beta_k \phi} {\cal L} ]_f   
         \nonumber \\    & \quad
             - \sum_h   {\cal H}^{-1}    (\partial_{\beta_k \phi \phi' \phi_h} {\cal L})   {\cal H}^{-1}   [  {\cal H}^{-1} ]_{gh}   
         \nonumber \\    & \quad
             - \sum_{h,f}   {\cal H}^{-1}    (\partial_{\phi \phi' \phi_h \phi_f} {\cal L})   {\cal H}^{-1}   [  {\cal H}^{-1} ]_{gh} 
               [  {\cal H}^{-1} (\partial_{\beta_k \phi} {\cal L})  ]_{f}  .
    \end{align*}

    
    \item[(iii)]
    Moreover,
    \begin{align*}
        \sup_{(\beta,S) \in {\cal SB}_r(\beta^0,\phi^0)}
             \left\|  \partial_{\beta \beta  \beta}  {\cal L}^*(\beta,S) \right\| &=
              {\cal O}_P\left( (NT)^{1/2+1/(2q)+\epsilon}   \right), \\
       \sup_{(\beta,S) \in {\cal SB}_r(\beta^0,\phi^0)}
             \left\| \partial_{\beta \beta  S}  {\cal L}^*(\beta,S) \right\|_q &= {\cal O}_P\left( (NT)^{ 1/q+\epsilon}   \right), \\
        \sup_{(\beta,S) \in {\cal SB}_r(\beta^0,\phi^0)}
                \left\| \partial_{\beta S S}  {\cal L}^*(\beta,S) \right\|_q &={\cal O}_P\left( (NT)^{ 1/(2q)+\epsilon}   \right), \\         
       \sup_{(\beta,S) \in {\cal SB}_r(\beta^0,\phi^0)}
            \left\| \partial_{\beta S S S}  {\cal L}^*(\beta,S) \right\|_q &= {\cal O}_P\left( (NT)^{ 1/(2q)+2 \epsilon}   \right), \\         
        \sup_{(\beta,S) \in {\cal SB}_r(\beta^0,\phi^0)}   \left\| \partial_{SSSS}  {\cal L}^*(\beta,S) \right\|_q 
             &=  {\cal O}_P\left( (NT)^{  2 \epsilon}   \right) . 
    \end{align*}    
   
    \end{itemize}
\end{lemma}

\begin{proof}[\bf Proof of Lemma~\ref{lemma:Lstar}]
     \#Part $(i)$:   
     According to the definition  \eqref{DefLS},
     ${\cal L}^*(\beta,S) = {\cal L}(\beta,\Phi(\beta,S)) - \Phi(\beta,S)' S$,
     where $\Phi(\beta,S)$ solves the FOC, 
     ${\cal S}(\beta,\Phi(\beta,S)) = S$, i.e. ${\cal S}(\beta,.)$
     and $\Phi(\beta,.)$ are inverse functions for every $\beta$. 
      Taking the derivative of      
      ${\cal S}(\beta,\Phi(\beta,S)) = S$ wrt to both  $S$ and $\beta$ yields
      \begin{align}
         [\partial_S \Phi(\beta,S)'] [\partial_\phi {\cal S}(\beta,\Phi(\beta,S))']   &= \mathbbm{1} ,
        \nonumber    \\
      [\partial_\beta {\cal S}(\beta,\Phi(\beta,S))'] +
           [\partial_\beta \Phi(\beta,S)'] [\partial_\phi {\cal S}(\beta,\Phi(\beta,S))']   &= 0.
        \label{LegendrePDE}
     \end{align}       
     By definition, ${\cal S}={\cal S}(\beta^0,\phi^0)$.   
     Therefore, $\Phi(\beta,S)$ is the unique function that satisfies
     the boundary condition $\Phi(\beta^0,{\cal S})=\phi^0$
     and the system of partial differential equations (PDE) in \eqref{LegendrePDE}. Those PDE's can equivalently be written as
     \begin{align}
          \partial_S \Phi(\beta,S)' &= - [{\cal H}(\beta, \Phi(\beta,S))]^{-1} ,
        \nonumber \\
           \partial_\beta \Phi(\beta,S)' &=  [\partial_{\beta \phi'} {\cal L}(\beta,\Phi(\beta,S))]   [{\cal H}(\beta, \Phi(\beta,S))]^{-1} .
        \label{LegendrePDE2}
     \end{align}
     This shows that $\Phi(\beta,S)$ (and thus ${\cal L}^*(\beta,S)$) are well-defined in any neighborhood
     of $(\beta,S)=(\beta^0,{\cal S})$ in which ${\cal H}(\beta, \Phi(\beta,S))$ is invertible
     (inverse function theorem).
      Lemma~\ref{lemma:assA1add} shows that  ${\cal H}(\beta, \phi)$
      is invertible in ${\cal B}(r_\beta, \beta^0) \times {\cal B}_q(r_\phi, \phi^0)$, wpa1.
     The inverse function theorem thus guarantee that
     $\Phi(\beta,S)$ and ${\cal L}^*(\beta,S)$ are well-defined in 
     $ {\cal SB}_r(\beta^0,\phi^0)$.
      The partial derivatives of ${\cal L}^*(\beta,S)$ of up to fourth order can be expressed  as continuous transformations of the partial
     derivatives of ${\cal L}(\beta,\phi)$ up to fourth order (see e.g. proof of part $(ii)$ of the lemma). Hence, 
     ${\cal L}^*(\beta,S)$ is four times continuously differentiable because ${\cal L}(\beta,\phi)$ is  four times continuously differentiable.

     \bigskip

     \noindent
   \#Part $(ii)$:  
   Differentiating ${\cal L}^*(\beta,S) = {\cal L}(\beta,\Phi(\beta,S)) - \Phi(\beta,S)' S$ wrt $\beta$ and $S$
   and using the FOC of the maximization over $\phi$ in the definition of  ${\cal L}^*(\beta,S)$ gives
   $\partial_\beta {\cal L}^*(\beta,S) = \partial_\beta {\cal L}(\beta,\Phi(\beta,S))$ and
   $\partial_S {\cal L}^*(\beta,S) = - \Phi(\beta,S)$, respectively. Evaluating this expression
   at $(\beta,S)=(\beta^0,{\cal S})$ gives the first two statements of part $(ii)$.
   
   Using $\partial_S {\cal L}^*(\beta,S) = - \Phi(\beta,S)$, the PDE  \eqref{LegendrePDE2} can be written as
     \begin{align*}
          \partial_{SS'} {\cal L}^*(\beta,S) &=  {\cal H}^{-1}(\beta, \Phi(\beta,S)) ,
        \nonumber \\
           \partial_{\beta S'} {\cal L}^*(\beta,S) &= -  [\partial_{\beta \phi'} {\cal L}(\beta,\Phi(\beta,S))]   {\cal H}^{-1}(\beta, \Phi(\beta,S)) .
     \end{align*}
      Evaluating this expression
   at $(\beta,S)=(\beta^0,{\cal S})$ gives the next two statements of part $(ii)$.
   
   Taking the derivative  of $\partial_\beta {\cal L}^*(\beta,S) = \partial_\beta {\cal L}(\beta,\Phi(\beta,S))$ wrt to $\beta$
   and using the second equation of \eqref{LegendrePDE2} gives the next statement 
   when evaluated at $(\beta,S)=(\beta^0,{\cal S})$.
   
   Taking the derivative  of
    $\partial_{SS'} {\cal L}^*(\beta,S) =  - [\partial_{\phi \phi'} {\cal L}(\beta, \Phi(\beta,S))]^{-1}$ wrt to $S_g$
    and using the first equation of \eqref{LegendrePDE2} gives the next statement 
   when evaluated at $(\beta,S)=(\beta^0,{\cal S})$.
   
    Taking the derivative of
    $\partial_{SS'} {\cal L}^*(\beta,S) =  - [\partial_{\phi \phi'} {\cal L}(\beta, \Phi(\beta,S))]^{-1}$  wrt to $\beta_k$ 
    and using the second equation of \eqref{LegendrePDE2} gives 
    \begin{align}
         \partial_{\beta_k S S'} {\cal L}^*(\beta,S) &=       
             {\cal H}^{-1}(\beta,\phi) [\partial_{\beta_k \phi' \phi} {\cal L}(\beta,\phi)]           {\cal H}^{-1}(\beta,\phi)
        \nonumber \\ & \quad
             +  \sum_{g}     {\cal H}^{-1}(\beta,\phi)  [\partial_{\phi_g \phi' \phi} {\cal L}(\beta,\phi)] 
                       {\cal H}^{-1}(\beta,\phi)    
                \{        {\cal H}^{-1}(\beta,\phi)     [ \partial_{\beta_k \phi} {\cal L}(\beta,\phi)] \}_g ,
        \label{PartialBSS}
    \end{align}    
    where $\phi = \Phi(\beta,S)$.
     This becomes the next statement 
   when evaluated at $(\beta,S)=(\beta^0,{\cal S})$.
   
   We omit the proofs for 
   $ \partial_{\beta_k \beta_l S'} {\cal L}^*$,
            $\partial_{\beta_k \beta_l S} {\cal L}^*$,
     $\partial_{SS' S_g S_h} {\cal L}^*$ and
     $\partial_{\beta_k SS' S_g} {\cal L}^*$ because they are analogous.

  \bigskip
   
      \noindent
   \#Part $(iii)$: We only show the result 
   for $  \left\| \partial_{\beta S S}  {\cal L}^*(\beta,S) \right\|_q$, the proof of the other statements is analogous.
   By equation \eqref{PartialBSS} 
   \begin{align*}
         \left\| \partial_{\beta S S}  {\cal L}^*(\beta,S) \right\|_q
         &\leq 
          \left\|   {\cal H}^{-1}(\beta,\phi) \right\|_q^2
          \left\| \partial_{\beta \phi \phi} {\cal L}(\beta,\phi) \right\|_q
           +   \left\|   {\cal H}^{-1}(\beta,\phi) \right\|_q^3  
              \left\| \partial_{\phi \phi \phi} {\cal L}(\beta,\phi) \right\|_q
              \left\| \partial_{\beta \phi'} {\cal L}(\beta,\phi) \right\|_q ,
   \end{align*}
   where $\phi = \Phi(\beta,S)$. 
   Then, by Lemma~\ref{lemma:assA1add} 
    \begin{align*}
          \sup_{(\beta,S) \in {\cal SB}_r(\beta^0,\phi^0)}
             \left\| \partial_{\beta S S}  {\cal L}^*(\beta,S) \right\|_q
     & \leq 
      \sup_{\beta \in {\cal B}(r_\beta, \beta^0)}  \sup_{\phi \in {\cal B}_q(r_\phi, \phi^0)}         
       \bigg[  
          \left\|   {\cal H}^{-1}(\beta,\phi) \right\|_q^2
          \left\| \partial_{\beta \phi \phi} {\cal L}(\beta,\phi) \right\|_q
        \nonumber \\ &          +   \left\|   {\cal H}^{-1}(\beta,\phi) \right\|_q^3  
              \left\| \partial_{\phi \phi \phi} {\cal L}(\beta,\phi) \right\|_q
                \left\| \partial_{\beta \phi'} {\cal L}(\beta,\phi) \right\|_q
           \bigg]
       = {\cal O}\left( (NT)^{1/(2q) + \epsilon} \right).    
   \end{align*}
   To derive the rest of the bounds 
   we can use that the expressions from
    part $(ii)$  hold not only for $(\beta^0, {\cal S})$,
    but also for other values $(\beta,S)$,
    provided that $(\beta,\Phi(\beta,S)$ is used as the argument on the rhs expressions.
\end{proof}

\subsection{Proofs of Theorem \ref{th:ScoreExpansion}, Corollary~\ref{cor:LimitBeta}, and Theorem~\ref{th:consistency}}

\begin{proof}[\bf Proof of Theorem \ref{th:ScoreExpansion},
    Part 1: Expansion of $\widehat \phi(\beta)$.]
     Let $\beta = \beta_{NT}  \in {\cal B}(\beta^0,r_\beta)$. 
     A Taylor expansion of $\partial_S {\cal L}^*(\beta,0)$ around $(\beta^0,{\cal S})$ gives
     \begin{align*}
          \widehat \phi(\beta) =   -  \partial_S {\cal L}^*(\beta,0)  
                           =  - \partial_S {\cal L}^*
                              - (\partial_{S \beta'} {\cal L}^*) \Delta \beta 
                              + (\partial_{S S'} {\cal L}^*) {\cal S}
                              - \frac 1 2 \sum_g (\partial_{S S' S_g} {\cal L}^*)  {\cal S}  {\cal S}_g
                              + R^\phi(\beta) ,
     \end{align*}
     where we first expand in $\beta$ holding $S={\cal S}$ fixed, and then expand in $S$. 
      For any $v \in \mathbbm{R}^{\dim \phi}$
     the remainder term satisfies
     \begin{align*}
         v' R^\phi(\beta) &= v' \bigg\{ - \frac 1 2 \sum_k [\partial_{S \beta' \beta_k} {\cal L}^*(\tilde \beta,{\cal S})] (\Delta \beta) (\Delta \beta_k) 
       + \sum_k [\partial_{S S' \beta_k} {\cal L}^*(\beta^0,\tilde S)] {\cal S} (\Delta \beta_k)
       \nonumber \\ & \qquad  \quad
               + \frac 1 6 \sum_{g,h} [\partial_{S S' S_g S_h} {\cal L}^*(\beta^0,\bar S)]  {\cal S}  {\cal S}_g {\cal S}_h \bigg\} ,
     \end{align*}
     where $\tilde \beta$ is between $\beta^0$ and $\beta$, 
     and $\tilde S$ and $\bar S$ are between $0$ and ${\cal S}$.
     By part $(ii)$ of Lemma~\ref{lemma:Lstar},
    \begin{align*}
        \widehat \phi(\beta) - \phi^0 =   {\cal H}^{-1}  ( \partial_{\phi \beta'}  {\cal L} )   \Delta \beta
        +    {\cal H}^{-1}    {\cal S}
        + \ft 1 2  {\cal H}^{-1} \sum_g
                  (\partial_{\phi \phi' \phi_g}  {\cal L}  )  {\cal H}^{-1} {\cal S}
                   (   {\cal H}^{-1} {\cal S} )_g +  R^\phi(\beta).
    \end{align*} 
    Using that the vector norm $\|.\|_{q/(q-1)}$ is the dual to the vector norm $\|.\|_q$,
     Assumption~\ref{ass:A1}, and
      Lemmas~\ref{lemma:assA1add} and \ref{lemma:Lstar} yields
     \begin{align*}
          \left\| R^\phi(\beta) \right\|_q &=
           \sup_{\|v\|_{q/(q-1)}=1}  v' R^\phi(\beta)
       \nonumber \\
          &\leq  
            \frac 1 2 \left\| \partial_{S \beta \beta} {\cal L}^*(\tilde \beta,{\cal S}) \right\|_q  
                    \| \Delta \beta \|^2 
       + \left\| \partial_{S S \beta} {\cal L}^*(\beta^0,\tilde S) \right\|_q \| {\cal S} \|_q \|\Delta \beta \| 
       + \frac 1 6  \left\| \partial_{S S S S} {\cal L}^*(\beta^0,\bar S) \right\|_q
        \| {\cal S} \|_q^3
   \nonumber \\    
           &= {\cal O}_P\left[  (NT)^{1/q+\epsilon} r_\beta  \| \Delta \beta \| + (NT)^{-1/4+1/q+\epsilon}    \| \Delta \beta \|
              + (NT)^{-3/4+3/(2q)+2\epsilon}  \right]
   \nonumber \\    
           &=
           o_P\left( (NT)^{-1/2 + 1/(2q)} \right) + o_P\left( (NT)^{1/(2q)}  \| \beta - \beta^0\| \right) ,
     \end{align*}
     uniformly over $\beta \in {\cal B}(\beta^0,r_\beta)$ by Lemma~\ref{lemma:Lstar}.
\end{proof}

\begin{proof}[\bf Proof of Theorem \ref{th:ScoreExpansion},
    Part 2: Expansion of profile score.]
    Let $\beta = \beta_{NT} \in {\cal B}(\beta^0,r_\beta)$. 
      A Taylor expansion of $\partial_\beta {\cal L}^*(\beta,0) $ around $(\beta^0,{\cal S})$ gives
     \begin{align*}
         \partial_\beta {\cal L}(\beta,\widehat \phi(\beta)) = \partial_\beta {\cal L}^*(\beta,0) 
           = \partial_\beta {\cal L}^* +  (\partial_{\beta \beta'} {\cal L}^*) \Delta \beta
              -  (\partial_{\beta S'} {\cal L}^*) {\cal S}
              + \frac 1 2 \sum_g  (\partial_{\beta S' S_g} {\cal L}^*) {\cal S} {\cal S}_g
              + R_1(\beta) ,
     \end{align*}
     where we first expand in $\beta$ for fixed $S={\cal S}$, and then expand in $S$. 
     For any $v \in \mathbbm{R}^{\dim \beta}$
     the remainder term satisfies
     \begin{align*}
         v' R_1(\beta) &= v' \bigg\{ \frac 1 2 \sum_k [\partial_{\beta \beta' \beta_k} {\cal L}^*(\tilde \beta,{\cal S})] (\Delta \beta) (\Delta \beta_k) 
       - \sum_k [\partial_{\beta  \beta_k S'} {\cal L}^*(\beta^0,\tilde S)] {\cal S} (\Delta \beta_k)
       \nonumber \\ & \quad  \qquad 
               - \frac 1 6 \sum_{g,h} [\partial_{\beta S' S_g S_h} {\cal L}^*(\beta^0,\bar S)]  {\cal S}  {\cal S}_g {\cal S}_h  \bigg\},
     \end{align*}
     where $\tilde \beta$ is between $\beta^0$ and $\beta$,
     and $\tilde S$ and $\bar S$ are between $0$ and ${\cal S}$.
     By Lemma~\ref{lemma:Lstar},
     \begin{align*}
         \partial_\beta {\cal L}(\beta,\widehat \phi(\beta)) 
        &= \partial_\beta {\cal L} +
           \left[  \partial_{\beta \beta'} {\cal L}
                                          + ( \partial_{\beta \phi'} {\cal L} ) {\cal H}^{-1}
                                            ( \partial_{\phi' \beta} {\cal L} ) \right] (\beta-\beta^0)
              +  (\partial_{\beta \phi'}  {\cal L})  {\cal H}^{-1} {\cal S}
      \nonumber \\ & \quad        
              + \frac 1 2 \sum_g  \left( \partial_{\beta \phi' \phi_g}  {\cal L}
               +  [\partial_{\beta \phi'}  {\cal L}] \,  {\cal H}^{-1}
            [\partial_{\phi \phi' \phi_g}  {\cal L}]  \right)
              [ {\cal H}^{-1} {\cal S}]_g
               {\cal H}^{-1} {\cal S}
              +  R_1 (\beta) ,
     \end{align*}     
     where for any $v \in \mathbbm{R}^{\dim \beta}$,
     \begin{align*}
          \left\|  R_1(\beta) \right\| &=
           \sup_{\|v\|=1}  v' R_1(\beta)
          \nonumber \\
          & \leq \frac 1 2 \left\| \partial_{\beta \beta \beta} {\cal L}^*(\tilde \beta,{\cal S}) \right\|
                \| \Delta \beta \|^2 
       + (NT)^{1/2-1/q} \left\| \partial_{\beta  \beta S} {\cal L}^*(\beta^0,\tilde S) \right\|_q
       \| {\cal S}  \|_q \| \Delta \beta \|
       \nonumber \\ & \quad   
               + \frac 1 6 (NT)^{1/2-1/q} \left\| \partial_{\beta S S S} {\cal L}^*(\beta^0,\bar S) \right\|_q
                 \| {\cal S} \|_q^3
          \nonumber \\
           &= {\cal O}_P\left[ (NT)^{1/2+1/(2q)+\epsilon} r_\beta \| \Delta \beta \|
                 + (NT)^{1/4+1/(2q)+\epsilon}  \| \Delta \beta \| 
                 + (NT)^{-1/4+1/q+2\epsilon}
              \right]
          \nonumber \\
          &= o_P(1) + o_P(\sqrt{NT} \| \beta - \beta^0\|) ,
     \end{align*}
   uniformly over  $\beta \in {\cal B}(\beta^0,r_\beta)$ by Lemma~\ref{lemma:Lstar}. 
   We can also write
     \begin{align*}
         d_\beta {\cal L}(\beta,\widehat \phi(\beta)) 
         &= \partial_\beta {\cal L} -  \sqrt{NT} \, \overline W (\Delta \beta)
              +  (\partial_{\beta \phi'} \overline {\cal L}) \overline {\cal H}^{-1} {\cal S}
              +  (\partial_{\beta \phi'} \widetilde {\cal L}) \overline {\cal H}^{-1} {\cal S}
              -  (\partial_{\beta \phi'} \overline {\cal L}) \overline {\cal H}^{-1}   \widetilde {\cal H}    \overline {\cal H}^{-1} {\cal S}
      \nonumber \\ & \quad        
              + \frac 1 2 \sum_g  \left( \partial_{\beta \phi' \phi_g} \overline {\cal L}
               +  [\partial_{\beta \phi'} \overline {\cal L}] \, \overline {\cal H}^{-1}
            [\partial_{\phi \phi' \phi_g} \overline {\cal L}]  \right)
              [\overline {\cal H}^{-1} {\cal S}]_g
              \overline {\cal H}^{-1} {\cal S}
              +  R (\beta) ,
      \nonumber \\        
        &= U  -  \sqrt{NT} \, \overline W (\Delta \beta) +  R (\beta) ,
     \end{align*}   
     where we decompose the term linear in ${\cal S}$ into multiple terms by using that
     $$ -  (\partial_{\beta S'} {\cal L}^*) = (\partial_{\beta \phi'} {\cal L}) {\cal H}^{-1}
      = \left[(\partial_{\beta \phi'} \overline {\cal L}) + (\partial_{\beta \phi'} \widetilde {\cal L}) \right] 
      \left[ \overline {\cal H}^{-1} -  \overline {\cal H}^{-1}   \widetilde {\cal H}    \overline {\cal H}^{-1} + \ldots \right].$$
    The new remainder term is
     \begin{align*}
        R(\beta) &=   R_1(\beta) + 
         ( \partial_{\beta \beta'} \widetilde {\cal L}  ) \Delta \beta
              +    \left[ ( \partial_{\beta \phi'} {\cal L} ) {\cal H}^{-1}
                                            ( \partial_{\phi' \beta} {\cal L} )
                       - ( \partial_{\beta \phi'} \overline {\cal L} ) \overline {\cal H}^{-1}
                                            ( \partial_{\phi' \beta} \overline {\cal L} )    \right]       \Delta \beta   
       \nonumber \\ & \quad  \qquad                                   
          +  (\partial_{\beta \phi'} {\cal L}) \left[ {\cal H}^{-1} 
          - \left(  \overline {\cal H}^{-1} -  \overline {\cal H}^{-1}   \widetilde {\cal H}    \overline {\cal H}^{-1}  \right)  \right] {\cal S}  
         -   (\partial_{\beta \phi'} \widetilde {\cal L})        \overline {\cal H}^{-1}   \widetilde {\cal H}    \overline {\cal H}^{-1}  {\cal S} 
       \nonumber \\ & \quad \qquad
          +  \frac 1 2
          \left[   \sum_g  \partial_{\beta \phi' \phi_g}  {\cal L}
              [ {\cal H}^{-1} {\cal S}]_g
               {\cal H}^{-1} {\cal S}     
              -    \sum_g \partial_{\beta \phi' \phi_g} \overline {\cal L}
              [\overline {\cal H}^{-1} {\cal S}]_g
              \overline {\cal H}^{-1} {\cal S}                        
         \right]        
       \nonumber \\ & \quad \qquad
          +  \frac 1 2  
         \bigg[    \sum_g    [\partial_{\beta \phi'}  {\cal L}] \,  {\cal H}^{-1}
            [\partial_{\phi \phi' \phi_g}  {\cal L}]  
              [ {\cal H}^{-1} {\cal S}]_g
               {\cal H}^{-1} {\cal S}
               -  \sum_g    [\partial_{\beta \phi'} \overline {\cal L}] \, \overline {\cal H}^{-1}
            [\partial_{\phi \phi' \phi_g} \overline {\cal L}]  
              [\overline {\cal H}^{-1} {\cal S}]_g
              \overline {\cal H}^{-1} {\cal S}
         \bigg]  .
     \end{align*} 
     By Assumption~\ref{ass:A1} and
      Lemma~\ref{lemma:assA1add},
     \begin{align*}
          \left\|  R(\beta) \right\|  
           &\leq \left\|  R_1(\beta) \right\|
      +  \left\| \partial_{\beta \beta'} \widetilde {\cal L}  \right\|  \|\Delta \beta \|
           +    \left\|  \partial_{\beta \phi'} {\cal L}  \right\| 
                \left\|  {\cal H}^{-1} - \overline {\cal H}^{-1} \right\|
                  \left\| \partial_{\phi' \beta} {\cal L}  \right\| \|  \Delta \beta \|   
       \nonumber \\ & \quad   
              +    \left\|  \partial_{\beta \phi'} \widetilde {\cal L}  \right\| 
                \left\|   \overline {\cal H}^{-1} \right\|
                \left(  \left\| \partial_{\phi' \beta} {\cal L}  \right\| +   \left\| \partial_{\phi' \beta} \overline {\cal L}  \right\|  \right) \|  \Delta \beta \|      
       \nonumber \\ & \quad                                     
          +  \left\| \partial_{\beta \phi'} {\cal L} \right\| 
          \left\| {\cal H}^{-1} 
          - \left(  \overline {\cal H}^{-1} -  \overline {\cal H}^{-1}   \widetilde {\cal H}    \overline {\cal H}^{-1}  \right)  \right\| \| {\cal S} \|  
         +   \left\|       \overline {\cal H}^{-1} \right\|^2 
          \left\| \partial_{\beta \phi'} \widetilde {\cal L} \right\| 
           \left\|   \widetilde {\cal H}  \right\| \left\| {\cal S} \right\| 
       \nonumber \\ & \quad 
          +  \frac 1 2
           \left\| \partial_{\beta \phi \phi }  {\cal L} \right\|
            \left( \left\| {\cal H}^{-1}  \right\|  +  \left\| \overline {\cal H}^{-1}  \right\| \right)
            \left\|    {\cal H}^{-1} - \overline {\cal H}^{-1} \right\|  \| {\cal S} \|^2     
       \nonumber \\ & \quad 
          + \frac 1 2 
           \left\| \overline {\cal H}^{-1}  \right\|^2
           \left\| \partial_{\beta \phi \phi} \widetilde {\cal L} \right\|
              \| {\cal S} \|^2                              
       \nonumber \\ & \quad 
          +   \frac 1 2  
         \bigg\|    \sum_g    [\partial_{\beta \phi'}  {\cal L}] \,  {\cal H}^{-1}
            [\partial_{\phi \phi' \phi_g}  {\cal L}]  
              [ {\cal H}^{-1} {\cal S}]_g
               {\cal H}^{-1} {\cal S}
               -  \sum_g    [\partial_{\beta \phi'} \overline {\cal L}] \, \overline {\cal H}^{-1}
            [\partial_{\phi \phi' \phi_g} \overline {\cal L}]  
              [\overline {\cal H}^{-1} {\cal S}]_g
              \overline {\cal H}^{-1} {\cal S}
             \bigg\|
          \nonumber \\
           &=   \left\|  R_1(\beta) \right\|
              + o_P(1)
              + o_P(\sqrt{NT} \| \beta - \beta^0\|)
     + {\cal O}_P\left[  
                (NT)^{-1/8 + \epsilon + 1/(2q)} \right]
          \nonumber \\
          &= o_P(1) + o_P(\sqrt{NT} \| \beta - \beta^0\|) ,
     \end{align*}
  uniformly over $\beta \in {\cal B}(\beta^0,r_\beta)$. 
   Here we use that
   \begin{align*}
         & \left\|    \sum_g    [\partial_{\beta \phi'}  {\cal L}] \,  {\cal H}^{-1}
            [\partial_{\phi \phi' \phi_g}  {\cal L}]  
              [ {\cal H}^{-1} {\cal S}]_g
               {\cal H}^{-1} {\cal S}
               -  \sum_g    [\partial_{\beta \phi'} \overline {\cal L}] \, \overline {\cal H}^{-1}
            [\partial_{\phi \phi' \phi_g} \overline {\cal L}]  
              [\overline {\cal H}^{-1} {\cal S}]_g
              \overline {\cal H}^{-1} {\cal S}
             \right\|
         \nonumber \\
            &\leq
                \left\| \partial_{\beta \phi'}    {\cal L} \right\|
                 \left\| {\cal H}^{-1} - \overline {\cal H}^{-1} \right\|
               \left( \left\| {\cal H}^{-1} \right\|  + \left\| \overline {\cal H}^{-1} \right\|  \right)  \|{\cal S}\|
                 \left\|  \sum_g  \partial_{\phi \phi' \phi_g}  {\cal L} \, [{\cal H}^{-1}  {\cal S}]_g \right\|   
         \nonumber \\  & \quad +    
                \left\| \partial_{\beta \phi'}    {\cal L} \right\|
                 \left\| {\cal H}^{-1} - \overline {\cal H}^{-1} \right\|
               \left\| \overline {\cal H}^{-1} \right\|   \|{\cal S}\|
                 \left\|  \sum_g  \partial_{\phi \phi' \phi_g}  {\cal L} \, [\overline {\cal H}^{-1}  {\cal S}]_g \right\|   
         \nonumber \\  & \quad +           
                \left\| \partial_{\beta \phi'}  \widetilde {\cal L} \right\|
                \left\| \overline {\cal H}^{-1} \right\|^2    \|{\cal S}\|
                 \left\|  \sum_g  \partial_{\phi \phi' \phi_g}  {\cal L} \, [\overline {\cal H}^{-1}  {\cal S}]_g \right\|   
         \nonumber \\  & \quad +           
                \left\| \partial_{\beta \phi'}  \overline {\cal L} \right\|
                \left\| \overline {\cal H}^{-1} \right\|    
                 \left\|  \sum_{g,h}  \partial_{\phi \phi_g \phi_h}  \widetilde {\cal L} 
                 \, [\overline {\cal H}^{-1}  {\cal S}]_g [\overline {\cal H}^{-1}  {\cal S}]_h \right\|   .
  \end{align*}   
\end{proof}   

\begin{proof}[\bf Proof of Corollary~\ref{cor:LimitBeta}]
     $\widehat \beta$ solves the FOC
     \begin{align*}
           \partial_{\beta} {\cal L}(\widehat \beta, \widehat \phi(\widehat \beta)) &= 0.
     \end{align*}
     By  $\left\| \widehat \beta - \beta^0 \right\| = o_P(r_\beta)$ 
     and Theorem~\ref{th:ScoreExpansion},
     \begin{align*}
          0 = \partial_{\beta} {\cal L}(\widehat \beta, \widehat \phi(\widehat \beta))  =  U
            - \overline W \, \sqrt{NT} (\widehat \beta-\beta^0)
            +o_P(1) + o_P(\sqrt{NT} \| \widehat \beta-\beta^0 \|).
   \end{align*}
   Thus, 
    $\sqrt{NT} (\widehat \beta-\beta^0) =  \overline W^{-1} U+ o_P(1) + o_P(\sqrt{NT} \| \widehat \beta-\beta^0 \|)  
    =  \overline W_{\infty}^{-1} U+ o_P(1) + o_P(\sqrt{NT} \| \widehat \beta-\beta^0 \|) $,
    where we use that $\overline W = \overline W_{\infty} + o_P(1) $ is invertible wpa1
    and that $\overline W^{-1} = \overline W_{\infty}^{-1} + o_P(1)$.
    We conclude that $\sqrt{NT} (\widehat \beta-\beta^0) = {\cal O}_P(1)$ because $U = {\cal O}_P(1)$,
    and therefore 
      $\sqrt{NT} (\widehat \beta-\beta^0) =   \overline W_{\infty}^{-1} U+ o_P(1)$.
\end{proof}

\begin{proof}[\bf Proof of Theorem~\ref{th:consistency}]
      \# Consistency of $\widehat \phi(\beta)$:  
     Let $\eta = \eta_{NT} >0$ be such that  $\eta=o_P(r_\phi)$, 
     $(NT)^{-1/4+1/(2q)} = o_P(\eta)$,
     and $(NT)^{1/(2q)} r_\beta = o_P(\eta)$.
      For $\beta \in {\cal B}(r_\beta, \beta^0)$, define         
     \begin{align}
          \widehat \phi^*(\beta) &:= \argmin_{\{ \phi: \; \|\phi - \phi^0\|_q \leq   \eta \} }
          \| {\cal S}(\beta,\phi) \|_q .
          \label{OptOverS}
     \end{align}
     Then, $ \| {\cal S}(\beta, \widehat \phi^*(\beta)) \|_q \leq  \| {\cal S}(\beta,\phi^0) \|_q$, and therefore by a Taylor expansion of ${\cal S}(\beta,\phi^0)$ around $\beta=\beta^0$,
     \begin{align*}
         \| {\cal S}(\beta, \widehat \phi^*(\beta)) - {\cal S}(\beta,\phi^0) \|_q
         &\leq  \| {\cal S}(\beta, \widehat \phi^*(\beta)) \|_q + \| {\cal S}(\beta,\phi^0) \|_q
         \leq 2 \| {\cal S}(\beta,\phi^0) \|_q
       \nonumber \\  
                &\leq 2 \| {\cal S} \|_q + 2   \left\| \partial_{\phi \beta'} {\cal L}(\tilde \beta,\phi^0) \right\|_q \|\beta - \beta^0\|
       \nonumber \\  
           &=   {\cal O}_P\left[ (NT)^{-1/4+1/(2q)} + (NT)^{1/(2q)} \|\beta - \beta^0\| \right] ,
     \end{align*}
     uniformly over $\beta \in {\cal B}(r_\beta, \beta^0)$, where  $\tilde \beta$ is between $\beta^0$ and $\beta$, 
     and we use Assumption~\ref{ass:A1}$(v)$ and
      Lemma~\ref{lemma:assA1add}.
     Thus,
     \begin{align*}
        \sup_{\beta \in {\cal B}(r_\beta, \beta^0)}  \| {\cal S}(\beta, \widehat \phi^*(\beta)) - {\cal S}(\beta,\phi^0) \|_{q}
        =  {\cal O}_P\left[ (NT)^{-1/4+1/(2q)} + (NT)^{1/(2q)} r_\beta \right].
     \end{align*}     
    By  a Taylor expansion of $\Phi(\beta,S)$  around $S={\cal S}(\beta,\phi^0)$,
     \begin{align*}
             \left\| \widehat \phi^*(\beta) - \phi^0 \right\|_q
           &=  \left\|  \Phi(\beta, {\cal S}(\beta, \widehat \phi^*(\beta))) - \Phi(\beta,{\cal S}(\beta,\phi^0)) \right\|_q
             \leq 
                \left\| \partial_S \Phi(\beta,\tilde S)' \right\|_q
             \left\| {\cal S}(\beta, \widehat \phi^*(\beta)) - {\cal S}(\beta,\phi^0) \right\|_q
         \nonumber \\
            &=     \left\| {\cal H}^{-1} (\beta, \Phi(\beta,\tilde S)) \right\|_q
                  \left\| {\cal S}(\beta, \widehat \phi^*(\beta)) - {\cal S}(\beta,\phi^0) \right\|_q
             = {\cal O}_P(1) 
               \left\| {\cal S}(\beta, \widehat \phi^*(\beta)) - {\cal S}(\beta,\phi^0) \right\|_q ,
     \end{align*}
     where $\tilde S$ is between ${\cal S}(\beta, \widehat \phi^*(\beta))$ and ${\cal S}(\beta,\phi^0)$ and we use Lemma~\ref{lemma:assA1add}$(i)$.
     Thus,
     \begin{align*}
       \sup_{\beta \in {\cal B}(r_\beta, \beta^0)}    \left\| \widehat \phi^*(\beta) - \phi^0 \right\|_q 
      =  {\cal O}_P\left[ (NT)^{-1/4+1/(2q)} + (NT)^{1/(2q)} r_\beta \right] = o_P(\eta).
     \end{align*}
     This shows that $\widehat \phi^*(\beta)$ is an interior solution of the minimization problem
     \eqref{OptOverS}, wpa1.
     Thus,
     ${\cal S}(\beta,\widehat \phi^*(\beta))=0$, because  the objective function ${\cal L}(\beta,\phi)$ is strictly concave and differentiable, and therefore
     $\widehat \phi^*(\beta) = \widehat \phi(\beta)$. We conclude that 
       $\displaystyle \sup_{\beta \in {\cal B}(r_\beta, \beta^0)}
             \left\| \widehat \phi(\beta) - \phi^0 \right\|_q
          = {\cal O}_P(\eta) 
             = o_P(r_\phi)$.
       
       \bigskip
          \noindent
   \# Consistency of $\widehat \beta$:   
    We have already shown that Assumption~\ref{ass:A1}$(ii)$ is satisfied, in addition
   to the remaining parts of Assumption~\ref{ass:A1}, which we assume. 
   The bounds on the spectral norm in 
   Assumption~\ref{ass:A1}$(vi)$
   and in part $(ii)$ of
   Lemma~\ref{lemma:assA1add}
   can be used to show that $U = {\cal O}_P(  (NT)^{1/4} )$.  

   First, we consider the case ${\rm dim}(\beta) = 1$ first. The extension to ${\rm dim}(\beta) > 1$
   is discussed below.  Let $\eta = 2 (NT)^{-1/2}   \overline W^{-1}  | U |$.
   Our goal is to show that $\widehat \beta \in [ \beta^0 - \eta , \beta^0 + \eta]$.   
   By Theorem~\ref{th:ScoreExpansion},
    \begin{align*}
        \partial_{\beta} {\cal L}(\beta^0 + \eta, \widehat \phi(\beta^0 + \eta))
          &=   U
            - \overline W \, \sqrt{NT} \eta
            + o_P(1) + o_P( \sqrt{NT}  \eta ) 
           = o_P( \sqrt{NT}  \eta)
            - \overline W \, \sqrt{NT} \eta ,
       \nonumber \\
        \partial_{\beta} {\cal L}(\beta^0 - \eta, \widehat \phi(\beta^0 - \eta))
          &=   U
            + \overline W \, \sqrt{NT}  \eta
            + o_P(1) + o_P( \sqrt{NT}  \eta ) 
           = o_P( \sqrt{NT}  \eta)
            + \overline W \, \sqrt{NT}   \eta ,
   \end{align*}
   and therefore for sufficiently large $N,T$
   \begin{align*}
          \partial_{\beta} {\cal L}(\beta^0 + \eta, \widehat \phi(\beta^0 + \eta))
          &\leq 
         0
         \leq
         \partial_{\beta} {\cal L}(\beta^0 - \eta, \widehat \phi(\beta^0 - \eta)) .
   \end{align*}
   Thus, since $\partial_{\beta} {\cal L}(\widehat \beta, \widehat \phi(\widehat \beta)) = 0,$
   for sufficiently large $N,T$,
   \begin{align*}
          \partial_{\beta} {\cal L}(\beta^0 + \eta, \widehat \phi(\beta^0 + \eta))
          &\leq 
         \partial_{\beta} {\cal L}(\widehat \beta, \widehat \phi(\widehat \beta)) 
         \leq
         \partial_{\beta} {\cal L}(\beta^0 - \eta, \widehat \phi(\beta^0 - \eta)) .
   \end{align*}
   The profile objective ${\cal L}(\beta, \widehat \phi(\beta))$ is strictly concave in $\beta$ because  ${\cal L}(\beta,\phi)$ is strictly concave in $(\beta,\phi)$. Thus,
    $\partial_{\beta} {\cal L}(\beta, \widehat \phi(\beta))$ is strictly decreasing.
    The previous set of inequalities implies that  for sufficiently large $N,T$
   \begin{align*}
        \beta^0 + \eta &\geq \widehat \beta  \geq \beta^0 - \eta .
   \end{align*}
   We conclude that $\| \widehat \beta - \beta^0 \| \leq \eta =  {\cal O}_P(  (NT)^{-1/4} )$.
   This concludes the proof for ${\rm dim}(\beta) = 1$. 
   
   To generalize the proof to ${\rm dim}(\beta) > 1$ 
   we define $\beta_{\pm} = \beta^0 \pm \eta \, \frac{\widehat \beta- \beta^0 }{\|\widehat \beta - \beta^0\|} $.
   Let $\langle \beta_- , \beta_+ \rangle = \{ r \beta_- + (1-r) \beta_+ \; |\;  r \in [0,1] \}$ be the line segment
   between $\beta_-$ and $\beta_+$. By restricting attention to values $\beta \in \langle \beta_- , \beta_+ \rangle$
   we can repeat the above argument for the case ${\rm dim}(\beta) = 1$
   and thus show that $\widehat \beta \in \langle \beta_- , \beta_+ \rangle$, which implies
   $\| \widehat \beta - \beta^0 \| \leq \eta =  {\cal O}_P(  (NT)^{-1/4} )$.
\end{proof}

\subsection{Proof of Theorem \ref{th:DeltaExpansion}}

\begin{proof}[\bf Proof of Theorem \ref{th:DeltaExpansion}]
   A Taylor expansion of 
   $\Delta(\beta,\phi)$ around $(\beta^0, \phi^0)$ yields
  \begin{align*}
   \Delta(\beta,\phi)
    = \Delta
    +[\partial_{\beta'} \Delta] (\beta - \beta^0)
    +[ \partial_{\phi'} \Delta] (\phi - \phi^0)
     + \ft 1 2 (\phi - \phi^0)' [\partial_{\phi \phi'} \Delta]
           (\phi - \phi^0)
         + R_1^{\Delta}(\beta,\phi),
   \end{align*} 
with remainder term
\begin{align*}
    R_1^{\Delta}(\beta,\phi)
     &= \ft 1 2 (\beta - \beta^0)' 
       [\partial_{\beta \beta'} \Delta(\bar \beta,\phi)] (\beta - \beta^0) 
      +   (\beta - \beta^0)' 
       [\partial_{\beta \phi'} \Delta(\beta^0,\tilde \phi)] (\phi - \phi^0) 
  \nonumber \\ & \quad     
      + \ft 1 6
      \sum_g
       (\phi - \phi^0)' 
       [\partial_{\phi \phi' \phi_g} \Delta(\beta^0,\bar \phi)]
           (\phi - \phi^0) 
        [\phi - \phi^0]_g
       ,
\end{align*}
where
$\bar \beta$  is between $\beta$ and $\beta^0$, and 
$\tilde \phi$ and $\bar \phi$ are between $\phi$ and $\phi^0$. 

By assumption, $\| \widehat \beta - \beta^0 \| = o_P( (NT)^{-1/4} )$, and by the expansion of $\widehat \phi = \widehat \phi(\widehat \beta)$
in Theorem~\ref{th:ScoreExpansion},
\begin{align*}
   \| \widehat \phi - \phi^0 \|_q
   &\leq
 \left\| {\cal H}^{-1} \right\|_q  \left\| {\cal S} \right\|_q
          +  \left\| {\cal H}^{-1} \right\|_q
            \left\| \partial_{\phi \beta'}  {\cal L} \right\|_q
             \left\| \widehat \beta-\beta^0 \right\|_q
  + \ft 1 2 \left\| {\cal H}^{-1} \right\|_q^3 
      \left\| \partial_{\phi \phi \phi}   {\cal L}  \right\|_q
      \left\| {\cal S} \right\|_q^2
                       + \left\| R^{\phi}(\widehat \beta) \right\|_q
  \nonumber \\
        &= {\cal O}_P( (NT)^{-1/4 + 1/(2q)} ).          
\end{align*}
 Thus, for 
$\widehat R_1^{\Delta} := R_1^{\Delta}(\widehat \beta,\widehat \phi)$,
\begin{align*}
  \left| \widehat R_1^{\Delta} \right|
   &\leq 
   \ft 1 2 \| \widehat \beta - \beta^0 \|^2
   \sup_{\beta \in {\cal B}(r_\beta, \beta^0)}  \sup_{\phi \in {\cal B}_q(r_\phi, \phi^0)}
   \left\| \partial_{\beta \beta'} \Delta(\beta,\phi) \right\|
   \nonumber \\ & \quad     
     +  
       (NT)^{1/2-1/q}
       \| \widehat \beta - \beta^0 \|
      \| \widehat \phi - \phi^0 \|_q
      \sup_{\beta \in {\cal B}(r_\beta, \beta^0)}  \sup_{\phi \in {\cal B}_q(r_\phi, \phi^0)}
   \left\| \partial_{\beta \phi'} \Delta(\beta,\phi) \right\|_q
   \nonumber \\ & \quad     
     + \ft 1 6 
       (NT)^{1/2-1/q}
      \| \widehat \phi - \phi^0 \|_q^3 
      \sup_{\beta \in {\cal B}(r_\beta, \beta^0)}  \sup_{\phi \in {\cal B}_q(r_\phi, \phi^0)}
   \left\| \partial_{\phi \phi \phi} \Delta(\beta,\phi) \right\|_q
   \nonumber \\  
      &= o_P(1/\sqrt{NT}).
\end{align*}

Again by the expansion of 
$\widehat \phi = \widehat \phi(\widehat \beta)$
from Theorem~\ref{th:ScoreExpansion},
\begin{align}
  \widehat \delta - \delta & =  \Delta(\widehat \beta,\widehat \phi) - \Delta  
   =
    \left( \partial_{\beta'}  {\Delta}
      +
      [\partial_{\phi} \Delta]'    {\cal H}^{-1} [\partial_{\phi \beta'}  {\cal L}]
    \right) (\widehat \beta - \beta^0)  
  \nonumber \\ & \quad      
     + [\partial_{\phi} \Delta]'    {\cal H}^{-1}
      \left(
      {\cal S} + \ft 1 2
      \sum_{g=1}^{\dim \phi}
                  [\partial_{\phi \phi' \phi_g}  {\cal L}]    {\cal H}^{-1} {\cal S}
                   [   {\cal H}^{-1} {\cal S} ]_g             \right) 
             + \ft 1 2 \, {\cal S}'  {\cal H}^{-1}
         [\partial_{\phi \phi'}  \Delta]
          {\cal H}^{-1} {\cal S} + R_2^{\Delta},
   \label{ResDeltaSoFar}       
\end{align}
where
\begin{align*}
  \left| R_2^{\Delta} \right| &= 
  \left| R_1^{\Delta} + 
[\partial_{\phi} \Delta]' R^\phi(\widehat \beta)
   + \ft 1 2 (\widehat \phi - \phi^0 +  {\cal H}^{-1} {\cal S} )'
       [\partial_{\phi \phi'}  \Delta]
           (\widehat \phi - \phi^0 -  {\cal H}^{-1} {\cal S}) \right|
  \nonumber \\
    &\leq  
  \left| R_1^{\Delta} \right| + 
  (NT)^{1/2-1/q}
  \left\| \partial_{\phi} \Delta \right\|_q
  \left\| R^\phi(\widehat \beta) \right\|_q
\nonumber \\  & \quad 
   + \ft 1 2   (NT)^{1/2-1/q}
   \left\| \widehat \phi - \phi^0 +  {\cal H}^{-1} {\cal S} \right\|_q
   \left\| \partial_{\phi \phi'}  \Delta \right\|_q
    \left\| \widehat \phi - \phi^0 -  {\cal H}^{-1} {\cal S} \right\|_q  
 \nonumber \\ & 
    = o_P(1/\sqrt{NT}),          
\end{align*}
that uses  
$ \left\| \widehat \phi - \phi^0 -  {\cal H}^{-1} {\cal S} \right\|_q  
= {\cal O}_P\left( (NT)^{-1/2 + 1/q + \epsilon} \right)$.
From equation \eqref{ResDeltaSoFar},
the terms of the expansion for  $\widehat \delta - \delta$ are analogous to the terms of the 
expansion for the score in Theorem~\ref{th:ScoreExpansion}, with
$\Delta(\beta,\phi)$ taking the role of
$\frac 1 {\sqrt{NT}} \partial_{\beta_k} {\cal L}(\beta,\phi)$. 
%
\end{proof}

\section{Proofs of Appendix~\ref{app:ProofMain} (Theorem~\ref{th:connection})}\label{sec:s3}

\begin{proof}[\bf Proof of Theorem~\ref{th:connection}, Part $(i)$]
     Assumption~\ref{ass:A1}$(i)$  is satisfied because  
     $\lim_{N,T \rightarrow \infty} \frac{\dim \phi} {\sqrt{NT}} =  \lim_{N,T \rightarrow \infty} \frac{N+T} {\sqrt{NT}} = \kappa + \kappa^{-1}$. 
     
     Assumption~\ref{ass:A1}$(ii)$  is satisfied because $ \ell_{it}(\beta,\pi)$ and $(v'\phi)^2$ are four times continuously differentiable and the same is
     true for ${\cal L}(\beta,\phi)$.
     
     Let $\overline{\mathcal{D}} 
     =\diag \left(  \overline {\cal H}_{(\alpha \alpha)}^*, \overline {\cal H}_{(\gamma \gamma)}^*
    \right)$. Then, 
      $\left\| \overline{\mathcal{D}}^{-1} \right\|_{\infty} = {\cal O}_P(1)$ by Assumption~\ref{ass:PanelA1}$(v)$.
     By the properties of the matrix norms and    Lemma~\ref{lemma:HessianAdditive},
     $ \left\| \overline {\cal H}^{-1} - \overline{\mathcal{D}}^{-1} \right\|_{\infty}
      \leq (N+T)  \left\| \overline {\cal H}^{-1} - \overline{\mathcal{D}}^{-1} \right\|_{\max} =  {\cal O}_P(1)$.
      Thus,     $\left\| \overline {\cal H}^{-1} \right\|_q \leq \left\| \overline {\cal H}^{-1} \right\|_{\infty}
       \leq \left\| \overline{\mathcal{D}}^{-1} \right\|_{\infty}
          + \left\| \overline {\cal H}^{-1} - \overline{\mathcal{D}}^{-1} \right\|_{\infty}
          = {\cal O}_P(1)$ by  Lemma~\ref{lemma:matrix-q-norm}  and the triangle inequality. We conclude that  Assumption~\ref{ass:A1}$(iv)$ holds.    

    We now show that the assumptions of  Lemma~\ref{lemma:BasicRegularityPanel}  
    are satisfied: 
    \begin{itemize}
    \item[(i)]
    By Lemma~\ref{lemma:mixing_inequality}, $\chi_i = \frac 1 {\sqrt{T}} \sum_{t} \partial_{\beta_k} \ell_{it}$
    satisfies $\mathbb{E}_\phi( \chi_i^2 ) \leq B$. Thus, by independence across $i$ 
    \begin{align*}
           \mathbb{E}_\phi\left[ \left( \frac 1 {\sqrt{NT}} \sum_{i,t} \partial_{\beta_k} \ell_{it}
           \right)^2 \right] 
           &= \mathbb{E}_\phi\left[ \left( \frac 1 {\sqrt{N}} \sum_{i}  \chi_i
           \right)^2 \right] 
            = \frac 1 N \sum_i   \mathbb{E}_\phi  \chi_i^2 \leq B,
    \end{align*}
    and therefore
    $ \frac 1 {\sqrt{NT}} \sum_{i,t} \partial_{\beta_k} \ell_{it} = {\cal O}_P(1)$.
    Analogously,
     $ \frac 1 {NT} \sum_{i,t} \left\{ \partial_{\beta_k \beta_l} \ell_{it} 
                    - \mathbb{E}_\phi\left[ \partial_{\beta_k \beta_l} \ell_{it}  \right] \right\} 
                    = {\cal O}_P(1/\sqrt{NT}) = o_P(1)$.
   Next,                 
   \begin{align*}
        & \mathbb{E}_\phi \left( \sup_{\beta \in {\cal B}(r_\beta, \beta^0)}  \sup_{\phi \in {\cal B}_q(r_\phi, \phi^0)}
                  \frac 1 {NT} \sum_{i,t} \partial_{\beta_k \beta_l \beta_m} \ell_{it}( \beta, \pi_{it}) \right)^2   
     \nonumber \\   
        & \leq \mathbb{E}_\phi \left( \sup_{\beta \in {\cal B}(r_\beta, \beta^0)}  \sup_{\phi \in {\cal B}_q(r_\phi, \phi^0)}
                  \frac 1 {NT} \sum_{i,t} \left| \partial_{\beta_k \beta_l \beta_m} \ell_{it}( \beta, \pi_{it}) \right| \right)^2
         \leq      \mathbb{E}_\phi \left(  
                  \frac 1 {NT} \sum_{i,t} M(Z_{it}) \right)^2     
     \nonumber \\   
        & \leq \mathbb{E}_\phi  
                  \frac 1 {NT} \sum_{i,t} M(Z_{it})^2
                   =  \frac 1 {NT} \sum_{i,t} \mathbb{E}_\phi   M(Z_{it})^2 = {\cal O}_P(1),
   \end{align*} 
   and therefore 
   $ \sup_{\beta \in {\cal B}(r_\beta, \beta^0)}  \sup_{\phi \in {\cal B}_q(r_\phi, \phi^0)}
                  \frac 1 {NT} \sum_{i,t} \partial_{\beta_k \beta_l \beta_m} \ell_{it}( \beta, \pi_{it})  = {\cal O}_P(1)$.
    A similar argument gives          
    $\frac 1 {NT} \sum_{i,t} \partial_{\beta_k \beta_l} \ell_{it} = {\cal O}_P(1)$.
    
    \item[(ii)]
    For $\xi_{it}(\beta,\phi) = \partial_{\beta_k \pi} \ell_{it}( \beta, \pi_{it})$
      or $\xi_{it}(\beta,\phi) = \partial_{\beta_k \beta_l \pi} \ell_{it}( \beta, \pi_{it})$,
    \begin{align*}
        &  \mathbb{E}_\phi
         \left[ \sup_{\beta \in {\cal B}(r_\beta, \beta^0)}  \sup_{\phi \in {\cal B}_q(r_\phi, \phi^0)}
           \frac 1 T \sum_t \left| \frac 1 N \sum_i  \xi_{it}(\beta,\phi) \right|^q \right]
        \nonumber \\   &
        \leq \mathbb{E}_\phi
         \left[ \sup_{\beta \in {\cal B}(r_\beta, \beta^0)}  \sup_{\phi \in {\cal B}_q(r_\phi, \phi^0)}
           \frac 1 T \sum_t \left( \frac 1 N \sum_i  \left| \xi_{it}(\beta,\phi) \right| \right)^q \right]
        \nonumber \\   
        &\leq \mathbb{E}_\phi
         \left[  
           \frac 1 T \sum_t \left( \frac 1 N \sum_i  M(Z_{it}) \right)^q \right]
        \leq \mathbb{E}_\phi
         \left[  
           \frac 1 T \sum_t  \frac 1 N \sum_i  M(Z_{it})^q \right]   
        \nonumber \\   
          &=  \frac 1 T \sum_t  \frac 1 N \sum_i  \mathbb{E}_\phi M(Z_{it})^q = {\cal O}_P(1),
    \end{align*}
    i.e.  $\sup_{\beta \in {\cal B}(r_\beta, \beta^0)}  \sup_{\phi \in {\cal B}_q(r_\phi, \phi^0)}
           \frac 1 T \sum_t \left| \frac 1 N \sum_i  \xi_{it}(\beta,\phi) \right|^q
            = {\cal O}_P\left( 1 \right)$.
    Analogously, it follows that
    $\sup_{\beta \in {\cal B}(r_\beta, \beta^0)}  \sup_{\phi \in {\cal B}_q(r_\phi, \phi^0)}
           \frac 1 N \sum_i \left| \frac 1 T \sum_t  \xi_{it}(\beta,\phi) \right|^q
             = {\cal O}_P\left( 1 \right)$.

      \item[(iii)]
 For $\xi_{it}(\beta,\phi) = \partial_{\pi^r} \ell_{it}( \beta, \pi_{it})$, 
    with $r \in \{3,4\}$,
  or $\xi_{it}(\beta,\phi) = \partial_{\beta_k \pi^r} \ell_{it}( \beta, \pi_{it})$,  
   with $r \in \{2,3\}$,
   or $\xi_{it}(\beta,\phi) = \partial_{\beta_k \beta_l \pi^2} \ell_{it}( \beta, \pi_{it})$,
   \begin{align*}
       & \mathbb{E}_\phi \left[ 
       \left( \sup_{\beta \in {\cal B}(r_\beta, \beta^0)}  \sup_{\phi \in {\cal B}_q(r_\phi, \phi^0)}
      \max_i  \frac 1 T \sum_t | \xi_{it}(\beta,\phi) |  \right)^{(8+\nu)} \right]
     \nonumber  \\
       &= \mathbb{E}_\phi \left[
       \max_i  \left( \sup_{\beta \in {\cal B}(r_\beta, \beta^0)}  \sup_{\phi \in {\cal B}_q(r_\phi, \phi^0)}
       \frac 1 T \sum_t | \xi_{it}(\beta,\phi) |  \right)^{(8+\nu)} \right]
     \nonumber  \\
       &\leq \mathbb{E}_\phi \left[
       \sum_i  \left( \sup_{\beta \in {\cal B}(r_\beta, \beta^0)}  \sup_{\phi \in {\cal B}_q(r_\phi, \phi^0)}
       \frac 1 T \sum_t | \xi_{it}(\beta,\phi) |  \right)^{(8+\nu)} \right]
       \leq   \mathbb{E}_\phi \left[
       \sum_i  \left( 
       \frac 1 T \sum_t  M(Z_{it}) \right)^{(8+\nu)} \right]
     \nonumber  \\
       &\leq   \mathbb{E}_\phi \left[
       \sum_i  
       \frac 1 T \sum_t  M(Z_{it})^{(8+\nu)} \right]
       = \sum_i  
       \frac 1 T \sum_t  \mathbb{E}_\phi M(Z_{it})^{(8+\nu)} = {\cal O}_P(N).
   \end{align*}
   Thus,
   $\sup_{\beta \in {\cal B}(r_\beta, \beta^0)}  \sup_{\phi \in {\cal B}_q(r_\phi, \phi^0)}
      \max_i  \frac 1 T \sum_t | \xi_{it}(\beta,\phi) | 
        = {\cal O}_P\left( N^{1/(8+\nu)} \right)  = {\cal O}_P\left( N^{2 \epsilon} \right)$.
     Analogously, it follows that
     $\sup_{\beta \in {\cal B}(r_\beta, \beta^0)}  \sup_{\phi \in {\cal B}_q(r_\phi, \phi^0)}
      \max_t \frac 1 N \sum_i | \xi_{it}(\beta,\phi) | 
        = {\cal O}_P\left( N^{2 \epsilon} \right)$.
  
     \item[(iv)]  Let $\chi_t =  \frac 1 {\sqrt{N}} \sum_i \partial_{\pi} \ell_{it}.$ By cross-sectional independence and 
     $\mathbb{E}_\phi (\partial_{\pi} \ell_{it})^8 \leq \mathbb{E}_\phi M(Z_{it})^8= {\cal O}_P(1), $ $\mathbb{E}_\phi \chi_t^8  = {\cal O}_P(1)$ uniformly
     over $t$. Thus,
     $\mathbb{E}_\phi \frac 1 T \sum_t \chi_t^8
      = {\cal O}_P(1)$
      and therefore 
      $ \frac 1 T \sum_t \left| \frac 1 {\sqrt{N}} \sum_i \partial_{\pi} \ell_{it}  \right|^q
             = {\cal O}_P\left( 1 \right)$, with $q=8$.
     
     Let $\chi_i = \frac 1 {\sqrt{T}} \sum_t \partial_{\pi} \ell_{it}( \beta^0, \pi^0_{it})$. By  Lemma~\ref{lemma:mixing_inequality} 
     and  $\mathbb{E}_\phi (\partial_{\pi} \ell_{it})^{8+\nu} \leq \mathbb{E}_\phi M(Z_{it})^{8+\nu}= {\cal O}_P(1)$,  $\mathbb{E}_\phi \chi_i^8  = {\cal O}_P(1)$ uniformly over $i$.
     Here we use $\mu> 4/[1-8/(8+\nu)] = 4 (8+\nu)/\nu$ that is imposed in Assumption~\ref{ass:A1}.
    Thus, $\mathbb{E}_\phi\frac 1 N \sum_i   \chi_i^8  = {\cal O}_P(1)$
     and therefore 
     $ \frac 1 N \sum_i \left| \frac 1 {\sqrt{T}} \sum_t \partial_{\pi} \ell_{it}  \right|^q
             = {\cal O}_P\left( 1 \right)$, with $q=8$.
      
     The proofs for       
     $ \frac 1 T \sum_t \left| \frac 1 {\sqrt{N}} \sum_i  \partial_{\beta_k \pi} \ell_{it}  - 
               \mathbb{E}_\phi\left[ \partial_{\beta_k \pi} \ell_{it}  \right] \right|^2
             = {\cal O}_P\left( 1 \right) $  
     and
     $ \frac 1 N \sum_i \Big| \frac 1 {\sqrt{T}} \sum_t  \partial_{\beta_k \pi} \ell_{it} - 
               \mathbb{E}_\phi\left[ \partial_{\beta_k \pi} \ell_{it} \right] \Big|^2
             = {\cal O}_P\left( 1 \right)$   
     are analogous.           
     
    \item[(v)]  It follows by  the independence of  $\{ (\ell_{i1}, \ldots, \ell_{iT}) : 1 \leq i \leq N \}$ across $i$, conditional on $\phi$, in Assumption~\ref{ass:A1}$(ii)$.

    \item[(vi)] Let $\xi_{it} = \partial_{\pi^r} \ell_{it}( \beta^0, \pi^0_{it})
                           -  \mathbb{E}_\phi\left[ \partial_{\pi^r} \ell_{it} \right]$, with $r \in \{2,3\}$,
                     or $\xi_{it} = \partial_{\beta_k \pi^2} \ell_{it}( \beta^0, \pi^0_{it})
                           -  \mathbb{E}_\phi\left[ \partial_{\beta_k \pi^2} \ell_{it} \right]$.
For $\tilde \nu = \nu$,
$\max_i  \, \mathbb{E}_\phi\left[ \xi_{it}^{8+\tilde \nu} \right] = {\cal O}_P(1)$
by assumption. By Lemma~\ref{lemma:Cov_mixing_bound},
\begin{align*}
    \left|   \sum_s
   \mathbb{E}_\phi\left[ \xi_{it} \xi_{is}  \right] \right|
   &=   \sum_s \left| {\rm Cov}_\phi\left( \xi_{it}, \, \xi_{is} \right) \right|
   \nonumber \\ &
   \leq \, \sum_s \, [8 \; a(|t-s|)]^{1-2/(8+\nu)} \,
           \left[ \mathbb{E}_\phi |\xi_t|^{8+\nu} \right]^{1/(8+\nu)}
           \left[ \mathbb{E}_\phi |\xi_{s}|^{8+\nu} \right]^{1/(8+\nu)}
   \nonumber \\
   &= \tilde C \sum_{m=1}^\infty 
          m^{-\mu [1-2/(8+\nu)]}
     \leq \tilde C \sum_{m=1}^\infty 
          m^{-4} = \tilde C \pi^4/90,
\end{align*}
where $\tilde C$ is a constant.
Here we use that $\mu > 4 (8+\nu)/\nu$ implies
$\mu [1-2/(8+\nu) > 4$.
We thus have shown
$\max_{i} \,  \max_t \sum_s \mathbb{E}_\phi\left[ \xi_{it} \xi_{js}  \right]
  \leq \tilde C \pi^4/90 =: C$. 
  
  Analogous to the proof of part (iv), we can use Lemma~\ref{lemma:mixing_inequality} 
  to obtain 
  $\max_i \, \mathbb{E}_\phi    \left\{ \left[ \frac 1 {\sqrt{T}} 
                            \sum_t      \xi_{it}  
                         \right]^8    \right\}
                           \leq C$,
   and independence across $i$
   to obtain  
   $\max_t \, \mathbb{E}_\phi    \left\{ \left[ \frac 1 {\sqrt{N}} 
                            \sum_i      \xi_{it}  
                         \right]^8    \right\}
                           \leq C$.                 
   Similarly, by Lemma~\ref{lemma:mixing_inequality} 
                      \begin{align*}
                        \max_{i,j} \, \mathbb{E}_\phi    \left\{ \left[ \frac 1 {\sqrt{T}} 
                            \sum_t   \left[  \xi_{it} \xi_{jt} - \mathbb{E}_\phi\left(   \xi_{it} \xi_{jt}  \right) \right]
                         \right]^4    \right\}
                           &\leq C ,
                     \end{align*}
   which requires $\mu> 2/[1-4/(4+\nu/2)]$, which is 
   implied by the assumption that $\mu > 4 (8+\nu)/\nu$.

       \item[(vii)] We have already shown that $\left\| \overline {\cal H}^{-1} \right\|_q
               = {\cal O}_P\left(  1\right)$.                  
  \end{itemize}     
  Therefore, we can apply  Lemma~\ref{lemma:BasicRegularityPanel},
  which shows that Assumption~\ref{ass:A1}$(v)$ and $(vi)$ hold.
   We have already shown that Assumption~\ref{ass:A1}$(i)$, $(ii)$, $(iv)$, $(v)$ and $(vi)$ hold.
    One can also check that
    $(NT)^{-1/4+1/(2q)} = o_P(r_\phi)$
     and $(NT)^{1/(2q)} r_\beta = o_P(r_\phi)$ are satisfied.
    In addition, ${\cal L}(\beta,\phi)$ is strictly concave.
    We can therefore invoke 
    Theorem~\ref{th:consistency} to show that
     Assumption~\ref{ass:A1}$(iii)$ holds and that
       $\| \widehat \beta - \beta^0\| = {\cal O}_P( (NT)^{-1/4})$.
\end{proof}
 
 \begin{proof}[\bf Proof of Theorem~\ref{th:connection}, Part $(ii)$]
For any $N\times T$ matrix $A$ we define the $N \times T$ matrix $\mathbbm{P}A$
as follows
\begin{align}
   (\mathbbm{P}A)_{it} &= \alpha^*_{i} + \gamma^*_{t}  ,   
       &
  ( \alpha^* , \gamma^*)
   \in \argmin_{\alpha,\gamma}                   
    \sum_{i,t}   \mathbb{E}_\phi( - \partial_{\pi^2} \ell_{it} )    \left( A_{it} -  \alpha_{i} 
                      - \gamma_{t}  
    \right)^2 .       
    \label{DefP}
\end{align}
Here, the minimization is over $\alpha \in \mathbbm{R}^N$
and $\gamma \in \mathbbm{R}^T$. The operator $\mathbbm{P}$ is a linear projection,
i.e. we have $\mathbbm{P} \mathbbm{P}= \mathbbm{P}$.
It is also convenient to define
\begin{align}
     \widetilde{\mathbbm{P}} A = \mathbbm{P} \tilde A , \quad
     \text{ where } \quad
     \tilde A_{it} = \frac{ A_{it} } { \mathbb{E}_\phi( - \partial_{\pi^2} \ell_{it} ) }.
    \label{DefTildeP} 
\end{align}
$\widetilde{\mathbbm{P}}$ is a linear operator, but
not a projection. Note that $\Lambda$ and $\Xi$ defined in
\eqref{DefProLa} and  \eqref{DefProXi}
can be written as $\Lambda =  \widetilde{\mathbbm{P}}  A$
and $\Xi_k = \widetilde{\mathbbm{P}}  B_k$,
where $A_{it}=- \partial_{\pi} \ell_{it}$ and
$B_{k,it}=- \mathbb{E}_\phi( \partial_{\beta_k \pi} \ell_{it} )$,
for $k=1,\ldots,\dim \beta$.\footnote{$B_k$ and $\Xi_k$
are $N \times T$ matrices with entries $B_{k,it}$ and $\Xi_{k,it}$, respectively,
while $B_{it}$ and $\Xi_{it}$ are $\dim \beta$-vectors with entries $B_{k,it}$ and $\Xi_{k,it}$.}

By Lemma~\ref{RewriteSumIT}$(ii)$,
\begin{align*}
   \overline W &= - \, \frac 1 {\sqrt{NT}} \,
                \left( \partial_{\beta \beta'} \overline {\cal L}
                     + [\partial_{\beta \phi'} \overline {\cal L}] \; \overline {\cal H}^{-1} \;
             [\partial_{\phi \beta'} \overline {\cal L}]  \right)
    =   -  \frac 1 {NT}  \sum_{i=1}^N
     \sum_{t=1}^T \left[ \mathbb{E}_\phi \left(
            \partial_{\beta \beta'} \ell_{it}  \right)
              +  \mathbb{E}_\phi  \left( - \partial_{\pi^2} \ell_{it} 
              \right) \Xi_{it} \Xi'_{it}  \right] .
\end{align*}

By  Lemma~\ref{RewriteSumIT}$(i)$,
\begin{align*}
      U^{(0)} &=  
                     \partial_{\beta} {\cal L}
                   + [\partial_{\beta \phi'} \overline {\cal L} ]\, \overline {\cal H}^{-1} {\cal S} 
           =   \frac 1 {\sqrt{NT}}
                      \sum_{i,t} 
     \left(
          \partial_{\beta} \ell_{it}
        -  \Xi_{it} \, \partial_{\pi} \ell_{it} 
         \right) 
         = \frac 1 {\sqrt{NT}}
        \sum_{i=1}^N \sum_{t=1}^T
           D_{\beta} \ell_{it}   .
\end{align*}                      
We decompose $U^{(1)}= U^{(1a)}  + U^{(1b)}$, with
\begin{align*}
      U^{(1a)} &=  
      [\partial_{\beta \phi'} \widetilde{\cal L}] \overline {\cal H}^{-1} {\cal S}
        -  [\partial_{\beta \phi'} \overline {\cal L}] \,
                      \overline {\cal H}^{-1} \, \widetilde {\cal H} \,
                      \overline {\cal H}^{-1} \, {\cal S}   ,
    \nonumber \\
       U^{(1b)} &=  
                 \sum_{g=1}^{\dim \phi}
           \left( \partial_{\beta \phi' \phi_g} \overline {\cal L}
               +  [\partial_{\beta \phi'} \overline {\cal L}] \, \overline {\cal H}^{-1}
            [\partial_{\phi \phi' \phi_g} \overline {\cal L}]  \right)
              \overline {\cal H}^{-1} {\cal S}
               [\overline {\cal H}^{-1} {\cal S}]_g/2 .
\end{align*}                      
By Lemma~\ref{RewriteSumIT}$(i)$ and $(iii)$,
\begin{align*}
  U^{(1a)}   &= - \frac 1 {\sqrt{NT}}
      \sum_{i,t}   \Lambda_{it}
      \left(
         \partial_{\beta \pi} \tilde \ell_{it}
         +  \Xi_{it} \,
         \partial_{\pi^2} \tilde \ell_{it}
      \right)    
    =  - \frac 1 {\sqrt{NT}}
     \sum_{i=1}^N \sum_{t=1}^T \,  \Lambda_{it}  \, 
       \left[  D_{\beta \pi} \ell_{it}   - \mathbb{E}_\phi( D_{\beta \pi} \ell_{it}  ) \right],
\end{align*}
and
\begin{align*}
    U^{(1b)}      &=  \frac 1 {2 \, \sqrt{NT}}
            \sum_{i,t}   
            \Lambda_{it}^2
        \left[  \mathbb{E}_\phi( \partial_{\beta \pi^2}  \ell_{it})
             +  
              [\partial_{\beta \phi'} \overline {\cal L}] \, \overline {\cal H}^{-1}
            \mathbb{E}_\phi( \partial_\phi \partial_{\pi^2}  \ell_{it}  )
        \right] ,
\end{align*}
where for each $i,t$, $\partial_\phi \partial_{\pi^2}  \ell_{it}$ is a
$\dim \phi$-vector, which can be written as
$\partial_\phi   \partial_{\pi^2}  \ell_{it} 
=  { A 1_T \choose A' 1_N }$
for an $N \times T$ matrix $A$ with
elements $A_{j\tau} = \partial_{\pi^3}  \ell_{j \tau}$
if $j=i$ and $\tau=t$, and $A_{j\tau}=0$ otherwise. Thus, 
Lemma~\ref{RewriteSumIT}$(i)$ gives
$   [\partial_{\beta \phi'} \overline {\cal L}] \, \overline {\cal H}^{-1}
             \partial_\phi   \partial_{\pi^2}   \ell_{it} 
          =   -\sum_{j,\tau}    \Xi_{j\tau}
               1(i = j) 1(t = \tau)
              \partial_{\pi^3}  \ell_{it} =
             - \Xi_{it}
              \partial_{\pi^3} \ell_{it}$. Therefore
\begin{align*}        
        U^{(1b)}            &=
         \frac 1 {2 \, \sqrt{NT}}
            \sum_{i,t}   
             \Lambda_{it}^2
      \mathbb{E}_\phi   \left(  \partial_{\beta \pi^2} \ell_{it}
             -  
             \Xi_{it}
              \partial_{\pi^3}  \ell_{it}
        \right)
        =   \frac 1 {2 \, \sqrt{NT}}
     \sum_{i=1}^N \sum_{t=1}^T       \Lambda_{it}^2 \, \mathbb{E}_\phi(
           D_{\beta \pi^2} \ell_{it}).
\end{align*} 
\end{proof}
    
\begin{proof}[\bf Proof of Theorem~\ref{th:connection}, Part $(iii)$]    
     Showing that Assumption~\ref{ass:A2} is satisfied is
     analogous to the proof of Lemma~\ref{lemma:BasicRegularityPanel}
     and of part $(ii)$ of this Theorem.
     
       In the proof of Theorem~\ref{th:BothEffects} we show that 
       Assumption~\ref{ass:PanelA1}  implies that $U={\cal O}_P(1)$.
       This fact together with part $(i)$ of this theorem 
       show that
  Corollary~\ref{cor:LimitBeta} is applicable, so that 
     $\sqrt{NT} (\widehat \beta - \beta^0) = \overline W_{\infty}^{-1} U + o_P(1) = {\cal O}_P(1)$,
     and we can apply Theorem~\ref{th:DeltaExpansion}.
     
     By Lemma~\ref{RewriteSumIT} and the result
     for $\sqrt{NT} (\widehat \beta - \beta^0)$, 
     \begin{align}
       & \sqrt{NT}  \left[ \partial_{\beta'}  \overline {\Delta}
      +
      (\partial_{\phi'} \overline \Delta)  
        \overline {\cal H}^{-1} (\partial_{\phi \beta'} \overline {\cal L})
    \right] (\widehat \beta - \beta^0)
          =  \left[ \frac 1 {NT}
        \sum_{i,t}
         \mathbb{E}_\phi( D_{\beta} {\Delta_{it}} )  \right]'
         \overline W_\infty^{-1} \left( U^{(0)} + U^{(1)} \right)  + o_P(1).
         \label{DeltaDerivation1}
     \end{align}
    We apply Lemma~\ref{RewriteSumIT} to
     $U^{(0)}_\Delta$ and $U^{(1)}_\Delta$ defined in
     Theorem~\ref{th:DeltaExpansion} to give
   \begin{align}   
        \sqrt{NT} \, U^{(0)}_\Delta &=        - \frac 1 {\sqrt{NT}} \sum_{i,t}
          \mathbb{E}_\phi( \Psi_{it} )
          \partial_{\pi} \ell_{it}  ,
    \nonumber \\ 
        \sqrt{NT}  \, U^{(1)}_\Delta
        &=  \frac 1 {\sqrt{NT}}\sum_{i,t}  \Lambda_{it} 
      \left[   \mathbb{E}_\phi( \Psi_{it} )   \partial_{\pi^2} \ell_{it} 
          -    \Psi_{it} \mathbb{E}_\phi( \partial_{\pi^2}  \ell_{it} )
         \right]
   \nonumber \\     & \qquad       
    + 
     \frac 1 {2 \, \sqrt{NT}} \sum_{i,t}    \Lambda_{it}^2 \left[
            \mathbb{E}_\phi(  \partial_{\pi^2} \Delta_{it} )
    -  
            \mathbb{E}_\phi( \partial_{\pi^3} \ell_{it} )
          \mathbb{E}_\phi( \Psi_{it} ) 
        \right].
            \label{DeltaDerivation2}     
\end{align}
The derivation of \eqref{DeltaDerivation1}
and \eqref{DeltaDerivation2} is analogous to the proof of the part $(ii)$ of the Theorem.
Combining Theorem~\ref{th:DeltaExpansion} with equations \eqref{DeltaDerivation1}
and \eqref{DeltaDerivation2} gives the result.
\end{proof}

\section{Proofs of Appendix~\ref{sec:InverseH} (Lemma~\ref{lemma:HessianAdditive})}\label{sec:s4}

The following Lemmas are useful to prove  Lemma~\ref{lemma:HessianAdditive}. Let ${\cal L}^*(\beta,\phi)  = (NT)^{-1/2} \sum_{i,t} \ell_{it}(\beta, \alpha_i + \gamma_t)$.

\begin{lemma}
   \label{lemma:ChoiceB}
    If the statement of Lemma~\ref{lemma:HessianAdditive} holds
    for some constant $b>0$, then it holds for any constant $b>0$.
\end{lemma}

\begin{proof}[\bf Proof of Lemma~\ref{lemma:ChoiceB}]
  Write 
    $\overline {\cal H} = \overline {\cal H}^* + \frac b {\sqrt{NT}} v v'$,
    where $\overline {\cal H}^* = \mathbb{E}_\phi \left[ -
 \frac{\partial^2}{\partial \phi \partial \phi'}
  {\cal L}^* \right]$. Since $\overline {\cal H}^* v = 0$,
   $$
       \overline {\cal H}^{-1} = 
       \left( \overline {\cal H}^* \right)^\dagger + 
          \left( \frac b {\sqrt{NT}} v v' \right)^\dagger
             =  \left( \overline {\cal H}^* \right)^\dagger + 
                   \frac{\sqrt{NT}} {b \| v v' \|^2}  v v' 
              =      \left( \overline {\cal H}^* \right)^\dagger + 
                   \frac{\sqrt{NT}} {b (N+T)^2}  v v' ,
   $$
where $\dagger$ refers to the  Moore-Penrose pseudo-inverse.
Thus, if $\overline {\cal H}_1$ is the expected Hessian for $b=b_1>0$
and $\overline {\cal H}_2$ is the expected Hessian for $b=b_2>0$, 
  $ \left\| \overline {\cal H}^{-1}_1 - \overline {\cal H}^{-1}_2 \right\|_{\rm max}
       = \left\| \left( \frac 1 {b_1} - \frac 1 {b_2} \right) \frac{\sqrt{NT}} {(N+T)^2}  v v'  
       \right\|_{\rm max}
       =   \mathcal{O}\left( (NT)^{-1/2} \right).
   $    
\end{proof}

\begin{lemma}
    \label{lemma:HH}
     Let Assumption~\ref{ass:PanelA1} hold
     and let $0< b \leq  b_{\min} \left( 1 +
     \frac {\max(N,T)} {\min(N,T)} \frac{ b_{\max} } {b_{\min}} \right)^{-1}.$
   Then, 
    $$ \left\| \overline {\cal H}_{(\alpha \alpha)}^{-1}
    \overline {\cal H}_{(\alpha \gamma)} 
        \right\|_\infty  < 1 -  \frac{b} {b_{\max}}, \ \text{ and } \ 
      \left\|  \overline {\cal H}_{(\gamma \gamma)}^{-1} \,
           \overline {\cal H}_{(\gamma \alpha)}   \right\|_\infty < 
            1 -  \frac{b} {b_{\max}}.$$
\end{lemma}

\begin{proof}[\bf Proof of Lemma~\ref{lemma:HH}]
Let $h_{it} = \mathbb{E}_\phi( - \partial_{\pi^2} \ell_{it} )$, and define
\begin{align*}
     \tilde h_{it} &=
     h_{it} - b
           - \frac 1 {b^{-1} + \sum_j \left( \sum_{\tau} h_{j\tau} \right)^{-1}}
          \sum_j \frac{h_{jt} - b} {\sum_{\tau} h_{j \tau}} .
\end{align*}
By definition, $\overline {\cal H}_{(\alpha \alpha)} = \overline {\cal H}_{(\alpha \alpha)}^*
   + b 1_N 1_N'/\sqrt{NT}$
and    $\overline {\cal H}_{(\alpha \gamma)} = \overline {\cal H}_{(\alpha \gamma)}^*
   - b 1_N 1_T'/\sqrt{NT}$.
The   matrix $ \overline {\cal H}_{(\alpha \alpha)}^*$ is diagonal
with elements $\sum_t h_{it}/\sqrt{NT}$.
The matrix $\overline {\cal H}_{(\alpha \gamma)}^*$
has elements $h_{it}/\sqrt{NT}$.
The Woodbury identity states that
$$ \overline {\cal H}_{(\alpha \alpha)}^{-1}
       = \overline {\cal H}_{(\alpha \alpha)}^{* -1}
         -  \overline {\cal H}_{(\alpha \alpha)}^{* -1}
          1_N
         \left( \sqrt{NT} \, b^{-1}
          + 1_N' \overline {\cal H}_{(\alpha \alpha)}^{* -1}
            1_N  \right)^{-1}
           1_N'
           \overline {\cal H}_{(\alpha \alpha)}^{* -1} .
$$
Then,  
$ \overline {\cal H}_{(\alpha \alpha)}^{-1} \,
           \overline {\cal H}_{(\alpha \gamma)} =\overline {\cal H}_{(\alpha \alpha)}^{* -1}
           \tilde H /\sqrt{NT}$, where $\tilde H$ is the $N \times T$ matrix
           with elements $\tilde h_{it}$.
Therefore 
\begin{align*}
   \left\| \overline {\cal H}_{(\alpha \alpha)}^{-1}
    \overline {\cal H}_{(\alpha \gamma)}  \right\|_\infty
   &= \max_i \frac {\sum_t \left|  \tilde h_{it} \right|}
      {\sum_{t} h_{it}} .
\end{align*}
Assumption~\ref{ass:PanelA1}$(iv)$ guarantees
that $b_{\max} \geq h_{it} \geq b_{\min}$, which implies
$h_{jt} - b \geq b_{\min} - b > 0$, and
\begin{align*}
   \tilde h_{it}
    & > h_{it} - b
           - \frac 1 {b^{-1}}
          \sum_j \frac{h_{jt} - b} {\sum_{\tau} h_{j \tau}}
     \geq  b_{\min}
    - b \, \left( 1 + \frac N T \, \frac{b_{\max}} {b_{\min}} \right)
    \geq 0   .
\end{align*}
We  conclude that 
\begin{align*}
   \left\| \overline {\cal H}_{(\alpha \alpha)}^{-1}
    \overline {\cal H}_{(\alpha \gamma)}  \right\|_\infty
   &= \max_i \frac {\sum_t \tilde h_{it}}
      {\sum_{t} h_{it}}
   = 1 - \min_i
     \frac {1} {\sum_{t} h_{it}}  \sum_t
      \left( b
           + \frac 1 {b^{-1} + \sum_j \left( \sum_{\tau} h_{j\tau} \right)^{-1}}
          \sum_j \frac{h_{jt} - b} {\sum_{\tau} h_{j \tau}}
      \right)
   \nonumber \\
   &< 1 -  \frac{b} {b_{\max}} .
\end{align*}
Analogously,
$ \left\| \overline {\cal H}_{(\gamma \gamma)}^{-1}
   \overline {\cal H}_{(\gamma \alpha)} \right\|_\infty
   < 1 -  \frac{b} {b_{\max}}$.
\end{proof}

\begin{proof}[\bf Proof of Lemma~\ref{lemma:HessianAdditive}]
   We choose 
   $b < b_{\min} \left( 1 +
     \max( \kappa^2, \kappa^{-2} ) \frac{ b_{\max} } {b_{\min}} \right)^{-1}$.
   Then,   
   $b \leq b_{\min} \left( 1 +
     \frac {\max(N,T)} {\min(N,T)} \frac{ b_{\max} } {b_{\min}} \right)^{-1}$ for large enough $N$ and $T$,
   so that   
   Lemma~\ref{lemma:HH} becomes applicable. The choice of $b$
   has no effect on the general validity of the lemma
   for all $b>0$ by Lemma~\ref{lemma:ChoiceB}.
   
   By the inversion formula for partitioned matrices,
   \begin{align*}
      \overline {\cal H}^{-1}
        &= \left( \begin{array}{c@{\;\;\;\;\;\;}c}
                    A   & -  A \, \overline {\cal H}_{(\alpha \gamma)} \,
                     \overline {\cal H}_{(\gamma \gamma)}^{-1} \\
                     - \overline {\cal H}_{(\gamma \gamma)}^{-1} \,
                     \overline {\cal H}_{(\gamma \alpha)} \, A
                     & \overline {\cal H}_{(\gamma \gamma)}^{-1}
                     + \overline {\cal H}_{(\gamma \gamma)}^{-1}  \,
                     \overline {\cal H}_{(\gamma \alpha)} \, A \,
                     \overline {\cal H}_{(\alpha \gamma)} \,
                      \overline {\cal H}_{(\gamma \gamma)}^{-1}
                     \end{array} \right) ,
   \end{align*}
with $A := (\overline {\cal H}_{(\alpha \alpha)}
          - \overline {\cal H}_{(\alpha \gamma)}
           \overline {\cal H}_{(\gamma \gamma)}^{-1}
           \overline {\cal H}_{(\gamma \alpha)})^{-1}$.
   The Woodbury identity states that
   \begin{align*}
      \overline {\cal H}_{(\alpha \alpha)}^{-1}
       &= \overline {\cal H}_{(\alpha \alpha)}^{* -1}
         - \underbrace{ \overline {\cal H}_{(\alpha \alpha)}^{* -1}
         1_N
         \left( \sqrt{NT}/b
          + 1_N' \overline {\cal H}_{(\alpha \alpha)}^{* -1}
            1_N  \right)^{-1}
           1_N'
           \overline {\cal H}_{(\alpha \alpha)}^{* -1}
           }_{=: C_{(\alpha \alpha)}} ,
     \nonumber \\
      \overline {\cal H}_{(\gamma \gamma)}^{-1}
       &= \overline {\cal H}_{(\gamma \gamma)}^{* -1}
         - \underbrace{ \overline {\cal H}_{(\gamma \gamma)}^{* -1}
          1_T
         \left( \sqrt{NT}/b
          + 1_T' \overline {\cal H}_{(\gamma \gamma)}^{* -1}
            1_T  \right)^{-1}
           1_T'
           \overline {\cal H}_{(\gamma \gamma)}^{* -1}
            }_{=: C_{(\gamma \gamma)}} .
   \end{align*}
   By Assumption \ref{ass:PanelA1}$(v)$,
   $\| \overline {\cal H}_{(\alpha \alpha)}^{* -1} \|_\infty = {\cal O}_P(1)$,
   $\| \overline {\cal H}_{(\gamma \gamma)}^{* -1} \|_\infty = {\cal O}_P(1)$,
   $\| \overline {\cal H}_{(\alpha \gamma)}^{*} \|_{\max} = {\cal O}_P(1/\sqrt{NT})$.
   Therefore\footnote{
   Here and in the following me make use of the inequalities
   $\|AB\|_{\max} < \|A\|_\infty \|B\|_{\max}$,
   $\|AB\|_{\max} < \|A\|_{\max} \|B'\|_{\infty}$,
   $\|A\|_{\infty} \leq n \|A\|_{\max}$,
   which hold for any $m \times n$ matrix $A$ and
   $n \times p$ matrix $B$.}
   \begin{align*}
      \| C_{(\alpha \alpha)} \|_{\max}
       &\leq  \| \overline {\cal H}_{(\alpha \alpha)}^{* -1} \|^2_\infty
         \left\| 1_N
           1_N' \right\|_{\max}
            \left( \sqrt{NT}/b
          + 1_N' \overline {\cal H}_{(\alpha \alpha)}^{* -1}
            1_N  \right)^{-1}
        = {\cal O}_P(1/\sqrt{NT}),
     \nonumber \\
      \| \overline {\cal H}_{(\alpha \alpha)}^{-1} \|_\infty
       &\leq \| \overline {\cal H}_{(\alpha \alpha)}^{* -1} \|_\infty
              + N \| C_{(\alpha \alpha)} \|_{\max}
        = {\cal O}_P(1).
   \end{align*}
   Analogously, $\| C_{(\gamma \gamma)} \|_{\max} = {\cal O}_P(1/\sqrt{NT})$
   and
   $ \| \overline {\cal H}_{(\gamma \gamma)}^{-1} \|_\infty = {\cal O}_P(1)$.
   Furthermore, 
   $ \| \overline {\cal H}_{(\alpha \gamma)} \|_{\max} \leq
      \| \overline {\cal H}_{(\alpha \gamma)}^* \|_{\max}
       + b/\sqrt{NT} = {\cal O}_P(1/\sqrt{NT})$.
   Define
   \begin{align*}
       B &:= \left( \mathbbm{1}_N
          - \overline {\cal H}_{(\alpha \alpha)}^{-1}
           \overline {\cal H}_{(\alpha \gamma)}
           \overline {\cal H}_{(\gamma \gamma)}^{-1}
           \overline {\cal H}_{(\gamma \alpha)} \right)^{-1} - \mathbbm{1}_N
        = \sum_{n=1}^\infty
                \left(  \overline {\cal H}_{(\alpha \alpha)}^{-1}
           \overline {\cal H}_{(\alpha \gamma)}
           \overline {\cal H}_{(\gamma \gamma)}^{-1}
           \overline {\cal H}_{(\gamma \alpha)}  \right)^n .
   \end{align*}
   Then, $A = \overline {\cal H}_{(\alpha \alpha)}^{-1}
                +   \overline {\cal H}_{(\alpha \alpha)}^{-1} B
     = \overline {\cal H}_{(\alpha \alpha)}^{* -1}
       - C_{(\alpha \alpha)}
      +  \overline {\cal H}_{(\alpha \alpha)}^{-1} B$.
   By Lemma~\ref{lemma:HH},
   $\| \overline {\cal H}_{(\alpha \alpha)}^{-1}
           \overline {\cal H}_{(\alpha \gamma)}
           \overline {\cal H}_{(\gamma \gamma)}^{-1}
           \overline {\cal H}_{(\gamma \alpha)} \|_{\infty}
     \leq \| \overline {\cal H}_{(\alpha \alpha)}^{-1}
           \overline {\cal H}_{(\alpha \gamma)}\|_{\infty}
          \| \overline {\cal H}_{(\gamma \gamma)}^{-1}
           \overline {\cal H}_{(\gamma \alpha)} \|_{\infty}
      < \left( 1 -  \frac{b} {b_{\max}} \right)^2 < 1$, and
   \begin{align*}
       \| B \|_{\max}
        &\leq \sum_{n=0}^\infty
            \left( \|  \overline {\cal H}_{(\alpha \alpha)}^{-1}
           \overline {\cal H}_{(\alpha \gamma)}
            \overline {\cal H}_{(\gamma \gamma)}^{-1}
           \overline {\cal H}_{(\gamma \alpha)}  \|_\infty \right)^n          
           \| \overline {\cal H}_{(\alpha \alpha)}^{-1} \|_\infty
           \| \overline {\cal H}_{(\alpha \gamma)} \|_\infty
           \| \overline {\cal H}_{(\gamma \gamma)}^{-1} \|_\infty
           \| \overline {\cal H}_{(\gamma \alpha)} \|_{\max}
      \nonumber \\
        &\leq \left[ \sum_{n=0}^\infty 
        \left( 1 -  \frac{b} {b_{\max}} \right)^{2n} \right] \, T \,     
          \| \overline {\cal H}_{(\alpha \alpha)}^{-1} \|_\infty
           \| \overline {\cal H}_{(\gamma \gamma)}^{-1} \|_\infty
           \| \overline {\cal H}_{(\gamma \alpha)} \|^2_{\max}
       =   {\cal O}_P(1/\sqrt{NT}).
   \end{align*}
   By the triangle inequality,
   \begin{align*}
      \| A \|_\infty \leq
       \| \overline {\cal H}_{(\alpha \alpha)}^{-1}  \|_\infty
         + N \|  \overline {\cal H}_{(\alpha \alpha)}^{-1} \|_\infty
           \| B \|_{\max} =
           {\cal O}_P(1).
   \end{align*}
   Thus, for the different blocks of
   \begin{align*}
      \overline {\cal H}^{-1}
      - \left( \begin{array}{cc}
                     \overline {\cal H}_{(\alpha \alpha)}^* &
                     0  \\
                     0 &
                     \overline {\cal H}_{(\gamma \gamma)}^*
                \end{array} \right)^{-1}
        &= \left( \begin{array}{c@{\;\;\;\;\;\;}c}
                  A
                - \overline {\cal H}_{(\alpha \alpha)}^{* -1}
              & -  A \, \overline {\cal H}_{(\alpha \gamma)} \,
                     \overline {\cal H}_{(\gamma \gamma)}^{-1} \\
                     - \overline {\cal H}_{(\gamma \gamma)}^{-1} \,
                     \overline {\cal H}_{(\gamma \alpha)} \, A
                     & \overline {\cal H}_{(\gamma \gamma)}^{-1}  \,
                     \overline {\cal H}_{(\gamma \alpha)} \, A \,
                     \overline {\cal H}_{(\alpha \gamma)} \,
                      \overline {\cal H}_{(\gamma \gamma)}^{-1}
                    - C_{(\gamma \gamma)}
                     \end{array} \right),
   \end{align*}
   we find
   \begin{align*}
      \left\| A - \overline {\cal H}_{(\alpha \alpha)}^{* -1} \right\|_{\max}
      &=  \left\|
         \overline {\cal H}_{(\alpha \alpha)}^{-1} B
                - C_{(\alpha \alpha)} \right\|_{\max}
            \nonumber \\ &    
         \leq  \|
         \overline {\cal H}_{(\alpha \alpha)}^{-1} \|_\infty
            \| B \|_{\max}
                - \| C_{(\alpha \alpha)} \|_{\max}
        =  {\cal O}_P(1/\sqrt{NT})   ,
    \nonumber \\
    \left\| -  A \, \overline {\cal H}_{(\alpha \gamma)} \,
                     \overline {\cal H}_{(\gamma \gamma)}^{-1} \right\|_{\max}
         &\leq
         \| A \|_\infty
         \| \overline {\cal H}_{(\alpha \gamma)} \|_{\max}
          \| \overline {\cal H}_{(\gamma \gamma)}^{-1} \|_{\infty}
         = {\cal O}_P(1/\sqrt{NT}) ,
     \nonumber \\
      \left\| \overline {\cal H}_{(\gamma \gamma)}^{-1}  \,
                     \overline {\cal H}_{(\gamma \alpha)} \, A \,
                     \overline {\cal H}_{(\alpha \gamma)} \,
                      \overline {\cal H}_{(\gamma \gamma)}^{-1}
                    - C_{(\gamma \gamma)} \right\|_{\max}
        &\leq
        \| \overline {\cal H}_{(\gamma \gamma)}^{-1}  \|^2_\infty
        \| \overline {\cal H}_{(\gamma \alpha)} \|_\infty
        \| A \|_\infty
        \| \overline {\cal H}_{(\alpha \gamma)} \|_{\max}
         + \| C_{(\gamma \gamma)} \|_{\max}
      \nonumber \\
        &\leq
         N
        \| \overline {\cal H}_{(\gamma \gamma)}^{-1}  \|^2_\infty
        \| A \|_\infty
        \| \overline {\cal H}_{(\alpha \gamma)} \|^2_{\max}
         + \| C_{(\gamma \gamma)} \|_{\max}
         = {\cal O}_P(1/\sqrt{NT}).
   \end{align*}
   The bound ${\cal O}_P(1/\sqrt{NT})$ for the $\max$-norm
   of each block of the matrix yields the same bound for the $\max$-norm
   of the matrix itself.
\end{proof}

\section{Useful Lemmas}\label{sec: s5}

\subsection{Some Properties of Stochastic Processes}

Here we collect some known properties of 
$\alpha$-mixing processes, which are useful for our proofs.

\begin{lemma}
    \label{lemma:Cov_mixing_bound}
     Let $\{\xi_t\}$ be an $\alpha$-mixing process with mixing coefficients $a(m)$.
     Let $\mathbb{E} |\xi_t|^p < \infty$
     and $\mathbb{E} |\xi_{t+m}|^q < \infty$ for some
     $p,q \geq 1$ and $1/p + 1/q<1$. Then,
     \begin{align*}
         \left| {\rm Cov}\left( \xi_t, \, \xi_{t+m} \right) \right|
         \, \leq \, 8 \; a(m)^{1/r} \,
           \left[ \mathbb{E} |\xi_t|^p \right]^{1/p}
           \left[ \mathbb{E} |\xi_{t+m}|^q \right]^{1/q} ,
     \end{align*}
     where $r=(1-1/p-1/q)^{-1}$.
\end{lemma}

\begin{proof}[\bf Proof of Lemma~\ref{lemma:Cov_mixing_bound}]
    See, for example, Proposition 2.5 in
    Fan and Yao~\cite*{FanYao2003}.
\end{proof}

The following result is a simple modification of Theorem 1 in Cox and Kim~\cite*{CoxKim1995}. 

\begin{lemma}
    \label{lemma:mixing_inequality}
     Let $\{\xi_t\}$ be an $\alpha$-mixing process with mixing coefficients $a(m)$.
     Let $r \geq 1$ be an integer, and let $\delta>2 r$, $\mu> r/(1-2 r/ \delta)$, $c>0$ and $C>0$.
     Assume that
     $\sup_t \mathbb{E}  \left| \xi_t \right|^\delta  \leq C$ and that
     $a(m) \leq c \, m^{- \mu}$ for all $m \in \{1,2,3,\ldots\}$.
     Then there exists a constant $B>0$ depending on $r$, $\delta$, $\mu$, $c$ and $C$,
     but not depending on $T$ or any other distributional characteristics of $\xi_t$, such that
     for any $T>0$,
     \begin{align*}
          \mathbb{E}\left[  \left( \frac 1 {\sqrt{T}} \sum_{t=1}^T \xi_t \right)^{2 r} \right]
          &\leq B .
     \end{align*}
\end{lemma}


The following is a central limit theorem for martingale difference sequences.

\begin{lemma}
    \label{lemma:martingale_diff_CLT}
    Consider the scalar process $\xi_{it} = \xi_{NT,it}$, $i=1,\ldots,N$, $t=1,\ldots,T$.
    Let  $\{(\xi_{i1}, \ldots , \xi_{iT}): 1 \leq i \leq N\}$ be independent across~$i$,
    and be a martingale difference sequence for each~$i$, $N$, $T$.
    Let $\mathbb{E} | \xi_{it} |^{2+\delta}$ 
    be uniformly bounded across $i,t,N,T$ 
    for some $\delta>0$.
    Let $\overline \sigma = \overline \sigma_{NT} > \Delta > 0$
    for all sufficiently large $NT$, and let
    $\frac 1 {NT}  \sum_{i,t} \xi_{it}^2 - \overline \sigma^2 \rightarrow_P 0$ as $NT \rightarrow \infty$.\footnote{%
     Here can allow for an arbitrary sequence of $(N,T)$ with $NT \rightarrow \infty$.    }
    Then,
    \begin{align*}
          \frac 1 {\overline \sigma \, \sqrt{NT}} \sum_{i,t} \xi_{it}
         \to_d {\cal N}(0,1).
    \end{align*}

\end{lemma}

\begin{proof}[\bf Proof of Lemma~\ref{lemma:martingale_diff_CLT}]
    Define $\xi_m = \xi_{M,m} = \xi_{NT,it}$, with $M = NT$ and $m = T(i-1)+t \in \{1,\ldots,M\}$.
    Then $\{ \xi_m, \, m=1,\ldots,M \}$ is a martingale difference sequence. With this
    redefinition the statement of the Lemma is
    equal to Corollary~5.26 in White~\cite*{White2001}, which is based on
    Theorem~2.3 in Mcleish~\cite*{Mcleish1974}, and which shows that
    $\frac 1 {\overline \sigma \, \sqrt{M}} \sum_{m=1}^M \xi_{m}
         \to_d {\cal N}(0,1)$.
\end{proof}

\subsection{Some Bounds for the Norms of Matrices and Tensors}

The following lemma provides bounds for the matrix norm $\|.\|_q$ in terms
of the  matrix norms $\|.\|_1$, $\|.\|_2$, $\|.\|_\infty$,
and a bound for $\|.\|_2$ in terms of $\|.\|_q$ and $\|.\|_{q/(q-1)}$.
For sake of clarity we use notation
$\|.\|_2$ for the spectral norm in this lemma, which everywhere else is denoted by $\|.\|$, without any index.
Recall that $\|A\|_\infty = \max_i \sum_j |A_{ij}|$ and $\|A\|_1 = \|A'\|_\infty$.

\begin{lemma}
    \label{lemma:matrix-q-norm}
    For any matrix $A$  we have
    \begin{align*}
         \| A \|_q  &\leq \| A \|_1^{1/q} \|A\|_\infty^{1-1/q} ,  && \text{for $q \geq 1$,}
         \\
         \| A \|_q  &\leq \| A \|_2^{2/q} \|A\|_\infty^{1-2/q} , && \text{for $q \geq 2$,}
         \\
         \| A \|_2 &\leq \sqrt{\| A \|_q \| A \|_{q/(q-1)} } , && \text{for $q \geq 1$.}
    \end{align*}
    Note also that $\| A \|_{q/(q-1)} = \| A' \|_q$ for $q \geq 1$. Thus, for a symmetric matrix $A$,
    we have 
    $\| A \|_2 \leq \| A \|_q  \leq \| A \|_\infty$ for any $q \geq 1$.
\end{lemma}
\begin{proof}[\bf Proof of Lemma~\ref{lemma:matrix-q-norm}]
    The statements follow from the fact that $\log \|A\|_q$ is a convex function of 
    $1/q$, which is a consequence of the Riesz-Thorin theorem.
    For more details and references see e.g. Higham~\cite*{Higham1992}.
\end{proof}

The following lemma shows that the norm $\| . \|_q$ applied to higher-dimensional tensors with a special structure can be expressed in terms of
matrix norms $\| . \|_q$. In our panel application all higher dimensional tensors have such a special structure, since they are obtained
as partial derivatives wrt to $\alpha$ and $\gamma$ from the likelihood function.

\begin{lemma}
    \label{lemma:SpectralNormTensors}
    Let $a$ be an $N$-vector with entries $a_i$,
    let $b$ be a $T$-vector with entries $b_t$,
    and let $c$ be an $N \times T$ matrix with entries  $c_{it}$.
    Let $A$ be an $\underbrace{ N \times N \times \ldots \times N}_{p \text{ times}}$  tensor
    with entries
    $$A_{i_1 i_2 \ldots i_p}
       = \left\{ \begin{array}{ll}
              a_{i_1} & \text{ if $i_1 = i_2 = \ldots = i_p$,}
            \nonumber \\
               0    & \text{ otherwise.}
          \end{array} \right.
    $$
    Let $B$ be an $\underbrace{ T \times T \times \ldots \times T}_{r \text{ times}}$  tensor
    with entries
    $$B_{t_1 t_2 \ldots t_r}
       = \left\{ \begin{array}{ll}
              b_{t_1} & \text{ if $t_1 = t_2 = \ldots = t_r$,}
            \nonumber \\
               0    & \text{ otherwise.}
          \end{array} \right.
    $$
    Let $C$ be an $\underbrace{ N \times N \times \ldots \times N}_{p \text{ times}}
     \times \underbrace{ T \times T \times \ldots \times T}_{r \text{ times}}$  tensor
    with entries
    $$C_{i_1 i_2 \ldots i_p t_1 t_2 \ldots t_r}
       = \left\{ \begin{array}{ll}
              c_{i_1 t_1} & \text{ if $i_1 = i_2 = \ldots = i_p$ and $t_1 = t_2 = \ldots = t_r$,}
            \nonumber \\
               0    & \text{ otherwise.}
          \end{array} \right.
    $$
    Let $\widetilde C$ be an $\underbrace{ T \times T \times \ldots \times T}_{r \text{ times}}
     \times \underbrace{ N \times N \times \ldots \times N}_{p \text{ times}}$  tensor
    with entries
    $$\widetilde C_{t_1 t_2 \ldots t_r i_1 i_2 \ldots i_p}
       = \left\{ \begin{array}{ll}
              c_{i_1 t_1} & \text{ if $i_1 = i_2 = \ldots = i_p$ and $t_1 = t_2 = \ldots = t_r$,}
            \nonumber \\
               0    & \text{ otherwise.}
          \end{array} \right.
    $$
    Then,
    \begin{align*}
         \| A \|_q &=  \max_i  | a_i | , && \text{for $p \geq 2$,} 
         \\
         \| B \|_q &= \max_t  | b_t | ,  && \text{for $r \geq 2$,} 
         \\
         \| C \|_q &\leq \| c \|_q ,  && \text{for $p \geq 1$, $r \geq 1$,} 
         \\
         \| \widetilde C \|_q &\leq \| c' \|_q ,  && \text{for $p \geq 1$, $r \geq 1$,} 
    \end{align*}
    where $\|.\|_q$ refers to the $q$-norm defined in \eqref{DefNorm}
    with $q \geq 1$.
\end{lemma}

\begin{proof}[\bf Proof of Lemma~\ref{lemma:SpectralNormTensors}]
     Since the vector norm $\|.\|_{q/(q-1)}$ is dual to the vector norm $\|.\|_q$ we can rewrite the
     definition of the tensor norm $\left\| C \right\|_q$ as follows
     \begin{align*}
        \left\| C \right\|_q
      &= \max_{\|u^{(1)}\|_{q/(q-1)}=1} 
           \max_{\begin{minipage}{2cm}\center \scriptsize  $\|u^{(k)}\|_q=1$ \\ $k=2,\ldots,p$ \end{minipage}} 
            \max_{\begin{minipage}{2cm}\center \scriptsize $\|v^{(l)}\|_q=1$ \\ $l=1,\ldots,r$ \end{minipage}} 
            \nonumber \\ & \qquad \qquad
         \left|   \sum_{i_1 i_2 \ldots i_p=1}^N \sum_{t_1 t_2 \ldots t_r=1}^T  
         u^{(1)}_{i_1} u^{(2)}_{i_2} \cdots u^{(p)}_{i_p}
         v^{(1)}_{i_1} v^{(2)}_{t_2} \cdots v^{(r)}_{t_r}
         C_{i_1 i_2 \ldots i_p t_1 t_2 \ldots t_r} \,
         \right|  .
     \end{align*}
    The specific structure
     of $C$ yields
     \begin{align*}
        \left\| C \right\|_q
              &= \max_{\|u^{(1)}\|_{q/(q-1)}=1} 
           \max_{\begin{minipage}{2cm}\center \scriptsize  $\|u^{(k)}\|_q=1$ \\ $k=2,\ldots,p$ \end{minipage}} 
            \max_{\begin{minipage}{2cm}\center \scriptsize $\|v^{(l)}\|_q=1$ \\ $l=1,\ldots,r$ \end{minipage}} 
         \left|   \sum_{i=1}^N \sum_{t=1}^T  
         u^{(1)}_{i} u^{(2)}_{i} \cdots u^{(p)}_{i}
         v^{(1)}_{t} v^{(2)}_{t} \cdots v^{(r)}_{t}
         c_{it} \,
         \right|  
      \nonumber \\
        &\leq  \; \; \max_{\|u\|_{q/(q-1)} \leq 1}  \; \;
            \max_{\|v\|_q \leq 1} 
         \left|   \sum_{i=1}^N \sum_{t=1}^T  
         u_i  v_i c_{it} \,
         \right|  = \| c \|_q       ,  
    \end{align*}
    where  we define $u \in \mathbbm{R}^N$
    with elements
    $u_i = u^{(1)}_{i} u^{(2)}_{i} \cdots u^{(p)}_{i}$
    and $v \in \mathbbm{R}^T$
    with elements
    $v_t = v^{(1)}_{t} v^{(2)}_{t} \cdots v^{(r)}_{t}$,
    and we use that  $\|u^{(k)}\|_q=1$, for $k=2,\ldots,p$,
    and $\|v^{(l)}\|_q=1$, for $l=2,\ldots,r$, implies
    $|u_i| \leq  |u^{(1)}_{i}|$
    and $|v_t| \leq | v^{(1)}_{t}|$,
    and therefore    
    $ \|u\|_{q/(1-q)} \leq \|u^{(1)} \|_{q/(1-q)} = 1$
    and $\|v\|_q \leq \| v^{(1)} \|_q = 1$. The proof of 
      $\| \widetilde C \|_q \leq \| c' \|_q$ is analogous.
    
    Let $A^{(p)} = A$, as defined above, for a particular value of $p$. 
    For $p=2$, $A^{(2)}$ is a diagonal $N \times N$ matrix
    with diagonal elements $a_i$, so that 
    $ \| A^{(2)} \|_q \leq   \| A^{(2)} \|_1^{1/q} \| A^{(2)} \|_\infty^{1-1/q}  = \max_i  | a_i |$.  For $p>2$,
     \begin{align*}
        \left\| A^{(p)} \right\|_{q}
      &=  \max_{\|u^{(1)}\|_{q/(q-1)}=1}  \max_{\begin{minipage}{2cm}\center \scriptsize  $\|u^{(k)}\|_q=1$ \\ $k=2,\ldots,p$ \end{minipage}} 
         \left|   \sum_{i_1 i_2 \ldots i_p=1}^N   
         u^{(1)}_{i_1} u^{(2)}_{i_2} \cdots u^{(p)}_{i_p}
         A_{i_1 i_2 \ldots i_p} \,
         \right| 
      \nonumber \\
        &=   \max_{\|u^{(1)}\|_{q/(q-1)}=1}  \max_{\begin{minipage}{2cm}\center \scriptsize  $\|u^{(k)}\|_q=1$ \\ $k=2,\ldots,p$ \end{minipage}} 
         \left|   \sum_{i,j=1}^N   
         u^{(1)}_{i} u^{(2)}_{i} \cdots u^{(p-1)}_{i} u^{(p)}_{j}
         A^{(2)}_{ij} \,
         \right|  
      \nonumber \\
        &\leq  \; \;  \max_{\|u\|_{q/(q-1)} \leq 1}  \; \;
            \max_{\|v\|_q = 1} \;
         \left|   \sum_{i=1}^N \sum_{t=1}^T  
         u_i  v_i A^{(2)}_{ij} \,
         \right|  = \| A^{(2)} \|_q \leq      \max_i  | a_i |,     
    \end{align*}   
    where we define $u \in \mathbbm{R}^N$
    with elements
    $u_i = u^{(1)}_{i} u^{(2)}_{i} \cdots u^{(p-1)}_{i}$
    and $v = u^{(p)}$,
    and we use that  $\|u^{(k)}\|_p=1$, for $k=2,\ldots,p-1$,
    implies $|u_i| \leq  |u^{(1)}_{i}|$ and therefore
    $ \|u\|_{q/(q-1)} \leq  \|u^{(1)}\|_{q/(q-1)} = 1$.
    We have thus shown 
    $\left\| A^{(p)} \right\| \leq  \max_i  | a_i |$.
    From the definition of $\left\| A^{(p)} \right\|_{q}$ above,
    we obtain $\left\| A^{(p)} \right\|_{q} \geq  \max_i  | a_i |$
    by choosing all $u^{(k)}$ equal to the standard basis vector,
    whose $i^*$'th component equals one, where $i^* \in  \argmax_i  | a_i |$.
   Thus, $\left\| A^{(p)} \right\|_q = \max_i  | a_i |$ for $p \geq 2$.
    The proof for  $\| B \|_q = \max_t  | b_t |$ is analogous.    
\end{proof}

The following lemma provides an asymptotic bound for the spectral norm of $N \times T$ matrices, 
whose entries are mean zero, and cross-sectionally independent and weakly time-serially dependent conditional on $\phi$.

\begin{lemma}
    \label{th:SpectralNorm}
     Let $e$ be an $N \times T$ matrix with entries $e_{it}$.
      Let
      $\bar \sigma^2_i = \frac 1 T \sum_{t=1}^T \E(e_{it}^2 ) $,     
      let $\Omega$ be the $T \times T$ matrix with entries
     $\Omega_{ts} = \frac 1 N \sum_{i=1}^N\E( e_{it} e_{is})$, 
     and let
     $\eta_{ij} 
        = \frac 1 {\sqrt{T}} \sum_{t=1}^T\left[ e_{it} e_{jt} -\E(e_{it} e_{jt}) \right]$.
     Consider asymptotic sequences where $N,T \rightarrow \infty$
    such that $N/T$ converges to a finite positive constant. Assume that
    \begin{itemize}
        \item[(i)] The distribution of $e_{it}$ is independent across $i$, conditional on $\phi,$ and
     satisfies $\E(e_{it})=0$.
    
        \item[(ii)] 
        $\frac 1 N \sum_{i=1}^N  \left( \bar \sigma^2_i \right)^4={\cal O}_P(1)$, \; \;
        $\frac 1  T {\rm Tr}( \Omega^4 ) = {\cal O}_P(1)$,   \; \;
         $\frac {1} {N} \sum_{i=1}^N\E
          \left( \eta_{ii}^4 \right) = {\cal O}_P(1)$,  \; \;
          $\frac {1} {N^2} \sum_{i,j=1}^N\E
          \left( \eta_{ij}^4 \right) = {\cal O}_P(1)$.
              
    \end{itemize} 
     Then, $\E \|e\|^8 = {\cal O}_P(N^5) $, and therefore
     $\| e \| = {\cal O}_P(N^{5/8})$.
\end{lemma}

\begin{proof}[\bf Proof of Lemma~\ref{th:SpectralNorm}]
    Let $ \|. \|_{\rm F}$ be the Frobenius norm of a matrix, i.e. $\| A \|_F = \sqrt{ {\rm Tr}(A A') }$.
    For $\bar \sigma^4_i  = (\bar \sigma^2_i)^2$, $\bar \sigma^8_i  = (\bar \sigma^2_i)^4$  and $\delta_{jk} = 1(j=k)$,
    \begin{align*}
        \| e \|^8 &=  \| e e' e e' \|^2 \leq  \| e e' e e'  \|_{\rm F}^2 
          =   \sum_{i,j=1}^N \left( \sum_{k=1}^N \sum_{t,\tau=1}^T e_{it} e_{kt} e_{k\tau} e_{j \tau} \right)^2
       \nonumber \\    &
       = T^2  \sum_{i,j=1}^N     
            \left[  \sum_{k=1}^N 
             \left( \eta_{ik} + T^{1/2} \delta_{ik} \bar \sigma^2_i \right) 
              \left( \eta_{jk} + T^{1/2} \delta_{jk} \bar \sigma^2_j  \right)   \right]^2
       \nonumber \\   
      &= T^2  \sum_{i,j=1}^N     
            \left( \sum_{k=1}^N  \eta_{ik} \eta_{jk} 
                  + 2 T^{1/2} \eta_{ij} \bar \sigma^2_i
                  + T \delta_{ij}  \bar \sigma^4_i \right)^2
       \nonumber \\   
      &\leq 3  T^2  \sum_{i,j=1}^N  \left[
            \left(  \sum_{k=1}^N  \eta_{ik} \eta_{jk}  \right)^2
                  + 4 T \eta_{ij}^2 \bar \sigma^4_i
                  + T^2 \delta_{ij}  \bar \sigma^8_i 
             \right]     
       \nonumber \\   
       &=  3  T^2  \sum_{i,j=1}^N  \left(  \sum_{k=1}^N  \eta_{ik} \eta_{jk}  \right)^2
            + 12 T^3   \sum_{i,j=1}^N \bar \sigma^4_i \eta_{ij}^2
            + 3 T^3 \sum_{i=1}^N  \bar \sigma^8_i  ,
     \end{align*}
    where we used that $(a+b+c)^2 \leq 3 (a^2+b^2+c^3)$.  
   By the Cauchy Schwarz inequality,
     \begin{align*}
        \E \|e\|^8 &\leq 
         3  T^2\E \left[ \sum_{i,j=1}^N  \left(  \sum_{k=1}^N  \eta_{ik} \eta_{jk}  \right)^2 \right]
            + 12 T^3  \sqrt{ \left( N \sum_{i=1}^N \bar \sigma^8_i  \right)
           \left( \sum_{i,j=1}^N\E(\eta_{ij}^4) \right) }
            + 3 T^3 \sum_{i=1}^N  \bar \sigma^8_i 
        \nonumber \\
          &=     3  T^2\E \left[ \sum_{i,j=1}^N  \left(  \sum_{k=1}^N  \eta_{ik} \eta_{jk}  \right)^2 \right]
            + {\cal O}_P( T^3 N^2) + {\cal O}_P( T^3 N ). 
     \end{align*}
    Moreover,
     \begin{align*}
        &\E  \left[ \sum_{i,j=1}^N  \left(  \sum_{k=1}^N  \eta_{ik} \eta_{jk}  \right)^2 \right]
        = \sum_{i,j,k,l=1}^N  \E ( \eta_{ik} \eta_{jk}  \eta_{il} \eta_{jl}  )
         = \sum_{i,j,k,l=1}^N  \E (  \eta_{ij} \eta_{jk}  \eta_{kl}     \eta_{li}  )
       \nonumber \\
          &\leq  \Bigg| \sum_{\begin{minipage}{2.5cm}\center \scriptsize $i,j,k,l$ \\ mutually different \end{minipage}}  
              \E\left(   \eta_{ij} \eta_{jk}  \eta_{kl}     \eta_{li} \right)  \Bigg|
               +  4 \left|  \sum_{i,j,k=1}^N  a_{ijk} \E (   \eta_{ii}    \eta_{ij} \eta_{jk}  \eta_{ki}   ) \right|, 
       \nonumber \\
          &\leq \Bigg| \sum_{\begin{minipage}{2.5cm}\center \scriptsize $i,j,k,l$ \\ mutually different \end{minipage}}  
              \E\left(   \eta_{ij} \eta_{jk}  \eta_{kl}     \eta_{li} \right)  \Bigg|
               +  4 \left\{  \left[ \sum_{i,j,k=1}^N\E (   \eta_{ii}^4 ) \right] 
                                \left[ \sum_{i,j,k=1}^N\E (   \eta_{ij}^4 ) \right]^3
                                \right\}^{1/4}
       \nonumber \\
          &= \Bigg| \sum_{\begin{minipage}{2.5cm}\center \scriptsize $i,j,k,l$ \\ mutually different \end{minipage}}  
              \E\left(   \eta_{ij} \eta_{jk}  \eta_{kl}     \eta_{li} \right)  \Bigg|
               +  4 N^3 \left\{  \left[ \frac 1 N \sum_{i=1}^N\E (   \eta_{ii}^4 ) \right] 
                                \left[ \frac 1 {N^2} \sum_{i,j=1}^N\E (   \eta_{ij}^4 ) \right]^3
                                \right\}^{1/4}
       \nonumber \\
          &= \Bigg| \sum_{\begin{minipage}{2.5cm}\center \scriptsize $i,j,k,l$ \\ mutually different \end{minipage}}  
              \E\left(   \eta_{ij} \eta_{jk}  \eta_{kl}     \eta_{li} \right)  \Bigg|
               +   {\cal O}_P(N^3).
    \end{align*}
    where in the second step we just renamed the indices and used that $\eta_{ij}$ is symmetric
    in $i,j$; and $a_{ijk} \in [0,1]$ in the second line is a combinatorial pre-factor;
    and in the third step we applied the Cauchy-Schwarz inequality. 
    
    Let  $\Omega_i$ be the $T \times T$ matrix with entries $\Omega_{i,ts} = \E( e_{it} e_{is} )$ such that $\Omega = \frac 1 N \sum_{i=1}^N \Omega_i$.
    For $i,j,k,l$
    mutually different,
    \begin{align*}
          \E\left(   \eta_{ij} \eta_{jk}  \eta_{kl}     \eta_{li} \right)         
         &= \frac 1{T^2} \sum_{t,s,u,v=1}^T
        \E( e_{it} e_{jt} e_{js} e_{ks} e_{ku} e_{lu} e_{lv} e_{iv} )
     \nonumber \\    
         &= \frac 1{T^2}  \sum_{t,s,u,v=1}^T
             \E( e_{iv} e_{it} )
             \E( e_{jt} e_{js} )
             \E( e_{ks} e_{ku} )
             \E( e_{lu} e_{lv}  ) 
     =  \frac 1 {T^2}  {\rm Tr}( \Omega_i \Omega_j \Omega_k \Omega_l ) \geq 0
    \end{align*} 
    because $\Omega_i \geq 0$ for all $i$.
    Thus,
     \begin{align*}
       \Bigg|  \sum_{\begin{minipage}{2.5cm}\center \scriptsize $i,j,k,l$ \\ mutually different \end{minipage}}  
              \E\left(   \eta_{ij} \eta_{jk}  \eta_{kl}     \eta_{li} \right)   \bigg|
             &=  
        \sum_{\begin{minipage}{2.5cm}\center \scriptsize $i,j,k,l$ \\ mutually different \end{minipage}}  
              \E\left(   \eta_{ij} \eta_{jk}  \eta_{kl}     \eta_{li} \right)   
         =     \frac 1 {T^2} 
            \sum_{\begin{minipage}{2cm}\center \scriptsize $i,j,k,l$ \\ mut. different \end{minipage}}   
              {\rm Tr}( \Omega_i \Omega_j \Omega_k \Omega_l )
         \nonumber \\
            &\leq    \frac 1 {T^2} 
            \sum_{i,j,k,l=1}^N   
              {\rm Tr}( \Omega_i \Omega_j \Omega_k \Omega_l )
         =    \frac {N^4} {T^2} 
              {\rm Tr}( \Omega^4) = {\cal O}_P(N^4/T).
     \end{align*}
     Combining all the above results gives 
     $\E \|e\|^8 = {\cal O}_P(N^5)$, since $N$ and $T$ are assumed to grow at the same rate.
\end{proof}

\subsection{Verifying the Basic Regularity Conditions in Panel Models}

The following Lemma provides sufficient conditions under which the panel fixed effects estimators 
in the main text  satisfy the high-level regularity conditions in Assumptions~\ref{ass:A1}$(v)$ and $(vi)$.

\begin{lemma}
   \label{lemma:BasicRegularityPanel}
   Let
   ${\cal L}(\beta,\phi) 
   = \frac 1 {\sqrt{NT}} \left[ \sum_{i,t} \ell_{it}( \beta, \pi_{it}) - \frac b 2 (v' \phi)^2 \right]$,
   where $\pi_{it} = \alpha_i + \gamma_t$, 
   $\alpha=(\alpha_1,\ldots,\alpha_N)'$,
   $\gamma=(\gamma_1,\ldots,\gamma_T)$,
     $\phi=(\alpha',\gamma')'$,
     and $v=(1_N',1_T')'$.
      Assume that  $\ell_{it}( . , .)$ is four times continuously differentiable
      in an appropriate neighborhood of the true parameter values $(\beta^0,\phi^0)$.
      Consider limits as $N,T \rightarrow \infty$ with $N/T \rightarrow \kappa^2>0$.  
      Let $4<q \leq 8$ and $0 \leq \epsilon < 1/8 - 1/(2q)$. 
     Let $r_\beta = r_{\beta,NT} >0$,
     $r_\phi = r_{\phi,NT}>0$,
     with $r_\beta = o\left[ (NT)^{-1/(2q)-\epsilon} \right]$
     and $r_\phi = o\left[ (NT)^{ -\epsilon}  \right]$.
     Assume that
  \begin{itemize}
      \item[(i)] For $k,l,m \in \{1,2,\ldots,\dim \beta\}$,
        \begin{align*}
                 & \frac 1 {\sqrt{NT}} \sum_{i,t} \partial_{\beta_k} \ell_{it} = {\cal O}_P(1),
             \ \ 
                   \frac 1 {NT} \sum_{i,t} \partial_{\beta_k \beta_l} \ell_{it} = {\cal O}_P(1),
             \ \ 
                   \frac 1 {NT} \sum_{i,t} \left\{ \partial_{\beta_k \beta_l} \ell_{it} 
                    -\E\left[ \partial_{\beta_k \beta_l} \ell_{it} \right] \right\} 
                    = o_P(1),
             \\ 
              & \sup_{\beta \in {\cal B}(r_\beta, \beta^0)}  \sup_{\phi \in {\cal B}_q(r_\phi, \phi^0)}
                  \frac 1 {NT} \sum_{i,t} \partial_{\beta_k \beta_l \beta_m} \ell_{it}( \beta, \pi_{it}) = {\cal O}_P(1).
        \end{align*}

     \item[(ii)] Let  $k,l \in \{1,2,\ldots,\dim \beta\}$.
      For $\xi_{it}(\beta,\phi) = \partial_{\beta_k \pi} \ell_{it}( \beta, \pi_{it})$
      or $\xi_{it}(\beta,\phi) = \partial_{\beta_k \beta_l \pi} \ell_{it}( \beta, \pi_{it})$,
      \begin{align*}
            \sup_{\beta \in {\cal B}(r_\beta, \beta^0)}  \sup_{\phi \in {\cal B}_q(r_\phi, \phi^0)}
           \frac 1 T \sum_t \left| \frac 1 N \sum_i  \xi_{it}(\beta,\phi) \right|^q
            &= {\cal O}_P\left( 1 \right) ,
         \\
            \sup_{\beta \in {\cal B}(r_\beta, \beta^0)}  \sup_{\phi \in {\cal B}_q(r_\phi, \phi^0)}
           \frac 1 N \sum_i \left| \frac 1 T \sum_t  \xi_{it}(\beta,\phi) \right|^q
            &= {\cal O}_P\left( 1 \right) .
      \end{align*}

      \item[(iii)]
      Let  $k,l \in \{1,2,\ldots,\dim \beta\}$.
  For $\xi_{it}(\beta,\phi) = \partial_{\pi^r} \ell_{it}( \beta, \pi_{it})$, 
    with $r \in \{3,4\}$,
  or $\xi_{it}(\beta,\phi) = \partial_{\beta_k \pi^r} \ell_{it}( \beta, \pi_{it})$,  
   with $r \in \{2,3\}$,
   or $\xi_{it}(\beta,\phi) = \partial_{\beta_k \beta_l \pi^2} \ell_{it}( \beta, \pi_{it})$, 
  \begin{align*} 
     \sup_{\beta \in {\cal B}(r_\beta, \beta^0)}  \sup_{\phi \in {\cal B}_q(r_\phi, \phi^0)}
      \max_i  \frac 1 T \sum_t | \xi_{it}(\beta,\phi) | 
        &= {\cal O}_P\left( N^{2 \epsilon} \right) ,
   \nonumber \\
     \sup_{\beta \in {\cal B}(r_\beta, \beta^0)}  \sup_{\phi \in {\cal B}_q(r_\phi, \phi^0)}
      \max_t \frac 1 N \sum_i | \xi_{it}(\beta,\phi) | 
        &= {\cal O}_P\left( N^{2 \epsilon} \right) .
  \end{align*} 
  
     \item[(iv)] Moreover,
     \begin{align*}
           & \frac 1 T \sum_t \left| \frac 1 {\sqrt{N}} \sum_i \partial_{\pi} \ell_{it} \right|^q
            = {\cal O}_P\left( 1 \right) ,
          \ \  
           \frac 1 N \sum_i \left| \frac 1 {\sqrt{T}} \sum_t \partial_{\pi} \ell_{it} \right|^q
            = {\cal O}_P\left( 1 \right) ,
          \\
           & \frac 1 T \sum_t \left| \frac 1 {\sqrt{N}} \sum_i  \partial_{\beta_k \pi} \ell_{it} - 
              \E\left[ \partial_{\beta_k \pi} \ell_{it} \right] \right|^2
            = {\cal O}_P\left( 1 \right) ,
          \\  
           & \frac 1 N \sum_i \left| \frac 1 {\sqrt{T}} \sum_t  \partial_{\beta_k \pi} \ell_{it} - 
              \E \left[ \partial_{\beta_k \pi} \ell_{it} \right] \right|^2
            = {\cal O}_P\left( 1 \right) .
     \end{align*}       
     
    \item[(v)] The sequence $\{(\ell_{i1}, \ldots, \ell_{iT}) : 1 \leq i \leq N \}$
        is independent across $i$ conditional on $\phi$.

    \item[(vi)]  
              Let  $k \in \{1,2,\ldots,\dim \beta\}$.
             For $\xi_{it} = \partial_{\pi^r} \ell_{it}
                           -  \mathbb{E}_\phi\left[ \partial_{\pi^r} \ell_{it} \right]$, with $r \in \{2,3\}$,
                     or $\xi_{it} = \partial_{\beta_k \pi^2} \ell_{it}
                           -  \mathbb{E}_\phi\left[ \partial_{\beta_k \pi^2} \ell_{it} \right]$, and  some $\tilde \nu>0$,
                   \begin{align*}
                         & \max_i  \, \mathbb{E}_\phi\left[ \xi_{it}^{8+\tilde \nu} \right] \leq C ,
                         \ \
                        \max_{i} \,  \max_t \sum_s \mathbb{E}_\phi\left[ \xi_{it} \xi_{is}  \right]  \leq C ,
                         \ \
                          \max_i \, \mathbb{E}_\phi    \left\{ \left[ \frac 1 {\sqrt{T}} 
                            \sum_t      \xi_{it}  
                         \right]^8    \right\}
                           \leq C ,
                         \\
                        & \max_t \, \mathbb{E}_\phi    \left\{ \left[ \frac 1 {\sqrt{N}} 
                            \sum_i      \xi_{it}  
                         \right]^8    \right\}
                           \leq C ,
                         \ \
                        \max_{i,j} \, \mathbb{E}_\phi    \left\{ \left[ \frac 1 {\sqrt{T}} 
                            \sum_t   \left[  \xi_{it} \xi_{jt} - \mathbb{E}_\phi\left(   \xi_{it} \xi_{jt}  \right) \right]
                         \right]^4    \right\}
                           \leq C ,
                     \end{align*}
                      uniformly in  $N,T$, where $C >0 $ is a constant.  
                     
       \item[(vii)] $\left\| \overline {\cal H}^{-1} \right\|_q
               = {\cal O}_P \left(  1\right)$.                  
  \end{itemize}
    Then, Assumptions~\ref{ass:A1}$(v)$ and $(vi)$ are satisfied with the same
    parameters $q$, $\epsilon$, $r_\beta = r_{\beta,NT}$
    and $r_\phi = r_{\phi,NT}$ used here.
\end{lemma}   

\begin{proof}[\bf Proof of Lemma~\ref{lemma:BasicRegularityPanel}]
     The penalty term $(v'\phi)^2$  is quadratic in $\phi$
     and does not depend on $\beta$. This term thus only enters  $\partial_\phi {\cal L}(\beta,\phi) $
     and  $\partial_{\phi \phi'} {\cal L}(\beta,\phi) $, but it does not effect any other partial derivative
     of ${\cal L}(\beta,\phi)$. Furthermore, the contribution of the penalty drops out
     of ${\cal S} = \partial_\phi {\cal L}(\beta^0,\phi^0)$, because
     we impose the normalization $v' \phi^0 =0$. It also drops out of
     $\tilde {\cal H}$, because it contributes the same to ${\cal H}$ and $\overline {\cal H}$.
     We can therefore  ignore the penalty term for the purpose of proving the lemma 
     (but it is necessary to satisfy the assumption
     $\left\| \overline {\cal H}^{-1} \right\|_q
               = {\cal O}_P \left(  1\right)$).

   \# Assumption $(i)$ implies that
    $\|   \partial_{\beta} {\cal L} \| = {\cal O}_P(1)$, 
    $\left\| \partial_{\beta \beta'}  {\cal L} \right\| = {\cal O}_P(\sqrt{NT})$,
    $\left\| \partial_{\beta \beta'} \widetilde {\cal L}  \right\| = o_P( \sqrt{NT} )$,
    and
    $\displaystyle \sup_{\beta \in {\cal B}(r_\beta, \beta^0)}  \sup_{\phi \in {\cal B}_q(r_\phi, \phi^0)} 
        \left\|  \partial_{\beta \beta \beta}  {\cal L}(\beta,\, \phi) \right\|
                 = {\cal O}_P\left( \sqrt{NT} \right)$.
    Note that it does not matter which norms we use here because $\dim \beta$ is fixed.             
         
    \# By Assumption $(ii)$,  
    $\left\| \partial_{\beta \phi'}  {\cal L} \right\|_q =    {\cal O}_P \left( (NT)^{1/(2q)} \right)$
    and
    $\displaystyle \sup_{\beta \in {\cal B}(r_\beta, \beta^0)}  \sup_{\phi \in {\cal B}_q(r_\phi, \phi^0)}  
        \left\|  \partial_{\beta \beta \phi}  {\cal L}(\beta,\, \phi) \right\|_q
                 =   {\cal O}_P\left(  (NT)^{1/(2q)} \right)$.
    For example, 
    $\partial_{\beta_k \alpha_i}  {\cal L}
     = \frac 1 {\sqrt{NT}} \sum_t \partial_{\beta_k \pi} \ell_{it}$
     and therefore
   $$\left\| \partial_{\beta_k \alpha}  {\cal L} \right\|_q
   = 
   \left( \sum_i \left| \frac 1 {\sqrt{NT}} \sum_t \partial_{\beta_k \pi} \ell_{it} \right|^q \right)^{1/q}
   = {\cal O}_P \left( N^{1/q} \right) = {\cal O}_P \left( (NT)^{1/(2q)} \right).$$  
  Analogously,       
  $\left\| \partial_{\beta_k \gamma}  {\cal L} \right\|_q  
    =  {\cal O}_P \left( (NT)^{1/(2q)} \right)$,
  and therefore
  $\left\| \partial_{\beta_k \phi}  {\cal L} \right\|_q
  \leq 
   \left\| \partial_{\beta_k \alpha}  {\cal L} \right\|_q
   + \left\| \partial_{\beta_k \gamma}  {\cal L} \right\|_q 
   =  {\cal O}_P \left( (NT)^{1/(2q)} \right)$. 
  This also implies that
   $\left\| \partial_{\beta \phi'}  {\cal L} \right\|_q =    {\cal O}_P \left( (NT)^{1/(2q)} \right)$ because $\dim \beta$ is fixed.
   
   \# By Assumption $(iii)$, 
   $\left\|  \partial_{\phi \phi \phi}  {\cal L} \right\|_q   =     {\cal O}_P\left(  (NT)^{\epsilon} \right)$,
     $\left\| \partial_{\beta \phi \phi}  {\cal L} \right\|_q   =  {\cal O}_P( (NT)^{\epsilon} )$,
     \\
      $\displaystyle \sup_{\beta \in {\cal B}(r_\beta, \beta^0)}  \sup_{\phi \in {\cal B}_q(r_\phi, \phi^0)} 
        \left\|  \partial_{\beta \beta \phi \phi}  {\cal L}(\beta,\, \phi) \right\|_q
                 =   {\cal O}_P\left(  (NT)^{\epsilon} \right)$,
     $\displaystyle \sup_{\beta \in {\cal B}(r_\beta, \beta^0)}  \sup_{\phi \in {\cal B}_q(r_\phi, \phi^0)}  
        \left\|  \partial_{\beta \phi \phi \phi}  {\cal L}(\beta,\, \phi) \right\|_q
                 =  {\cal O}_P\left(  (NT)^{\epsilon} \right)$,
      and           
     $\displaystyle  \sup_{\beta \in {\cal B}(r_\beta, \beta^0)}  \sup_{\phi \in {\cal B}_q(r_\phi, \phi^0)}  \left\| \partial_{\phi \phi \phi \phi} {\cal L}(\beta,\phi) \right\|_q    =    {\cal O}_P\left(  (NT)^{\epsilon} \right)$.    
     For example, 
     \begin{align*}
         \left\|  \partial_{\phi \phi \phi}  {\cal L} \right\|_q
         &\leq  \left\|  \partial_{\alpha \alpha \alpha}  {\cal L} \right\|_q
                +   \left\|  \partial_{\alpha \alpha \gamma}  {\cal L} \right\|_q
                +   \left\|  \partial_{\alpha \gamma \alpha}  {\cal L} \right\|_q
                +   \left\|  \partial_{\alpha \gamma \gamma}  {\cal L} \right\|_q 
          \nonumber \\ & \qquad \qquad   
                +    \left\|  \partial_{\gamma \alpha \alpha}  {\cal L} \right\|_q
                +   \left\|  \partial_{\gamma \alpha \gamma}  {\cal L} \right\|_q
                +   \left\|  \partial_{\gamma \gamma \alpha}  {\cal L} \right\|_q
                +   \left\|  \partial_{\gamma \gamma \gamma}  {\cal L} \right\|_q 
          \nonumber \\ &
          \leq     \left\|  \partial_{\pi \alpha \alpha}  {\cal L} \right\|_q
                +   \left\|  \partial_{\pi \gamma \gamma}  {\cal L} \right\|_q    
                +  3  \left\|  \partial_{\pi \alpha \gamma}  {\cal L} \right\|_q
                +  3  \left\|  \partial_{\pi \gamma \alpha}  {\cal L} \right\|_q
       \nonumber \\         
         &\leq       \left\|  \partial_{\pi \alpha \alpha}  {\cal L} \right\|_\infty
                +   \left\|  \partial_{\pi \gamma \gamma}  {\cal L} \right\|_\infty 
                + 3    \left\|  \partial_{\pi \alpha \gamma}  {\cal L} \right\|_\infty^{1-1/q} 
                                  \left\|  \partial_{\pi  \gamma \alpha}  {\cal L} \right\|_\infty^{1/q} 
                + 3    \left\|  \partial_{\pi \alpha \gamma}  {\cal L} \right\|_\infty^{1/q} 
                                  \left\|  \partial_{\pi  \gamma \alpha}  {\cal L} \right\|_\infty^{1-1/q} 
       \nonumber \\         
          &=   \frac 1 {\sqrt{NT}} \Bigg[    \max_i  \left|  \sum_t   \partial_{\pi^3} \ell_{it}  \right|  
                 + \max_t \left|  \sum_i   \partial_{\pi^3} \ell_{it}  \right| 
                  + 3  \left( \max_i  \sum_t |  \partial_{\pi^3} \ell_{it}  |   \right)^{1-1/q}
                                 \left(  \max_t \sum_t |  \partial_{\pi^3} \ell_{it}  |  \right)^{1/q}   
              \nonumber \\ & \qquad \qquad     \qquad    
                  + 3  \left( \max_i  \sum_t |  \partial_{\pi^3} \ell_{it}  |   \right)^{1/q}
                                 \left(  \max_t \sum_t |  \partial_{\pi^3} \ell_{it}  |  \right)^{1-1/q}                  
                                  \Bigg]
       \nonumber \\         
          &\leq  \frac 1 {\sqrt{NT}} \Bigg[    \max_i  \sum_t   | \partial_{\pi^3} \ell_{it} | 
                 + \max_t  \sum_i |  \partial_{\pi^3} \ell_{it}  |
                  + 3  \left( \max_i  \sum_t |  \partial_{\pi^3} \ell_{it}  |   \right)^{1-1/q}
                                 \left(  \max_t \sum_t |  \partial_{\pi^3} \ell_{it}  |  \right)^{1/q}   
              \nonumber \\ & \qquad \qquad     \qquad    
                  + 3  \left( \max_i  \sum_t |  \partial_{\pi^3} \ell_{it}  |   \right)^{1/q}
                                 \left(  \max_t \sum_t |  \partial_{\pi^3} \ell_{it}  |  \right)^{1-1/q}                  
                                  \Bigg]
       = {\cal O}_P(N^{2 \epsilon}) = {\cal O}_P( (NT)^{\epsilon} ).
     \end{align*}
     Here, we use Lemma~\ref{lemma:SpectralNormTensors} to bound
     the norms of the 3-tensors in terms of the norms of matrices, e.g.
     $ \left\|  \partial_{\alpha \alpha \gamma}  {\cal L} \right\|_q \leq  \left\|  \partial_{\pi \alpha \gamma}  {\cal L} \right\|_q$, because $\partial_{\alpha_i \alpha_j \gamma_t}  {\cal L} = 0$ if $i \neq j$
     and $\partial_{\alpha_i \alpha_i \gamma_t}  {\cal L} 
      = (NT)^{-1/2} \partial_{\pi \alpha_i \gamma_t}$.\footnote{%
      With a slight abuse of notation
      we write $\partial_{\pi \alpha \gamma}  {\cal L} $ for the $N\times T$
      matrix with entries $(NT)^{-1/2} \partial_{\pi^3}  \ell_{it}    = (NT)^{-1/2} \partial_{\pi^3}  \ell_{it} $,
      and analogously for $\partial_{\pi \alpha \alpha}  {\cal L}$,
      $ \partial_{\pi \gamma \gamma}  {\cal L} $,
      and $ \partial_{\pi \gamma \alpha}  {\cal L}$.}
     Then, we use Lemma~\ref{lemma:matrix-q-norm} to bound $q$-norms in terms
     of $\infty$-norms, and then explicitly expressed those $\infty$-norm in terms of the elements
     of the matrices. Finally, we use that
     $\left|  \sum_i   \partial_{\pi^3} \ell_{it}  \right|  \leq  \sum_i   | \partial_{\pi^3} \ell_{it} | $
     and $\left|  \sum_t   \partial_{\pi^3} \ell_{it}  \right|  \leq  \sum_t   | \partial_{\pi^3} \ell_{it} | $,
     and apply Assumption $(iii)$.
     
     \# By Assumption $(iv)$,
     $\| {\cal S} \|_q = {\cal O}_P \left( (NT)^{-1/4 + 1/(2q)} \right)$
     and  $\left\| \partial_{\beta \phi'}  \widetilde {\cal L} \right\|
                 =   {\cal O}_P \left( 1 \right)$.
       For example,          
       \begin{align*}
             \| {\cal S} \|_q
            &=   \frac{1}{\sqrt{NT}} 
             \left( \sum_i \left|   \sum_t \partial_{\pi} \ell_{it}  \right|^q
            +  \sum_t \left|   \sum_i \partial_{\pi} \ell_{it}  \right|^q
             \right)^{1/q}
         = {\cal O}_P\left( N^{-1/2+1/q} \right) = {\cal O}_P \left( (NT)^{-1/4 + 1/(2q)} \right).
       \end{align*}    
       
   \# By Assumption $(v)$ and $(vi)$, $\| \widetilde {\cal H} \| =  {\cal O}_P\left( (NT)^{-3/16} \right) =  o_P \left( (NT)^{-1/8} \right)$
   and $ \left\| \partial_{\beta \phi \phi}  \widetilde {\cal L} \right\|  =  {\cal O}_P\left( (NT)^{-3/16} \right)
                 =   o_P \left( (NT)^{-1/8} \right)$. 
      We now show it $\| \widetilde {\cal H} \| $. The proof for   
      $ \left\| \partial_{\beta \phi \phi}  \widetilde {\cal L} \right\|$ is analogous.         
 
     By the triangle inequality,           
    \begin{align*}
        \| \widetilde {\cal H} \|
        &= 
      \left\|  \partial_{\phi \phi'} {\cal L} - \mathbb{E}_\phi\left[ \partial_{\phi \phi'} {\cal L} \right] \right\|      
      \leq   \left\|  \partial_{\alpha \alpha'} {\cal L} - \mathbb{E}_\phi\left[ \partial_{\alpha \alpha'} {\cal L} \right] \right\|
          +  \left\|  \partial_{\gamma \gamma'} {\cal L} - \mathbb{E}_\phi\left[ \partial_{\gamma \gamma'} {\cal L} \right] \right\|
               + 2  \left\|  \partial_{\alpha \gamma'} {\cal L} - \mathbb{E}_\phi\left[ \partial_{\alpha \gamma'} {\cal L} \right] \right\|.
    \end{align*}             
    Let $\xi_{it} = \partial_{\pi^2} \ell_{it}
                           -  \mathbb{E}_\phi\left[ \partial_{\pi^2} \ell_{it} \right]$.
    Since $\partial_{\alpha \alpha'} {\cal L}$ is a diagonal matrix with diagonal entries
    $\frac 1 {\sqrt{NT}} \sum_t \xi_{it}$,
    $\left\|  \partial_{\alpha \alpha'} {\cal L} - \mathbb{E}_\phi\left[ \partial_{\alpha \alpha'} {\cal L} \right] \right\|
         =  \max_i  \frac 1 {\sqrt{NT}}
            \sum_t \xi_{it}$, and therefore
    \begin{align*}
         \mathbb{E}_\phi
         \left\|  \partial_{\alpha \alpha'} {\cal L} - \mathbb{E}_\phi\left[ \partial_{\alpha \alpha'} {\cal L} \right] \right\|^8
         &=    \mathbb{E}_\phi\left[    \max_i  \left( \frac 1 {\sqrt{NT}}
            \sum_t \xi_{it} \right)^8 \right]
        \nonumber \\    &
        \leq   \mathbb{E}_\phi\left[    \sum_i  \left( \frac 1 {\sqrt{NT}}
            \sum_t \xi_{it} \right)^8 \right]
         \leq   C N   \left( \frac {1} {\sqrt{N}} \right)^8  = {\cal O}_P(N^{-3}).
    \end{align*}
    Thus, 
     $\left\|  \partial_{\alpha \alpha'} {\cal L} - \mathbb{E}_\phi\left[ \partial_{\alpha \alpha'} {\cal L} \right] \right\|
      = {\cal O}_P(N^{-3/8})$.
     Analogously, 
     $ \left\|  \partial_{\gamma \gamma'} {\cal L} - \mathbb{E}_\phi\left[ \partial_{\gamma \gamma'} {\cal L} \right] \right\|
       =   {\cal O}_P(N^{-3/8})$.
       
    Let $\xi$ be the $N \times T$ matrix with entries  $\xi_{it}$. We now show that $\xi$
    satisfies all the regularity condition of Lemma~\ref{th:SpectralNorm} with $e_{it} = \xi_{it}$.
    Independence across $i$ is assumed. Furthermore, 
      $\bar \sigma^2_i = \frac 1 T \sum_{t=1}^T \mathbb{E}_\phi(\xi_{it}^2 ) \leq C^{1/4}$
      so that $\frac 1 N \sum_{i=1}^N  \left( \bar \sigma^2_i \right)^4={\cal O}_P(1)$.
      For $\Omega_{ts} = \frac 1 N \sum_{i=1}^N \mathbb{E}_\phi( \xi_{it} \xi_{is})$,
    \begin{align*}
         \frac 1  T {\rm Tr}( \Omega^4 ) 
         \leq \| \Omega \|^4 \leq   \| \Omega \|_\infty^4
         = \left( \max_t \sum_s \mathbb{E}_\phi\left[ \xi_{it} \xi_{is}  \right] \right)^4  &\leq C
         = {\cal O}_P(1).
    \end{align*}  
     For
     $\eta_{ij} 
        = \frac 1 {\sqrt{T}} \sum_{t=1}^T\left[ \xi_{it} \xi_{jt} - \mathbb{E}_\phi(\xi_{it} \xi_{jt}) \right]$
     we assume $\mathbb{E}_\phi \eta_{ij}^4 \leq C$, which implies   
     $\frac {1} {N} \sum_{i=1}^N \mathbb{E}_\phi
          \left( \eta_{ii}^4 \right) = {\cal O}_P(1)$ and
          $\frac {1} {N^2} \sum_{i,j=1}^N \mathbb{E}_\phi
          \left( \eta_{ij}^4 \right) = {\cal O}_P(1)$. Then,  Lemma~\ref{th:SpectralNorm} gives
      $\| \xi \| = {\cal O}_P(N^{5/8})$.
        Note that  $\xi = \frac 1 {\sqrt{NT}} \partial_{\alpha \gamma'} {\cal L} - \mathbb{E}_\phi\left[ \partial_{\alpha \gamma'} {\cal L} \right]$ and therefore
   $\left\|  \partial_{\alpha \gamma'} {\cal L} - \mathbb{E}_\phi\left[ \partial_{\alpha \gamma'} {\cal L} \right] \right\|
    =     {\cal O}_P(N^{-3/8})$.
   We conclude that  $ \| \widetilde {\cal H} \| =     {\cal O}_P(N^{-3/8}) = {\cal O}_P\left( (NT)^{-3/16} \right)$.
              
  \# Moreover, for  $\xi_{it} = \partial_{\pi^2} \ell_{it}
                           -  \mathbb{E}_\phi\left[ \partial_{\pi^2} \ell_{it} \right]$
     \begin{align*}
         \mathbb{E}_\phi \|\widetilde {\cal H}\|^{8+\tilde \nu}_\infty &=\mathbb{E}_\phi  \left( \frac 1 {\sqrt{NT}} \max_i \sum_t | \xi_{it} | \right)^{8+\tilde \nu}
         = \mathbb{E}_\phi \max_i \left( \frac 1 {\sqrt{NT}}  \sum_t | \xi_{it} | \right)^{8+\tilde \nu}
        \nonumber \\  &
        \leq \mathbb{E}_\phi \sum_i \left( \frac 1 {\sqrt{NT}}  \sum_t | \xi_{it} | \right)^{8+\tilde \nu}  \leq \mathbb{E}_\phi \sum_i
         \left( \frac {T} {\sqrt{NT}} \right)^{8+\tilde \nu}
         \left( \frac 1 {T}  \sum_t | \xi_{it} |^{8+\tilde \nu} \right) = {\cal O}_P(N) ,
     \end{align*}
     and therefore $ \|\widetilde {\cal H}\|_\infty =  o_P(N^{1/8})$.
  Thus, by Lemma~\ref{lemma:matrix-q-norm} 
     \begin{align*}
          \| \widetilde {\cal H} \|_q  &\leq \| \widetilde {\cal H} \|_2^{ 2/q} \|\widetilde {\cal H}\|_\infty^{1-2/q}    
          = o_P\left(   N^{1/8 [-6/q + (1 - 2/q)] } \right)
           = o_P\left(   N^{-1/q + 1/8} \right) = o_P(1),
     \end{align*}
     where we use that $q \leq 8$.
  
  \# Finally we show that
  $ \left\|  \sum_{g,h=1}^{\dim \phi} 
                \partial_{\phi \phi_g \phi_h} \widetilde {\cal L} \,
                [ \overline {\cal H}^{-1}  {\cal S} ]_g
                [ \overline {\cal H}^{-1}  {\cal S} ]_h
                   \right\|   
                      =   o_P \left( (NT)^{-1/4}  \right) $. First,
    \begin{align*}
        &\left\|  \sum_{g,h=1}^{\dim \phi} 
                \partial_{\phi \phi_g \phi_h} \widetilde {\cal L} \,
                [ \overline {\cal H}^{-1}  {\cal S} ]_g
                [ \overline {\cal H}^{-1}  {\cal S} ]_h
                   \right\|   
       \nonumber \\             & \qquad 
       \leq          \left\|  \sum_{g,h=1}^{\dim \phi} 
                \partial_{\alpha \phi_g \phi_h} \widetilde {\cal L} \,
                [ \overline {\cal H}^{-1}  {\cal S} ]_g
                [ \overline {\cal H}^{-1}  {\cal S} ]_h
                   \right\|   
               +  \left\|  \sum_{g,h=1}^{\dim \phi} 
                \partial_{\gamma \phi_g \phi_h} \widetilde {\cal L} \,
                [ \overline {\cal H}^{-1}  {\cal S} ]_g
                [ \overline {\cal H}^{-1}  {\cal S} ]_h
                   \right\|    .
    \end{align*}   
    Let  $(v, w)' := \overline {\cal H}^{-1}  {\cal S}$, where
    $v$ is a $N$-vector and $w$ is a $T$-vector. 
    We assume $\left\| \overline {\cal H}^{-1} \right\|_q
               = {\cal O}_P \left(  1\right)$.
     By Lemma~\ref{lemma:assA1add} this also implies
       $\left\| \overline {\cal H}^{-1} \right\|
               =   {\cal O}_P \left(  1\right)$  
    and           $\| {\cal S} \| = {\cal O}_P\left(  1\right)$.
    Thus, 
    $\|v\| \leq \left\| \overline {\cal H}^{-1} \right\| \| {\cal S} \| = {\cal O}_P\left(  1\right)$,
    $\|w\| \leq \left\| \overline {\cal H}^{-1} \right\| \| {\cal S} \| = {\cal O}_P\left(  1\right)$,
    $\|v\|_\infty \leq \|v\|_q \leq \left\| \overline {\cal H}^{-1} \right\|_q \| {\cal S} \|_q =  {\cal O}_P \left( (NT)^{-1/4 + 1/(2q)} \right)$,
    $\|w\|_\infty \leq \|w\|_q \leq \left\| \overline {\cal H}^{-1} \right\|_q \| {\cal S} \|_q =  {\cal O}_P \left( (NT)^{-1/4 + 1/(2q)} \right)$. Furthermore, by an analogous argument to the above proof for $\| \widetilde {\cal H} \|,$ Assumption $(v)$ and $(vi)$  imply that
    $ \left\|  \partial_{\pi \alpha \alpha'} \widetilde {\cal L} \right\| =   {\cal O}_P(N^{-3/8})$,
    $  \left\|  \partial_{\pi \alpha \gamma'} \widetilde {\cal L} \right\| =   {\cal O}_P(N^{-3/8})$,
    $ \left\|  \partial_{\pi \gamma \gamma'} \widetilde {\cal L} \right\| =   {\cal O}_P(N^{-3/8})$.
      Then, 
    \begin{align*}
         \sum_{g,h=1}^{\dim \phi} 
                \partial_{\alpha_i \phi_g \phi_h} \widetilde {\cal L} \,
                [ \overline {\cal H}^{-1}  {\cal S} ]_g
                [ \overline {\cal H}^{-1}  {\cal S} ]_h
         & =  \sum_{j,k=1}^N (\partial_{\alpha_i \alpha_j \alpha_k} \widetilde {\cal L}) v_j v_k        
              + 2 \sum_{j=1}^N \sum_{t=1}^T (\partial_{\alpha_i \alpha_j \gamma_t} \widetilde {\cal L}) v_j w_t   
              + \sum_{t,s=1}^T (\partial_{\alpha_i \gamma_t \gamma_s} \widetilde {\cal L}) w_t w_s   
         \nonumber \\
         &=  \sum_{j=1}^N (\partial_{\pi^2 \alpha_i} \widetilde {\cal L}) v_i^2        
              + 2   \sum_{t=1}^T (\partial_{\pi \alpha_i  \gamma_t} \widetilde {\cal L}) v_i w_t   
              + \sum_{t=1}^T (\partial_{\pi \alpha_i \gamma_t} \widetilde {\cal L}) w_t^2  ,
    \end{align*}               
    and therefore
    \begin{align*}
         \left\|  \sum_{g,h=1}^{\dim \phi} 
                \partial_{\alpha \phi_g \phi_h} \widetilde {\cal L} \,
                [ \overline {\cal H}^{-1}  {\cal S} ]_g
                [ \overline {\cal H}^{-1}  {\cal S} ]_h
                   \right\|  
  \leq
            \left\|  \partial_{\pi \alpha \alpha'} \widetilde {\cal L} \right\|
               \| v \| \|v\|_\infty
        + 2  \left\|  \partial_{\pi \alpha \gamma'} \widetilde {\cal L} \right\|
                 \| w \|        \|v\|_\infty
            +   \left\|  \partial_{\pi \alpha \gamma'} \widetilde {\cal L} \right\|
               \| w \| \|w\|_\infty     &
        \nonumber \\           
             \quad \quad \quad \quad = {\cal O}_P( N^{-3/8}  )  {\cal O}_P \left( (NT)^{-1/4 + 1/(2q)} \right)
              =  {\cal O}_P \left( (NT)^{-1/4 - 3/16 + 1/(2q)} \right) = o_P\left( (NT)^{-1/4}  \right) ,&
    \end{align*}
    where we use that $q>4$.
    Analogously,
    $ \left\|  \sum_{g,h=1}^{\dim \phi} 
                \partial_{\gamma \phi_g \phi_h} \widetilde {\cal L} \,
                [ \overline {\cal H}^{-1}  {\cal S} ]_g
                [ \overline {\cal H}^{-1}  {\cal S} ]_h
                   \right\| = o_P\left( (NT)^{-1/4}  \right)$
     and thus also 
             $ \left\|  \sum_{g,h=1}^{\dim \phi} 
                \partial_{\phi \phi_g \phi_h} \widetilde {\cal L} \,
                [ \overline {\cal H}^{-1}  {\cal S} ]_g
                [ \overline {\cal H}^{-1}  {\cal S} ]_h
                   \right\|   
                      =   o_P \left( (NT)^{-1/4}  \right) $.\footnote{%
    Given the structure of this last part of the proof of Lemma~\ref{lemma:BasicRegularityPanel} 
    one might wonder why, instead of
    $ \left\|  \sum_{g,h=1}^{\dim \phi} 
                \partial_{\phi \phi_g \phi_h} \widetilde {\cal L} \,
                [ \overline {\cal H}^{-1}  {\cal S} ]_g
                [ \overline {\cal H}^{-1}  {\cal S} ]_h
                   \right\|   
                      =   o_P \left( (NT)^{-1/4}  \right),$
    we did not directly impose
    $\sum_g \left\|  \partial_{\phi_g \phi \phi'} \widetilde {\cal L}  \right\|=   o_P\left( (NT)^{-1/(2q)} \right)$
    as a  high-level condition in Assumption~\ref{ass:A1}$(vi)$. 
    While this alternative high-level assumption would indeed be more elegant and sufficient
    to derive our results, it would not be satisfied for panel models,
    because it involves bounding 
   $ \sum_i \left\|  \partial_{\alpha_i \gamma \gamma'} \widetilde {\cal L}  \right\|$
 and   $ \sum_t \left\|  \partial_{\gamma_t \alpha \alpha'} \widetilde {\cal L}  \right\|$,
   which was avoided in the proof of Lemma~\ref{lemma:BasicRegularityPanel}.
    }
\end{proof}

\subsection{A Useful Algebraic Result}

Let $\widetilde{\mathbbm{P}}$ be the linear operator
defined in equation \eqref{DefTildeP}, and and let $\mathbbm{P}$ be the related  projection
operator  defined in \eqref{DefP}. 
Lemma \ref{RewriteSumIT} shows how
in the context of panel data models some expressions that 
appear in the general expansion
of Appendix~\ref{app:expansion} can be conveniently  expressed using 
the operator $\widetilde{\mathbbm{P}}$. This lemma is used extensively in the proof
of part $(ii)$ of Theorem~\ref{th:connection}.

\begin{lemma}
   \label{RewriteSumIT}
 Let $A$, $B$ and $C$ be $N \times T$ matrices,
 and let the expected incidental parameter
 Hessian $\overline {\cal H}$ be invertible. Define
 the $N+T$ vectors ${\cal A}$ and ${\cal B}$
 and the $(N+T) \times (N+T)$ matrix ${\cal C}$
 as follows\footnote{Note that
 $A 1_T$ is simply the $N$-vectors with entries $\sum_t A_{it}$
 and $A' 1_N$ is simply the $T$-vector with entries $\sum_i A_{it}$,
 and analogously for $B$ and $C$.}   
 \begin{align*}
     {\cal A} &=  \frac 1 {NT}  { A 1_T
                                \choose A' 1_N } ,
      &
     {\cal B} &=   \frac 1 {NT}  { B 1_T
                                \choose B' 1_N } ,
      &    
     {\cal C} = \frac1 {NT} \left( \begin{array}{@{}c@{}c@{}}
        \diag \left(  C 1_T \right)
       &  C
      \\  C'
       &   \diag \left( C' 1_N \right)
      \end{array} 
      \right)  .
 \end{align*}
Then,
\begin{itemize}
    \item[(i)] \qquad  $\displaystyle {\cal A}' \;
       \overline {\cal H}^{-1} \;
      {\cal B}    
       =      \frac 1 {(NT)^{3/2}} \sum_{i,t}  
       (  \widetilde{\mathbbm{P}}  A )_{it} B_{it} =   \frac 1 {(NT)^{3/2}} \sum_{i,t}  
       (  \widetilde{\mathbbm{P}}  B )_{it} A_{it}$,
       
    \item[(ii)] \qquad  $\displaystyle {\cal A}' \;
       \overline {\cal H}^{-1} \;
      {\cal B}    
       =      \frac 1 {(NT)^{3/2}} \sum_{i,t} 
     \mathbb{E}_\phi( - \partial_{\pi^2} \ell_{it} )  (  \widetilde{\mathbbm{P}}  A )_{it}
       (  \widetilde{\mathbbm{P}}  B )_{it}$,
       
    \item[(iii)]   \qquad
       $\displaystyle {\cal A}' \;
       \overline {\cal H}^{-1} \;
       {\cal C} \;
        \overline {\cal H}^{-1} \;
      {\cal B}    
     =  \frac 1 {(NT)^2} \sum_{i,t}    (  \widetilde{\mathbbm{P}}  A )_{it}
       C_{it}
        (  \widetilde{\mathbbm{P}}  B  )_{it}$.  
\end{itemize}

\end{lemma}

\begin{proof}[\bf{Proof}] 
    Let  $\tilde \alpha^*_{i} +
                       \tilde \gamma^*_{t}  =  (\mathbbm{P} \tilde A)_{it}
                        =  ( \widetilde{\mathbbm{P}}  A  )_{it}$,
    with $\tilde A$ as defined in equation \eqref{DefTildeP}. The first order condition of the
    minimization problem in the definition of $(\mathbbm{P} \tilde A)_{it}$
    can be written as      
    $\frac 1 {\sqrt{NT}} \overline {\cal H}^{*} { \tilde \alpha^* \choose
                      \tilde  \gamma^* }
      =  {\cal A}$.
     One solution to this equation is  
     $ { \tilde \alpha^* \choose
                      \tilde  \gamma^* } = \sqrt{NT} \, \overline {\cal H}^{-1}  {\cal A}$
     (this is the solution that imposes the normalization
     $\sum_i \tilde \alpha^*_{i}  = \sum_t  \tilde  \gamma^*$,
     but this is of no importance in the following).  
Thus,
\begin{align*}
   &  \sqrt{NT} \; {\cal A}' \;
       \overline {\cal H}^{-1} \;
      {\cal B} 
   =     { \tilde \alpha^* \choose
                      \tilde  \gamma^* }'
         {\cal B}   
   =    \frac 1 {NT} \left[ \sum_{i,t}      \tilde \alpha^{*}_i    B_{it}
                   +   \sum_{i,t}      \tilde \gamma^{*}_t      B_{it} \right]
        = \frac 1 {NT} \sum_{i,t}  
      (  \widetilde{\mathbbm{P}}  A )_{it} B_{it}     .      
\end{align*}   
This gives the first equality of Statement $(i)$. 
The second equality of Statement $(i)$ follows by symmetry.  Statement $(ii)$
is a special case of of Statement $(iii)$ with ${\cal C}= \frac 1 {\sqrt{NT}} \overline {\cal H}^*$, so we only need to
prove Statement $(iii)$.

Let  $\alpha^*_{i} + \gamma^*_{t}  =  (\mathbbm{P} \tilde B)_{it}
                        =  ( \widetilde{\mathbbm{P}}  B  )_{it}$,
  where $    \tilde B_{it} =  \frac{ B_{it} } { \mathbb{E}_\phi( - \partial_{\pi^2} \ell_{it} ) }$.       
By an argument analogous to the one given above, we can choose    
${ \alpha^* \choose
                       \gamma^* } = \sqrt{NT}    \,  
    \overline {\cal H}^{-1} {\cal B}$
as one solution to the minimization problem.    
Then,
\begin{align*}
    NT \;   {\cal A}' \;
       \overline {\cal H}^{-1} \;
       {\cal C} \;
        \overline {\cal H}^{-1} \;
      {\cal B}            
      &=  \frac 1 {NT}  \sum_{i,t}  
       \left[ \tilde \alpha_i^{*}
       C_{it} \alpha_i^*
          + \tilde \alpha_i^{*}
       C_{it} \gamma_t^*
         +  \tilde \gamma_t^{*}
       C_{it} \alpha_i^*
          + \tilde \gamma_t^{*}
       C_{it} \gamma_t^*
      \right] 
 \\ & 
    =  \frac 1 {NT}   \sum_{i,t}    (  \widetilde{\mathbbm{P}}  A )_{it}
       C_{it}
        (  \widetilde{\mathbbm{P}}  B  )_{it}.
\end{align*}
\end{proof}

\begin{table}[ht]
\begin{center}
\scalebox{0.8}{\includegraphics[angle=-90]{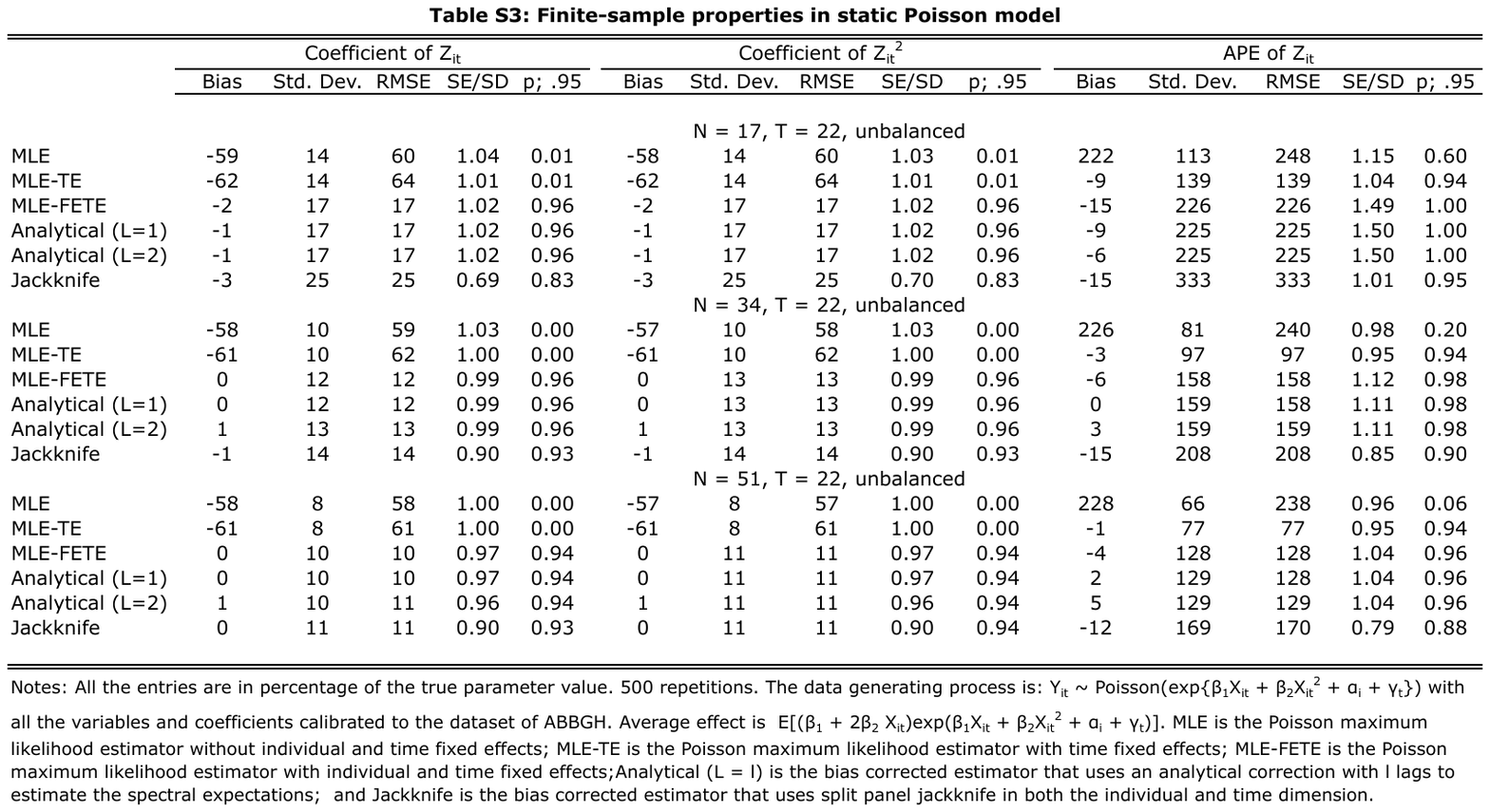}}
\end{center}
\end{table}

\begin{table}[ht]
\begin{center}
\scalebox{0.9}{\includegraphics{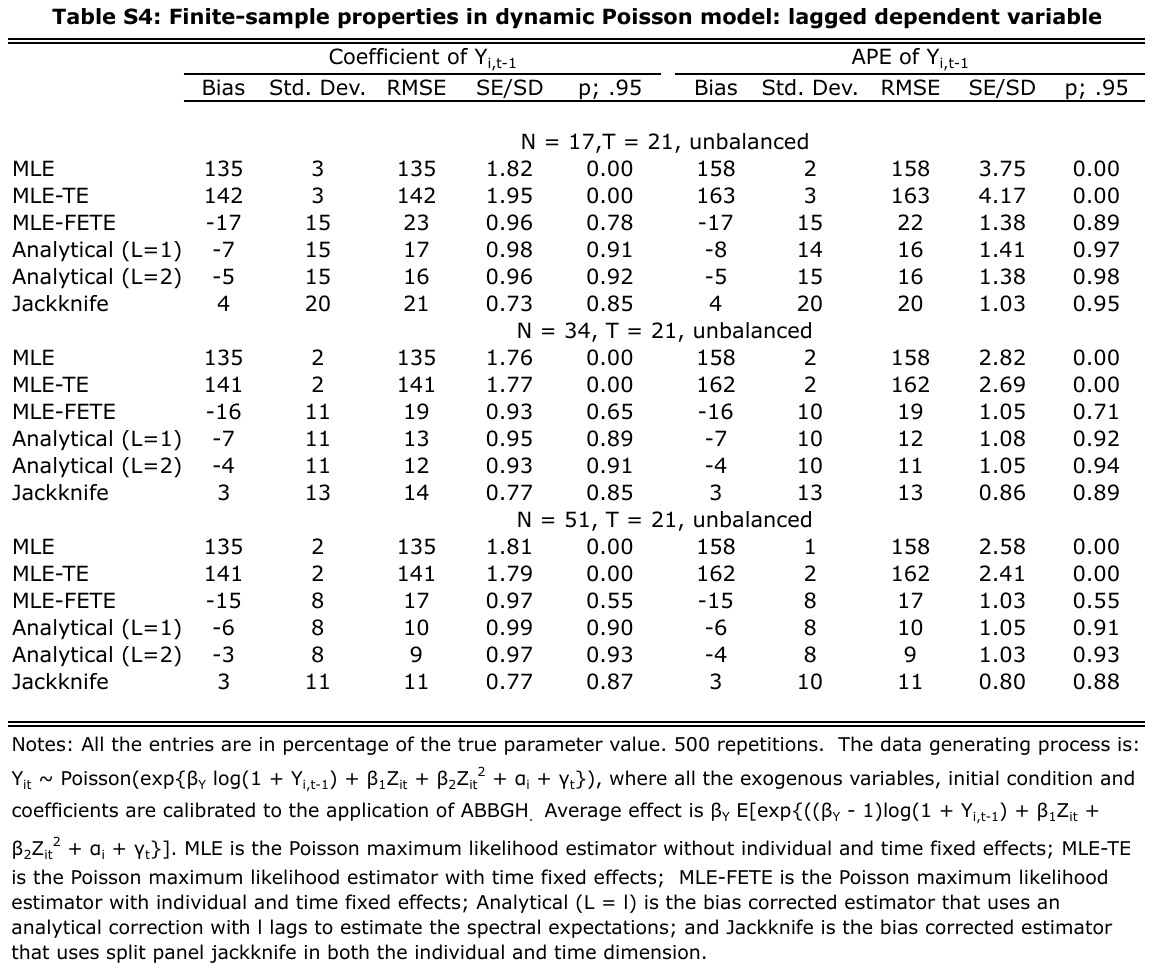}}
\end{center}
\end{table}

\begin{table}[ht]
\begin{center}
\scalebox{0.8}{\includegraphics[angle=-90]{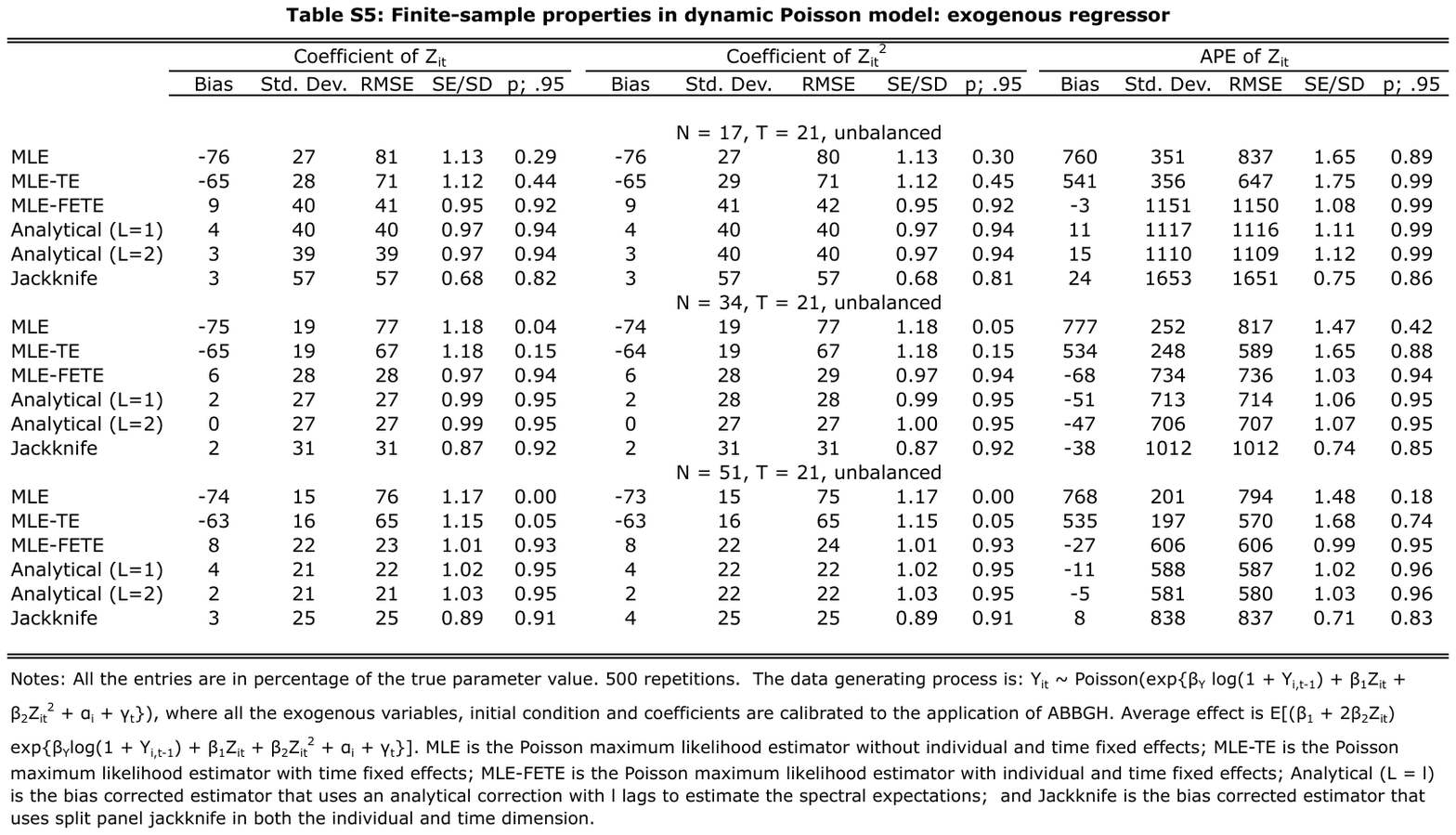}}
\end{center}
\end{table}


\end{document}